\documentclass[10pt, aps, prx, superscriptaddress, showpacs, nofootinbib, longbibliography, floatfix, reprint]{revtex4-2}

\usepackage[utf8]{inputenc}     
\usepackage[english]{babel}     

\usepackage{amsmath, amssymb, amsthm, mathtools}    
\usepackage{mathrsfs}                               
\usepackage{bbm}                                    
\usepackage{bm}                                     
\usepackage{latexsym}                               
\usepackage{xfrac}                                  
\usepackage{stmaryrd}                               
\usepackage[normalem]{ulem}
\usepackage{thmtools}
\usepackage{thm-restate}

\usepackage{physics}                                
\usepackage{siunitx}             

\usepackage{booktabs}              
\usepackage{longtable}             
\usepackage{multirow}              
\usepackage[table,dvipsnames]{xcolor}         
\usepackage{diagbox}

\newcolumntype{R}[1]{>{\raggedleft\arraybackslash}p{#1}}
\newcolumntype{C}[1]{>{\centering\arraybackslash}p{#1}}
\newcolumntype{L}[1]{>{\raggedleft\arraybackslash}p{#1}}

\usepackage{graphicx}              
\usepackage{epstopdf}              
\usepackage{wrapfig}
\usepackage{url}                   
\usepackage{hyperref}              
\hypersetup{
    colorlinks=true,               
    linktocpage=true,
    linkcolor=magenta!90!black,
    citecolor=black!60     
}

\usepackage{enumitem}              

\usepackage[caption=false]{subfig}
\newcommand{\phantomsubfloat}[1]{
    {%
        \captionsetup[subfigure]{labelformat=empty}
        \subfloat[][]{#1}
    }%
}

\usepackage{cleveref}
\crefname{equation}{Eq.}{Eqs.}
\crefname{section}{Sec.}{Secs.}
\crefname{subsection}{Sec.}{Secs.}
\crefname{appendix}{App.}{Apps.}
\crefname{figure}{Fig.}{Figs.}
\Crefname{figure}{Figure}{Figures}
\crefname{table}{Table}{Tables}
\crefname{result}{Result}{Results}
\crefname{algorithm}{Algorithm}{Algorithms}
\crefname{proposition}{Proposition}{Propositions}
\crefname{definition}{Definition}{Definitions}
\crefname{remark}{Remark}{Remarks}
\crefname{theorem}{Theorem}{Theorems}
\crefname{lemma}{Lemma}{Lemmas}
\crefname{corollary}{Corollary}{Corollaries}

\usepackage{comment}               
\usepackage{verbatim}              
\usepackage{cancel}                
\usepackage{soul}                  
\usepackage{indentfirst}           
\allowdisplaybreaks                

\usepackage{tikz}
\usetikzlibrary{quantikz2}
\usetikzlibrary{tikzmark}

\newtheorem{definition}{Definition}

\declaretheorem[name=Theorem]{theorem}
\declaretheorem[name=Lemma]{lemma}
\declaretheorem[name=Proposition]{proposition}
\declaretheorem[name=Remark]{remark}

\DeclareMathOperator*{\diag}{diag}
\DeclareMathOperator*{\supp}{supp}

\DeclareMathOperator*{\spn}{span}
\DeclareMathOperator*{\wt}{wt}

\DeclareMathOperator*{\colspan}{colspan}
\DeclareMathOperator*{\rs}{rs}
\newcommand{\trans}[1]{#1^{\mathsf{T}}}
\newcommand{\invtrans}[1]{#1^{-\mathsf{T}}}
\DeclarePairedDelimiter\db{\llbracket}{\rrbracket}

\newcommand{\lxor}{\mathbin{\vcenter{\hbox{$\veebar$}}}}    

\newcommand{\biglor}{\bigvee}

\definecolor{coral}{RGB}{254,125,106}

\usepackage{xurl}

\newcommand{\RM}{\mathrm{RM}}
\newcommand{\sddots}{\raisebox{0pt}{$\scalebox{.5}{$\ddots$}$}}

\newcommand{\bl}{\mathbf{l}}

\newcommand{\bz}{\mathbf{z}}

\newcommand{\ra}{\rangle}
\newcommand{\la}{\langle}
\newcommand{\rmx}{\mathrm{x}}
\newcommand{\rmz}{\mathrm{z}}
\newcommand{\enc}{\text{enc}}

\newcommand{\F}{\mathbb{F}}
\newcommand{\GF}{\mathrm{GF}}

\renewcommand{\S}{\mathrm{S}}
\newcommand{\GL}{\mathrm{GL}}
\newcommand{\GA}{\mathrm{GA}}
\newcommand{\Sp}{\mathrm{Sp}}

\newcommand{\nocontentsline}[3]{}
\newcommand{\tocless}[2]{\vspace{4ex}\bgroup\let\addcontentsline=\nocontentsline#1{#2}\egroup}

\makeatletter
\renewenvironment{acknowledgments}{%
  \begingroup
  \let\orig@addcontentsline\addcontentsline
  \renewcommand{\addcontentsline}[3]{}%
  \acknowledgments@sw{%
    \expandafter\section\expandafter*\expandafter{\acknowledgmentsname}%
  }{%
    \par
  }%
}{%
  \par
  \endgroup
}
\@booleantrue\acknowledgments@sw
\makeatother

\usepackage{makecell}
\usepackage{quantikz}
\usepackage{adjustbox}
\usepackage{tikz}
\usetikzlibrary{tikzmark}
\usepackage{enumitem}

\begin{document}
\title{Entangling logical qubits without physical operations}
\author{Jin Ming Koh}
\affiliation{Department of Physics, Harvard University, Cambridge, Massachusetts 02138, USA}
\author{Anqi Gong}
\affiliation{Institute for Theoretical Physics, ETH Z\"urich, 8093 Z\"urich, Switzerland}
\author{Andrei~C.~Diaconu}
\affiliation{Department of Physics, Harvard University, Cambridge, Massachusetts 02138, USA}
\author{Daniel~Bochen~Tan}
\affiliation{Department of Physics, Harvard University, Cambridge, Massachusetts 02138, USA}
\author{Alexandra~A.~Geim}
\affiliation{Department of Physics, Harvard University, Cambridge, Massachusetts 02138, USA}
\author{Michael~J.~Gullans}
\affiliation{Joint Center for Quantum Information and Computer Science, University of Maryland, College Park, Maryland 20742, USA}
\affiliation{National Institute of Standards and Technology, Gaithersburg, MD 20899, USA}
\author{Norman~Y.~Yao}
\affiliation{Department of Physics, Harvard University, Cambridge, Massachusetts 02138, USA}
\author{Mikhail~D.~Lukin}
\affiliation{Department of Physics, Harvard University, Cambridge, Massachusetts 02138, USA}
\author{Shayan Majidy}
\email{smajidy@fas.harvard.edu}
\affiliation{Department of Physics, Harvard University, Cambridge, Massachusetts 02138, USA}

\begin{abstract}
Fault-tolerant logical entangling gates are essential for scalable quantum computing, but are limited by the error rates and overheads of physical two-qubit gates and measurements. To address this limitation we introduce \emph{phantom codes}---quantum error-correcting codes that realize entangling gates between all logical qubits in a codeblock purely through relabelling of physical qubits during compilation, yielding perfect fidelity with no spatial or temporal overhead. We present a systematic study of such codes. First, we identify phantom codes using complementary numerical and analytical approaches. We exhaustively enumerate all $2.71\times10^{10}$ inequivalent CSS codes up to $n=14$ and identify additional instances up to $n=21$ via SAT-based methods. We then construct higher-distance phantom-code families using quantum Reed--Muller codes and the binarization of qudit codes. Across all identified codes, we characterize other supported fault-tolerant logical Clifford and non-Clifford operations. Second, through end-to-end noisy simulations with state preparation, full QEC cycles, and realistic physical error rates, we demonstrate scalable advantages of phantom codes over the surface code across multiple tasks. We observe one–to–two–order-of-magnitude reduction in logical infidelity at comparable qubit overhead for GHZ-state preparation and Trotterized many-body simulation tasks, given a modest preselection acceptance rate. Our work establishes phantom codes as a viable architectural route to fault-tolerant quantum computation with scalable benefits for workloads with dense local entangling structure, and introduces general tools for systematically exploring the broader landscape of quantum error-correcting codes.
\end{abstract}

\maketitle

\tocless\section{Introduction
\label{sec:introduction}}

\begin{figure*}
    \centering
    \includegraphics[width=\linewidth]{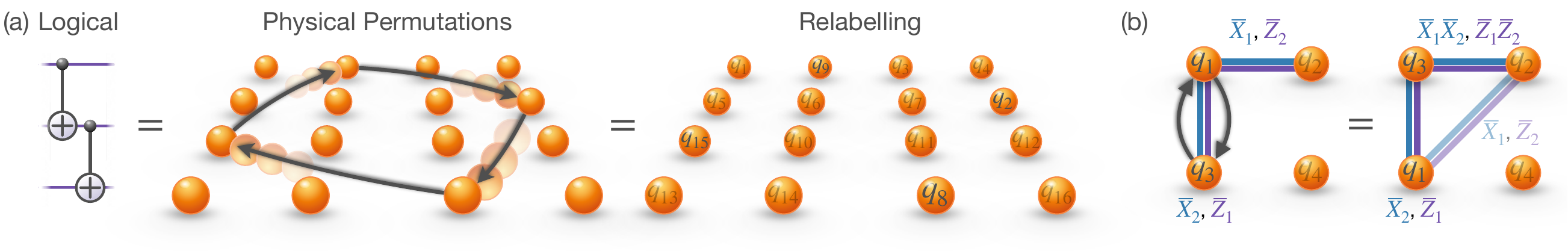}
    \phantomsubfloat{\label{fig:phantomgates/simple}}
    \phantomsubfloat{\label{fig:phantomgates/422}}
    \vspace{-1cm}
    \caption{
    \textbf{Illustration of logical entanglement via relabelling.} 
    \textbf{(a)} In phantom codes, logical entangling gates are realized by physical qubit permutations that are absorbed at the compilation stage as a relabelling, without performing any physical qubit permutation.
    \textbf{(b)} The $\db{4,2,2}$ code is the smallest phantom code. It has logical operators $\overline{X}_1 = XXII$, $\overline{X}_2 = XIXI$, $\overline{Z}_1 = ZIZI$, $\overline{Z}_2 = ZZII$. Permuting qubits 1 and 3 (1 and 2) implements $\overline{\mathrm{CNOT}}_{12}$ ($\overline{\mathrm{CNOT}}_{21}$, not shown).
    }
    \label{fig:phantomgates}
\end{figure*}

To realize their potential, quantum computers must run deep circuits reliably in the presence of noise~\cite{majidy2024building}. Achieving this at large scales requires the use of quantum error correction (QEC)~\cite{gottesman2024surviving}. The last several years have seen remarkable progress in QEC across a broad landscape of platforms ranging from neutral-atom and superconducting qubits to trapped-ions and bosonic systems~\cite{bluvstein2025fault, lacroix2025scaling, bluvstein2024logical, reichardt2024demonstration, paetznick2024demonstration, google2025quantum, brock2025quantum}. Further progress hinges on reducing QEC overheads and enabling efficient computation with suppressed logical error rates.

Recent efforts to reduce QEC overhead focus on encoding as many logical qubits as possible within a codeblock while enforcing sparsity structure, namely in high-rate quantum low-density parity-check (qLDPC) codes~\cite{breuckmann2021quantum}. While this design choice yields compact quantum memories, they do not directly constrain the cost of computation. In practice, addressable logical operations dominate space–time overheads and error budgets in such codes~\cite{cohen2022low, yoder2025tour, bravyi2024high, xu2025fast, xu2024constant}. These observations suggest that reducing overhead requires jointly optimizing storage and computation, rather than treating logical gates as a downstream constraint. In this work, we pursue a complementary approach that places logical operations at the center of the design. We begin with a set of logical operations—in this case, targeted logical entangling gates—and ask which codes support them with minimal overhead. In conventional codes, such gates typically require repeated measurement rounds, dense layers of two-qubit gates, or ancillary codeblocks~\cite{fowler2012surface, bravyi2024high, xu2024constant, xu2025fast}. As observed in recent state-of-the-art experiments, logical entangling gates account for a substantial fraction of the total error budget~\cite{bluvstein2025fault, lacroix2025scaling}. This raises the question of how fault-tolerant quantum computation can be optimized as a whole, encompassing both logical operations and QEC components.

As a step towards this optimization, we explore the fundamental lower bound on the cost of a logical entangling gate. Interestingly, for certain codes, logical entangling gates can be realized solely via the permutation of physical qubits. This permutation generates no new physical entanglement, but instead redistributes existing correlations among subsystems. Crucially, such an approach can incur zero overhead and exhibits perfect fidelity if the permutations are directly absorbed into circuit compilation rather than executed on hardware (\cref{fig:phantomgates}).

Motivated by this observation, we designate a new class of QEC codes termed \emph{phantom codes}: stabilizer codes in which logical entangling gates between every ordered pair of logical qubits can be implemented solely via physical qubit permutations. In such codes, all in-block logical entangling gates vanish from the physical circuit after classical circuit compilation (\cref{fig:phantomgates/detailed})---hence the name ``phantom''. Of the existing QEC codes, only two satisfy this definition: the $\db{12,2,4}$ Carbon code~\cite{paetznick2024demonstration} and $\db{2^D,D,2}$ hypercube codes~\cite{vasmer2022morphing, hangleiter2025fault}. The former enabled repeated rounds of QEC in trapped-ion experiments~\cite{paetznick2024demonstration}, while the latter underpinned recent neutral-atom experiments achieving logical-over-physical performance gains in quantum circuit sampling~\cite{bluvstein2024logical}. Despite intriguing properties and recent experimental success, phantom codes remain largely unexplored. For example, it is unknown whether additional classes of phantom codes exist, what structural constraints they obey, and whether their architectural advantages persist for scalable applications under realistic noise conditions.

We present a systematic and extensive investigation of phantom codes, exploring both their key features and limitations. Our main results are two fold. First, via a combination of numerics and analytic construction, we substantially expand the phantom-code landscape, from a single known error-correcting code and a single error-detecting code family, to over a hundred thousand new phantom code instances and multiple error-correcting families. Second, we demonstrate the practical advantages of phantom codes for  applications in the presence of realistic noise. In particular, we benchmark a representative $\db{64,4,8}$ phantom code against surface-code baselines through end-to-end noisy simulations including both state preparation and full QEC cycles. Using realistic physical error rates, we investigate the logical fidelity of entangled state preparation (i.e.~of a GHZ state) and Trotterized many-body simulation; in both scenarios, we demonstrate one–to–two–order-of-magnitude improvements in the logical fidelity using a nearly identical number of physical qubits (and a preselection acceptance rate ${\sim} 24\%$), over a range of $8$ to $64$ logical qubits. Moreover, in addition to the logical entangling gates generated via recompilation, for all of the phantom codes discovered, we also identify fault-tolerant logical Clifford and non-Clifford gates. Finally, we develop practical tools for implementing non-LDPC phantom codes that also apply to other non-LDPC codes, including improved decoding and fault-tolerant state preparation.

\begin{figure*}
    \includegraphics[width=\linewidth]{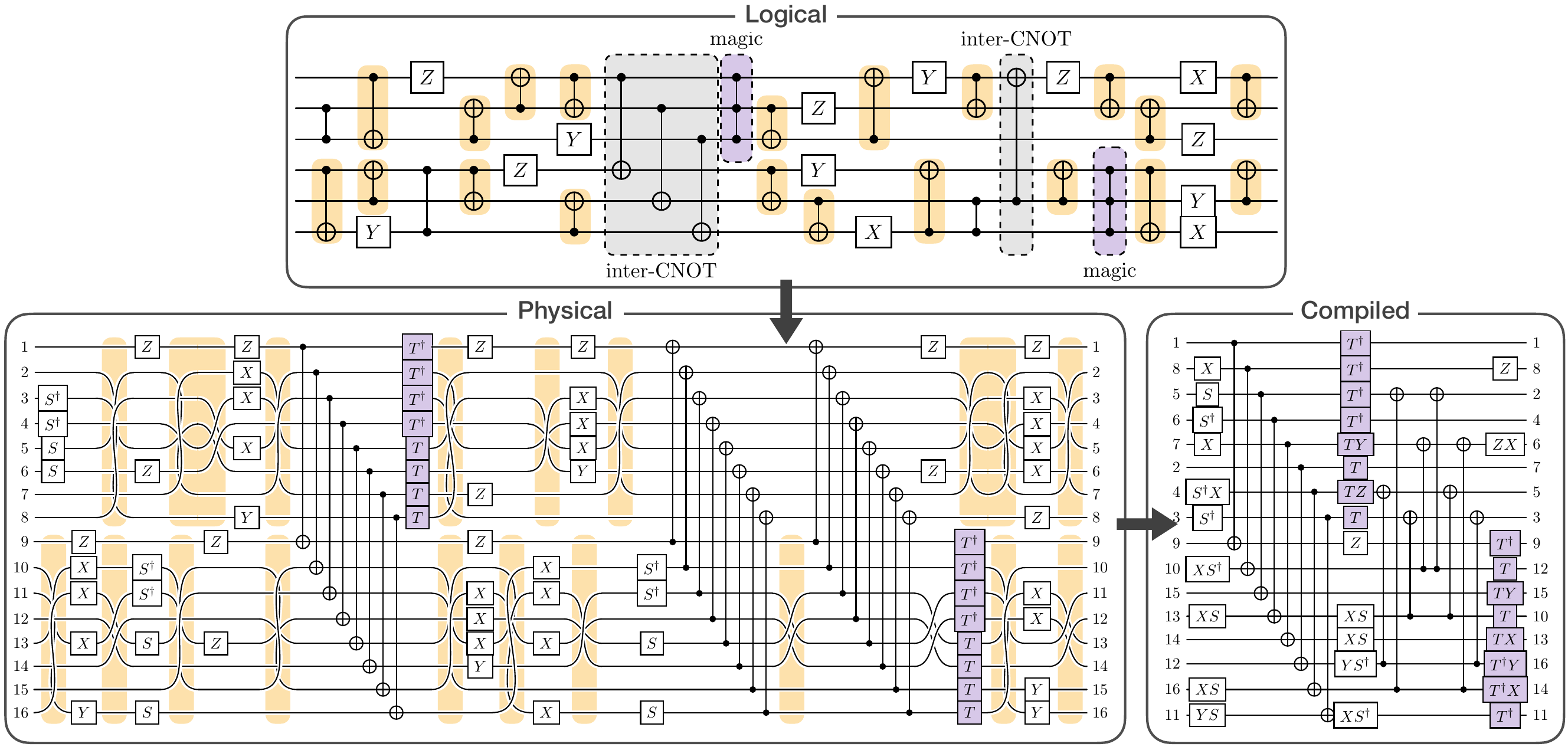}
    \vspace{-0.5cm}
    \caption{
    \textbf{Phantom codes in circuits with interblock and magic gates.} 
    The benefit of phantom codes is most evident when $\overline{\mathrm{CNOT}}$s appear alongside interblock and magic gates. Qubit permutations can be pulled through these operation without introducing operator spread, yielding substantial reductions in physical circuit size. In this example, the physical permutation required to implement the 20 in-block $\overline{\mathrm{CNOT}}$s (highlighted in orange) vanish from the final compiled circuit.
    }
    \label{fig:phantomgates/detailed}
\end{figure*}

\tocless\section{Phantom Codes\label{sec:phantom_codes}}

We begin by formally defining phantom codes and establishing their key structural properties. We then summarize our main results, explaining the implications of phantom codes for fault-tolerant design and of our methods for systematic QEC code discovery.
\medskip

\tocless\subsection{Definitions and key properties}

We focus on CSS stabilizer codes because scalable fault-tolerant architectures require a reliable interblock entangling mechanism, and the CSS structure guarantees access to a transversal interblock $\overline{\mathrm{CNOT}}$ gate. On CSS codes, we assume $X$- ($Z$-) logical operators consist only of $X$- ($Z$-) physical operators. In this setting, the only nontrivial in-block logical entangling gate realizable via qubit permutation alone is a $\overline{\mathrm{CNOT}}$.

\medskip
\begin{definition}[CSS phantom codes]
\label{def:phantom_code}
    An $\db{n,k,d}$ CSS code is phantom if the $\overline{\mathrm{CNOT}}_{ab}$ gate for every ordered pair of logical qubits $(a, b) \in [k]^2$ can be implemented via qubit permutations, for some choice of logical basis.
\end{definition}

\Cref{def:phantom_code} introduces a new class of codes. To help contextualize this notion, we clarify how phantom codes differ from three related ideas. First, although many stabilizer codes admit permutation symmetries, phantom codes satisfy a substantially stronger requirement. For example, the class excludes codes such as the $\db{2^{2r},\binom{2r}{r},2^{r}}$ qRM codes for $r > 1$ and bivariate bicycle codes~\cite{gong2024computation, bravyi2024high}, which support only \emph{products} of logical $\overline{\mathrm{CNOT}}$ gates via permutations. Because these products cannot be composed to produce individually addressable $\overline{\mathrm{CNOT}}$s, such codes do not qualify as phantom codes. Phantom codes are also distinct from permutation-invariant (PI) codes, whose codewords are invariant under all qubit permutations~\cite{pollatsek2004permutationally, ouyang2014permutation}.\footnote{Most PI codes encode a single logical qubit ($k = 1$) and are non-additive; no known PI code satisfies the phantom criterion, and none of the phantom codes identified here are PI.}

Second, the operation of phantom codes cannot  simply be understood as a form of logical Pauli frame tracking. While logical Clifford circuits permit all $\overline{\mathrm{CNOT}}$s to be absorbed into classical tracking of the Pauli frame, phantom codes allow entangling operations to be interleaved with non-Clifford (magic) gates (\cref{fig:phantomgates/detailed}) while still eliminating all in-block $\overline{\mathrm{CNOT}}$s from the compiled circuit\footnote{
Even for Clifford circuits, absorbing $\overline{\mathrm{CNOT}}$s by logical Pauli frame tracking induces logical operator spread, turning single-block measurements into joint measurements that require higher-overhead techniques (e.g.~lattice surgery).}, thereby yielding zero-cost, all-to-all entangling connectivity within each codeblock.

Third, phantom codes exclude constructions that rely on permutations in combination with local physical Clifford gates~\cite{sayginel2025fault, malcolm2025computing}. The reason is that zero-overhead, perfect-fidelity entangling gates are unique to \emph{pure} permutations, which are the only operations that commute through arbitrary circuits without operator spread. Single-qubit Cliffords break this property.

Phantom codes exhibit several structural properties; we highlight three here and discuss others in \cref{app:basics}. First, phantom codes enable arbitrary $\overline{\mathrm{CNOT}}$ circuits across \emph{multiple codeblocks} to be efficiently executed. This is achieved by combining the complete, basis-resolved set of zero-depth in-block logical $\overline{\mathrm{CNOT}}_{ab}$ gates with transversal interblock $\overline{\mathrm{CNOT}}$s. This capability is formalized in the following theorem (see \cref{app:basics/addressable_cnots} for the proof and \cref{fig:phantomgates/detailed} for an illustration).

\begin{restatable}[Efficient arbitrary $\overline{\mathrm{CNOT}}$ circuits across CSS phantom codeblocks]{theorem}{MainPropositionInterblockCNOT}
\label{thm:phantom_css_interblock_cnot_circuits}
    Any logical circuit of $\overline{\mathrm{CNOT}}$ gates acting on $2^a$ codeblocks, where $a \in \mathbb{N}$, of a CSS phantom code can be implemented in physical depth at most $4 (2^a - 1)$, up to a residual permutation of logical qubits. This depth reduces to at most $2 (2^a - 1)$ while maintaining the ordering of logical qubits when the $\overline{\mathrm{CNOT}}$ gates are unidirectional.
\end{restatable}

\begin{restatable}{remark}{MainRemarkInterblockCNOT}
\label{remark:phantom_css_interblock_cnot_circuits}
    The residual permutations of logical qubits in \cref{thm:phantom_css_interblock_cnot_circuits} incur zero cost when compiled away. In cases where they are necessary to be performed physically, they require depth at most $8k+8$ for any number of codeblocks of an $\db{n, k, d}$ CSS phantom code.
\end{restatable}

Second, the phantom property of CSS phantom codes is independent of the logical basis. This simplifies analysis and accelerates numerical searches.

Third, the weight distribution of physical Pauli operators within a logical Pauli equivalence class (e.g.~$\overline{X}_1$, $\overline{Z}_2$, etc.), generated by stabilizer multiplication, are identical across all $\overline{X}$ equivalence classes and likewise for all $\overline{Z}$ classes. This observation, for example, leads to the following parameter bound, which is proved in \cref{app:basics/weight}.

\begin{restatable}[Hamming bound for CSS phantom codes]{theorem}{PhantomCSSHammingBound}
\label{thm:phantom_css_hamming_bound}
    An $\db{n,k,d}$ CSS phantom code with $d = d_\mu$ for the sector $\mu \in \{X, Z\}$ satisfies
    \begin{equation}
        \eta (2^k - 1) \le B(n, d),
    \end{equation}
    where $\eta \geq 1$ is the number of weight-$d$ $\mu$-type logical operators of the same logical equivalence class of the code, and $B(n, d) \leq \binom{n}{d}$ is the maximum number of weight-$d$ binary strings of length $n$ such that the addition of any two strings is of weight at least $d$.
\end{restatable}

\Cref{thm:phantom_css_hamming_bound} constrains the possible $(n, k, d)$ parameters of CSS phantom codes. Certain phantom code families saturate this bound (see \cref{sec:code_discovery_construction/other}). We are not, however, aware of polynomial-time methods to compute $\eta$ and $B(n, d)$ in general, and using the $B(n, d) \leq \binom{n}{d}$ upper bound results in a loose Hamming bound for $d > 2$.

\tocless\subsection{Summary of results
\label{sec:phantom_codes/summary}}

To orient the reader, we briefly summarize the results of each section. We identify and construct phantom codes via four complementary approaches:
\medskip
\begin{itemize}[leftmargin=*]
    \item \textit{Exhaustive enumeration} (\cref{sec:code_discovery_construction/enumeration}). We enumerate all $2.71\times10^{10}$ CSS codes with block length $n\le14$ and identify every phantom code within this range.
    \item \textit{SAT-based search} (\cref{sec:code_discovery_construction/sat}). To reach $n>14$, we recast the search as a Boolean satisfiability problem, identify phantom codes up to $n=21$, and establish minimal-$n$ realizations across a range of $k$ and $d$.
    \item \textit{Quantum Reed--Muller (qRM) codes} (\cref{sec:code_discovery_construction/qrm}). To access higher $k$, we construct infinite families of phantom qRM codes, which generalize hypercube codes to higher distances, up to $d\le\sqrt{n}$.
    \item \textit{Binarization and concatenation scheme} (\cref{sec:code_discovery_construction/gf4}). To reach $d > \sqrt{n}$, we introduce a qudit-to-qubit construction 
    which yields higher-distance codes.
\end{itemize}

For each of the identified phantom codes, we identify additional Clifford and non-Clifford logical operations in \cref{sec:gates}. These operations include gates implemented via local Cliffords and qubit permutations, fold-type gates involving in-block two-qubit interactions, and magic gates arising from diagonal single-qubit rotations.

Finally, in \cref{sec:benchmarking}, we present end-to-end benchmarks comparing a representative phantom code with surface-code baselines under realistic noise ($p = 10^{-3}$). We focus on two representative primitives:
\begin{itemize}[leftmargin=*]
    \item \textit{GHZ state preparation} (\cref{fig:benchmark/ghz}). For $4$–$64$ logical qubits, the phantom code yields an approximately $56\times$ reduction in logical infidelity at a $24\%$ preselection acceptance rate, relative to surface codes with comparable physical-qubit counts.
    \item \textit{Trotterized many-body simulation} (\cref{fig:benchmarking/trotter}). For Hamiltonians with $8$-body terms (string operators that naturally arise in several simulation contexts~\cite{maskara2025fast,wietek2021stripes,jafari2019geometrically}) and $8$–$64$ logical qubits, we observe an approximately $94\times$ reduction in logical infidelity using nearly identical physical resources and the same $24\%$ acceptance rate.
\end{itemize}
\cref{sec:benchmarking} also introduces practical tools for implementing non-LDPC phantom codes, including fault-tolerant state-preparation strategies and decoders that track spatiotemporal error correlations; these tools are directly applicable to other non-LDPC codes.

\tocless\subsection{Implications}

The codes, benchmarks, methods and resources developed in this work carry several implications for fault-tolerant quantum computing.

First, phantom codes reduce logical error rates and overhead for circuits with dense local entanglement, bringing applications like fermionic simulation and correlated-phase preparation closer to practical realization. Together with recent works that lower the overhead of digital fermionic simulation~\cite{maskara2025fast, constantinides2025low}, these code- and algorithm-level advances markedly improve near-term experimental feasibility. This potential is accessible on hardware with long-range connectivity, including neutral-atom and trapped-ion platforms. Neutral-atom systems enable efficient parallel transversal logical entangling gates, whereas trapped-ion platforms lack a comparably parallel mechanism. Consequently, while phantom codes benefit both architectures, their impact is expected to be especially pronounced for trapped-ion processors.

Second, practical fault tolerance demands QEC codes that are co-designed for efficient storage and efficient logical computation. High rate and LDPC structure define one axis of optimization, supporting efficient state preparation, decoding, and low qubit overhead; an equally important axis is the availability of efficient logical gates, which distinguishes quantum memories from quantum computers. In this work, we target the zero-overhead extreme: although it appears to exclude high-rate qLDPC codes, the resulting codes are competitive in practical performance. This motivates developing QEC codes that optimize along both axes of this code design space.

The resources developed and compiled here enable systematic exploration of this design space. These include a complete database of $n \le 14$ CSS codes, SAT-based searches for specified logical gate structure, and an automated pipeline for extracting logical operations. Within this space, phantom codes occupy an extreme point corresponding to vanishing logical entangling cost. By relaxing the zero-cost constraint, the same framework probes nearby trade-offs among logical overhead, code distance, stabilizer weight, and logical gate sets, enabling the identification of codes tailored to other hardware connectivities and objectives, such as reducing the cost of logical magic.

Finally, our benchmarking results motivate revisiting which QEC codes are viable for specific applications. With appropriate decoding and state-preparation protocols, our work suggests that non-LDPC codes can outperform LDPC codes on tailored tasks despite the latter’s generic structural advantages.

\tocless\section{Code discovery \& construction
\label{sec:code_discovery_construction}}

We identify phantom codes via a combination of numerical search and analytic construction. The main ideas are discussed here and full technical details are given in the Appendices. \Cref{tab:gate_for_phantom_codes} summarizes representative phantom codes that outperform all previously known and newly generated codes in parameters, gate sets, or both.

Two major components of our numerical workflow employ SAT solvers \cite{biere2009handbook}. Accordingly, \cref{app:sat_preliminaries} provides a brief primer. In short, SAT solving formulates a problem as Boolean constraints and searches for a satisfying assignment or certifies that none exists.

\begin{table}
\centering
\begin{tabular}{l@{\hskip 6pt}c@{\hskip 6pt}c@{\hskip 6pt}c@{\hskip 6pt}c}
\toprule
Code        & Clifford & Fold & $d\!=\!2$ Magic & Found \\
\midrule
$\db{4,2,2}$   & $\mathrm{CZ}, H^{\otimes 2}$  & $SS$     & -- & \cite{vasmer2022morphing} \\
$\db{8,3,2}$   & $\mathrm{CZ}_{ij}$ & $S_i S_j$  & $\mathrm{C^2Z}$  & \cite{vasmer2022morphing} \\
$\db{16,4,2}$  & $\mathrm{CZ}_{ij}$  & $S_i S_j$ & $\mathrm{C^{3}Z}$ & \cite{vasmer2022morphing} \\
$\db{12,2,4}$  & $H^{\otimes 2}$   & $SS$, $\mathrm{CZ}$                 & -- & \cite{reichardt2024demonstration} \\
\midrule
$\db{7,3,2}$    & $\mathrm{CZ}_{ij}$    & $S_i S_j$   & --                     & {Enumeration} \\
$\db{12,2,2}$   & $\mathrm{CZ}, S_i$ & --  & $\mathrm{CS}$             & {Enumeration} \\
$\db{14,3,3}$    & --  & $S_i S_j$, $\mathrm{CZ}_{ij}$ & --                        & {Enumeration} \\
$\db{18,2,5}$     & $H^{\otimes 2}$  & $SS$, $\mathrm{CZ}$  & --               & {SAT search}\\
$\db{20,2,6}$   & $\mathrm{CZ}, H^{\otimes 2}$ & $SS$ & --      & {SAT search}\\
$\db{21,2,3}$   & {$\mathrm{CZ}, S_i$} & --   & --            & {SAT search}\\
$\db{16,3,4}$   & -- & $S_i S_j$, $\mathrm{CZ}_{ij}$  &  $\mathrm{C^2Z}$                       & {qRM} \\
$\db{32,4,4}$   & -- & $S_i S_j$, $\mathrm{CZ}_{ij}$ & $\mathrm{C^{3}Z}$                            & {qRM} \\
$\db{64,5,4}$   & --  & $S_i S_j$, $\mathrm{CZ}_{ij}$   & $\mathrm{C^{4}Z}$                     & {qRM} \\
$\db{64,4,8}$    & -- & $S_i S_j$, $\mathrm{CZ}_{ij}$ & $\mathrm{C^{3}Z}$                          & {qRM} \\
$\db{44,2,10}$   & $H^{\otimes 2}$ & $SS$, $\mathrm{CZ}$  &  --                       & {B\&C} \\
$\db{52,2,10}$   & $\mathrm{CZ},H^{\otimes 2}$ & ?  &  --                       & {B\&C} \\
$\db{68,2,14}$   & $\mathrm{CZ},H^{\otimes 2}$ & ?  &  --                       & {B\&C} \\
$\db{76,2,14}$   & $H^{\otimes 2}$ & $SS$, $\mathrm{CZ}$  &  --                       & {B\&C} \\
$\db{116,2,22}$   & $\mathrm{CZ},H^{\otimes 2}$ & ?  &  --                       & {B\&C} \\
\bottomrule
\end{tabular}
\caption{\textbf{Parameters and logical gate sets in representative phantom codes.}
Listed are phantom codes that outperform prior examples and all phantom codes found in this work in parameters, gate sets, or both. Columns list the code parameters, Clifford gates available via single-qubit Cliffords and qubit permutations (i.e.~code automorphisms), fold-diagonal Clifford gates, and magic gates. “--” denotes absence; “?” denotes unknown due to intractability. Fold-diagonal gates are non-exhaustive and omitted when already implementable by automorphisms. All magic gates are effectively distance-two. Any codes listed capable of implementing a $\mathrm{C}^{p}\mathrm{Z}$ gate can also implement $\mathrm{C}^{q}\mathrm{Z}$ gates for all $q<p$. Codes are obtained via exhaustive enumeration, SAT-based code discovery, quantum Reed--Muller (qRM) constructions, and binarized-qudit with concatenation (B\&C) constructions. The Hadamard-dual of these codes obtained via code deformation by $H^{\otimes n}$ are also phantom, supporting the corresponding logical gate sets conjugated by $H^{\otimes k}$.
}
\label{tab:gate_for_phantom_codes}
\end{table}

\tocless\subsection{Exhaustive code enumeration
\label{sec:code_discovery_construction/enumeration}}

Our first strategy is to construct a complete catalogue of inequivalent CSS codes to exhaustively identify all phantom codes up to some $n$. The dominant costs in this endeavor are (i) enumerating all admissible stabilizer groups defining candidate codes and (ii) identifying the equivalence class of each candidate code. The computational cost of both tasks scales superexponentially with $n$. We define equivalence up to qubit permutations and global Hadamards ($H^{\otimes n}$), as phantomness is preserved under these transformations. 

Our approach builds on the method of Ref.~\cite{cross2025small}, which enumerated stabilizer codes up to $n=9$. Starting from the trivial $\db{n,n,1}$ code, we iteratively reduce $k$ by appending linearly independent, commuting Pauli operators to the stabilizer group, with each admissible extension defining a new code. At each $k$, codes are grouped into equivalence classes based on the Tanner graphs that represent their stabilizer groups, and a single representative code is retained.

To reach larger $n$, we introduce CSS-specific simplifications that reduce both the branching factor of the search and the cost of equivalence testing. First, letting $r_{\mathrm{x}}$ and $r_{\mathrm{z}}$ respectively denote the ranks of the $X$- and $Z$-type stabilizer groups, code equivalence up to Hadamard duality allows us to restrict to $r_{\mathrm{x}} \le r_{\mathrm{z}}$, eliminating the remaining symmetric portion of the enumeration space. Second, CSS codes admit tripartite Tanner graphs~\cite{tillich2013quantum}, which simplifies canonical labelling and reduces memory requirement---an essential reduction given datasets spanning tens of~\si{\tera\byte}. Computing and then deduplicating canonical labels~\cite{mckay2014practical} of Tanner graphs reduces equivalence testing from $\mathcal{O}(I^2)$ pairwise isomorphism checks to $\mathcal{O}(I)$ canonical-labelling operations, where the number of candidate codes $I$ in an iteration exceeds $10^{10}$ at the scale we explored. 

Iterating this procedure yields all $2.71\times 10^{10}$ inequivalent CSS codes of block length $n \le 14$. By comparison, there are only $62\,156$ codes for $n \le 9$, roughly a millionth of this dataset. To our knowledge, this is the largest complete stabilizer code catalogue assembled to date.

We identify all phantom codes in this catalogue using Boolean constraint satisfaction. Given $X$- and $Z$-type stabilizer generator matrices $H_{\mathrm{x}} \in \mathbb{F}_2^{r_{\mathrm{x}} \times n}$ and $H_{\mathrm{z}} \in \mathbb{F}_2^{r_{\mathrm{z}} \times n}$ defining a CSS code, we compute $X$- and $Z$-type logical operators $L_{\mathrm{x}}, L_{\mathrm{z}} \in \mathbb{F}_2^{k \times n}$ by Gaussian elimination. We then introduce a permutation matrix $P^{ab}$ representing the qubit permutation implementing $\overline{\mathrm{CNOT}}_{ab}$ for each pair of distinct logical qubits $(a, b) \in [k]^2$ as a free variable. Constraints enforcing the preservation of the codespace and correctness of the induced logical actions are compactly included in the SAT problem instance, which we solve using the state-of-the-art solver \texttt{kissat}~\cite{biere2024cadical}. A satisfying assignment signifies that the code is phantom; unsatisfiability certifies that no implementation of the gates exist.

Applying this procedure, we identify $1.39\times 10^{5}$ CSS phantom codes up to $n = 14$---about one in every two hundred thousand CSS codes. Representative examples appear in \cref{tab:gate_for_phantom_codes}. Details of the enumeration algorithm, SAT formulation, and a breakdown of the results are given in \cref{app:enumeration}, along with an extension to non-CSS stabilizer codes.

\tocless\subsection{SAT-based code discovery
\label{sec:code_discovery_construction/sat}}

The exhaustive enumeration yielded a complete catalogue of CSS codes for $n \le 14$, including all phantom codes. For $n>14$, the number of CSS codes becomes prohibitive, so to identify phantom codes at larger $n$ and determine minimal block lengths at larger $k$ and $d$, we adopt a catalogue-free method.

Essentially, we convert the SAT check for phantomness into an automated code-discovery tool. Instead of fixing $H_{\mathrm{x}}$ and $H_{\mathrm{z}}$ as input, we treat them as free variables, allowing the SAT solver to construct codes that satisfy the required permutation gate-set constraints. We impose a standard form on $H_{\mathrm{x}}, H_{\mathrm{z}}$ so that a logical basis $L_{\mathrm{x}}, L_{\mathrm{z}}$ is explicit. Distance-$d$ constraints are enforced by demanding all Pauli errors of weight ${<} \, d$ (outside the stabilizers) to produce a nonzero syndrome. The full formulation, including non-CSS extensions, is provided in \cref{app:code_discovery}.

Our SAT search uncovers CSS phantom codes up to $n = 21$, including several with best-known parameters or gate sets (see  \cref{tab:gate_for_phantom_codes}). Unsatisfiability certificates allow us to determine minimal block lengths $n$ for a given $k$ and $d$ (see \cref{tab:min_n_phantom}). Compared with all CSS codes, phantom codes are near-identical in minimal block lengths for $k = 2$, exhibit small deviations for $k = 3$, and a larger deviation for the single $k = 4$ case, reflecting the increasing severity of the gate-set constraint on phantom codes.

\begin{table}
\centering
\vspace{4pt}
\setlength{\tabcolsep}{5.5pt}
\begin{tabular}{c|ccccc|ccc|c}
\toprule
 $k$  & \multicolumn{5}{c|}{$2$} & \multicolumn{3}{c|}{$3$} & $4$  \\
 $d$ & $2$ & $3$ & $4$ & $5$ & $6$ & $2$ & $3$ & $4$ & $2$\\
\midrule
Phantom   & 4 & 11 & 12 & 18 & 20 & 7 & 14 & 15 & 15\\
General   & 4 & 10 & 12 & 18 & 20 & 6 & 11 & 14 & 6 \\
\bottomrule
\end{tabular}
\caption{\textbf{Minimal lengths of CSS phantom and general CSS codes.}
Entries list the smallest block length $n$ for each logical-qubit count $k$ and distance $d$, for which a code exists.}
\label{tab:min_n_phantom}
\end{table}

\tocless\subsection{Phantom quantum Reed--Muller codes
\label{sec:code_discovery_construction/qrm}}

Our numerical methods are tractable only for $k \le 4$. To obtain phantom codes at larger $k$ and scalable code families beyond the reach of enumeration or SAT code discovery, we turn to an analytic construction based on quantum Reed–Muller (qRM) codes. qRM codes are CSS codes derived from pairs of classical Reed–Muller codes, with several families such as the $\db{2^D, D, 2}$ hypercube codes as special cases. Our full derivation is detailed in \cref{app:qrm} and uses the polynomial formalism \cite{macwilliams1977theory,gong2024computation}, which we review in \cref{app:qrm/polymonial}. Here we outline the construction and summarize the key properties.

Our construction begins with the qRM family $\db{2^m, {m\choose l}, \min(2^{m-l}, 2^l)}$, where $m$ sets the number of variables and $l$ the monomial degree defining the $X$-type logical generators~\cite{gong2024computation, reichardt2024demonstration}. These codes possess large automorphism groups and support products of $\overline{\mathrm{CNOT}}$ gates implemented via permutations, which cannot be composed to produce individually addressable $\overline{\mathrm{CNOT}}$s. Therefore these qRM codes are not generally phantom.

To obtain phantom codes, we fix selected logical qubits to $\ket*{\overline{0}}$ or $\ket*{\overline{+}}$, promoting the logical operators to $Z$- or $X$-type stabilizers. The resulting codes have parameters $\db{2^m, m-l+1, \min(2^{m-l}, 2^l)}$\footnote{Concatenating a phantom code with a $k=1$ inner quantum code preserves phantomness. While the code parameters can resemble concatenating a hypercube code with a repetition code~\cite{hangleiter2025fault}, our qRM construction yields more balanced codes (see~\cref{app:qrm/construction}).}. This choice is guided by the code's polynomial representation and admits a degree of flexibility to balance $X$- and $Z$-type stabilizer ranks, corresponding to different gauge-fixings of the parent qRM code. When $m = 2l$ and the $X$- and $Z$-distances are equal, such balancing can further improve logical error rates, a feature leveraged in our benchmarking (\cref{sec:benchmarking}).

A simple mnemonic for the phantom qRM code parameters is: starting from the hypercube parameters $\db{2^D, D, 2}$, each reduction of $k$ by one doubles the distance $d$ at fixed $n$, up to $d = \sqrt{n}$. Representative examples appear in \cref{tab:gate_for_phantom_codes}, and \cref{fig:phantomqRM} illustrates two permutation $\overline{\mathrm{CNOT}}$s for the smallest ($d = 4$) error-correcting qRM code.

\begin{figure}
    \centering
    \includegraphics[width=\columnwidth]{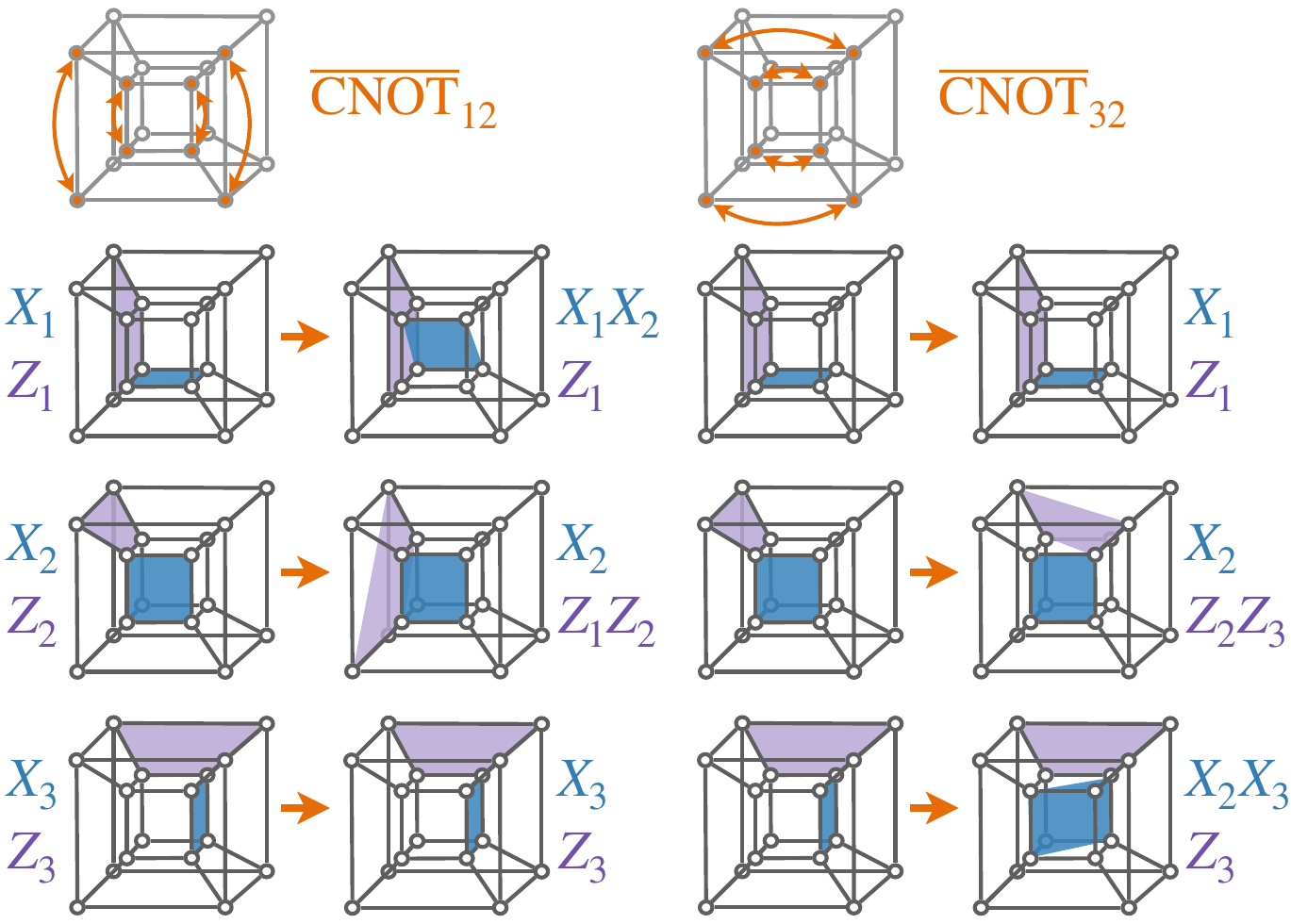}
    \caption{\textbf{The smallest error-correcting phantom qRM code.}
    Logical $\overline{\mathrm{CNOT}}$ gates for the $\db{16,3,4}$ phantom qRM code. The code is obtained by promoting three same-type logical operators ($X$ or $Z$) of the parent $\db{16,6,4}$ code to stabilizers. The resulting reduction in the number of logical qubits enables arbitrary $\overline{\mathrm{CNOT}}$ gates via qubit relabelling; examples of $\overline{\mathrm{CNOT}}_{12}$ and $\overline{\mathrm{CNOT}}_{32}$ are shown.}
    \label{fig:phantomqRM}
\end{figure}

Beyond their phantom property, these codes offer several additional architectural advantages. They are compatible with preselection-based fault-tolerant state preparation~\cite{gong2024computation}, enabling high-fidelity codeblock initialization despite their high-weight stabilizers. They support the full logical Clifford group via fold-$\overline{S}_i \overline{S}_j$ gates and teleported Hadamards from an all-$\ket{\overline{+}}$ state. Finally, they admit a $d = 2$ transversal magic-gate scheme: by temporarily projecting into the $\db{2^{m-l+1}, m-l+1, 2}$ hypercube codes, applying their transversal non-Clifford gate, and returning to the phantom qRM codespace (see \cref{app:qrm}).

\tocless\subsection{Binarization and concatenation scheme
\label{sec:code_discovery_construction/gf4}}

While the phantom qRM family supports arbitrary $k$, its distance is bounded by $d \le \sqrt{n}$. To surpass this bound at $k = 2$, we introduce an analytic construction based on high-distance self-dual $\GF(4)$ qudit codes encoding one logical qudit~\cite{macwilliams1978self}. The underlying $\GF(4)$ codes achieve extremal distances at small block lengths.

To convert such a qudit code into a phantom code, we first binarize it by representing each $\GF(4)$ symbol by a pair of qubits~\cite{gottesman2024surviving}. We then concatenate each qubit pair with the $\db{4,2,2}$ phantom code (\cref{fig:qudit}). Both steps of this binarize-and-concatenate (B\&C) scheme are essential: binarization alone does not yield a phantom code, as the $\db{4,2,2}$ layer supplies a subset of required substructure for phantomness.

These codes inherit key properties from the starting qudit codes. Self-duality of the qudit code guarantees a nontrivial logical Hadamard on the resulting qubit code; Hermiticity additionally yields a $\overline{\mathrm{CZ}}$ gate. As binarization does not decrease distance, their large qudit distances translate into high qubit distances. Additionally, concatenation with $\db{4,2,2}$ introduces many low-weight stabilizers, though completing the stabilizer group still requires some generators of weight at least the code distance. More generally, any $k = 2$ CSS phantom code can be used for concatenation, but the $\db{4,2,2}$ is favourable in encoding rate. Proofs and technical details of a phantom code family arising from quadratic-residue $\GF(4)$ codes are provided in \cref{app:binarize_concatenate}.

\begin{figure}
    \centering
    \includegraphics[width=0.9\columnwidth]{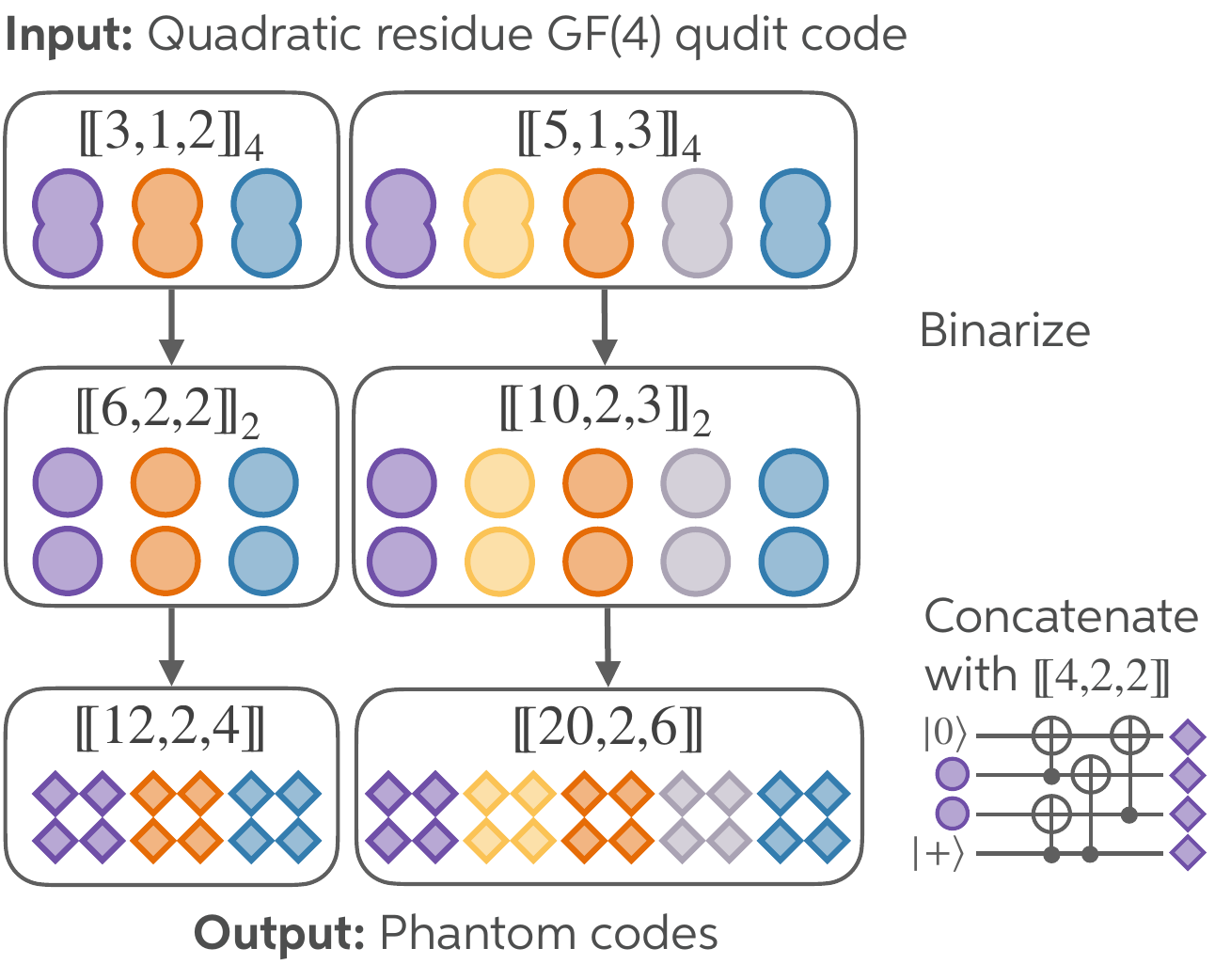}
    \caption{
    \textbf{Schematic construction of phantom codes from quadratic-residue qudit codes.}  
    Phantom codes can be constructed by starting from a quadratic-residue $\db{n,k,d}_4$ code over $\mathrm{GF}(4)$ and binarizing it to obtain a qubit code with parameters $\db{2n, 2k, d}_2$. This intermediate code is not itself phantom, but it serves as the outer code in a concatenation with the $\db{4,2,2}$ qubit code, and the resulting concatenated code is phantom. A detailed derivation is given in \cref{app:binarize_concatenate}.
    }
    \label{fig:qudit}
\end{figure}

\tocless\subsection{Other code constructions
\label{sec:code_discovery_construction/other}}

Finally, we discuss additional strategies for constructing phantom codes, highlighting a subset here and deferring full details to \cref{app:other_constructions}.

\medskip
\begin{itemize}[leftmargin=*]
    \item \textit{Simple concatenation}. Concatenating a phantom outer code $\db{n_{\mathrm{out}}, k, d_{\mathrm{out}}}$ with a high-distance inner code $\db{n_{\mathrm{in}}, 1, d_{\mathrm{in}}}$ (e.g., a surface code) yields a phantom code with parameters $\db{n_{\mathrm{in}} n_{\mathrm{out}}, k, d_{\mathrm{in}} d_{\mathrm{out}}}$. Logical $\overline{\mathrm{CNOT}}$s are performed by relabellings at the outer level.
    \item \textit{Hypergraph product (HGP) codes}. The HGP~\cite{tillich2013quantum} increases both $k$ and $d$ while preserving permutation symmetries of the input classical codes~\cite{berthusen2025automorphism}. Phantom codes can be constructed from, e.g., the HGP of a classical simplex and a repetition code (\cref{app:other_constructions/hypergraph}). The smallest members of this family, $\db{7,2,2}$ and $\db{49,3,4}$, have lower rates than the phantom qRM codes.
    \item \textit{Punctured hypercube codes.} Removing one qubit from the $\db{2^D, D, 2}$ hypercube codes ($D>2$) and truncating the affected stabilizers yields phantom codes with parameters $\db{2^D - 1, D, 2}$. Both punctured and unpunctured versions saturate the Hamming bound for CSS phantom codes (\cref{thm:phantom_css_hamming_bound}), though puncturing can reduce the gate set. For example, the $\db{15,4,2}$ code admits $\overline{\mathrm{CCZ}}$ but not $\overline{\mathrm{CCCZ}}$ (see \cref{app:other_constructions}).
\end{itemize}

\tocless\section{Logical gates beyond phantom
\label{sec:gates}}

Practical QEC codes must support versatile fault-tolerant logical gates. We therefore identify additional Clifford and non-Clifford logical operations supported by phantom codes beyond their defining permutation $\overline{\mathrm{CNOT}}$s. We begin with a fundamental constraint: 

\medskip
\begin{restatable}[No strictly transversal logical gates non-commuting with permutation logical gates]{theorem}{MainTheoremTransversalGates}
\label{thm:no_noncommuting_transversal_gates}
    A stabilizer code supporting a logical gate $\overline{U}$ via qubit permutations can admit no strictly transversal logical gate $\overline{V}$ acting on any number of codeblocks $b$ where $\comm{U^{\otimes b}}{V} \neq 0$. 
\end{restatable}

On $b$ codeblocks of an $\db{n, k, d}$ code, the strictly transversal gates of \cref{thm:no_noncommuting_transversal_gates} have the form $\overline{V} = \prod_{i = 1}^n W_{i i \cdots i}$, with $W_{i i \cdots i}$ acting on the $i^\text{th}$ qubit(s) of the codeblock(s). For phantom codes, this result precludes, for example, strictly transversal $\overline{H}$, $\overline{S}$, $\overline{\mathrm{CZ}}$, and magic gates. Additional logical gates beyond the permutation $\overline{\mathrm{CNOT}}$s can arise from non-uniform single-qubit gates, qubit permutations, and in-block multi-qubit interactions. The three strategies described below address these mechanisms.

We first identify the logical gates arising from physical qubit permutations and local Cliffords, i.e.~the \textit{automorphism Cliffords}. Following Refs.~\cite{breuckmann2024fold, sayginel2025fault}, we analyze the permutation automorphisms of the extended symplectic stabilizer generator matrix to determine the set of logical gates. Details are given in \cref{app:gates/auto}; here we illustrate the approach using the $\db{4,2,2}$ phantom code, which intersects several well-studied families~\cite{criger2016noise, majidy2023unification, kubica2015unfolding}. Its stabilizers are generated by $\smash{X^{\otimes 4}}$ and $\smash{Z^{\otimes 4}}$, which in extended symplectic form are
\begin{align}
\bigg[
\underbrace{
\begin{array}{cccc}
    1&1&1&1 \\
    0&0&0&0
\end{array}
}_{H^{(\mathrm{x})}}
\Bigg|
\underbrace{
\begin{array}{cccc}
    0&0&0&0 \\
    1&1&1&1
\end{array}
}_{H^{(\mathrm{z})}}
\Bigg|
\underbrace{
\begin{array}{cccc}
    1&1&1&1 \\
    1&1&1&1
\end{array}
}_{H^{(\mathrm{x})} \oplus H^{(\mathrm{z})}}
\bigg]. 
\end{align}

Column permutations that preserve the row span\footnote{Technically, they must satisfy an additional condition (\cref{app:gates/auto}).} represent physical $\mathrm{SWAP}$, $H$, and $S$ combinations that act logically. For example, swapping columns $1$--$4$ with $5$--$8$ leaves the matrix invariant and corresponds to an $H^{\otimes 4}$, yielding the logical action $\smash{\overline{H}}^{\otimes 2}\,\overline{\mathrm{SWAP}}$. Thus, automorphism Cliffords follow directly from permutation symmetries of the extended stabilizer matrix.

Second, we identify logical diagonal magic gates obtained via non-uniform single-qubit physical $Z$ rotations~\cite{webster2022xp, webster2023transversal}. Fixing a level of the diagonal Clifford hierarchy (e.g.~$T$ or $\smash{\sqrt{T}}$), we solve for integer-power assignments of this gate on the physical qubits whose action preserves every logical computational-basis state up to a state-dependent phase. Each qubit may receive a distinct rotation angle. These yield logical diagonal gates at the chosen hierarchy level (see \cref{app:gates/diagonal}).

Finally, we extend the diagonal-rotation method to identify fold-diagonal gates~\cite{webster2023transversal, sayginel2025fault}. The key idea is to embed a $\db{n,k,d}$ code into a larger $\db{n+\binom{n}{2},k,d'}$ code and search for (non-uniform) single-qubit physical $Z$-basis rotations that act as logical operators on the enlarged code. Using the construction of \cref{app:gates/diagonal/fold}, any such operator induces a logical operation on the original code that is implemented by patterned one- and two-qubit diagonal interactions. Restricting to depth-one circuits yields fold gates, which in favourable cases preserve the $X$-sector code distance~\cite{malcolm2025computing}.

Together, these mechanisms generate a wide class of logical operations on phantom codes. For example, fold-$\overline{S}_i \overline{S}_j$ gates combined with automorphism operations implementing $\smash{\overline{H}}^{\otimes k}$ generate the full logical Clifford group, as realized by the $\db{20,2,6}$ phantom code.

\tocless\section{Numerical benchmarking
\label{sec:benchmarking}}

Lastly, we assess whether phantom codes can be practically useful in realistic fault-tolerant settings. Two structural properties make their practical usefulness non-obvious. First, the phantom codes identified in this study possess high-weight stabilizers and are non-LDPC, which generally precludes high-fidelity codeblock initialization and syndrome extraction using standard bare-ancilla techniques. Second, in generic logical circuits not all $\overline{\mathrm{CNOT}}$s can be in-block.

To this end, we benchmark a concrete phantom code against a surface-code baseline using end-to-end noisy simulations. We focus on the $\db{64,4,8}$ phantom qRM code, whose parameters are suitable for near-term experiments, and compare its performance to rotated surface codes of varying $d$ at both near-term ($10^{-3}$) and projected  ($5\times10^{-4}$) physical error rates. [Results at a higher $3 \times 10^{-3}$ error rate are provided in \cref{app:benchmarking/current}.] Our study is organized around three benchmarks of increasing complexity: (i) repeated in-block logical $\overline{\mathrm{CNOT}}$ circuits on a single codeblock; (ii) logical GHZ state preparation across multiple codeblocks; and (iii) Trotterized many-body quantum simulation. These benchmarks probe regimes successively dominated by in-block entanglement and by interblock operations, with circuit widths far exceeding the per-codeblock logical dimension $k$.

\tocless\subsection{Operation of codes and error correction
\label{sec:benchmarking/operation}}

We begin by specifying how the phantom and surface codes are operated, including state-preparation protocols, syndrome-extraction strategies, and decoding procedures, since this common framework underlies all our benchmarks.

Because all performance comparisons depend on the strength of the reference architecture, we use a fully optimized, state-of-the-art surface-code implementation as our baseline. The surface code is LDPC and has benefited from ${\sim} 25$ years of development: it supports low-overhead state preparation, efficient hook-error-free bare-ancilla syndrome extraction, a near-optimal matching decoder, and transversal interblock $\overline{\mathrm{CNOT}}$ gates.

Crucially, we use modern techniques of correlated decoding and algorithmic fault tolerance~\cite{cain2024correlated,zhou2025low} to ensure a strong surface-code baseline. These methods account for error propagation across codeblocks and allow the number of syndrome-extraction rounds per transversal $\overline{\mathrm{CNOT}}$ to be reduced from $\mathcal{O}(d)$ to $\mathcal{O}(1)$ while maintaining fault tolerance. Concretely, we employ a recent matching-based correlated decoder~\cite{cain2025fast} and, following prior results, use a single round of syndrome extraction per transversal $\overline{\mathrm{CNOT}}$. Technical details of the surface-code implementation are given in \cref{app:benchmarking/surface}.

We next describe how we operate the $\db{64,4,8}$ phantom code. Its non-LDPC structure necessitates syndrome-extraction, state-preparation, and decoding procedures compatible with high-weight stabilizers: \medskip
\begin{itemize}[leftmargin=*]
    \item Because bare-ancilla syndrome extraction is not viable, we employ Steane-style quantum error correction, in which errors are coherently copied onto ancilla codeblocks that are then measured transversally.
    \item Bare-ancilla codeblock initialization is likewise precluded, so we develop a short-depth, fault-tolerant state preparation protocol extending Ref.~\cite{gong2024computation}. Our scheme uses a compact four-to-one codeblock certification circuit, exploits qRM permutation automorphisms for fault tolerance, and operates as a state factory via error detection and preselection (see \cref{app:benchmarking/phantom/state_prep}). At two-qubit error rates of $10^{-3}$ and $5\times10^{-4}$, strict syndrome-free preselection yields acceptance rates of ${\sim}24\%$ and ${\sim}49\%$, respectively; and allowing weight-one residual errors increases these to ${\sim}40\%$ and ${\sim}64\%$. We focus on the stricter preselection protocol for comparison discussions in this section.
    \item To accurately decode large circuits, an efficient decoder that tracks spatiotemporal error correlations is needed. The $\db{64,4,8}$ code is non-matchable, and existing approaches such as code-capacity list decoding~\cite{gong2024computation,gong2024improved} do not extend directly to the correlated setting. We therefore develop spatiotemporal sliding-window correlated list and most-likely-error (MLE) decoders that account for error propagation through transversal $\overline{\mathrm{CNOT}}$ gates and Steane error-correction gadgets. Although MLE decoding is not polynomial-time in general, our protocol operates only on constant-sized windows\footnote{Constant-sized windows (around each transversal $\overline{\mathrm{CNOT}}$ gate in our implementation) suffice because verified-ancilla Steane QEC limits correlations to a constant-sized spacetime neighbourhood.},
    yielding decoding costs that scale linearly with logical circuit size on a fixed code. Details are given in \cref{app:benchmarking/phantom/decoding}.
\end{itemize}

Additionally, as described in \cref{sec:code_discovery_construction/qrm}, phantom qRM codes with identical $\db{n, k, d}$ parameters can be obtained by different choices of gauge-fixed logical operators, which can favour the preparation fidelity of either $\smash{\ket*{\overline{0}}}^{\otimes k}$ or $\smash{\ket*{\overline{+}}}^{\otimes k}$. For the numerics presented here, we adopt a balanced construction in which neither state is disadvantaged (see \cref{def:balanced_64-4-8}).

We employ a circuit-level noise model that includes single- and two-qubit gate errors, idle errors, and reset and measurement errors, with relative error rates calibrated to recent neutral-atom array experiments~\cite{bluvstein2025fault} (see also comparable error rates in Refs.~\cite{manetsch2025tweezer, muniz2025repeated, senoo2025high}). Full technical details of our benchmarking appear in \cref{app:benchmarking}.

\tocless\subsection{Single-codeblock repeated $\overline{\mathrm{CNOT}}$ circuits
\label{sec:benchmarking/cx}}

Our first benchmark isolates the most favourable regime for phantom codes, in which all $\overline{\mathrm{CNOT}}$s are in-block. While in-block $\overline{\mathrm{CNOT}}$s on phantom codes incur no cost after compilation, state preparation is generally more costly than on surface codes. This benchmark investigates when the reduction in gate cost outweighs the increased initialization overhead.

\begin{figure}
    \centering
    \includegraphics[width=\linewidth]{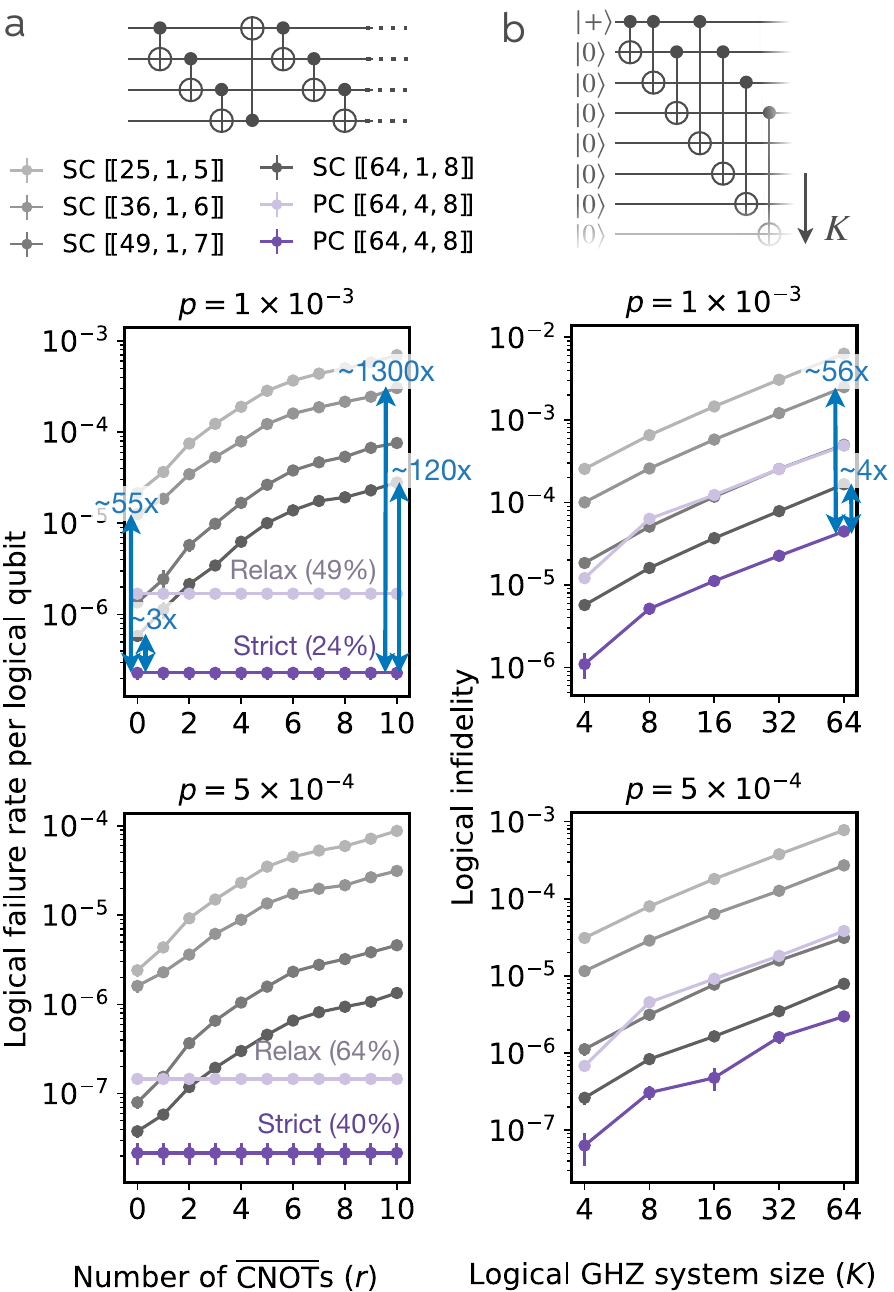}
    \phantomsubfloat{\label{fig:benchmark/cx}}
    \phantomsubfloat{\label{fig:benchmark/ghz}}
    \vspace{-0.75cm}
    \caption{\textbf{Benchmarking phantom and surface codes in $\bm{\overline{\mathrm{CNOT}}}$ circuits and logical GHZ state preparation.} 
    \textbf{(a)} Repeated in-block $\overline{\mathrm{CNOT}}$ circuits on four logical qubits hosted on a single $\db{64,4,8}$ phantom codeblock or four surface codeblocks ($d = 5$--$8$).
    The numerics include fault-tolerant state preparation and failure rates are averaged over $\ket*{\overline{0}}^{\otimes 4}$ and $\ket*{\overline{+}}^{\otimes 4}$ initial states.
    \textbf{(b)} Logical GHZ state preparation across multiple codeblocks, probing performance as the fraction of in-block permutation $\overline{\mathrm{CNOT}}$s decreases.
    Dark and pale purple denote strict and relaxed preselection, respectively. 
    Error bars are 98\% confidence intervals. 
    }
    \label{fig:benchmark}
\end{figure}

To probe this trade-off, we consider circuits consisting of fault-tolerant state preparation followed by repeated in-block $\overline{\mathrm{CNOT}}$s acting on four logical qubits hosted on a single $\db{64,4,8}$ codeblock or four surface codes (see \cref{fig:benchmark/cx}). We benchmark surface codes of distance $d = 5$--$8$. In all cases, results are averaged over $\smash{\ket*{\overline{0}}}^{\otimes 4}$ and $\smash{\ket*{\overline{+}}}^{\otimes 4}$ initial logical states to remove basis dependence.

We first assess logical state-preparation performance alone, corresponding to the $r = 0$ limit of \cref{fig:benchmark/cx}. The phantom code uses $64$ data qubits and $64 \times 3$ ancilla qubits for its state-preparation factory. Accounting for preselection overhead, it achieves a ${\sim} 55 \times$ reduction in logical failure rate over the $d=6$ surface code with a comparable spatial footprint ($144$ data and $140$ ancilla qubits). Even relative to the larger $d=8$ surface code, requiring roughly twice as many physical qubits, the phantom code achieves a ${\sim} 3 \times$ advantage.

As in-block $\overline{\mathrm{CNOT}}$s are added, the logical failure rate of the surface code grows linearly, whereas the phantom code's remains unchanged because its $\overline{\mathrm{CNOT}}$s require no physical operations. At near-term error rates, the resulting improvement reaches ${\sim} 1300\times$ (${\sim} 120 \times$) over the $d=6$ ($d = 8$) surface code on the deepest (depth-$10$) circuit studied. The logical gate time on the phantom code also remains effectively zero after state preparation, while it increases linearly with logical circuit depth on the surface code. The relative logical performance between phantom and surface codes remains similar at projected future physical error rates.

\tocless\subsection{Multiple-codeblock GHZ state preparation
\label{sec:benchmarking/ghz}}

Practical applications often require the use of multiple codeblocks, so our second benchmark studies logical performance as the fraction of in-block $\overline{\mathrm{CNOT}}$s implemented via permutations is reduced.

To control the ratio of in-block permutation and interblock transversal $\overline{\mathrm{CNOT}}$ gates in a simple and interpretable way, we benchmark logical GHZ state preparation across an increasing number of codeblocks. We consider logical system sizes $K = 4$--$64$ (see \cref{fig:benchmark/ghz}). For $K = 4$, all three entangling $\overline{\mathrm{CNOT}}$ gates are in-block and can be implemented as qubit permutations, giving a permutation-to-transversal gate ratio of $3\!:\!0$. As $K$ increases, entangling operations between codeblocks become necessary; for $K = 64$ this ratio decreases to $3\!:\!61$. We expect the phantom-code advantage to diminish as the contribution of transversal $\overline{\mathrm{CNOT}}$s increases.

Compared to the single-codeblock benchmark of \cref{sec:benchmarking/cx}, GHZ state preparation introduces two additional factors for the phantom code. First, interblock $\overline{\mathrm{CNOT}}$ gates necessitate active QEC, which we perform via Steane syndrome extraction as described in \cref{sec:benchmarking/operation}, incurring an overhead of two ancillary codeblocks per data codeblock. Second, the logical circuit requires codeblocks initialized in mixed computational- and Hadamard-basis logical states, namely the $\ket*{\overline{+000}}$ state. Fault-tolerant preparation of mixed-basis logical states on $k > 1$ codes are generally costly. To implement this efficiently we (i) prepare two phantom codeblocks in $\smash{\ket*{\overline{0}}}^{\otimes k}$ and $\smash{\ket*{\overline{+}}}^{\otimes k}$, (ii) addressably teleport logical qubits between the codeblocks using at most two transversal $\overline{\mathrm{CNOT}}$s, and (iii) measure one of the codeblocks transversally with preselection to suppress errors. We simulate the entire noisy process. Technical details are in \cref{app:benchmarking/ghz}.

We compare the phantom code against surface codes matched either in spatial footprint or in code distance. For $K=64$, a $d = 6$ surface code requires 64 codeblocks of $36$ data plus $35$ ancilla qubits each, for a total of $4544$ qubits. Each $\db{64,4,8}$ phantom codeblock requires two ancillary codeblocks for Steane QEC. To be conservative, we allocate one state-preparation factory per two data codeblocks, noting that this overhead could decrease on future hardware. At $K=64$, this amounts to a footprint of $4608$ qubits. At this scale, the phantom code achieves a ${\sim} 56 \times$ reduction in logical infidelity over the surface code. Compared to the $d = 8$ surface code, the phantom code retains ${\sim} 4 \times$ infidelity advantage while using ${\sim} 1.8 \times$ fewer physical qubits. Similar relative behaviour in logical fidelity persists at projected future physical error rates. Notably, the phantom-code advantage remains for all $K$, despite the increasing fraction of interblock $\overline{\mathrm{CNOT}}$s.

\tocless\subsection{Trotterized many-body quantum simulation
\label{sec:benchmarking/trotter}}

We conclude our benchmarking study with an application-motivated test of large-scale quantum many-body simulation, asking whether phantom codes can provide a sustained advantage as system size increases and physical error rates improve. This benchmark is motivated by the natural suitability of phantom codes for implementing $\overline{\mathrm{CNOT}}$ ladders, in which $\overline{\mathrm{CNOT}}$s act sequentially on neighbouring logical qubits. Such structures arise in Hamiltonian simulation, where terms of the form $\exp(-i\theta Z^{\otimes L})$ decompose into paired $\overline{\mathrm{CNOT}}$ ladders surrounding a single-qubit rotation. When these ladders span multiple codeblocks, a logarithmic-depth hypercube layout of $\overline{\mathrm{CNOT}}$ gates minimizes the number of required interblock transversal $\overline{\mathrm{CNOT}}$s, which we exploit.

\begin{figure}
    \includegraphics[width=\linewidth]{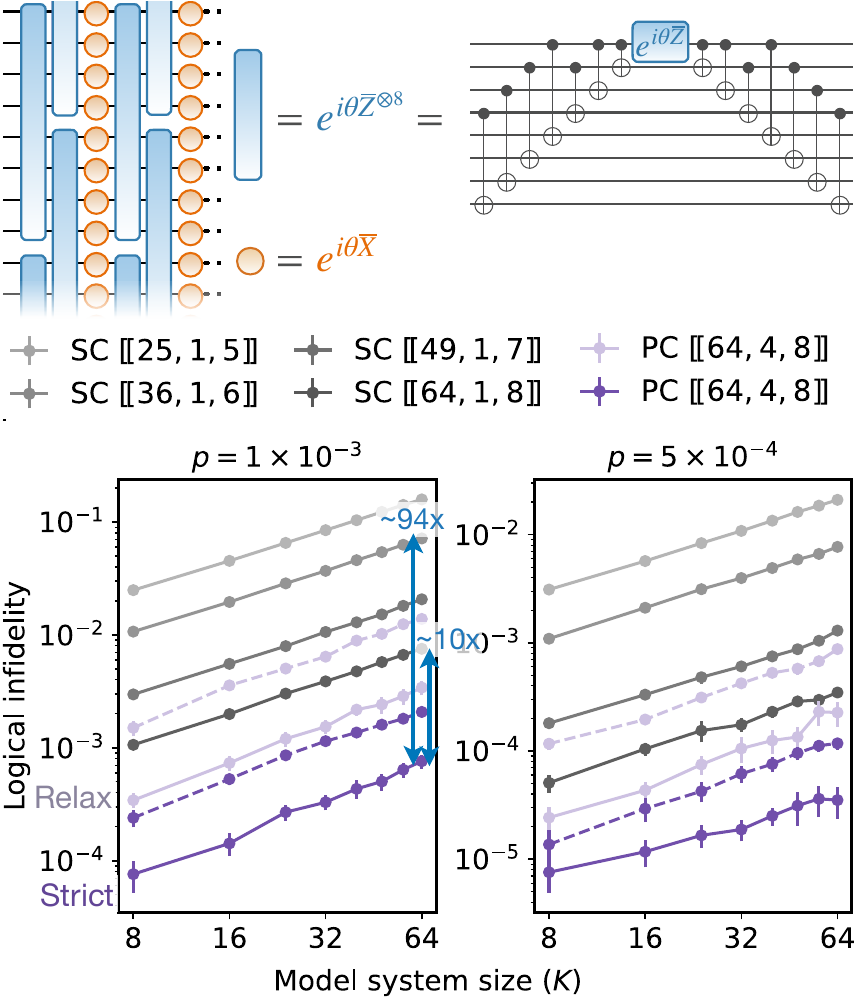}
    \caption{\textbf{Benchmarking phantom and surface codes in many-body quantum simulation.} 
    Logical circuits perform Trotterized time-evolution of a Hamiltonian with eight-body Ising interaction terms and a transverse field, on up to $64$ logical qubits over eight Trotter steps. 
    The largest system size uses $16$ codeblocks of the $\db{64,4,8}$ phantom code, or $64$ surface codeblocks of distance $d = 5$--$8$. 
    Dark and pale purple denote strict and relaxed preselection for state preparation on the phantom code, respectively; solid and dashed lines denote sliding-window MLE and correlated list decoding.
    Error bars are 98\% confidence intervals.
    }
    \label{fig:benchmarking/trotter}
\end{figure}

Concretely, we simulate Trotterized time-evolution generated by eight-body Ising Hamiltonian terms $\exp(-i\theta Z^{\otimes 8})$ in a brickwork pattern, interleaved with single-qubit transverse-field terms $\exp(-i\phi X)$ (see \cref{fig:benchmarking/trotter}). For the phantom code, each eight-body term use three in-block and four physical $\overline{\mathrm{CNOT}}$s. We study systems with $K = 8$--$64$ logical qubits over eight Trotter steps. The largest circuit contains ${>} 2400$ logical gates and have a two-qubit logical gate depth of $96$ at $K=64$.

At these widths, such quantum simulation circuits are generally classically intractable. For $K=64$, the $\db{64,4,8}$ phantom code again requires $4608$ qubits, inclusive of ancillary codeblocks for Steane QEC and state-preparation factories. This is to be compared with $4544$ ($8128$) qubits for the $d=6$ ($d=8$) surface code. To make simulations tractable, we fix $\theta = \phi = \pi/2$; in physical experiments, smaller angles would be used. Such rotations can be implemented using the STAR protocol~\cite{akahoshi2024partially}, whose application to the $\db{64,4,8}$ phantom code is discussed in \cref{app:qrm/STAR} (its use for surface codes was detailed in Ref.~\cite{akahoshi2024partially}). 

We use two polynomial-time spatiotemporal sliding-window decoders for the phantom code: a most-likely-error (MLE) decoder, and a correlated list decoder that is ${\sim} 100 \times$ faster but less accurate (see \cref{app:benchmarking/phantom/decoding}). The surface code continues to use the matching-based correlated decoder as described in \cref{sec:benchmarking/operation} and \cref{app:benchmarking/surface}. We use one round of syndrome extraction per transversal $\overline{\mathrm{CNOT}}$ for both codes.

As shown in \cref{fig:benchmarking/trotter}, the phantom code outperforms the surface code across near-term and projected future physical error rates, and this advantage is maintained as the system size increases. At near-term error rates, with sliding-window MLE decoding, the phantom codes achieves a ${\sim} 94 \times$ (${\sim} 10 \times$) reduction in logical infidelity compared against the $d = 6$ ($d = 8$) surface code, with similar relative improvements at future error rates. Even with relaxed preselection, the phantom code outperforms the $d = 8$ surface code. Although the sliding-window correlated list decoder sacrifices decoding accuracy, the phantom code still exceeds the $d = 8$ ($d = 7$) surface code in logical fidelity under strict (relaxed) preselection.

\tocless\section{Discussion and outlook
\label{sec:discussion}}

As emphasized, reducing QEC overhead requires jointly optimizing for storage and computation: what ultimately matters are the operations (including memory) that a QEC code enables. High rate and LDPC structure define one axis, enabling efficient state preparation, decoding, and low qubit overhead. An equally important axis is efficient logical gates, which distinguish quantum memories from quantum computers. Here, we identify a broad class of quantum error-correcting codes that realize logical entanglement without physical overhead or infidelity, and show that these codes can outperform leading fault-tolerant architectures for well-chosen tasks. These phantom codes bring workloads with dense local entangling structure closer to practical realization and enable new opportunities for logical algorithms in, for example, trapped-ion and neutral-atom platforms. The resources developed in this work broaden the scope of code discovery, enabling a systematic cataloguing of small-footprint CSS codes and direct solver-based search for codes with targeted properties.

Our results reveal that phantomness is a stringent structural condition that selects a highly symmetric corner of the space of codes and impose two practical consequences. First, all phantom codes discovered here are non-LDPC, necessitating Steane-style QEC, preselected logical-state factories, and the development of decoders presented in this work. Second, the phantom codes we identify encode $k = \mathcal{O}(\log n)$ logical qubits, limiting the fraction of phantom $\overline{\mathrm{CNOT}}$s in large-scale algorithms. As a result, scalable advantages arise in workloads where entangling operations occur in dense local patches and logical blocks are connected via transversal $\overline{\mathrm{CNOT}}$s.

These results open several new research directions, one of which is motivated by the observed limitations. The absence of LDPC phantom codes hint that permutation-implemented $\overline{\mathrm{CNOT}}$s may be incompatible with bounded-weight or geometrically local stabilizer generators, motivating either new constructions or formal no-go theorems. Likewise, the empirical bound $k = \mathcal{O}(\log n)$ raise the question of whether large automorphism groups can coexist with higher encoding rates. Finally, the fact that all transversal magic gates we found are effectively distance-two may reflect a deeper structural obstruction. Resolving which features are fundamental would clarify the asymptotic limits of permutation-implemented $\overline{\mathrm{CNOT}}$s.

A second direction concerns mapping the broader space of code design. Phantom codes minimize the temporal cost of logical entanglement, at the apparent expense of low encoding rates, whereas many prominent qLDPC codes exhibit the opposite trade-off. This contrast suggests an existence of a continuum between large-automorphism, low-rate codes and high-rate qLDPC constructions. To what extent this continuum exists, and how resource-optimal operating points are distributed along it, remains an open problem.

More broadly, the resources developed and assembled in this work—including a complete database of CSS codes up to $n = 14$, SAT-based search methods for identifying codes with specified logical gate structure, and an automated pipeline for extracting logical operations—provide the tools needed to explore this space systematically. These tools can be used, for example, to identify low-overhead codes tailored to specific hardware connectivities, such as on superconducting architectures, or to optimize alternative objectives, including reducing the cost of logical magic.

A third direction concerns software and decoding. The large automorphism groups of phantom codes suggest automorphism-aware compilers that treat permutations as free, pack $\overline{\mathrm{CNOT}}$-dense subroutines into phantom codeblocks, and schedule interblock gates to minimize noise. Many workloads beyond the Trotterized circuits benchmarked contain a high density of logical entangling gates. Identifying other algorithms that benefit from phantom codes would help focus these compiler developments.

A final direction concerns experimental implementation. Small-angle rotations in our Trotterized circuits can be realized using the STAR protocol~\cite{akahoshi2024partially}, opening the possibility of near-term fermionic simulations at realistic error rates. Beyond dynamics, phantom codes naturally support highly entangled logical states, which are central to measurement-based quantum computation~\cite{briegel2009measurement} and quantum communication tasks~\cite{majidy2025scalable}. Demonstrating such states with near-zero entangling-gate cost would validate the phantom-code paradigm.

\medskip
\medskip

\begin{acknowledgments}

We dedicate this work to the memory of Ray Laflamme (1960–2025), whose pioneering contributions and leadership have left a lasting legacy in quantum computing.

We would like to thank Gefen Baranes, Juan Pablo Bonilla Ataides, Madelyn Cain, Daniel Gottesmann, Milan Kornjača, Aleksander Kubica, Jonathan Lu, Nishad Maskara, Rohan Mehta, Chris Pattinson, John Preskill, Michael Vasmer, Adam Wills, Qian Xu, and Harry Zhou for useful discussions. 
We acknowledge financial support from
the IARPA and the Army Research Office, under the Entangled Logical Qubits program (Cooperative Agreement Number W911NF-23-2-0219);
the DARPA ONISQ MeasQuIT program (grant number HR0011-24-9-0359);
the U.S. Department of Energy (DOE Quantum Systems Accelerator Center, contract number 7568717);
the Center for Ultracold Atoms (an NSF Physics Frontier Center);
the National Science Foundation (grant numbers CCF-2313084, QLCI grant OMA-2120757 and OMA-2016245, and NQVL: Pilot:DLPQC grant 2410716 and Design:ORAQL grant 2533041);
Wellcome Leap Quantum for Bio program;
and QuEra Computing.
We acknowledge additional support from
the A*STAR Graduate Academy (J.M.K.),
a Harvard Quantum Initiative postdoctoral fellowship (D.B.T.),
a Simons Investigator award (N.Y.Y.),
and a Banting Postdoctoral Fellowship (S.M.).

\end{acknowledgments}

\clearpage
\pagebreak

\let\oldaddcontentsline\addcontentsline
\renewcommand{\addcontentsline}[3]{}
\bibliography{bib}
\let\addcontentsline\oldaddcontentsline

\clearpage
\pagebreak

\appendix
\onecolumngrid

\begin{center}
    \textbf{APPENDICES}
\end{center}

\tableofcontents

\clearpage
\pagebreak

\section{Phantom code basics, conditions, and properties}
\label{app:basics}

\subsection{Review of binary symplectic representation of Pauli operators and Clifford gates}
\label{app:basics/symplectic}

\textit{Pauli operators.} We make pervasive use of the symplectic representation in this work. In this representation, a Pauli operator on $n$ qubits is encoded as a binary vector $(\mathbf{x} \mid \mathbf{z}) \in \mathbb{F}_2^{2n}$, where $\mathbf{x}$ and $\mathbf{z}$ indicate the positions of $X$ and $Z$ components, respectively. For example, $XYIZ$ is represented as $(1100 \mid 0101)$, as $Y = iXZ$ contributes to both parts. Commutation relations are captured by the symplectic inner product---defining\footnote{Throughout our manuscript, the sizes of the identity matrix $\mathbb{I}$ and symplectic metric $\Omega$ are inferred from context where unambiguous, and arithmetic on binary vectors and matrices are modulo-$2$ unless stated otherwise.}
\begin{equation}
    \Omega = \begin{bmatrix}0 & \mathbb{I} \\ \mathbb{I} & 0\end{bmatrix},
\end{equation}
two Pauli operators represented by symplectic vectors $\mathbf{v}_1 = (\mathbf{x}_1 \mid \mathbf{z}_1)$ and $\mathbf{v}_2 = (\mathbf{x}_2 \mid \mathbf{z}_2)$ commute if and only if $\langle \mathbf{v}_1, \mathbf{v}_2 \rangle_S := \mathbf{v}_1 \Omega \mathbf{v}_2^\mathrm{T} = \mathbf{x}_1 \cdot \mathbf{z}_2 + \mathbf{x}_2 \cdot \mathbf{z}_1 = 0$. Finally, relevant to our setting here, a qubit permutation $\pi \in \S_n$ represented by an $n \times n$ permutation matrix $P$ acts on the symplectic representation of Pauli operators in a block-wise manner, $(\mathbf{x} \mid \mathbf{z}) \mapsto (\mathbf{x} \mid \mathbf{z})(P \oplus P) = (\mathbf{x}P \mid \mathbf{z}P)$.

\textit{Clifford gates and circuits.} Clifford circuits are linear maps between Pauli operators under conjugation. That is, $U P U^\dag = P'$ for a Clifford unitary $U$ and Pauli operators $P$ and $P'$. In the symplectic representation, on a register of $k$ qubits, a Clifford circuit $U$ can be represented by a binary Pauli transfer matrix $F(U) \in \Sp(2k, \mathbb{F}_2)$. That $F(U)$ is a symplectic matrix means that $\trans{F(U)} \Omega F(U) = \Omega$. For example, on $k = 1$ qubit, the Hadamard gate has the representation
\begin{equation}
    F(H) 
    = {
        \left[
        \begin{array}{c|c}
            F_{\mathrm{xx}}(H) & F_{\mathrm{xz}}(H) \\
            \hline
            F_{\mathrm{zx}}(H) & F_{\mathrm{zz}}(H) \\
        \end{array}
        \right]
    }
    = {
        \left[
        \begin{array}{c|c}
            0 & 1 \\
            \hline
            1 & 0 \\
        \end{array}
        \right],
    }
\end{equation}
and on a register of $k = 3$ qubits, a $\mathrm{CNOT}$ gate controlled by the first qubit and targeting the second has the representation 
\begin{equation}
    F(\mathrm{CNOT}_{12}) 
    = {
        \left[
        \begin{array}{c|c}
            F_{\mathrm{xx}}(\mathrm{CNOT}_{12}) & F_{\mathrm{xz}}(\mathrm{CNOT}_{12}) \\
            \hline
            F_{\mathrm{zx}}(\mathrm{CNOT}_{12}) & F_{\mathrm{zz}}(\mathrm{CNOT}_{12}) \\
        \end{array}
        \right]
    }
    = {
        \left[
        \begin{array}{ccc|ccc}
            1 & 1 & 0 & 0 & 0 & 0 \\
            0 & 1 & 0 & 0 & 0 & 0 \\
            0 & 0 & 1 & 0 & 0 & 0 \\
            \hline
            0 & 0 & 0 & 1 & 0 & 0 \\
            0 & 0 & 0 & 1 & 1 & 0 \\
            0 & 0 & 0 & 0 & 0 & 1 \\
        \end{array}
        \right].
    }
\end{equation}

The rows of the $F$ matrix correspond to input Pauli operators and the columns to their transformed outputs under conjugation, with basis elements ordered as $X_1, \dots, X_k, Z_1, \dots, Z_k$. Here, the defining action of $H$ is $X \mapsto Z$ and $Z \mapsto X$, and that of $\mathrm{CNOT}_{12}$ is $X_1 \mapsto X_1 X_2$, $Z_2 \mapsto Z_1 Z_2$, with all other basis Pauli operators remaining unchanged. This same formalism applies on logical qubits, but with physical Pauli operators replaced by logical Pauli operators (i.e.~$X_1, \dots, X_k, Z_1, \dots, Z_k$ replaced by $\overline{X}_1, \dots, \overline{X}_k, \overline{Z}_1, \dots, \overline{Z}_k$).

Composition of Clifford gates or circuits in the symplectic representation corresponds to multiplying their Pauli transfer matrices. The Pauli transfer matrix of a circuit $U$ comprising first the subcircuit $U_1$, then $U_2$, henceforth till $U_l$ is given by $F(U) = F(U_l) \cdots F(U_2) F(U_1)$. 

\textit{CNOT and bias-preserving circuits.} A gate or circuit that is bias-preserving is one that maps $X$- ($Z$-) type Pauli operators solely to $X$- ($Z$-) type Pauli operators under conjugation. Modulo Pauli operators, the CNOT gate is the only generator of the Clifford group that is bias-preserving; hence bias-preserving Clifford circuits are exactly those that are generated by CNOT gates (i.e.~CNOT circuits). 

The bias-preserving property means that the off-diagonal blocks of the symplectic representation of CNOT gates and circuits vanish: $F_{\mathrm{xz}}(\mathrm{CNOT}_{ij}) = F_{\mathrm{zx}}(\mathrm{CNOT}_{ij}) = 0$ for any $\mathrm{CNOT}_{ij}$ gate, and more generally $F_{\mathrm{xz}}(C) = F_{\mathrm{zx}}(C) = 0$ for any CNOT circuit $C$. Moreover, as a consequence of reversibility of the circuit, the diagonal submatrices must be invertible. The binary Pauli transfer matrix $F(C)$ for a CNOT circuit on $k$ qubits takes the form
\begin{equation}
    F(C) 
    = \diag\left(A, \invtrans{A}\right),
    \label{eq:bias_preserving_circuit_symplectic}
\end{equation}
where $A \in \GL(k, \mathbb{F}_2)$, and we use the shorthand for the transpose of the inverse $\invtrans{A} = \trans{(A^{-1})} = (\trans{A})^{-1}$. Such a CNOT circuit transforms $X$- and $Z$-type Pauli operators under conjugation as $\trans{\vb{x}} \mapsto A \trans{\vb{x}}$ and $\trans{\vb{z}} \mapsto \invtrans{A} \trans{\vb{z}}$ in the symplectic representation. Equivalently, the CNOT circuit effects the action $\ket{a} \mapsto \ket{a \trans{A}}$ for computational basis states $\ket{a}$ where $a \in \{0, 1\}^k$. Because $F(C)$ has the block-diagonal form in \cref{eq:bias_preserving_circuit_symplectic}, it suffices to consider the top-left $A$ block when analyzing $\mathrm{CNOT}$ circuits $C$, a fact that simplifies proofs in subsequent sections.

\subsection{Conditions for permutation logical Clifford gate sets and phantomness on stabilizer codes}
\label{app:basics/phantom_gen}

An $\db{n,k,d}$ stabilizer code is specified by a full-row-rank stabilizer generator matrix $H \in \mathbb{F}_2^{r \times 2 n}$ whose rows comprise a generating set of stabilizers in symplectic representation. Selecting an orthonormal logical basis fixes $Q_\mathrm{x}, Q_\mathrm{z} \in \mathbb{F}_2^{k \times 2n}$, whose rows are the symplectic representations of the logical operators $\{\overline{X}_i\}_{i = 1}^k$ and $\{\overline{Z}_i\}_{i = 1}^k$, such that $Q_\mathrm{x} \Omega \trans{Q_\mathrm{z}} = \mathbb{I}$. We can now express the conditions that the code has to satisfy to possess a given logical Clifford gate set implemented by a certain set of qubit permutations:

\begin{proposition}[Permutation logical Clifford gate set on a stabilizer code]
\label{prop:phantom_gateset_symplectic_gen}
    Given an $\db{n,k,d}$ stabilizer code with stabilizer generators $H \in \mathbb{F}_2^{r \times 2n}$ and an orthonormal logical basis $Q_\mathrm{x}, Q_\mathrm{z} \in \mathbb{F}_2^{k \times 2n}$ in symplectic form ($Q_\mathrm{x} \Omega \trans{Q_\mathrm{z}} = \mathbb{I}$), stack $Q := \binom{Q_\mathrm{x}}{Q_\mathrm{z}}$, and denote by $P^{(i)} \in \mathbb{F}_2^{n \times n}$ the permutation matrix representing $\pi^{(i)} \in \S_n$. Then $\{\pi^{(i)}\}_{i = 1}^p$ implements the logical Clifford gate set $\{\smash{\overline{U}}^{(i)}\}_{i = 1}^p$ on this code in this logical basis iff:
    \begin{subequations}\begin{align}
        Q \big( P^{(i)} \oplus P^{(i)} \big) \Omega \trans{Q} &= F\big(U^{(i)}\big) \Omega, \label{eq:phantom_logic}\\
        H \big( P^{(i)} \oplus P^{(i)} \big) \Omega \trans{H} &= 0, \label{eq:phantom_comm1}\\
        H \big( P^{(i)} \oplus P^{(i)} \big) \Omega \trans{Q} &= 0, \label{eq:phantom_comm2}\\
        Q \big( P^{(i)} \oplus P^{(i)} \big) \Omega \trans{H} &= 0, \label{eq:phantom_comm3}
    \end{align}\end{subequations}
    for every $i \in [p]$.
\end{proposition}

Above, \cref{eq:phantom_logic} ensures that each permutation transforms logical operators correctly, modulo stabilizers, and hence implements the desired logical action; while \cref{eq:phantom_comm1,eq:phantom_comm2,eq:phantom_comm3} demand that commutation relations among stabilizer and logical operators are maintained, so that the stabilizer group and hence codespace is preserved. 

By definition, phantom codes support a complete set of individually addressable $\overline{\mathrm{CNOT}}$ gates performed by qubit permutations, thus leading to \cref{prop:phantom_code_symplectic_gen}.

\begin{proposition}[Phantomness of a stabilizer code]
\label{prop:phantom_code_symplectic_gen}
    An $\db{n,k,d}$ stabilizer code is phantom iff it satisfies \cref{prop:phantom_gateset_symplectic_gen} for the gate set $\{\overline{\mathrm{CNOT}}_{a b}: (a, b) \in [k]^2, a \neq b\}$ for some choice of logical basis $Q_\mathrm{x}, Q_\mathrm{z}$.
\end{proposition}

There is an additional consideration when using \cref{prop:phantom_gateset_symplectic_gen} to determine whether a stabilizer code supports a specific permutation logical gate set, or \cref{prop:phantom_code_symplectic_gen} to determine whether a code is phantom: the logical bases that enable the desired gate set to be implemented via permutations are generically unknown a priori. While a choice of $Q_\mathrm{x}, Q_\mathrm{z}$ can always be written given $H$ via basis completion (i.e.~Gaussian elimination), this arbitrary choice may not support the specific desired gate set, and one must allow changes in logical basis to determine whether the code can support the gate set. To do so, we note that every pair of orthonormal logical bases $Q := \binom{Q_\mathrm{x}}{Q_\mathrm{z}}$ and $L := \binom{L_\mathrm{x}}{L_\mathrm{z}}$ of a code is related by a logical Clifford circuit, and accordingly, $Q = R L$ where $R \in \Sp(2k, \mathbb{F}_2)$---see \cref{app:basics/symplectic} for a review of the symplectic representation. Such a rotation of logical bases preserves the anticommutation relations, $Q_\mathrm{x} \Omega \trans{Q_\mathrm{z}} = \mathbb{I} \Longleftrightarrow L_\mathrm{x} \Omega \trans{L_\mathrm{z}} = \mathbb{I}$. Therefore, in using \cref{prop:phantom_gateset_symplectic_gen} or \cref{prop:phantom_code_symplectic_gen} to determine the supported gate set or phantomness of a stabilizer code, a fixed logical basis $L$ can first be computed from $H$ and $Q$ set as $Q = R L$, where $R$ representing the Clifford transform is allowed to be freely chosen.

\subsection{Conditions for permutation \texorpdfstring{$\overline{\mathrm{CNOT}}$}{CNOT} gate sets and phantomness on CSS codes}
\label{app:basics/phantom_css}

We now specialize \cref{prop:phantom_gateset_symplectic_gen,prop:phantom_code_symplectic_gen} to CSS codes. CSS codes are defined by stabilizer generators that each comprise purely $X$ or $Z$ Pauli operators, and their logical bases can similarly be chosen such that $\overline{X}_i$ ($\overline{Z}_i$) logicals comprise only $X$ ($Z$) operators, which we refer to as being in CSS form. Thus, a more compact ``half-symplectic'' binary representation is natural, where the zero $X$- or $Z$-portion of the symplectic representation is removed. An $\db{n,k,d}$ CSS code is then specified by full-row-rank stabilizer generator matrices $H_\mathrm{x} \in \mathbb{F}_2^{r_\mathrm{x} \times n}$ and $H_\mathrm{z} \in \mathbb{F}_2^{r_\mathrm{z} \times n}$, whose rows encode generating set of $X$- and $Z$-type stabilizers, respectively, and likewise selecting an orthonormal CSS logical basis fixes $Q_\mathrm{x}, Q_\mathrm{z} \in \mathbb{F}_2^{k \times n}$ such that $Q_\mathrm{x} \trans{Q_\mathrm{z}} = \mathbb{I}$. 

First, we establish results that restrict the possible scope of permutation logical gates on CSS codes, as formalized in \cref{prop:only_cnot_circuits_css_logical_basis,remark:not_only_cnot_swap_non_css_logical_basis}.

\begin{proposition}[Qubit permutations can achieve only $\overline{\mathrm{CNOT}}$ circuits with CSS logical basis]
\label{prop:only_cnot_circuits_css_logical_basis}
    The only logical action implementable by qubit permutations on a stabilizer code are $\overline{\mathrm{CNOT}}$ circuits when using a CSS logical basis, modulo logical Pauli gates.
\end{proposition}

\begin{proof}
    A CSS logical basis defines logical operators that are pure $X$-or $Z$-type. Qubit permutations are bias-preserving and cannot induce Pauli basis changes---they map $X$- ($Z$)-type Pauli operators exactly onto $X$- ($Z$)-type Pauli operators, and so cannot change the $X$- or $Z$-type of logical operators on the code. The only bias-preserving Clifford circuits are $\mathrm{CNOT}$ circuits (see \cref{app:basics/symplectic}), modulo single-qubit Pauli gates. Additionally, logical action at the third or higher level of the Clifford hierarchy is not possible as qubit permutations on stabilizer codes contain no magic.
\end{proof}

\begin{remark}
\label{remark:not_only_cnot_swap_non_css_logical_basis}
    \Cref{prop:only_cnot_circuits_css_logical_basis} does not hold for stabilizer, even CSS, codes when a non-CSS logical basis is used.
\end{remark}

\begin{proof}
    We give a counter-example. Consider a CSS logical basis of an arbitrary $k = 2$ code, $\{\overline{X}_1, \overline{X}_2, \overline{Z}_1, \overline{Z}_2\}$, and suppose $\overline{\mathrm{CNOT}}_{12}$ is available through a qubit permutation $\pi$ in this basis. Now consider the non-CSS logical basis
    \begin{equation*}
        \overline{X}_1' = \overline{X}_1 \overline{Z}_1, 
            \qquad
        \overline{X}_2' = \overline{X}_2 \overline{Z}_2, 
            \qquad
        \overline{Z}_1' = \overline{X}_1 \overline{X}_2 \overline{Z}_2, 
            \qquad
        \overline{Z}_2' = \overline{X}_1 \overline{X}_2 \overline{Z}_1.
    \end{equation*}
    In this logical basis, $\pi$ has action $\overline{X}_1' = \overline{X}_1 \overline{Z}_1 \mapsto \overline{X}_1 \overline{X}_2 \overline{Z}_1 = \overline{Z}_2'$, which cannot be generated by $\overline{\mathrm{CNOT}}$ gates.
\end{proof}

\Cref{remark:not_only_cnot_swap_non_css_logical_basis} suggests that, to assess logical operations beyond $\overline{\mathrm{CNOT}}$ circuits via qubit permutations on a CSS code, one can employ a non-CSS logical basis on the code. However, the availability of transversal (i.e.~depth-one) $\overline{\mathrm{CNOT}}$ gates between codeblocks generally require CSS logical bases to be chosen. Therefore, demanding the general ability to efficiently entangle distinct codeblocks through transversal $\overline{\mathrm{CNOT}}$ gates places us within the regime of \cref{prop:only_cnot_circuits_css_logical_basis}, which constrains available permutation logical actions to $\overline{\mathrm{CNOT}}$ circuits. With this limitation, \cref{prop:phantom_gateset_symplectic_gen} simplifies to \cref{prop:phantom_gateset_symplectic_css}, and \cref{prop:phantom_code_symplectic_css} follows.

\begin{proposition}[Permutation $\overline{\mathrm{CNOT}}$-circuit gate sets on a CSS code]
\label{prop:phantom_gateset_symplectic_css}
    Given an $\db{n,k,d}$ CSS code with stabilizer generators $H_\mathrm{x} \in \mathbb{F}_2^{r_\mathrm{x} \times n}$, $H_\mathrm{z} \in \mathbb{F}_2^{r_z \times n}$ and an orthonormal CSS logical basis $Q_\mathrm{x}, Q_\mathrm{z} \in \mathbb{F}_2^{k \times n}$ in binary form ($Q_\mathrm{x} \trans{Q_\mathrm{z}} = \mathbb{I}$), denote by $P^{(i)} \in \mathbb{F}_2^{n \times n}$ the permutation matrix representing $\pi^{(i)} \in \S_n$. Then $\{\pi^{(i)}\}_{i = 1}^p$ implements the logical gate set $\{\smash{\overline{U}}^{(i)}\}_{i = 1}^p$ on this code in this logical basis iff:
    \begin{subequations}\begin{align}
        Q_\mathrm{x} P^{a_i b_i} \trans{Q_\mathrm{z}} &= F_\mathrm{xx}\big(U^{(i)}\big), \\
        Q_\mathrm{z} P^{a_i b_i} \trans{Q_\mathrm{x}} &= F_\mathrm{zz}\big(U^{(i)}\big), \\
        H_\mathrm{x} P^{a_i b_i} \trans{H_\mathrm{z}}
        &= H_\mathrm{z} P^{a_i b_i} \trans{H_\mathrm{x}} = 0, \\
        H_\mathrm{x} P^{a_i b_i} \trans{Q_\mathrm{z}} 
        &= H_\mathrm{z} P^{a_i b_i} \trans{Q_\mathrm{x}} = 0, \\
        Q_\mathrm{x} P^{a_i b_i} \trans{H_\mathrm{z}}
        &= Q_\mathrm{z} P^{a_i b_i} \trans{H_\mathrm{x}} = 0,
    \end{align}\end{subequations}
    for every $i \in [p]$, where each $U^{(i)} \in \langle \mathrm{CNOT}_jk: j,k \in [k]^2, j \neq k \rangle$.
\end{proposition}

The issue of an unknown suitable choice of logical basis $Q_\mathrm{x}, Q_\mathrm{z}$ again arises when using \cref{prop:phantom_gateset_symplectic_css} to determine whether a CSS code supports a specific gate set. We address this by noting that every pair of orthonormal CSS logical bases $Q_\mathrm{x}, Q_\mathrm{z}$ and $L_\mathrm{x}, L_\mathrm{z}$ is related by a $\overline{\mathrm{CNOT}}$ circuit---namely, $Q_\mathrm{x} = R L_\mathrm{x}$ and $Q_\mathrm{z} = \invtrans{R} L_\mathrm{z}$ where $R \in \GL(k, \mathbb{F}_2)$. The $\overline{\mathrm{CNOT}}$ circuit $C$ performing the logical basis change is then described by $F(C) = \diag(R, \invtrans{R})$; see \cref{app:basics/symplectic} for a review of the symplectic representation. Such a rotation of logical basis preserves the anticommutation relations, $Q_\mathrm{x} \trans{Q_\mathrm{z}} = \mathbb{I} \Longleftrightarrow L_\mathrm{x} \trans{L_\mathrm{z}} = \mathbb{I}$. Therefore, in using \cref{prop:phantom_gateset_symplectic_css} to determine the supported permutation gate set of a CSS code, a fixed logical basis $L_\mathrm{x}, L_\mathrm{z}$ can first be computed from $H_\mathrm{x}, H_\mathrm{z}$, and $Q_\mathrm{x}, Q_\mathrm{z}$ then set as the rotated versions of $L_\mathrm{x}, L_\mathrm{z}$, where $R$ are allowed to be freely chosen.

Interestingly, while the specific permutation gate set of a CSS code is logical basis-dependent, the phantomness of the code is basis-independent. We formalize this in \cref{prop:phantomness_css_logical_basis_independence}.

\begin{proposition}[Phantomness is independent of CSS logical basis]
    \label{prop:phantomness_css_logical_basis_independence}
    If a stabilizer code is phantom in some CSS logical basis, it is phantom in every CSS logical basis.
\end{proposition}

\begin{proof}
Any two CSS logical bases are related by a logical $\overline{\mathrm{CNOT}}$ circuit, and a phantom code can implement any $\overline{\mathrm{CNOT}}$ circuit through qubit permutations. Suppose the code is phantom in some CSS logical basis $L_\mathrm{x}, L_\mathrm{z}$. Then, in some other CSS logical basis $L_\mathrm{x}', L_\mathrm{z}'$, to perform an arbitrary $\overline{\mathrm{CNOT}}$, one does the following: (1) transform the logical basis to $L_\mathrm{x}, L_\mathrm{z}$ through qubit permutations; (2) implement the $\overline{\mathrm{CNOT}}$ as known; (3) transform the logical basis back to $L_\mathrm{x}', L_\mathrm{z}'$ through qubit permutations. The entire procedure comprises only qubit permutations, so the code is phantom also in logical basis $L_\mathrm{x}', L_\mathrm{z}'$.
\end{proof}

This basis independence allows a simplification of \cref{prop:phantom_code_symplectic_gen} to \cref{prop:phantom_code_symplectic_css}. Here, as explained in the main text and motivated above, we restrict our attention to CSS logical bases only on CSS codes, so as to guarantee the availability of transversal $\overline{\mathrm{CNOT}}$ gates between codeblocks.

\begin{proposition}[Phantomness of a CSS code]
\label{prop:phantom_code_symplectic_css}
    An $\db{n,k,d}$ CSS code is phantom iff it satisfies \cref{prop:phantom_gateset_symplectic_css} for the gate set $\{\overline{\mathrm{CNOT}}_{a b}: (a, b) \in [k]^2, a \neq b\}$ for an arbitrarily chosen CSS logical basis $Q_\mathrm{x}, Q_\mathrm{z}$.
\end{proposition}

Lastly, we remark that a smaller gate set can be used in the conditions of \cref{prop:phantom_code_symplectic_gen,prop:phantom_code_symplectic_css} when assessing whether a code is phantom. Doing so is advantageous for computational efficiency, for example in our code enumeration and SAT-based code discovery efforts described later in \cref{app:enumeration,app:code_discovery}.

\begin{remark}[Minimal-size gate set for phantomness of stabilizer codes]
\label{prop:phantom_code_minimal_gateset}
    In \cref{prop:phantom_code_symplectic_gen,prop:phantom_code_symplectic_css}, the smaller-sized gate set $\{\overline{\mathrm{CNOT}}_{12}, \overline{\mathrm{SWAP}}_{12}, \overline{\mathrm{SWAP}}_{23}, \ldots, \overline{\mathrm{SWAP}}_{(k - 1) (k)}\}$ can equivalently be used.
\end{remark}

\begin{proof}
    Either gate set generates the other.
\end{proof}

\subsection{Code equivalence under qubit permutation and global Hadamards}
\label{app:basics/equivalence}

As mentioned in the main text and to facilitate further technical discussion, it is useful to define a suitable notion of equivalence of codes. While prior literature considered equivalence of codes up to qubit permutations and local (single-qubit) Cliffords~\cite{cross2025small,yu2007graphical}, the phantom property of codes is not preserved under local Clifford deformations. It is most suitable instead to restrict the freedom of local Cliffords to global Hadamards, leading to \cref{def:code_equivalence}.

\begin{definition}[$\mathrm{\Pi H}$ equivalence of codes]
    \label{def:code_equivalence}
    Two $n$-qubit codes are $\mathrm{\Pi H}$-equivalent if one can be mapped onto the other through a qubit permutation and optionally a global Hadamard ($H^{\otimes n}$).
\end{definition}

This notion of equivalence preserves the $\db{n, k, d}$ parameters and crucially the phantom property of codes.

\begin{proposition}[Invariance of phantomness under $\mathrm{\Pi H}$ equivalence]
    \label{prop:phantom_equivalence_invariance}
    Any code that is $\mathrm{\Pi H}$-equivalent to a phantom code is also phantom.
\end{proposition}

\subsection{Uniform weight distribution of logical operators on phantom codes and the Hamming bound}
\label{app:basics/weight}

\begin{table}[t]
    \centering
    \renewcommand{\arraystretch}{1.3}
    \begin{tabular}{|C{5cm}|C{2cm}|C{3cm}|C{6.5cm}|}
        \hline
        Logical Pauli equivalence class & Weight $\abs{\vb{w}}$ & $\vb{w} = (w_\mathrm{x}, w_\mathrm{z}, w_\mathrm{y})$ & Class members (physical Pauli operators) \\
        \hline
        \multirow{2}{*}{$\overline{X_1}$} & 2 & $(2, 0, 0)$ & $X_{1} X_{4}, X_{2} X_{3}$ \\ 
        \cline{2-4}
         & 4 & $(0, 2, 2)$ & $Y_{1} Z_{2} Z_{3} Y_{4}, Z_{1} Y_{2} Y_{3} Z_{4}$ \\
        \hline
        \multirow{2}{*}{$\overline{X_2}$} & 2 & $(2, 0, 0)$ & $X_{1} X_{2}, X_{3} X_{4}$ \\
        \cline{2-4}
         & 4 & $(0, 2, 2)$ & $Y_{1} Y_{2} Z_{3} Z_{4}, Z_{1} Z_{2} Y_{3} Y_{4}$ \\
        \hline
        \multirow{2}{*}{$\overline{X_1 X_2}$} & 2 & $(2, 0, 0)$ & $X_{1} X_{3}, X_{2} X_{4}$ \\
        \cline{2-4}
         & 4 & $(0, 2, 2)$ & $Y_{1} Z_{2} Y_{3} Z_{4}, Z_{1} Y_{2} Z_{3} Y_{4}$ \\
        \hline
        \multirow{2}{*}{$\overline{Z_1}$} & 2 & $(0, 2, 0)$ & $Z_{1} Z_{2}, Z_{3} Z_{4}$ \\
        \cline{2-4}
         & 4 & $(2, 0, 2)$ & $X_{1} X_{2} Y_{3} Y_{4}, Y_{1} Y_{2} X_{3} X_{4}$ \\
        \hline
        \multirow{2}{*}{$\overline{Z_2}$} & 2 & $(0, 2, 0)$ & $Z_{1} Z_{4}, Z_{2} Z_{3}$ \\
        \cline{2-4}
         & 4 & $(2, 0, 2)$ & $X_{1} Y_{2} Y_{3} X_{4}, Y_{1} X_{2} X_{3} Y_{4}$ \\
        \hline
        \multirow{2}{*}{$\overline{Z_1 Z_2}$} & 2 & $(0, 2, 0)$ & $Z_{1} Z_{3}, Z_{2} Z_{4}$ \\
        \cline{2-4}
         & 4 & $(2, 0, 2)$ & $X_{1} Y_{2} X_{3} Y_{4}, Y_{1} X_{2} Y_{3} X_{4}$ \\
        \hline
        \multirow{2}{*}{$\overline{X_1 Z_2}$} & 2 & $(0, 0, 2)$ & $Y_{1} Y_{4}, Y_{2} Y_{3}$ \\
        \cline{2-4}
         & 4 & $(2, 2, 0)$ & $X_{1} Z_{2} Z_{3} X_{4}, Z_{1} X_{2} X_{3} Z_{4}$ \\
        \hline
        \multirow{2}{*}{$\overline{Z_1 X_2}$} & 2 & $(0, 0, 2)$ & $Y_{1} Y_{2}, Y_{3} Y_{4}$ \\
        \cline{2-4}
         & 4 & $(2, 2, 0)$ & $X_{1} X_{2} Z_{3} Z_{4}, Z_{1} Z_{2} X_{3} X_{4}$ \\
        \hline
        \multirow{2}{*}{$\overline{Y_1 Y_2}$} & 2 & $(0, 0, 2)$ & $Y_{1} Y_{3}, Y_{2} Y_{4}$ \\
        \cline{2-4}
         & 4 & $(2, 2, 0)$ & $X_{1} Z_{2} X_{3} Z_{4}, Z_{1} X_{2} Z_{3} X_{4}$ \\
        \hline
        $\overline{Y_1}$ & 3 & $(1, 1, 1)$ & $X_{1} Z_{3} Y_{4}, X_{2} Y_{3} Z_{4}, Y_{1} Z_{2} X_{4}, Z_{1} Y_{2} X_{3}$ \\
        \hline
        $\overline{Y_2}$ & 3 & $(1, 1, 1)$ & $X_{1} Y_{2} Z_{3}, Y_{1} X_{2} Z_{4}, Z_{1} X_{3} Y_{4}, Z_{2} Y_{3} X_{4}$ \\
        \hline
        $\overline{X_1 Y_2}$ & 3 & $(1, 1, 1)$ & $X_{1} Z_{2} Y_{3}, Y_{1} X_{3} Z_{4}, Y_{2} Z_{3} X_{4}, Z_{1} X_{2} Y_{4}$ \\
        \hline
        $\overline{Y_1 X_2}$ & 3 & $(1, 1, 1)$ & $X_{1} Y_{3} Z_{4}, X_{2} Z_{3} Y_{4}, Y_{1} Z_{2} X_{3}, Z_{1} Y_{2} X_{4}$ \\
        \hline
        $\overline{Y_1 Z_2}$ & 3 & $(1, 1, 1)$ & $X_{1} Z_{2} Y_{4}, Y_{1} Z_{3} X_{4}, Y_{2} X_{3} Z_{4}, Z_{1} X_{2} Y_{3}$ \\
        \hline
        $\overline{Z_1 Y_2}$ & 3 & $(1, 1, 1)$ & $X_{1} Y_{2} Z_{4}, Y_{1} X_{2} Z_{3}, Z_{1} Y_{3} X_{4}, Z_{2} X_{3} Y_{4}$ \\
        \hline
    \end{tabular}
    \caption{\textbf{Logical Pauli equivalence classes of the $\mathbf{\db{4,2,2}}$ code and their weights.} The last column lists explicitly all members of each class $\smash{\mathcal{P}_{\overline{P}, \vb{w}}}$; the size of the class $\abs{\smash{\mathcal{P}_{\overline{P}, \vb{w}}}}$ is the length of the list.}
    \label{table:422_all_logical}
\end{table}

As the weight, and in fact at a finer-grained level the $X$-, $Y$- and $Z$-operator weights, of a Pauli operator is invariant under qubit permutations, yet qubit permutations induce logical $\overline{\mathrm{CNOT}}$ actions on phantom codes, phantom codes must possess an underlying structure of weight uniformity on their logical operators. This observation ultimately enables the derivation of a Hamming bound on the parameters of phantom codes. We formalize these concepts and this line of argument in this section. 

\begin{definition}[Logical Pauli equivalence classes]
    For a logical Pauli operator $\overline{P}$ on a stabilizer code, we define:
    \begin{itemize}
        \item $\mathcal{P}_{\overline{P}}$ to be the set of all physical Pauli operators implementing the logical Pauli $\overline{P}$, i.e., the set obtained by multiplying a fixed representative of $\overline{P}$ by all elements of the stabilizer group of the code.
        \item $\mathcal{P}_{\overline{P},w}$ to be the subset of $\mathcal{P}_{\overline{P}}$ consisting of physical Pauli operators of weight $w$.
        \item $\mathcal{P}_{\overline{P}, \vb{w}}$ to be the subset of $\mathcal{P}_{\overline{P}}$ consisting of physical Pauli operators with weight vector $\vb{w} := (w_\mathrm{x}, w_\mathrm{z}, w_\mathrm{y})$, where $w_\mathrm{x}$, $w_\mathrm{z}$, and $w_\mathrm{y}$ count the numbers of $X$, $Z$, and $Y$ operators, respectively, in the physical Pauli implementation.
    \end{itemize}
\end{definition}

Note that the $\smash{\mathcal{P}_{\overline{P},w}}$ subsets are disjoint for different $w$---$\smash{\mathcal{P}_{\overline{P},w}} \cap \smash{\mathcal{P}_{\overline{P},w'}} = \empty$ when $w \neq w'$---and likewise $\smash{\mathcal{P}_{\overline{P},\vb{w}}} \cap \smash{\mathcal{P}_{\overline{P},\vb{w}'}} = \empty$ when $\vb{w} \neq \vb{w}'$. As the weight of a Pauli operator is the sum of its part-wise weights, $\abs{\vb{w}} = w_\mathrm{x} + w_\mathrm{z} + w_\mathrm{y}$, we have also that $\smash{\mathcal{P}_{\overline{P},w}} = \smash{\bigcup_{\vb{w}: \abs{\vb{w}} = w} \mathcal{P}_{\overline{P},\vb{w}}}$, where all subsets in the union are disjoint. That the distance of the code is $d$ implies that there can be no logical operators of weight less than $d$, therefore $\mathcal{P}_{\overline{P},w} = \emptyset$ for all $w < d$, and likewise $\mathcal{P}_{\overline{P}, \vb{w}} = \emptyset$ for all $\vb{w}$ with $\abs{\vb{w}} < d$. 

As a concrete example, we list comprehensively the (nontrivial) logical Pauli equivalence classes and their weights for the $\db{4,2,2}$ phantom code in \cref{table:422_all_logical}. We read, for instance, that $\abs{\smash{\mathcal{P}_{\overline{X}_1}}} = 4$ is naturally the same size as the stabilizer group of the code; $\abs{\smash{\mathcal{P}_{\overline{X}_1,2}}} = 2$ and $\abs{\smash{\mathcal{P}_{\overline{X}_1,4}}} = 2$. As the code is CSS and we have used a CSS logical basis, there is only a single possible weight vector for each weight: $\abs{\smash{\mathcal{P}_{\overline{X}_1,(2,0,0)}}} = \abs{\smash{\mathcal{P}_{\overline{X}_1,(0,2,2)}}} = 2$.

\begin{lemma}[Uniform Pauli weight distribution of logical operators on phantom stabilizer codes]
\label{prop:uniform_weight_distribution}
    For an $\db{n,k,d}$ phantom stabilizer code and any fixed weight vector $\vb{w}$, the size of the equivalence class $\abs{\smash{\mathcal{P}_{\overline{P}, \vb{w}}}}$, in other words the number of physical Pauli operators implementing the logical Pauli operator $\overline{P}$, is the same for all nontrivial $X$-type logical operators $\overline{P}$. (There are $2^k - 1$ such operators: $\overline{X}_1, \overline{X}_2, \ldots, \overline{X}_1 \overline{X}_2, \ldots, \overline{X}_1 \overline{X}_2 \cdots \overline{X}_k$.) The same statement holds for $Z$-type logical operators.
\end{lemma}

\begin{proof}
    Consider any two $X$-type logical operators $\smash{\overline{P}}$ and $\smash{\overline{P}}'$. These must be related by $\overline{\mathrm{CNOT}}$ transformations: there exist $\overline{\mathrm{CNOT}}$ circuits $C$ and $C'$ transforming $\smash{\overline{P}}$ to $\smash{\overline{P}}'$ and $\smash{\overline{P}}'$ to $\smash{\overline{P}}$, respectively (in fact $C'$ is the circuit $C$ reversed). As the code is phantom, $C$ and $C'$ can be implemented by qubit permutations. Any qubit permutation preserves the weight vector of any physical Pauli operator. Therefore, each physical Pauli operator in the equivalence class $\smash{\mathcal{P}_{\overline{P}, \vb{w}}}$ is mapped to a physical Pauli operator in $\smash{\mathcal{P}_{\overline{P}', \vb{w}}}$, and likewise from $\smash{\mathcal{P}_{\overline{P}', \vb{w}}}$ to $\smash{\mathcal{P}_{\overline{P}, \vb{w}}}$. This defines a bijective map between $\smash{\mathcal{P}_{\overline{P}, \vb{w}}}$ and $\smash{\mathcal{P}_{\overline{P}', \vb{w}}}$, so the two sets have the same cardinality. An analogous argument applies to $Z$-type logical operators.
\end{proof}

A consequence of this structure of weight uniformity is a Hamming bound for CSS phantom
codes, which limits the possible $\db{n,k,d}$ parameters of a CSS phantom code. This result has been presented as \cref{thm:phantom_css_hamming_bound} in the main text and is reproduced below for convenient reference. Here, we provide a proof for the bound.

\PhantomCSSHammingBound*

\begin{proof}
    Consider the case $\mu = X$, that is, the $X$-distance is limiting. By \cref{prop:uniform_weight_distribution}, all $2^k - 1$ nontrivial $X$-type logical operators $\overline{P}$ have equivalence classes of the same size at weight $d$, namely
    $\abs{\smash{\mathcal{P}_{\overline{P}, d}}} = \eta$. Thus, the left-hand side of the Hamming bound, $\eta (2^k - 1)$, counts the total number of weight-$d$ Pauli-$X$ physical operators that implement nontrivial logical $X$-type operators. Note that, because logical states on a code are orthogonal, these physical operators are all distinct. 
    
    On the other hand, the right-hand side $B(n, d)$ is an upper bound on this quantity, since each such physical Pauli-$X$ operator corresponds to a binary string of length $n$ and weight $d$, and because the code has distance $d$, the product of any two such physical operators must have weight at least $d$. This is precisely the constraint in the definition of $B(n, d)$. The case of $\mu = Z$ is analogous.
\end{proof}

\begin{table}[t]
    \newcommand{\colwidth}{1.5cm}
    \begin{tabular}{|C{0.4cm}|C{1cm}|C{1.1cm}|C{1.1cm}|C{1.1cm}|C{1.2cm}|C{1.4cm}|C{1.4cm}|C{1.4cm}|C{\colwidth}|C{\colwidth}|C{\colwidth}|C{\colwidth}|}
    \hline
    $d$ & $n=4$ & $n=5$ & $n=6$ & $n=7$ & $n=8$ & $n=9$ & $n=10$ & $n=11$ & $n=12$ & $n=13$ & $n=14$ & $n=15$ \\
    \hline
    2 & 6 (6) & 10 (10) & 15 (15) & 21 (21) & 28 (28) & 36 (36) & 45 (45) & 55 (55) & 66 (66) & 78 (78) & 91 (91) & 105 (105)\\
    3 & 1 (2) & 2 (3) & 4 (5) & 7 (7) & ${\ge} \,8$ (9) & 12 (12) & ${\ge} \,13$ (15) & ${\ge} \,17$ (18) & ${\ge} \,20$ (22) & 26 (26) & ${\ge} \,27$ (30) & 35 (35)\\
    4 & 1 (1) & 1 (2) & 3 (5) & 7 (8) & 14 (14) & ${\ge} \,16$ (21) & 30 (30) & ${\ge} \,31$ (41) & ${\ge} \,41$ (55) & ${\ge} \,53$ (71) & ${\ge} \,68$ (91) & ${\ge} \,85$ (113)\\
    5 & 0 (0) & 1 (1) & 1 (2) & 1 (3) & 2 (5) & 3 (8) & ${\ge} \,6$ (12) & ${\ge} \,8$ (16) & ${\ge} \,11$ (22) & ${\ge} \,14$ (28) & ${\ge} \,27$ (36) & ${\ge} \,22$ (45)\\
    \hline
    \end{tabular}
    \caption{\textbf{Computed exact values and bounds on $\bm{B(n,d)}$.} Used in the Hamming bound, $B(n,d)$ is the maximum number of length-$n$ weight-$d$ bitstrings with pairwise weight ${\geq} d$. Entries denoted with $\geq$ are best-known lower bounds obtained by the \texttt{Z3} SMT solver~\cite{moura2008z3} within a 12-hour time limit; the other numbers are proven optimal. Upper bounds of $\left\lfloor \tbinom{n}{t} / \tbinom{d}{t} \right\rfloor$, $t=\lfloor d/2 \rfloor + 1$, are in parentheses.}
    \label{table:fnd}
\end{table}

We make a few further remarks on the Hamming bound. First, in the presentation of \cref{thm:phantom_css_hamming_bound}, we implicitly assume that $\eta$ is known. If only $n, k, d$ are specified and $\eta$ is unknown, $\eta$ may be replaced by any lower bound on the size of the equivalence class of physical weight-$d$ Pauli operators implementing a type-$\mu$ logical Pauli. 

Secondly, the right-hand side of \cref{thm:phantom_css_hamming_bound} can in principle be tightened by computing sharper upper bounds on the number of type-$\mu$ physical Pauli operators that implement logical operators. Using the \texttt{Z3} SMT solver~\cite{moura2008z3}, we computed exact values of $B(n, d)$ for small $n$ and $d$, which are reported in \cref{table:fnd}. For $d = 2$, one has an exact closed-form formula $F(n,2) = n(n-1)/2$, since any two distinct weight-$2$ bitstrings can overlap in at most one position. The parentheses in \cref{table:fnd} enclose upper bounds $B(n,d) \leq \smash{\lfloor \binom{n}{t}/\binom{d}{t} \rfloor}$ where $t=\lfloor d/2\rfloor + 1$. To see this bound, identify the support of each weight-$d$ string with a block $S \subseteq [n]$ of size $d$. For two blocks $S$ and $T$, the condition $|S \oplus T|\le d$ implies $|S\cap T|\le t-1$, equivalently that no $t$-subset of $[n]$ is contained in more than one block. Each block contains $\smash{\binom{d}{t}}$ distinct $t$-subsets, while each $t$-subset can appear in at most one block. Counting $t$-subsets therefore yields the stated upper bound.

Lastly, \cref{thm:phantom_css_hamming_bound} is reminiscent of the Hamming bound for classical linear codes. We emphasize that, for a classical linear code, the codewords are only required to have weight at least $d$, whereas in the definition of $B(n,d)$, all strings are constrained to have weight exactly $d$. We comment that there is also a quantum Hamming bound for non-degenerate stabilizer codes~\cite{ekert1996error}, but some phantom codes are degenerate (e.g.~the $\db{20,2,6}$ code), so that bound does not directly apply.

\subsection{Efficient addressable \texorpdfstring{$\overline{\mathrm{CNOT}}$}{CNOT} gates and arbitrary \texorpdfstring{$\overline{\mathrm{CNOT}}$}{CNOT} circuits between CSS phantom codeblocks} 
\label{app:basics/addressable_cnots}

Phantom codes support super-efficient arbitrary in-block $\overline{\mathrm{CNOT}}$ circuits implemented by qubit permutations, which do not have to be physically performed as circuit compilation can commute the permutations through the rest of the circuit. We additionally show that individually addressable interblock $\overline{\mathrm{CNOT}}$ gates, and arbitrary $\overline{\mathrm{CNOT}}$ circuits spanning multiple codeblocks, enjoy benefits from the phantomness of CSS codes and can too be performed efficiently, as has been formalized in \cref{thm:phantom_css_interblock_cnot_circuits,remark:phantom_css_interblock_cnot_circuits} in the main text. Here we give the proofs of these results. 

The central idea is to interleave in-block permutation $\overline{\mathrm{CNOT}}$ gates between interblock transversal $\overline{\mathrm{CNOT}}$ gates~\cite{grassl2013leveraging} to obtain arbitrary interblock $\overline{\mathrm{CNOT}}$ connectivity. This decomposition is applied recursively from large to small scale on the desired logical circuit, with suitable parallelization of the required transversal $\overline{\mathrm{CNOT}}$ gates to control depth. As our codes are phantom and the qubit permutations implementing in-block $\overline{\mathrm{CNOT}}$ gates need not be physically performed, the only contribution to physical depth comes from the transversal $\overline{\mathrm{CNOT}}$s, which are depth one each.

\begin{lemma}[Efficient arbitrary $\overline{\mathrm{CNOT}}$ circuits between two CSS phantom codeblocks]
\label{lemma:phantom_css_interblock_cnot_circuits_two_codeblocks}
    Arbitrary $\overline{\mathrm{CNOT}}$ circuits between two codeblocks of a CSS phantom code can be performed in physical depth at most four, up to a residual permutation of logical qubits. This depth reduces to at most two whilst preserving the ordering of logical qubits when the $\overline{\mathrm{CNOT}}$ gates in the circuit are unidirectional.
\end{lemma}

\begin{proof}
    We consider two codeblocks of an $\db{n,k,d}$ CSS phantom code, for a total of $2k$ logical qubits. Recall that, in the symplectic representation, a $\mathrm{CNOT}$ circuit $C$ on $2k$ qubits is described by a symplectic matrix $F(C) \in \Sp(4k, \mathbb{F}_2)$ with a block-diagonal structure $F(C) = \diag(X, \invtrans{X})$ for $X \in \GL(2k, \mathbb{F}_2)$---see \cref{app:basics/symplectic} for a review. Henceforth we deal only the top-left block $X$, which contains all degrees of freedom characterizing the circuit. We consider cases of increasing complexity:
    \begin{itemize}
        
        \item $X$ is upper-block-triangular,
        \begin{equation}
            X = \left(\begin{array}{c|c} 
                    \mathbb{I} & U \\ 
                    \hline 
                    0 & \mathbb{I} 
                \end{array}\right),
        \end{equation}
        where each block is $k \times k$. First, if $U$ is invertible, then the circuit can be implemented using a single transversal $\overline{\mathrm{CNOT}}$ gate between the codeblocks, sandwiched by in-block $\overline{\mathrm{CNOT}}$ gates on the control codeblock:
        \begin{equation}
        \label{eq:cb_from_ib}
            X
            = \left(\begin{array}{c|c} 
                \mathbb{I} & U \\ 
                \hline 
                0 & \mathbb{I} 
            \end{array}\right)
            = \underbrace{\left(\begin{array}{c|c} 
                U & 0 \\ 
                \hline 
                0 & \mathbb{I} 
            \end{array}\right)}_{\text{in-block}}
            \underbrace{\left(\begin{array}{c|c} 
                \mathbb{I} & \mathbb{I} \\ 
                \hline 
                0 & \mathbb{I} 
            \end{array}\right)}_{\text{transversal}}
            \underbrace{\left(\begin{array}{c|c} 
                U^{-1} & 0 \\ 
                \hline 
                0 & \mathbb{I} 
            \end{array}\right)}_{\text{in-block}}.
        \end{equation}
        
        Otherwise, if $U$ is not invertible, by Ref.~\cite[Lemma~IX.7]{malcolm2025computing}, $U$ can be written as the sum of two $\GL(k, \mathbb{F}_2)$ matrices $U_1, U_2$. The block-triangular structure of these matrices enables us to rewrite the sum as a product, and we use the circuit decomposition of \cref{eq:cb_from_ib} for each of the terms,
        \begin{equation}\begin{split}
            \label{eq:cb_additive}
            X
            = \left(\begin{array}{c|c} 
                \mathbb{I} & U_1 + U_2 \\ 
                \hline 
                0 & \mathbb{I} 
            \end{array}\right)
            &= \left(\begin{array}{c|c} 
                \mathbb{I} & U_1 \\ 
                \hline 
                0 & \mathbb{I} 
            \end{array}\right)
            \left(\begin{array}{c|c} 
                \mathbb{I} & U_2 \\ 
                \hline 
                0 & \mathbb{I} 
            \end{array}\right) \\
            &= \underbrace{\left(\begin{array}{c|c} 
                U_1 & 0 \\ 
                \hline 
                0 & \mathbb{I} 
            \end{array}\right)}_{\text{in-block}}
            \underbrace{\left(\begin{array}{c|c} 
                \mathbb{I} & \mathbb{I} \\ 
                \hline 
                0 & \mathbb{I} 
            \end{array}\right)}_{\text{transversal}}
            \underbrace{\left(\begin{array}{c|c} 
                U_1^{-1} & 0 \\ 
                \hline 
                0 & \mathbb{I} 
            \end{array}\right)}_{\text{in-block}}
            \underbrace{\left(\begin{array}{c|c} 
                U_2 & 0 \\ 
                \hline 
                0 & \mathbb{I} 
            \end{array}\right)}_{\text{in-block}}
            \underbrace{\left(\begin{array}{c|c} 
                \mathbb{I} & \mathbb{I} \\ 
                \hline 
                0 & \mathbb{I} 
            \end{array}\right)}_{\text{transversal}}
            \underbrace{\left(\begin{array}{c|c} 
                U_2^{-1} & 0 \\ 
                \hline 
                0 & \mathbb{I} 
            \end{array}\right)}_{\text{in-block}}.
        \end{split}\end{equation}
        
        Therefore, the circuit can be implemented using two transversal $\overline{\mathrm{CNOT}}$ gates between the codeblocks, interleaved with in-block $\overline{\mathrm{CNOT}}$ gates on the control codeblock.

        \item $X$ is lower-block-triangular. The argument is analogous as for the upper-block-triangular case. Circuits with unidirectional $\overline{\mathrm{CNOT}}$ gates correspond precisely to either a lower- or upper-block-triangular $X$, hence proving the second part of \cref{lemma:phantom_css_interblock_cnot_circuits_two_codeblocks}.

        \item $X$ is not block-triangular. In this case, we invoke the block-PLDU decomposition\footnote{For invertible $X$, we can first apply a permutation matrix $P$, such that $P^{-1}X$ has an invertible upper-left quadrant $A$, then use the block-LDU decomposition $\left(\begin{array}{c|c} A & B\\ \hline C & D \end{array}\right)=\left(\begin{array}{c|c} I & 0\\ \hline CA^{-1} & I \end{array}\right)\left(\begin{array}{c|c} A & 0\\ \hline 0 & D-CA^{-1}B \end{array}\right)\left(\begin{array}{c|c} I & A^{-1}B\\ \hline 0 & I \end{array}\right)$.}:
        \begin{equation}
        \label{eq:PLDU}
            X 
            = P
            \left(\begin{array}{c|c} 
                \mathbb{I} & 0 \\ 
                \hline 
                L & \mathbb{I} 
            \end{array}\right)
            \underbrace{\left(\begin{array}{c|c} 
                C_1 & 0 \\ 
                \hline 
                0 & C_2 
            \end{array}\right)}_{\text{in-block}}
            \left(\begin{array}{c|c} 
                \mathbb{I} & U \\ 
                \hline 
                0 & \mathbb{I} 
            \end{array}\right),
        \end{equation}
        where $P$ is a permutation matrix, $L, U \in \mathbb{F}_2^{k \times k}$, and $C_1,C_2\in\GL(k,\mathbb{F}_2)$. The middle $\diag(C_1, C_2)$ term corresponds to in-block $\overline{\mathrm{CNOT}}$ gates on the two codeblocks, while the lower- and upper-block-triangular terms can be treated as described above. Thus, at most four transversal $\overline{\mathrm{CNOT}}$ gates between the codeblocks are needed, up to a permutation of logical qubits $P$.
        
    \end{itemize}
\end{proof}

Now we examine $\overline{\mathrm{CNOT}}$ circuits across a larger number of codeblocks, which is the general setting considered in \cref{thm:phantom_css_interblock_cnot_circuits} of the main text. We reproduce the statement here for reference and provide a proof.

\MainPropositionInterblockCNOT*

\begin{proof}
    There are $2^a$ codeblocks of an $\db{n, k, d}$ CSS phantom code, amounting to $2^a k$ logical qubits in total. Like before, let $X \in \GL(2^a k, \mathbb{F}_2)$ be the top-left block of the symplectic matrix representing the $\mathrm{CNOT}$ circuit. We invoke the block-PLDU decomposition on $X$ as in \cref{eq:PLDU}, and further decompose $C_1, C_2 \in \GL(2^{a - 1} k, \mathbb{F}_2)$ into PLDU forms, which involve permutation matrices $P_1, P_2$. We can then pull $P_1, P_2$ to the left,
    \begin{equation}
        X
        = P
        \left(\begin{array}{c|c} 
            \mathbb{I} & 0 \\ 
            \hline 
            L & \mathbb{I}
        \end{array}\right)
        \left(\begin{array}{c|c} 
            P_1 & 0 \\ 
            \hline 
            0 & P_2 
        \end{array}\right)
        \cdots \
        = P
        \left(\begin{array}{c|c} 
            P_1 & 0 \\ 
            \hline 
            0 & P_2 
        \end{array}\right)
        \left(\begin{array}{c|c} 
            \mathbb{I} & 0 \\ 
            \hline 
            P_2^{-1}LP_1 & \mathbb{I}
        \end{array}\right)
        \cdots \ .
    \end{equation}

    Merging $P$ and $\diag(P_1, P_2)$ leads to another permutation matrix. Repeating this process further, we obtain the following decomposition:
    \begin{equation}
    \label{eq:recursive_PLDU}
        X
        = P
        \left(\begin{array}{c|c} 
            \mathbb{I} & 0 \\ 
            \hline 
            L_0 & \mathbb{I} 
        \end{array}\right)
        \Biggl(\begin{array}{c|c} 
            \begin{smallmatrix}
                \mathbb{I} & 0 \\ 
                L_{1,1} & \mathbb{I} \\[1pt]
            \end{smallmatrix} & 0 \\ 
            \hline
            0 & \begin{smallmatrix}
                \\[0.5pt]
                \mathbb{I} & 0 \\ 
                L_{1,2} & \mathbb{I}
            \end{smallmatrix} 
        \end{array}\Biggr)
        \cdots
        \underbrace{\left(\begin{smallmatrix}
            C_{1} &&& \\ 
            & C_2 && \\ 
            && \sddots & \\ 
            &&& C_{2^{a}} 
        \end{smallmatrix}\right)}_{\text{in-block}} \\
        \cdots
        \Biggl(\begin{array}{c|c} 
            \begin{smallmatrix}
                \mathbb{I} & U_{1,1} \\ 
                0 & \mathbb{I} \\[1pt]
            \end{smallmatrix} & 0 \\ 
            \hline 
            0 & \begin{smallmatrix}
                \\[0.5pt]
                \mathbb{I} & U_{1,2} \\ 
                0 & \mathbb{I}
            \end{smallmatrix}
        \end{array}\Biggr)
        \left(\begin{array}{c|c} 
            \mathbb{I} & U_0 \\ 
            \hline 0 & \mathbb{I} 
        \end{array}\right),
    \end{equation}
    where $L_{i,j},U_{i,j}$ for $j \in [2^i]$ are binary square matrices of size $2^{a-i-1} k \times 2^{a-i-1} k$, and $C_i \in \GL(k,\mathbb{F}_2)$ for $i \in [2^a]$. The $C_i$ term in the middle corresponds to in-block $\overline{\mathrm{CNOT}}$ gates on the codeblocks independently. It then remains to implement the lower- and upper-block-triangular block matrices in this decomposition. To reduce depth, we seek to parallelize the required interblock $\overline{\mathrm{CNOT}}$ interactions as far as possible. This can be done through cyclic scheduling~\cite{malcolm2025computing}, which amounts to decomposing the block matrix into a sum of matchings on the complete directed graph of codeblocks. All $\overline{\mathrm{CNOT}}$ interactions in a matching are on disjoint pairs of codeblocks, so each matching is fully parallelizable; the overall depth of the implementation is then determined by the number of matchings (i.e.~rounds) required. As a concrete example, a cyclic scheduling for an upper-block-triangular term where $V$ is $4k \times 4k$ is
    \begin{equation}\begin{split}
        \left(\begin{array}{c|c} 
            \mathbb{I} & V \\ 
            \hline 0 & \mathbb{I} 
        \end{array}\right)
        &= \left(\begin{array}{c|c} 
            \mathbb{I} & V^{(1)} + V^{(2)} + V^{(3)} + V^{(4)} \\ 
            \hline 0 & \mathbb{I} 
        \end{array}\right)
        = \left(\begin{array}{c|c} 
            \mathbb{I} & V^{(1)} \\ 
            \hline 0 & \mathbb{I} 
        \end{array}\right)
        \left(\begin{array}{c|c} 
            \mathbb{I} & V^{(2)} \\ 
            \hline 0 & \mathbb{I} 
        \end{array}\right)
        \left(\begin{array}{c|c} 
            \mathbb{I} & V^{(3)} \\ 
            \hline 0 & \mathbb{I} 
        \end{array}\right)
        \left(\begin{array}{c|c} 
            \mathbb{I} & V^{(4)} \\ 
            \hline 0 & \mathbb{I} 
        \end{array}\right),
    \end{split}\end{equation}
    where
    \begin{equation}\begin{split}
        V
        &= \begin{psmallmatrix} 
            A_{1,1} & A_{1,2} & A_{1,3} & A_{1,4} \\
            A_{2,1} & A_{2,2} & A_{2,3} & A_{2,4} \\
            A_{3,1} & A_{3,2} & A_{3,3} & A_{3,4} \\
            A_{4,1} & A_{4,2} & A_{4,3} & A_{4,4} \\
        \end{psmallmatrix}
        = \underbrace{\begin{psmallmatrix} 
            A_{1,1} & 0 & 0 & 0 \\
            0 & A_{2,2} & 0 & 0 \\
            0 & 0 & A_{3,3} & 0 \\
            0 & 0 & 0 & A_{4,4} \\
        \end{psmallmatrix}}_{V^{(1)}}
        + \underbrace{\begin{psmallmatrix} 
            0 & A_{1,2} & 0 & 0 \\
            A_{2,1} & 0 & 0 & 0 \\
            0 & 0 & 0 & A_{3,4} \\
            0 & 0 & A_{4,3} & 0 \\
        \end{psmallmatrix}}_{V^{(2)}}
        + \underbrace{\begin{psmallmatrix} 
            0 & 0 & A_{1,3} & 0 \\
            0 & 0 & 0 & A_{2,4} \\
            A_{3,1} & 0 & 0 & 0 \\
            0 & A_{4,2} & 0 & 0 \\
        \end{psmallmatrix}}_{V^{(3)}}
        + \underbrace{\begin{psmallmatrix} 
            0 & 0 & 0 & A_{1,4} \\
            0 & 0 & A_{2,3} & 0 \\
            0 & A_{3,2} & 0 & 0 \\
            A_{4,1} & 0 & 0 & 0 \\
        \end{psmallmatrix}}_{V^{(4)}}.
    \end{split}\end{equation}
    
    Each $V^{(j)}$ round is implementable with fully parallelized unidirectional $\overline{\mathrm{CNOT}}$ circuits between disjoint pairs of codeblocks. In general, a $V$ that is $2^b k \times 2^b k$ in size requires $2^b$ rounds in its cyclic scheduling.

    In the decomposition of \cref{eq:recursive_PLDU}, $U_0$ is $2^{a-1} k \times 2^{a-1} k$ in size and therefore its cyclic scheduling takes $2^{a-1}$ rounds. Each round contains only lower- or upper-block-triangular submatrices, corresponding to unidirectional $\overline{\mathrm{CNOT}}$ circuits between disjoint pairs of codeblocks, and so takes depth at most two to implement by \cref{lemma:phantom_css_interblock_cnot_circuits_two_codeblocks} with parallelization. For the next term, the submatrices $U_{1,1}, U_{1,2}$ are $2^{a-2} k \times 2^{a-2} k$ in size each and their cyclic schedules can be performed in parallel; only $2^{a-2}$ rounds are needed in their schedules. Continuing this reasoning to all the upper-block-triangular matrices in \cref{eq:recursive_PLDU}, we find that the total number of rounds needed is at most $2^{a-1}+2^{a-2}+\cdots+1 = 2^a-1$, for a total depth of at most $2 (2^a-1)$. The same is true for the lower-block-triangular matrices. Therefore, the overall depth required is at most $4 (2^a-1)$ up to a residual permutation of logical qubits $P$.

    When the desired $\overline{\mathrm{CNOT}}$ circuit contains only unidirectional $\overline{\mathrm{CNOT}}$ gates, $X$ is itself either lower- or upper-block-triangular. Then the decomposition of \cref{eq:recursive_PLDU} contains only either the lower- or upper-block-triangular half, and there is no permutation $P$ present. Therefore, the overall depth required is at most $2 (2^a-1)$ whilst maintaining ordering of the logical qubits.
\end{proof}

While the worst-case depth up to residual permutation $P$ is $k$-independent, in cases where the permutation must be performed physically, $P$ could take a worst-case depth proportional to $k$ to implement. This has been stated in \cref{remark:phantom_css_interblock_cnot_circuits} of the main text, reproduced below for reference.

\MainRemarkInterblockCNOT*

\begin{proof}
We first decompose the permutation of logical qubits into two layers of involutions~\cite[Thm.~XI.4]{malcolm2025computing}. An involution layer can be understood as a product of $\overline{\mathrm{SWAP}}$s on disjoint supports that can be performed in parallel. The $\overline{\mathrm{SWAP}}$s that take place within the same codeblock can be compiled away for phantom codes (i.e.~by writing $\overline{\mathrm{SWAP}}_{ij} = \overline{\mathrm{CNOT}}_{ij} \overline{\mathrm{CNOT}}_{ji} \overline{\mathrm{CNOT}}_{ij}$ and the qubit permutations implementing the in-block $\overline{\mathrm{CNOT}}$ gates are pulled through the physical circuit), so we henceforth deal only with interblock $\overline{\mathrm{SWAP}}$s.

For each involution layer, we construct a graph where each codeblock serves as a vertex. We add an edge between two vertices if there exists interblock $\overline{\mathrm{SWAP}}$ gates between the two corresponding codeblocks\footnote{This is different from the multigraph construction in Ref.~\cite[Thm.~XI.4]{malcolm2025computing} where an edge is added for \emph{every} interblock $\overline{\mathrm{SWAP}}$ gate. A multigraph of degree $k$ could take up to $3k/2$ colours for edge colouring using Shannon's theorem. The reason that Ref.~\cite{malcolm2025computing} uses a multigraph is that a single $\overline{\mathrm{CNOT}}$ or $\overline{\mathrm{SWAP}}$ gate between two of SHYPS codeblocks is of constant depth; but were they to use the simple graph construction employed here, there will in general be arbitrary $\overline{\mathrm{CNOT}}$ or $\overline{\mathrm{SWAP}}$ circuits between two codeblocks, which would cost $\mathcal{O}(k)$ depth in the worst case on SHYPS codes.}. This graph has degree at most $k$, hence an edge-colouring of the graph requires at most $k+1$ colours by Vizing's theorem. An edge between two vertices represent a $\overline{\mathrm{SWAP}}$ circuit between two codeblocks; edges labelled by the same colour can be performed in parallel. For each colour, we use \cref{lemma:phantom_css_interblock_swap_circuits_two_codeblocks} to exactly implement the $\overline{\mathrm{SWAP}}$ circuits in depth four.

Therefore, the total depth required for an arbitrary permutation is $2 \times (k+1) \times 4 = 8k + 8$; this number is independent of the number of codeblocks.
\end{proof}

\begin{lemma}
\label{lemma:phantom_css_interblock_swap_circuits_two_codeblocks}
Arbitrary $\overline{\mathrm{SWAP}}$ circuits between two codeblocks of a CSS phantom code can be performed in physical depth at most four.
\end{lemma}

\begin{proof}
We start with a case wherein a $\overline{\mathrm{SWAP}}$ is to be implemented between the first logical qubits of the two codeblocks. A naive construction, for example, decomposes the $\overline{\mathrm{SWAP}}$ into three $\overline{\mathrm{CNOT}}$s and implements each using two transversal $\overline{\mathrm{CNOT}}$s, for a total depth of six. Here we show that a depth-four implementation is possible. We rely on the following decomposition

\renewcommand{\arraystretch}{1.3}
\vspace{-12pt}
\begin{equation}
\label{eq:swap_four_CX}
    \Biggl(\begin{array}{c|c} 
    \begin{smallmatrix}0&0\\0&1\end{smallmatrix} & \begin{smallmatrix}1&0\\0&0\end{smallmatrix} \\ \hline 
    \begin{smallmatrix}1&0\\0&0\end{smallmatrix} & 
    \begin{smallmatrix}0&0\\0&1\end{smallmatrix} 
    \end{array}\Biggr)
    =
    \overset{\raisebox{4pt}{\scriptsize 1}}{\underbrace{\Biggl(\begin{array}{c|c} 
    \begin{smallmatrix}1&\\&1\end{smallmatrix} & \begin{smallmatrix}1&\\&1\end{smallmatrix} \\ \hline  & 
    \begin{smallmatrix}1&\\&1\end{smallmatrix} 
    \end{array}\Biggr)}_{\text{transversal}}}
    \overset{\raisebox{4pt}{\scriptsize 2}}{\underbrace{\Biggl(\begin{array}{c|c} 
    \begin{smallmatrix}1&1\\1&0\end{smallmatrix} &  \\ \hline  & 
    \begin{smallmatrix}1&0\\1&1\end{smallmatrix} 
    \end{array}\Biggr)}_{\text{in-block}}}
    \overset{\raisebox{4pt}{\scriptsize 3}}{\underbrace{\Biggl(\begin{array}{c|c} 
    \begin{smallmatrix}1&\\&1\end{smallmatrix} &  \\ \hline   
    \begin{smallmatrix}1&\\&1\end{smallmatrix} &
    \begin{smallmatrix}1&\\&1\end{smallmatrix} 
    \end{array}\Biggr)}_{\text{transversal}}}
    \overset{\raisebox{4pt}{\scriptsize 4}}{\underbrace{\Biggl(\begin{array}{c|c} 
    \begin{smallmatrix}1&1\\1&0\end{smallmatrix} &  \\ \hline  & 
    \begin{smallmatrix}1&0\\1&1\end{smallmatrix} 
    \end{array}\Biggr)}_{\text{in-block}}}
    \overset{\raisebox{4pt}{\scriptsize 5}}{\underbrace{\Biggl(\begin{array}{c|c} 
    \begin{smallmatrix}1&\\&1\end{smallmatrix} &  \\ \hline   
    \begin{smallmatrix}1&\\&1\end{smallmatrix} &
    \begin{smallmatrix}1&\\&1\end{smallmatrix} 
    \end{array}\Biggr)}_{\text{transversal}}}
    \overset{\raisebox{4pt}{\scriptsize 6}}{\underbrace{\Biggl(\begin{array}{c|c} 
    \begin{smallmatrix}1&1\\1&0\end{smallmatrix} &  \\ \hline  & 
    \begin{smallmatrix}1&0\\0&1\end{smallmatrix} 
    \end{array}\Biggr)}_{\text{in-block}}}
    \overset{\raisebox{4pt}{\scriptsize 7}}{\underbrace{\Biggl(\begin{array}{c|c} 
    \begin{smallmatrix}1&\\&1\end{smallmatrix} & \begin{smallmatrix}1&\\&1\end{smallmatrix} \\ \hline  & 
    \begin{smallmatrix}1&\\&1\end{smallmatrix} 
    \end{array}\Biggr)}_{\text{transversal}}},
\end{equation}
\renewcommand{\arraystretch}{1.0}
which corresponds to the circuit
\begin{equation*}
    \begin{quantikz}[row sep={0.4cm,between origins}, column sep={0.5cm,between origins}]
        \lstick[2]{\small codeblock 1}
        & \swap{2} & 
        \\
        & & 
        \\[0.1cm]
        \lstick[2]{\small codeblock 2}
        & \targX{} & 
        \\
        & &
    \end{quantikz}
    \ = \
    \begin{quantikz}[row sep={0.4cm,between origins}, column sep={0.4cm,between origins}]
        & \ctrl{2} &    & & \ctrl{1} & \swap{1}    & & \targ{} &       & & \ctrl{1} & \swap{1}    & & \targ{} &       & & \ctrl{1} & \swap{1}    & & \ctrl{2} &      &
        \\
        & & \ctrl{2}    & & \targ{} & \targX{}     & & & \targ{}       & & \targ{} & \targX{}     & & & \targ{}       & & \targ{} & \targX{}     & & & \ctrl{2}      &
        \\[0.1cm]
        & \targ{} &     & & \targ{}   &  & &  \ctrl{-2} &     & & \targ{}     & & & \ctrl{-2} &     & & &             & & \targ{} &       &
        \\
        & & \targ{}     & & \ctrl{-1} &   & & &  \ctrl{-2}     & & \ctrl{-1}   & & & & \ctrl{-2}     & & &             & & & \targ{}       &
    \end{quantikz}.
\end{equation*}

To avoid clutter we have expressed the above for a $k = 2$ code, but the decomposition trivially generalizes for arbitrary $k$: matrices $1,3,5,7$ remain as transversal $\overline{\mathrm{CNOT}}$s, and the in-block permutation $\overline{\mathrm{CNOT}}$ (matrices $2,4,6$) receive additional identity block matrices for the remaining $k - 2$ logical qubits.

Next, we note that swapping all $k$ logical qubits between the two codeblocks amounts to completely swapping the two codeblocks, which can be performed by relabelling the codeblocks, plus the additional freedom of depth-zero in-block $\overline{\mathrm{SWAP}}$ operations. Therefore, up to codeblock relabelling, we only need to swap at most $\lfloor k/2 \rfloor$ logicals between the two codeblocks. Let $\ell \leq \lfloor k/2 \rfloor$ be the number of $\overline{\mathrm{SWAP}}$s we need to implement.

Since in-block $\overline{\mathrm{SWAP}}$s are free, without loss of generality, we can assume the action to be implemented is to swap the logical qubit $2i-1$ on the first codeblock with that on the second, for all $1\le i\le \ell$. The swapping of the logical qubit $2i-1$ uses the logical qubit $2i$ as in-block ancillary aid, as appears in matrices $2,4,6$ of \cref{eq:swap_four_CX}; the logical qubits $2\ell+1,\dots,k$ are left invariant. This finishes our constructive proof for a depth-four implementation.
\end{proof}

\subsection{Additional permutation logical gate and phantom code properties}
\label{app:basics/properties}

Here we discuss various additional properties of permutation logical gates and phantom codes.

\subsubsection{Structure of permutations implementing logical gates on stabilizer codes}
\label{app:basics/properties/permutations}

First, \cref{prop:permutation_cnot_swap_even_period,prop:permutation_cnot_swap_single_layer_swap} address the form of permutations that implement involutory logical gates (i.e.~a gate which squares to the logical identity) on stabilizer codes. 

\begin{proposition}[Even period of qubit permutations implementing involutory logical gates on stabilizer codes]
\label{prop:permutation_cnot_swap_even_period}
    Any permutation of physical qubits on a stabilizer code implementing an involutory logical gate must be of even period.
\end{proposition}

\begin{proof}
    Suppose otherwise, and let the odd period of the qubit permutation $\pi$ be $2 p + 1$. We denote the unitary representation of $\pi$ as $U_\pi$, such that $\pi^{2p + 1}$ is the trivial permutation and $U_\pi^{2p + 1} = \mathbb{I}$. Consider any code state $\ket{\psi}$. First, as the logical gate implemented by $\pi$ is of nontrivial logical action, it must be that $U_\pi \ket{\psi} \neq e^{i \phi} \ket{\psi}$ for any global phase $\phi$. But the logical gate is involutory and square to a logical identity, so we must have $U_\pi^2 \ket{\psi} = e^{i \theta} \ket{\psi}$ for some $\theta$. Now $U_\pi^{2p + 1} \ket{\psi} = U_\pi (U_\pi^2)^p \ket{\psi} = e^{i p \theta} U_\pi \ket{\psi} = \ket{\psi}$, a contradiction. So the period of $\pi$ must be even.
\end{proof}

\begin{proposition}[Single-layer qubit swaps for involutory logical gates on stabilizer codes]
\label{prop:permutation_cnot_swap_single_layer_swap}
    Suppose a stabilizer code does not support any nontrivial permutation of qubits with trivial logical action (i.e.~implements the logical identity). Then all permutation involutory logical gates on the code are implementable through period-two permutations (i.e.~qubit involutions), which correspond to single parallelizable layers of qubit swaps.
\end{proposition}

\begin{proof}
    Suppose there is an involutory logical gate on the code whose implementation requires a qubit permutation $\pi$ of period $p > 2$, such that $\pi^2$ is a nontrivial permutation. But the logical gate is involutory, so $\pi^2$ implements the logical identity but is a nontrivial permutation of physical qubits, a contradiction to the premise.
\end{proof}

\cref{prop:permutation_cnot_swap_even_period,prop:permutation_cnot_swap_single_layer_swap} apply to permutation $\overline{\mathrm{CNOT}}$ and $\overline{\mathrm{SWAP}}$ gates, which are involutory, on phantom codes. As we propose compiling away all physical qubit permutations by pulling them through the physical circuit, the structure of the permutations does not typically matter in practical use. However, in cases where one considers performing some of the qubit permutations physically, perhaps due to circuit compilation or hardware constraints, permutations that take the form of depth-one parallelizable qubit swaps as discussed in \cref{prop:permutation_cnot_swap_single_layer_swap} may be of interest.

\begin{remark}
    There exist stabilizer codes for which period ${>} \, 2$ qubit permutations are required to implement an involutory logical gate set.
\end{remark}

\begin{proof}
    We give the smallest example of a CSS code that requires period ${>} \, 2$ qubit permutations to implement a $\overline{\mathrm{CNOT}}$ gate. Restricting to period-two qubit permutations entails the loss of the $\overline{\mathrm{CNOT}}$. This minimal example is an $\db{10,2,3}$ code defined by stabilizer generator matrices:
    \begin{equation}
        H_\mathrm{x} = \left[\begin{array}{ccccccccccc}
            1 & 1 & 1 & 1 & 0 & 0 & 0 & 0 & 0 & 0 \\
            0 & 0 & 1 & 0 & 1 & 1 & 1 & 0 & 0 & 0 \\
            0 & 1 & 0 & 0 & 1 & 0 & 0 & 1 & 1 & 0 \\
            0 & 0 & 0 & 1 & 0 & 0 & 1 & 1 & 0 & 1
        \end{array}\right],
        \qquad
        H_\mathrm{z} = \left[\begin{array}{ccccccccccc}
            1 & 1 & 0 & 0 & 1 & 1 & 0 & 0 & 0 & 0 \\
            0 & 0 & 1 & 1 & 1 & 0 & 0 & 1 & 0 & 0 \\
            0 & 0 & 0 & 0 & 0 & 1 & 1 & 1 & 1 & 0 \\
            0 & 1 & 0 & 1 & 0 & 0 & 0 & 0 & 1 & 1
        \end{array}\right].
    \end{equation}

    This is the smallest-$n$ example found via exhaustive code enumeration; determination of the permutation logical gate set of the code was by SAT solving (see \cref{app:enumeration}). Our search also uncovered $\db{n,k}$ CSS codes for which restricting to period-two qubit permutations incurs a reduction in the size of their permutation logical gate set for all $n=11,12,13,14$ and $k=2,3$, and all $n=12,13,14$ for $k=4$; presumably examples likewise exist at higher $k$ but we did not cover those $k$ in our search.
\end{proof}

\subsubsection{Permutation \texorpdfstring{$\overline{\mathrm{CNOT}}$}{CNOT} gate-set-preserving logical basis changes on phantom codes}
\label{app:basics/properties/permutation_preserving_logical_basis_changes}

A phantom code realizes a complete set of individually addressable $\overline{\mathrm{CNOT}}$ logical gates via qubit permutations for some logical basis. A natural question to be asked is what other logical bases can be used on the code so that the \emph{same set} of qubit permutations likewise \emph{generate} a complete set of individually addressable $\overline{\mathrm{CNOT}}$ gates.

\begin{definition}[Phantom-gate-set-preserving logical basis change]
\label{definition:basis_change}
Let $M$ be a set of physical qubit permutations that realize the logical gate set $G = \{\overline{\mathrm{CNOT}}_{a b}: (a, b) \in [k]^2, a \neq b\}$ on a phantom stabilizer code for some choice of logical basis, and $\mathcal{M}$ be the group generated by $M$. A phantom-gate-preserving logical basis change is a transformation of the logical basis, so that a set of permutations $M' \subset \mathcal{M}$ still realizes $G$ on the code in the transformed logical basis.
\end{definition}

Note that, in \cref{definition:basis_change}, the logical action of each permutation in $M$ may be different in the transformed logical basis---i.e.~a permutation may now implement a different $\overline{\mathrm{CNOT}}$ gate or a product of $\overline{\mathrm{CNOT}}$ gates---but we demand that the complete set of $\overline{\mathrm{CNOT}}$ gates on the code can still be implemented by products of permutations in $M$.

We give two immediate remarks. First, phantom-gate-set-preserving logical basis changes, by their nature, form a group. Second, the notion of a phantom-gate-set-preserving logical basis change is stricter than that of a \emph{phantomness-preserving} logical basis change---a logical basis change of the first type is also necessarily one of the second type, but not the other way around. The reason is that a code that is phantom in a logical basis could still be phantom in a different logical basis using a completely \emph{different} set of qubit permutations, not in the span of the original, to implement the $\overline{\mathrm{CNOT}}$ gate set, but this case is not considered in \cref{definition:basis_change}. To re-iterate the point: phantom-gate-set-preserving logical basis changes as defined in \cref{definition:basis_change} suffice to preserve the phantom property of a stabilizer code but may not be necessary.

\Cref{definition:basis_change} places nontrivial restrictions on the admissible structure of logical basis change circuits, especially in regard to Hadamard-type gates in the circuits. We provide a definitive characterization of such circuits in \cref{prop:basis_change}.

\begin{theorem}
\label{prop:basis_change}
Phantom-gate-set-preserving logical basis-change circuits are generated by the following logical gates: (a)~for $k=2$ codes, $\smash{\overline{H}}^{\otimes 2}$, $\overline{\mathrm{CZ}}$, and $\overline{\mathrm{CNOT}}$ gates; or (b)~for $k\ge 3$ codes, $H^{\otimes k}$ and $\overline{\mathrm{CNOT}}$ gates.
\end{theorem}

\begin{proof}
    Let the logical basis change be $R \in \Sp(2k, \mathbb{F}_2)$ in the symplectic representation (see \cref{app:basics/symplectic}). Recall that $R$ being a symplectic matrix means $\trans{R} \Omega R = \Omega$. Let $Q = \smash{\binom{Q_\mathrm{x}}{Q_\mathrm{z}}}$ where $Q_\mathrm{x}, Q_\mathrm{z} \in \mathbb{F}_2^{k \times 2n}$ be a logical basis in which the code is phantom, and let $M$ be a set of permutations that implements the complete set of individually addressable $\overline{\mathrm{CNOT}}$ gates on the code. Denote by $\mathcal{M}$ the group generated by $M$. We assess whether there exists $M' \subset \mathcal{M}$ that implements the complete set of $\overline{\mathrm{CNOT}}$s in the transformed logical basis $Q' = R Q$. We consider different gate generators for $R$:
    \begin{itemize}
        
        \item $R$ is generated by $\overline{\mathrm{CNOT}}$ gates. This case subsumes the setting of \cref{prop:phantomness_css_logical_basis_independence}. Recall that $\mathrm{CNOT}$ circuits on $k$ qubits take the form $\diag(A, \invtrans{A})$ for $A \in \GL(k, \F_2)$ in the symplectic representation (see \cref{app:basics/symplectic}). Thus $R = \diag(C, \invtrans{C})$ for some $C \in \GL(k, \F_2)$. Let us consider a permutation $P^{(i)} \in M$, which implements a $\overline{\mathrm{CNOT}}$ gate described by $\diag(F_i, \invtrans{F_i})$ in the $Q$ logical basis. By \cref{eq:phantom_logic} of \cref{prop:phantom_gateset_symplectic_gen}, this means
        \begin{equation}
            Q \big( P^{(i)} \oplus P^{(i)} \big) \Omega \trans{Q}
            = \mqty[F_i & \\ & \invtrans{F_i}] \Omega.
        \end{equation}
    
        Now, in the $Q'$ logical basis,
        \begin{equation}\begin{split}
            Q' \big( P^{(i)} \oplus P^{(i)} \big) \Omega \trans{Q'}
            &= R Q \big( P^{(i)} \oplus P^{(i)} \big) \Omega \trans{Q} \trans{R} \\
            &= \mqty[C & \\ & \invtrans{C}] 
                \mqty[F_i & \\ & \invtrans{F_i}] \Omega 
                \mqty[\trans{C} & \\ & C^{-1}]
            = \mqty[C F_i C^{-1} & \\ & \invtrans{\left(C F_i C^{-1}\right)}] \Omega.
            \label{eq:basis_change_action_in_transformed_basis}
        \end{split}\end{equation}
    
        As the code is phantom in the $Q$ logical basis, $\spn(\{F_i\})$ for all the permutations in $M$ is precisely $\GL(k, \F_2)$, the space of $\overline{\mathrm{CNOT}}$ circuits on the $k$ logical qubits. Conjugating $\GL(k, \F_2)$ by a matrix $C \in \GL(k, \F_2)$ does not change the group. It must therefore be possible to implement a complete set of individually addressable $\overline{\mathrm{CNOT}}$ gates in the $Q'$ logical basis using some $M' \subset \mathcal{M}$, whose elements are generically products of permutations in $M$.
    
        \item In the same way as above, we can check that $R=\left[\begin{smallmatrix} & \mathbb{I} \\ \mathbb{I} & \end{smallmatrix}\right]$ corresponding to a $H^{\otimes k}$ circuit, and $R = \left[\begin{smallmatrix} \mathbb{I} & X \\ \mathbb{I} & 0 \end{smallmatrix}\right]$ for $k = 2$ where $\mathbb{I}, X$ are $2 \times 2$ Pauli matrices, are phantom-gate-set-preserving.
        
    \end{itemize}
    
    As phantom-gate-set-preserving logical basis changes form a group, all circuits generated by the gates above are phantom-gate-set-preserving. To show that only these circuits are admissible, consider an arbitrary basis change $R=\left[\begin{smallmatrix} A & B \\ C & D \end{smallmatrix}\right]$. Substituting into \cref{eq:basis_change_action_in_transformed_basis},
    \begin{equation}
        \begin{bmatrix} A & B \\ C & D \end{bmatrix} 
        \begin{bmatrix} F_i & \\ & \invtrans{F_i} \end{bmatrix}
        \Omega 
        \begin{bmatrix} \trans{A} & \trans{C} \\ \trans{B} & \trans{D} \end{bmatrix} =  
        \begin{bmatrix} 
            A F_i \trans{B} + B \invtrans{F_i} \trans{A} & A F_i \trans{D} + B \invtrans{F_i} \trans{C} 
            \\ 
            C F_i \trans{B} + D \invtrans{F_i} \trans{A} & C F_i \trans{D} + D \invtrans{F_i} \trans{C} \end{bmatrix},
    \end{equation}
    and we require this logical action to take the form 
    \begin{equation}
        \begin{bmatrix} F_{\sigma(i)} & \\  & \invtrans{F_{\sigma(i)}} \end{bmatrix} \Omega
        = \begin{bmatrix} & F_{\sigma(i)}  \\  \invtrans{F_{\sigma(i)}} & \end{bmatrix},
    \end{equation}
    for some $F_{\sigma(i)} \in \GL(k, \F_2)$. We treat the $k = 2$ and $k \geq 3$ cases separately:
    
    \begin{itemize}
        
        \item For $k=2$, we enumerate all admissible $R$. Concretely, we enumerate all $4\times 4$ binary matrix $R$ such that $\trans{R} \Omega R = \Omega$ and $A E \trans{B} + B \invtrans{E} \trans{A} = C E \trans{D} + D \invtrans{E} \trans{C} = 0$ for all $E\in\GL(k,\mathbb{F}_2)$.
        The results are 36 matrices in the six families below
        \begin{equation}
            \begin{bmatrix} E & \\ & E^{-\top} \end{bmatrix},\quad
            \begin{bmatrix} & E \\ E^{-\top} & \end{bmatrix},\quad
            \begin{bmatrix} E & EX \\ & E^{-\top} \end{bmatrix},\quad
            \begin{bmatrix} & E \\ E^{-\top} &E^{-\top} X \end{bmatrix},\quad
            \begin{bmatrix} E & \\ E^{-\top} X &E^{-\top} \end{bmatrix},\quad
            \begin{bmatrix} E & EX \\ E^{-\top} X & \end{bmatrix},
        \end{equation}
        where $E$ is any matrix in $\GL(k,\mathbb{F}_2)$ in each family. These $R$ are precisely circuits generated by $H \otimes H$, $\mathrm{CZ}$, and $\mathrm{CNOT}$ gates.

        \item For $k\geq3$, we note that the matrices $A$ and $B$ satisfies the conditions in \cref{lemma:basis_change_matrix}, to be proved later, and so do $C$ and $D$. (The full rank condition of the matrices is satisfied as $\trans{R} \Omega R = \Omega$.) Applying \cref{lemma:basis_change_matrix}, and noting that $A$ and $C$ cannot be both zero, $R$ can only be 
        \begin{equation}
            \begin{bmatrix} A & \\ & D \end{bmatrix},\quad
            \begin{bmatrix} & B \\ C & \end{bmatrix}.
        \end{equation}

        Further noting that $\trans{R} \Omega R = \Omega$, it must be that $D = \invtrans{A}$ in the first case with $A \in \GL(k,\mathbb{F}_2)$, or $C = \invtrans{B}$ with $B \in \GL(k,\mathbb{F}_2)$ in the second case. These correspond to a $\overline{\mathrm{CNOT}}$ circuit, and a $\overline{\mathrm{CNOT}}$ circuit followed by $\smash{\overline{H}}^{\otimes k}$, respectively.
    
    \end{itemize}
    
\end{proof}

For example, \cref{table:422_all_logical} shows the canonical choice of logical basis for the $\db{4,2,2}$ code.
However, we can choose a new basis corresponding to a basis-change circuit consisting of a single $\overline{\mathrm{CZ}}$ gate: $\overline{X}'_1=\overline{X_1Z_2}$, $\overline{X}'_2=\overline{X_2Z_1}$, $\overline{Z}'_1=\overline{Z_1}$, and $\overline{Z}'_2=\overline{Z_2}$.
Then, the same set of permutation still realizes all logical CNOT circuits.
We finish by proving a lemma required in the proposition above.

\begin{lemma} 
\label{lemma:basis_change_matrix}
Let $A$ and $B$ be binary matrices of size $k \times k$ with $k \ge 3$. 
If (1) the augmented matrix $[\begin{matrix} A \, | \, B \end{matrix}]$ has full rank, and
(2) for all $E \in \GL(k, \mathbb{F}_2)$, $A E \trans{B} + B \invtrans{E} \trans{A} = 0$;
then either $A = 0$ or $B = 0$.
\end{lemma}

\begin{proof}
When $E = \mathbb{I}$, condition (2) leads to $B \trans{A} = A \trans{B}$.
More generally, consider the elementary matrix $\mathbb{I} + E_{ij}$ for distinct indices $i \neq j$, where $E_{ij}$ has a 1 at position $(i,j)$ and zero everywhere else.
In $\mathbb{F}_2$, these matrices are self-inverse: $(\mathbb{I} + E_{ij})^2 = \mathbb{I} + 2 E_{ij} + E_{ij} E_{ij} = \mathbb{I}$.
Therefore, condition (2) leads to $A (\mathbb{I} + E_{ij}) \trans{B} = B \invtrans{(\mathbb{I} + E_{ij})} \trans{A} = B (\mathbb{I} + E_{ji}) \trans{A}$, i.e., $A E_{ij} \trans{B} = B E_{ji} \trans{A}$ after cancelling $B \trans{A} = A \trans{B}$. Recall that $E_{ij} = e_i \trans{e_j}$ where $e_i$ is a column vector with 1 at index $i$ and zero everywhere else. The condition becomes
\begin{equation}
\label{eq:dan_lemma_unequal}
    b_j \trans{a_i} = a_i \trans{b_j} \qquad \forall \, i,j\in\{1,...,k\}, \quad i\neq j,
\end{equation}
where $a_x$ and $b_x$ denote the $x^\text{th}$ columns of $A$ and $B$, respectively. Over $\mathbb{F}_2$, the equality $u \trans{v} = v \trans{u}$ holds iff $u$ and $v$ are linearly dependent, i.e., either $u=0$, $v=0$, or $u=v \neq 0$.

Suppose, for contradiction, that $A\neq 0$ and $B\neq 0$. Choose indices $i$ and $j$ such that $a_i\neq 0$ and $b_j\neq 0$.
\begin{itemize}
    \item Case of $i = j$. Then \cref{eq:dan_lemma_unequal} tells us that $a_\ell\in\spn(b_i) = \{0, b_i\}$ and $b_\ell\in\spn(a_i) = \{0, a_i\}$ for every $\ell\neq i$. This means $\colspan([\begin{matrix} A \, | \, B \end{matrix}]) \subseteq \spn(a_i,b_i)$, and accordingly $\rank([\begin{matrix} A \, | \, B \end{matrix}]) \leq \dim \spn(a_i,b_i) \leq 2$, contradicting condition (1) since $k \geq 3$.
    \item Case of $i\neq j$. Then \cref{eq:dan_lemma_unequal} constrains $a_i = b_j$. Moreover, we also have that $a_p\in\spn(b_j)=\spn(a_i)$ for every $p \neq j$, and $b_q\in\spn(a_i)$ for every $q \neq i$. Thus all columns of $[\begin{matrix} A \, | \, B \end{matrix}]$ lie in $\spn(a_i)$ except possibly the two columns $a_j$ and $b_i$. If $a_j = 0$, then $\colspan([\begin{matrix} A \, | \, B \end{matrix}]) \subseteq \spn(a_i,b_i)$ and $\rank([A\mid B]) \leq 2$. Otherwise, applying \cref{eq:dan_lemma_unequal} tells us that $b_i \in \spn(a_j)$. In this case, $\colspan([\begin{matrix} A \, | \, B \end{matrix}])\subseteq \spn(a_i,a_j)$ and again $\rank([\begin{matrix} A \, | \, B \end{matrix}]) \leq 2$. Either way, we find a contradiction with condition (1) since $k \geq 3$.
\end{itemize}

Thus $A$ and $B$ cannot both be nonzero. But they cannot both be zero because $[\begin{matrix} A \, | \, B \end{matrix}]$ is full-rank. So $A = 0$ or $B = 0$.
\end{proof}

\subsubsection{Limitations on strictly transversal logical gate sets on stabilizer codes}

That a stabilizer code supports a permutation logical gate places associated inherent limitations on the strictly transversal logical gates possible on the code. This has been formalized in \cref{thm:no_noncommuting_transversal_gates} in the main text. We reproduce the statement below for reference and provide a proof.

\MainTheoremTransversalGates*
\begin{proof}
    Suppose otherwise, that there exists an $\db{n, k, d}$ stabilizer code which supports $\overline{U}$ implemented by a qubit permutation $\pi$ and a strictly transversal $\overline{V}$ across $b$ codeblocks. As before, we denote the unitary representation of $\pi$ as $U_\pi$. The strict transversality of $\overline{V}$ entails an implementation $\overline{V} = \smash{\prod_{i = 1}^n W_{i i \cdots i}}$, where the physical gate $W$ act on corresponding qubits across the codeblocks. The inverse gate can then be performed as $\smash{\overline{V}}^\dag = \smash{\prod_{i = 1}^n W_{i i \cdots i}^\dag}$. We note that $U_\pi^{\otimes b}$ and $\smash{\prod_{i = 1}^n W_{i i \cdots i}}$ commute, as the same qubit permutation is applied to each codeblock.
    
    Consider the logical gate sequence $\overline{V} \smash{\overline{U}}^{\otimes b} \overline{V}^\dag$, which is implemented by $\smash{(\prod_{i = 1}^n W_{i i \cdots i})} U_\pi^{\otimes b} \smash{(\prod_{i = 1}^n W_{i i \cdots i}^\dag)} = U_\pi^{\otimes b} \smash{(\prod_{i = 1}^n W_{i i \cdots i})} \smash{(\prod_{i = 1}^n W_{i i \cdots i}^\dag)} = U_\pi^{\otimes b}$ at the physical level. But now we see that $U_\pi^{\otimes b}$ on the physical qubits implements both $\smash{\overline{V} \smash{\overline{U}}^{\otimes b} \overline{V}^\dag}$ and $\smash{\overline{U}}^{\otimes b}$, which have distinct logical actions and is therefore a contradiction when $U^{\otimes b}$ and $V$ do not commute.
\end{proof}

On CSS phantom codes, which support a complete set of individually addressable permutation $\overline{\mathrm{CNOT}}$ gates, \cref{thm:no_noncommuting_transversal_gates} rules out a large class of strictly transversal logical gates. This includes logical Hadamard ($\overline{H}$) or phase gates ($\overline{S}$), in-block and interblock $\overline{\mathrm{CZ}}$, and magic gates, for all or a subset of logical qubits. The sole exception is the $\smash{\overline{H}}^{\otimes 2} \overline{\mathrm{SWAP}}$ logical action, which are achievable transversally by $H^{\otimes 2}$, on some $k = 2$ codes (see e.g.~codes in \cref{tab:gate_for_phantom_codes}). Essentially, \cref{thm:no_noncommuting_transversal_gates} implies that, if other logical gates are available on phantom codes, then they cannot be strictly transversal: they must involve non-uniform operations on the qubits of the code or qubit permutations in addition to single-qubit operations.

\section{SAT preliminaries}
\label{app:sat_preliminaries}

In this work, we made extensive use of SAT solving to accomplish tasks such as checking the permutation gate sets supported by a code and thereby phantomness, and to discover phantom codes with desired parameters. To facilitate discussion of these methods, we first establish some background on SAT. 

To begin, as SAT instances are defined on Boolean variables subject to Boolean logic, while our problems and constraints are expressed mathematically in $\mathbb{F}_2$, a mapping is necessary. Conventionally, one associates the elements of $\mathbb{F}_2$ with Boolean values $0 \equiv \mathrm{FALSE}$, $1 \equiv \mathrm{TRUE}$, under which  addition over $\mathbb{F}_2$ corresponds to Boolean $\mathrm{XOR}$, $a + b \equiv a \lxor b$, and multiplication over $\mathbb{F}_2$ corresponds to Boolean $\mathrm{AND}$, $a + b \equiv a \land b$.

To declare $\mathbb{F}_2$ vectors or matrices as variables means to declare a Boolean variable for each of their entries. To restrict the matrix to a certain type, constraints are added into the SAT instance or structure in the matrix can be exploited. The most relevant cases are:
\begin{itemize}
    \item $A$ to be an $n \times n$ symmetric matrix in $\mathbb{F}_2$. Then $n (n + 1) / 2$ Boolean variables are declared to correspond to the upper triangular part of $A$; the lower triangular part is fixed by symmetry. 
    \item $A$ to be an $n \times n$ permutation matrix. Then $A \in \mathbb{F}_2^{n \times n}$ is declared, and the constraint that each row and column contains a single $1$ entry is imposed. Optionally, to restrict the permutation to be period-$2$, equivalently implementable by a single layer of swaps, either the constraint $A^2 = \mathbb{I}$ is directly imposed, or $A$ is declared to be symmetric and the row and column constraints imposed.
    \item $A \in \GL(n, \mathbb{F}_2)$. Then $A \in \mathbb{F}_2^{n \times n}$ and $A' \in \mathbb{F}_2^{n \times n}$ are declared and the constraint $A A' = \mathbb{I}$ is imposed (i.e.~$A'$ is the inverse of $A$). In contexts where row and column basis ordering does not matter, the diagonals of $A, A'$ can be set to be all ones.
    \item $A \in \Sp(2n, \mathbb{F}_2)$. Then $A \in \mathbb{F}_2^{2n \times 2n}$ is declared and the constraint $A \Omega \trans{A} = \Omega$ is imposed. 
\end{itemize}

The mapping described above converts $\mathbb{F}_2$ linear algebraic and arithmetic constraints prevalent in our problems to clauses involving $\mathrm{XOR}$s and $\mathrm{AND}$s. We convert these into a conjunctive normal form formulae through standard Tseitin transformation, which SAT solvers can read as input. Three computational outcomes are then possible: $\mathrm{SAT}$, which denotes that the formulae is satisfiable and a solution can be read off; $\mathrm{UNSAT}$, which denotes that the formulae is provably unsatisfiable and no solution exists; or that the solver does not finish.

We used the state-of-the-art SAT solver \textit{kissat}~\cite{biere2024cadical}, and the PySAT toolbox~\cite{ignatiev2018pysat} to aid in problem instance construction, throughout this study. We explored also the \textit{cryptominisat}~\cite{soos2009extending} solver, which natively supports $\mathrm{XOR}$ clauses and thus appears suitable for problems dominated by $\mathbb{F}_2$ matrix arithmetic constraints, but found no performance advantage. We were limited to $14$ days of solving time per SAT instance.

\section{Exhaustive code enumeration}
\label{app:enumeration}

We exhaustively enumerated all $n \leq 14$ CSS codes and filtered them for the phantom property. This amounted to $2.71 \times 10^{10}$ inequivalent codes in total, of which $132\,305$ are phantom codes with distance $d \geq 2$. A subset of these results were discussed in the main text, and we report detailed results in \cref{tab:enumeration_counts_phantom,tab:enumeration_counts_all_css}. As our approach extends in full generality to stabilizer codes, we first describe the general procedure and then explain the optimized specialization to CSS codes.

\subsection{Iteratively generating codes} 

Our general strategy builds on Sec.~6 of Ref.~\cite{cross2025small}. To enumerate $n$-qubit codes, we start from the trivial $\db{n,n,1}$ code whose stabilizer group is empty. Iteratively, for each $k = n - 1, \dots, 2, 1$, we build up the set of $\db{n, k}$ codes from $\db{n, k + 1}$ seed codes by taking each seed code and looping over all ways of appending a new (nontrivial) stabilizer generator into the stabilizer group. Accordingly, a valid $s$ must be linearly independent of and must commute with the stabilizers present. Equivalently, $s$ belongs to the logical group of the $\db{n, k + 1}$ seed code. We loop over the $2^{2 k + 2} - 1$ choices of $s$ to produce $\db{n, k}$ codes from each seed code. This procedure exhaustively generates every code, as every code can be obtained from some sequence of stabilizer generator additions starting from the trivial code.

\subsection{Efficient equivalence classification of codes via canonical forms} 

As the phantom property of codes is preserved under $\mathrm{\Pi H}$ equivalence (see \cref{prop:phantom_equivalence_invariance}), it suffices to store a single member of each $\mathrm{\Pi H}$ equivalence class of codes generated during enumeration. This avoids massive redundancy that would otherwise render enumeration infeasible. The deduplication of codes can be accomplished by computing a canonical form for each code, with the property that two codes are $\mathrm{\Pi H}$-equivalent iff their canonical forms are identical. Then, when generating $\db{n, k}$ codes in the enumeration process, we retain only ones with distinct canonical forms.

To obtain a canonical form of an $\db{n, k}$ code $\mathcal{C}$, we construct its expanded Tanner graph $G[\mathcal{C}]$, which is bipartite and comprises $n$ qubit vertices and $m = 2^{n - k} - 1$ stabilizer vertices enumerating the nontrivial stabilizer group elements. The $j^{\text{th}}$ stabilizer element $\bigotimes_{i_\mathrm{x} \in \mathcal{I}_\mathrm{x}} X_{i_\mathrm{x}} \bigotimes_{i_\mathrm{z} \in \mathcal{I}_\mathrm{z}} Z_{i_\mathrm{z}}$ contributes $X$- ($Z$-) coloured edges connecting the $j^{\text{th}}$ stabilizer vertex to the $i_\mathrm{x}^{\text{th}}$ ($i_\mathrm{z}^{\text{th}}$) qubit vertex for every $i_\mathrm{x} \in \mathcal{I}_\mathrm{x}$ ($i_\mathrm{z} \in \mathcal{I}_\mathrm{z}$). We compute the canonical labeling~\cite{babai1983canonical} of $G[\mathcal{C}]$ respecting vertex and edge colours, which removes the freedom in vertex and edge orderings, equivalently qubit and stabilizer permutations in $\mathcal{C}$. Lastly, we extract the biadjacency matrix $A_{\mathcal{C}} \in \smash{\mathbb{Z}_3^{m \times n}}$ of the canonical $G[\mathcal{C}]$ with edge colours encoded in the nonzero entries. The canonical form of $\mathcal{C}$ is then defined to be $\min(A_{\mathcal{C}}, A_{H^{\otimes n} \mathcal{C} H^{\otimes n}})$, where ordering is lexicographic, to remove the freedom of global Hadamards.

As the canonical form of a code is a bitvector, equivalently a byte array or string, equivalence checking or deduplication of a set of canonical forms is highly efficient. Empirically, the comparison cost of canonical forms is entirely negligible relative to other parts of the code enumeration procedure.

\subsection{Optimized enumeration of CSS codes}

CSS codes affords three key simplifications. First, the splitting of the stabilizers and logicals into pure $X$- and $Z$-sectors enables a reduction of the ``branching factor'' when generating codes. In particular, for each $\db{n, k + 1}$ seed code, one needs only loop over $2^{k + 1} - 1$ choices of appending an $X$- ($Z$-) logical into the $X$- ($Z$-) stabilizer group to generate $\db{n, k}$ codes. Second, it suffices to enumerate codes with stabilizer ranks $r_\mathrm{x} \leq r_\mathrm{z}$, as $r_\mathrm{x} > r_\mathrm{z}$ codes are $\mathrm{\Pi H}$-equivalent to them. Third, the expanded Tanner graph $G[\mathcal{C}]$ of an $\db{n, k}$ CSS code is simpler in structure and can be made smaller, in particular being tripartite and comprising $n$ qubit, $m_\mathrm{x} = 2^{r_\mathrm{x}} - 1$ $X$-stabilizer, and $m_\mathrm{z} = 2^{r_\mathrm{z}} - 1$ $Z$-stabilizer vertices. Edge colours are no longer needed, as the vertex colours of the $X$- and $Z$-stabilizer vertices suffice to encode the Pauli operator types of the stabilizers. It follows also that the canonical $G[\mathcal{C}]$ can be characterized more compactly by an ordered pair of two biadjacency matrices, $A[\mathcal{C}] = (A_\mathrm{x}[\mathcal{C}], A_\mathrm{z}[\mathcal{C}])$, where $A_\mathrm{x}[\mathcal{C}] \in \smash{\mathbb{F}_2^{m_\mathrm{x} \times n}}$, $A_\mathrm{z}[\mathcal{C}] \in \smash{\mathbb{F}_2^{m_\mathrm{z} \times n}}$. These simplifications make CSS code enumeration cheaper in compute time and storage costs, and accounts for our ability to reach higher $n = 14$ on CSS codes.

\subsection{Technical implementation leveraging massive compute}

We used the well-established \textit{bliss} algorithm~\cite{junttila2007engineering} to compute canonical labelings of graphs. As canonical labelling algorithms generally benefit from finer-grained colour classes, we additionally coloured stabilizer vertices according to their weight, which are invariant under qubit permutations. We massively parallelized the code enumeration process to utilize ${\sim}\, 10 \, 000$ cores across hundreds of nodes, over a wall-clock time of $1.5$ months.

In consideration of storage and memory performance, we stored the code data generated in \texttt{HDF5} format with \texttt{Zstandard} and \texttt{BLOSC}~\cite{blosc2009} compression, which allowed random-access chunked reading of large data files. Despite compression, the enumeration process used $>\SI{80}{\tera\byte}$ of storage during computation, and ${\sim} \, \SI{10}{\tera\byte}$ for the final database. We found that at such a scale of compute, low-level factors such as networked I/O latency and bandwidth, homogeneity of hardware, and load balancing between workers become important in addition to the intrinsic compute complexity of the algorithms employed; thus we invested considerable effort to fine-tune our implementation in these respects.

\subsection{Alternative approaches}

For comprehensiveness, we remark on several alternative approaches to code enumeration that we found to be less performant or were difficult to employ due to intrinsic limitations.

\begin{itemize} 
    \item First, in place of the described iterative method to generate codes, one could consider enumerating stabilizer generator matrices in standard form, as in \cref{eq:check_matrix_std_form_gen}. The limitation here is that there are $\smash{\mathcal{O}(2^{n^2})}$ such matrices to go over to generate $\db{n, k}$ codes. Even after reasonable optimizations---such as imposing stabilizer commutation, which constrains the submatrices in the standard form, and quotienting out qubit permutations, which allows restricting the columns of $\trans{\smash{\mqty(\trans{A_2} & \trans{C} & \trans{E})}}$ to be in nondecreasing lexicographic order---the number of matrices is astronomical: for example, ${>} \, 10^{17}$ for $\db{14, 2}$ CSS codes. Conceptually, the iterative approach mitigates this blow-up by deduplicating codes up to $\mathrm{\Pi H}$-equivalence at every $\db{n, k}$ stage of the process.
    
    \item Second, instead of deduplicating codes through canonical forms, one could consider partitioning them into equivalence classes by pairwise comparisons---in particular, checking isomorphism of their expanded Tanner graphs---and retaining a single member of each class. However, deduplicating $E$ codes in this manner requires $\mathcal{O}(E^2 \mathcal{T}_\mathrm{iso})$ time, in comparison to $\mathcal{O}(E \mathcal{T}_\mathrm{can})$ for the canonical form approach, where $\mathcal{T}_\mathrm{iso}$ and $\mathcal{T}_\mathrm{can}$ are characteristic runtimes for graph isomorphism and canonical labelling, respectively. As $\mathcal{T}_\mathrm{iso} \approx \mathcal{T}_\mathrm{can}$ and $E$ easily exceeds $10^9$, the canonical form approach is vastly faster.

    \item Third, instead of starting from the $\db{n, n}$ trivial code and decreasing $k$ iteratively, one could instead start from the $\db{n, 0}$ end of the spectrum and increasing $k$. The $\db{n, 0}$ codes in this context correspond to graph states. This approach, however, requires prior knowledge of the complete set of distinct graph states for each $n$ desired (up to $n = 14$ in this work), for which, to the best of our knowledge, there exists no publicly available database. We therefore found the $\db{n, n} \to \db{n, 1}$ direction more straightforward. 
    
\end{itemize}

\subsection{Code distances}

We computed the exact distance $d$ of each distinct code produced by brute force, which is computationally quick at the scale of codes examined here. This amounts to recording the minimum weight of $s \overline{P}$ over all choices of stabilizers $s$ and nontrivial Pauli logical operators $\overline{P}$, which totals $2^{n - k} \times (2^{2k} - 1)$ choices for an $\db{n, k}$ code---but as $n \leq 14$ this calculation is fast. An optimization is available for CSS codes: the $X$- and $Z$-distances, $d_\mathrm{x}$ and $d_\mathrm{z}$, can be computed separately, where $s$ and $\overline{P}$ are restricted to be $X$- or $Z$-type; then $d = \min(d_\mathrm{x}, d_\mathrm{z})$. 

\subsection{SAT solving for \texorpdfstring{$\overline{\mathrm{CNOT}}$}{CNOT} permutation gate sets and phantomness}

We analyzed each distinct CSS code for the permutation $\overline{\mathrm{CNOT}}$ gate sets they support using \cref{prop:phantom_gateset_symplectic_css,prop:phantom_code_symplectic_css}. In particular, we checked whether each code supports a complete set of individually addressable permutation $\overline{\mathrm{CNOT}}$ gates on $p = 2, 3, \ldots, k$ out of $k$ logical qubits. The highest case of $p = k$ corresponds to phantom codes; lower $p$ reflect partial phantomness, and we refer to codes that possess such gate sets as \emph{weak} phantom codes. For $p < k$ cases, we declare the logical basis rotation matrices $R_\mathrm{x}, R_\mathrm{z} \in \GL(k, \mathbb{F}_2)$ as free variables, whereas for $p = k$ this degree of freedom is unnecessary (see \cref{prop:phantomness_css_logical_basis_independence}); in all cases we declare permutation matrices $\{P^{a_i b_i}\}_{i = 1}^p$ as free variables. The lower-level encoding of these variables and constraints into a SAT problem instance, and the solving process, follows the prescription in \cref{app:sat_preliminaries}. Altogether $k - 1$ SAT problem instances ($p = 2, 3, \ldots, k$) are solved for each generated code. 

\subsection{Breakdown of code enumeration results}

We report an exhaustive breakdown of the numbers of $\mathrm{\Pi H}$-inequivalent $d \geq 2$ CSS codes, weak phantom, and phantom codes with $k = 2, 3, 4$ logical qubits of code sizes $n \leq 14$, stratified by their $(n, k, d_\mathrm{x}, d_\mathrm{z})$ parameters, in \cref{tab:enumeration_counts_phantom}. As described above, weak phantom codes are codes that support a complete set of individually addressable permutation $\overline{\mathrm{CNOT}}$ gates on $2 \leq p < k$ logical qubits. The stringency of the demanded gate set is reflected in the rapidly diminishing numbers of distinct codes as $p$ is increased.

For completeness, we present also an exhaustive breakdown of the numbers of $\mathrm{\Pi H}$-inequivalent CSS codes up to $n = 14$ at all $k$, stratified by their $(n, k, d)$ parameters, in \cref{tab:enumeration_counts_all_css}. These results may be of independent interest. Additionally, we provide a visualization of the number of inequivalent CSS codes at each $(n, k)$ in \cref{fig:num_codes_all_css}, which clearly illustrates the super-exponential growth in the number of codes with $n$ that underlies the core computational challenge of code enumeration. Generally, for a fixed $n$, the number of inequivalent codes increases with $k$ up to a (super-exponentially sharp) peak at $k^*$ and decreases thereafter; this peak location $k^*$ increases with $n$.

\medskip
\medskip

\begin{figure}[!ht]
    \centering
    \includegraphics[width=0.66\linewidth]{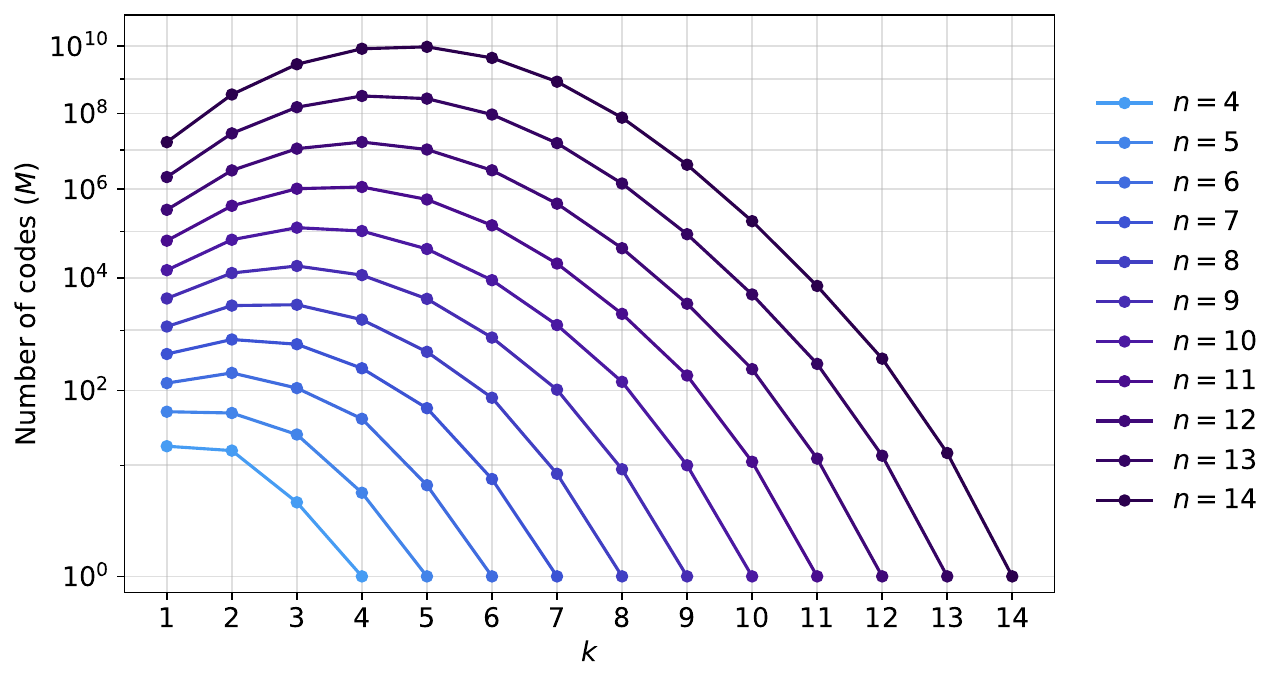}
    \caption{\textbf{Numbers of inequivalent $\mathbf{\db{n, k}}$ CSS codes}. The vertical axis scale is $\sqrt{\log M}$, chosen because a simple upper bound on the number of $n$-qubit CSS codes is $\smash{\mathcal{O}(2^{n^2})}$. This upper bound is loose as it neglects equivalency of codes.}
    \label{fig:num_codes_all_css}
\end{figure}

\clearpage
\pagebreak

{
\renewcommand{\arraystretch}{0.9}
\definecolor{lightmid}{gray}{0.7}
\newcommand{\graymidrule}{%
  \arrayrulecolor{lightmid}%
  \specialrule{0.3pt}{0.25ex}{0.80ex}
  \arrayrulecolor{black}%
}
\begin{table}[!ht]
    \begin{minipage}[t]{0.5\textwidth}
    \centering
    \scriptsize
    \begin{tabular}{C{0.6cm} C{0.6cm} C{0.6cm} R{1.5cm} R{1.5cm} R{1.5cm}}
        \multicolumn{6}{c}{$k = 2$} \\
        \toprule
        $n$ & $d_\mathrm{x}$ & $d_\mathrm{z}$ & $M$ & $M[K_1]$ & $M[K_2]$ \\
        \midrule
        4 & 2 & 
            2 & 1 & 1 & 1 \\
        \graymidrule
        5 & 2 & 
            2 & 2 & 2 & 2 \\
        \graymidrule
        6 & 2 & 
            2 & 15 & 15 & 9 \\
        \graymidrule
        \multirow{2}{*}{7} & \multirow{2}{*}{2} & 
            2 & 59 & 56 & 22 \\
        {} & {} & 
            3 & 2 & 2 & 2 \\
        \graymidrule
        \multirow{3}{*}{8} & \multirow{3}{*}{2} & 
            2 & 352 & 284 & 61 \\
        {} & {} & 
            3 & 20 & 17 & 8 \\
        {} & {} & 
            4 & 2 & 2 & 2 \\
        \graymidrule
        \multirow{3}{*}{9} & \multirow{3}{*}{2} & 
            2 & 1969 & 1219 & 164 \\
        {} & {} & 
            3 & 205 & 140 & 33 \\
        {} & {} & 
            4 & 11 & 11 & 11 \\
        \graymidrule
        \multirow{3}{*}{10} & \multirow{3}{*}{2} & 
            2 & 13229 & 5985 & 469 \\
        {} & {} & 
            3 & 2157 & 970 & 107 \\
        {} & {} & 
            4 & 99 & 98 & 56 \\
        \graymidrule
        10 & 3 & 
            3 & 5 & 2 & 0 \\
        \graymidrule
        \multirow{4}{*}{11} & \multirow{4}{*}{2} & 
            2 & 97043 & 30362 & 1371 \\
        {} & {} & 
            3 & 25389 & 6665 & 334 \\
        {} & {} & 
            4 & 967 & 794 & 229 \\
        {} & {} & 
            5 & 5 & 5 & 5 \\
        \graymidrule
        \multirow{2}{*}{11} & \multirow{2}{*}{3} & 
            3 & 103 & 12 & 1 \\
        {} & {} & 
            4 & 1 & 1 & 1 \\
        \graymidrule
        \multirow{5}{*}{12} & \multirow{5}{*}{2} & 
            2 & 830830 & 175410 & 4487 \\
        {} & {} & 
            3 & 338823 & 46711 & 1096 \\
        {} & {} & 
            4 & 13236 & 7163 & 960 \\
        {} & {} & 
            5 & 74 & 63 & 31 \\
        {} & {} & 
            6 & 4 & 4 & 4 \\
        \graymidrule
        \multirow{2}{*}{12} & \multirow{2}{*}{3} & 
            3 & 4885 & 231 & 5 \\
        {} & {} & 
            4 & 21 & 14 & 5 \\
        \graymidrule
        12 & 4 & 4 & 1 & 1 & 1 \\
        \graymidrule
        \multirow{5}{*}{13} & \multirow{5}{*}{2} & 
            2 & 8311808 & 1153978 & 16111 \\
        {} & {} & 
            3 & 5219100 & 352094 & 3601 \\
        {} & {} & 
            4 & 238551 & 66603 & 3907 \\
        {} & {} & 
            5 & 1342 & 810 & 176 \\
        {} & {} & 
            6 & 29 & 29 & 29 \\
        \graymidrule
        \multirow{2}{*}{13} & \multirow{2}{*}{3} & 
            3 & 187956 & 3375 & 35 \\
        {} & {} & 
            4 & 879 & 226 & 31 \\
        \graymidrule
        13 & 4 & 
            4 & 3 & 3 & 3 \\
        \graymidrule
        \multirow{5}{*}{14} & \multirow{5}{*}{2} & 
            2 & 100151088 & 9086873 & 67104 \\
        {} & {} & 
            3 & 94277832 & 2946714 & 12305 \\
        {} & {} & 
            4 & 5695469 & 686332 & 16650 \\
        {} & {} & 
            5 & 34936 & 10499 & 853 \\
        {} & {} & 
            6 & 391 & 374 & 196 \\
        \graymidrule
        \multirow{3}{*}{14} & \multirow{3}{*}{3} & 
            3 & 7246836 & 46369 & 161 \\
        {} & {} & 
            4 & 75374 & 5018 & 183 \\
        {} & {} & 
            5 & 4 & 4 & 3 \\
        \graymidrule
        14 & 4 & 
            4 & 65 & 54 & 31 \\
        \bottomrule
    \end{tabular}
    \end{minipage}%
    \begin{minipage}[t]{0.5\textwidth}
    \centering
    \scriptsize
    \begin{tabular}{C{0.6cm} C{0.6cm} C{0.6cm} R{1.5cm} R{1.5cm} R{1.5cm} R{1.5cm}}
        \multicolumn{7}{c}{$k = 3$} \\
        \toprule
        $n$ & $d_\mathrm{x}$ & $d_\mathrm{z}$ & $M$ & $M[K_1]$ & $M[K_2]$ & $M[K_3]$ \\
        \midrule
        6 & 2 & 
            2 & 2 & 2 & 2 & 0 \\
        \graymidrule
        \multirow{2}{*}{7} & \multirow{2}{*}{2} & 
            2 & 16 & 16 & 11 & 0 \\
        {} & {} & 
            3 & 1 & 1 & 1 & 1 \\
        \graymidrule
        \multirow{3}{*}{8} & \multirow{3}{*}{2} & 
            2 & 142 & 142 & 56 & 0 \\
        {} & {} & 
            3 & 4 & 4 & 4 & 2 \\
        {} & {} & 
            4 & 1 & 1 & 1 & 1 \\
        \graymidrule
        \multirow{3}{*}{9} & \multirow{3}{*}{2} & 
            2 & 1342 & 1242 & 302 & 0 \\
        {} & {} & 
            3 & 37 & 37 & 21 & 4 \\
        {} & {} & 
            4 & 4 & 4 & 4 & 4 \\
        \graymidrule
        \multirow{3}{*}{10} & \multirow{3}{*}{2} & 
            2 & 14173 & 10912 & 1445 & 0 \\
        {} & {} & 
            3 & 533 & 446 & 107 & 8 \\
        {} & {} & 
            4 & 26 & 26 & 23 & 12 \\
        \graymidrule
        \multirow{3}{*}{11} & \multirow{3}{*}{2} & 
            2 & 174822 & 100255 & 7287 & 0 \\
        {} & {} & 
            3 & 11304 & 5906 & 550 & 16 \\
        {} & {} & 
            4 & 219 & 219 & 122 & 34 \\
        \graymidrule
        11 & 3 & 
            3 & 6 & 4 & 3 & 0 \\
        \graymidrule
        \multirow{4}{*}{12} & \multirow{4}{*}{2} & 
            2 & 2542097 & 1000105 & 39042 & 0 \\
        {} & {} & 
            3 & 303742 & 79083 & 3042 & 32 \\
        {} & {} & 
            4 & 3859 & 3200 & 732 & 93 \\
        {} & {} & 
            5 & 2 & 2 & 2 & 0 \\
        \graymidrule
        \multirow{2}{*}{12} & \multirow{2}{*}{3} & 
            3 & 254 & 48 & 12 & 0 \\
        {} & {} & 
            4 & 2 & 2 & 2 & 0 \\
        \graymidrule
        \multirow{4}{*}{13} & \multirow{4}{*}{2} & 
            2 & 44258048 & 11273064 & 231983 & 0 \\
        {} & {} & 
            3 & 9507219 & 1110857 & 17283 & 68 \\
        {} & {} & 
            4 & 115581 & 55042 & 4454 & 260 \\
        {} & {} & 
            5 & 74 & 73 & 36 & 0 \\
        \graymidrule
        \multirow{2}{*}{13} & \multirow{2}{*}{3} & 
            3 & 54220 & 1987 & 89 & 0 \\
        {} & {} & 
            4 & 37 & 25 & 14 & 0 \\
        \graymidrule
        \multirow{5}{*}{14} & \multirow{5}{*}{2} & 
            2 & 930863172 & 149252977 & 1586674 & 0 \\
        {} & {} & 
            3 & 342695045 & 17000963 & 105686 & 148 \\
        {} & {} & 
            4 & 5530021 & 1096275 & 30013 & 759 \\
        {} & {} & 
            5 & 3620 & 2192 & 338 & 1 \\
        {} & {} & 
            6 & 20 & 20 & 15 & 2 \\
        \graymidrule
        \multirow{2}{*}{14} & \multirow{2}{*}{3} & 
            3 & 6236568 & 59587 & 489 & 1 \\
        {} & {} & 
            4 & 6377 & 1143 & 138 & 3 \\
        \graymidrule
        14 & 4 & 
            4 & 2 & 2 & 2 & 0 \\
        \bottomrule 
        \\[2pt]
        \multicolumn{7}{c}{$k = 4$} \\
        \toprule
        $n$ & $d_\mathrm{x}$ & $d_\mathrm{z}$ & $M$ & $M[K_1]$ & $M[K_2]$ & $M[K_3]$ \\
        \midrule
        6 & 2 & 
            2 & 1 & 1 & 1 & 0 \\
        \graymidrule
        7 & 2 & 
            2 & 3 & 3 & 3 & 0 \\
        \graymidrule
        8 & 2 & 
            2 & 39 & 39 & 29 & 0 \\
        \graymidrule
        \multirow{2}{*}{9} & \multirow{2}{*}{2} & 
            2 & 370 & 365 & 191 & 1 \\
        {} & {} & 
            3 & 1 & 1 & 1 & 0 \\
        \graymidrule
        \multirow{3}{*}{10} & \multirow{3}{*}{2} & 
            2 & 5732 & 5316 & 1512 & 5 \\
        {} & {} & 
            3 & 34 & 26 & 7 & 1 \\
        {} & {} & 
            4 & 1 & 1 & 1 & 0 \\
        \graymidrule
        \multirow{3}{*}{11} & \multirow{3}{*}{2} & 
            2 & 104987 & 82236 & 11970 & 26 \\
        {} & {} & 
            3 & 1288 & 783 & 96 & 6 \\
        {} & {} & 
            4 & 9 & 7 & 5 & 0 \\
        \graymidrule
        \multirow{3}{*}{12} & \multirow{3}{*}{2} & 
            2 & 2446585 & 1430851 & 104031 & 112 \\
        {} & {} & 
            3 & 71376 & 22840 & 1114 & 30 \\
        {} & {} & 
            4 & 338 & 257 & 64 & 2 \\
        \graymidrule
        12 & 3 & 
            3 & 6 & 4 & 3 & 0 \\
        \graymidrule
        \multirow{4}{*}{13} & \multirow{4}{*}{2} & 
            2 & 70097191 & 27711646 & 1002835 & 480 \\
        {} & {} & 
            3 & 4674483 & 665810 & 12209 & 150 \\
        {} & {} & 
            4 & 15455 & 7934 & 729 & 18 \\
        {} & {} & 
            5 & 2 & 2 & 2 & 0 \\
        \graymidrule
        13 & 3 & 
            3 & 859 & 114 & 14 & 0 \\
        \graymidrule
        \multirow{5}{*}{14} & \multirow{5}{*}{2} & 
            2 & 2467108696 & 620290888 & 11226575 & 2061 \\
        {} & {} & 
            3 & 335232851 & 20178079 & 141712 & 755 \\
        {} & {} & 
            4 & 1457061 & 346812 & 9937 & 123 \\
        {} & {} & 
            5 & 105 & 69 & 27 & 0 \\
        {} & {} & 
            6 & 2 & 2 & 2 & 1 \\
        \graymidrule
        \multirow{2}{*}{14} & \multirow{2}{*}{3} & 
            3 & 975140 & 8551 & 110 & 1 \\
        {} & {} & 
            4 & 54 & 22 & 3 & 0 \\
        \bottomrule
    \end{tabular}
    \end{minipage}%
    \caption{\textbf{Number of $\bm{d \geq 2}$ CSS, weakly phantom, and phantom codes at $\bm{k = 2, 3, 4}$ up to $\bm{n = 14}$.} Entries are the number of distinct codes of parameters $(n, k, d_\mathrm{x}, d_\mathrm{z})$ up to equivalence under qubit permutations and global Hadamard duality ($\mathrm{\Pi H}$ equivalence) found through exhaustive code enumeration. Counts cover codes where $d_\mathrm{x} \leq d_\mathrm{z}$; the $d_\mathrm{x} > d_\mathrm{z}$ counterparts are duals of $d_\mathrm{x} < d_\mathrm{z}$ codes. $M$, $M[K_1]$, $M[K_{p \geq 2}]$ denote number of CSS codes with no other constraints, supporting at least a single permutation $\overline{\mathrm{CNOT}}$, and supporting at least a complete set of individually addressable $\overline{\mathrm{CNOT}}$s on $p$ logical qubits, respectively; therefore $M[K_k]$ counts the number of CSS phantom codes. $M[K_4] = 0$ for all $n \leq 14$, $k = 4$.
    }
    \label{tab:enumeration_counts_phantom}
\end{table}
}

{
\renewcommand{\arraystretch}{0.9}
\definecolor{lightmid}{gray}{0.7}
\newcommand{\graymidrule}{%
  \arrayrulecolor{lightmid}%
  \specialrule{0.3pt}{0.25ex}{0.80ex}
  \arrayrulecolor{black}%
}
\newcommand{\blackmidrule}{%
  \arrayrulecolor{black}%
  \specialrule{0.3pt}{0.25ex}{0.80ex}
  \arrayrulecolor{black}%
}
\begin{table}[!ht]
    \begin{minipage}[t]{0.33\textwidth}
    \centering
    \scriptsize
    \begin{tabular}[t]{C{0.6cm} C{0.6cm} C{0.6cm} R{1.8cm}}
        \toprule
        $n$ & $k$ & $d$ & $M$ \\
        \midrule
        4 & 1 & 1 & 16 \\
        4 & 1 & 2 & 1 \\
        \graymidrule
        4 & 2 & 1 & 14 \\
        4 & 2 & 2 & 1 \\
        \graymidrule
        4 & 3 & 1 & 4 \\
        \graymidrule
        4 & 4 & 1 & 1 \\
        \graymidrule
        5 & 1 & 1 & 44 \\
        5 & 1 & 2 & 5 \\
        \graymidrule
        5 & 2 & 1 & 45 \\
        5 & 2 & 2 & 2 \\
        \graymidrule
        5 & 3 & 1 & 24 \\
        \graymidrule
        5 & 4 & 1 & 5 \\
        \graymidrule
        5 & 5 & 1 & 1 \\
        \graymidrule
        6 & 1 & 1 & 109 \\
        6 & 1 & 2 & 21 \\
        \graymidrule
        6 & 2 & 1 & 173 \\
        6 & 2 & 2 & 15 \\
        \graymidrule
        6 & 3 & 1 & 107 \\
        6 & 3 & 2 & 2 \\
        \graymidrule
        6 & 4 & 1 & 38 \\
        6 & 4 & 2 & 1 \\
        \graymidrule
        6 & 5 & 1 & 6 \\
        \graymidrule
        6 & 6 & 1 & 1 \\
        \graymidrule
        7 & 1 & 1 & 299 \\
        7 & 1 & 2 & 84 \\
        7 & 1 & 3 & 1 \\
        \graymidrule
        7 & 2 & 1 & 621 \\
        7 & 2 & 2 & 61 \\
        \graymidrule
        7 & 3 & 1 & 547 \\
        7 & 3 & 2 & 17 \\
        \graymidrule
        7 & 4 & 1 & 220 \\
        7 & 4 & 2 & 3 \\
        \graymidrule
        7 & 5 & 1 & 55 \\
        \graymidrule
        7 & 6 & 1 & 7 \\
        \graymidrule
        7 & 7 & 1 & 1 \\
        \graymidrule
        8 & 1 & 1 & 827 \\
        8 & 1 & 2 & 332 \\
        8 & 1 & 3 & 2 \\
        \graymidrule
        8 & 2 & 1 & 2481 \\
        8 & 2 & 2 & 374 \\
        \graymidrule
        8 & 3 & 1 & 2817 \\
        8 & 3 & 2 & 147 \\
        \graymidrule
        8 & 4 & 1 & 1514 \\
        8 & 4 & 2 & 39 \\
        \graymidrule
        8 & 5 & 1 & 415 \\
        8 & 5 & 2 & 3 \\
        \graymidrule
        8 & 6 & 1 & 77 \\
        8 & 6 & 2 & 1 \\
        \graymidrule
        8 & 7 & 1 & 8 \\
        \graymidrule
        8 & 8 & 1 & 1 \\
        \graymidrule
        9 & 1 & 1 & 2507 \\
        9 & 1 & 2 & 1385 \\
        9 & 1 & 3 & 24 \\
        \graymidrule
        9 & 2 & 1 & 10499 \\
        9 & 2 & 2 & 2185 \\
        \graymidrule
        9 & 3 & 1 & 16402 \\
        9 & 3 & 2 & 1383 \\
        \graymidrule
        9 & 4 & 1 & 11111 \\
        9 & 4 & 2 & 371 \\
        \graymidrule
        9 & 5 & 1 & 3803 \\
        9 & 5 & 2 & 46 \\
        \graymidrule
        9 & 6 & 1 & 733 \\
        9 & 6 & 2 & 4 \\
        \graymidrule
        9 & 7 & 1 & 103 \\
        \graymidrule
        9 & 8 & 1 & 9 \\
        \graymidrule
        9 & 9 & 1 & 1 \\
        \bottomrule
    \end{tabular}
    \end{minipage}%
    \hfill
    \begin{minipage}[t]{0.33\textwidth}
    \centering
    \scriptsize
    \begin{tabular}[t]{C{0.6cm} C{0.6cm} C{0.6cm} R{1.8cm}}
        \toprule
        $n$ & $k$ & $d$ & $M$ \\
        \midrule
        10 & 1 & 1 & 8204 \\
        10 & 1 & 2 & 6182 \\
        10 & 1 & 3 & 161 \\
        \graymidrule
        10 & 2 & 1 & 50158 \\
        10 & 2 & 2 & 15485 \\
        10 & 2 & 3 & 5 \\
        \graymidrule
        10 & 3 & 1 & 107996 \\
        10 & 3 & 2 & 14732 \\
        \graymidrule
        10 & 4 & 1 & 97088 \\
        10 & 4 & 2 & 5767 \\
        \graymidrule
        10 & 5 & 1 & 40178 \\
        10 & 5 & 2 & 897 \\
        \graymidrule
        10 & 6 & 1 & 8949 \\
        10 & 6 & 2 & 91 \\
        \graymidrule
        10 & 7 & 1 & 1232 \\
        10 & 7 & 2 & 4 \\
        \graymidrule
        10 & 8 & 1 & 135 \\
        10 & 8 & 2 & 1 \\
        \graymidrule
        10 & 9 & 1 & 10 \\
        \graymidrule
        10 & 10 & 1 & 1 \\
        \graymidrule
        11 & 1 & 1 & 30289 \\
        11 & 1 & 2 & 30793 \\
        11 & 1 & 3 & 1473 \\
        \graymidrule
        11 & 2 & 1 & 273325 \\
        11 & 2 & 2 & 123404 \\
        11 & 2 & 3 & 104 \\
        \graymidrule
        11 & 3 & 1 & 843528 \\
        11 & 3 & 2 & 186345 \\
        11 & 3 & 3 & 6 \\
        \graymidrule
        11 & 4 & 1 & 1028434 \\
        11 & 4 & 2 & 106284 \\
        \graymidrule
        11 & 5 & 1 & 539400 \\
        11 & 5 & 2 & 23181 \\
        \graymidrule
        11 & 6 & 1 & 136435 \\
        11 & 6 & 2 & 2163 \\
        \graymidrule
        11 & 7 & 1 & 19895 \\
        11 & 7 & 2 & 114 \\
        \graymidrule
        11 & 8 & 1 & 1984 \\
        11 & 8 & 2 & 5 \\
        \graymidrule
        11 & 9 & 1 & 171 \\
        \graymidrule
        11 & 10 & 1 & 11 \\
        \graymidrule
        11 & 11 & 1 & 1 \\
        \graymidrule
        12 & 1 & 1 & 128355 \\
        12 & 1 & 2 & 175795 \\
        12 & 1 & 3 & 14416 \\
        12 & 1 & 4 & 2 \\
        \graymidrule
        12 & 2 & 1 & 1768453 \\
        12 & 2 & 2 & 1182967 \\
        12 & 2 & 3 & 4906 \\
        12 & 2 & 4 & 1 \\
        \graymidrule
        12 & 3 & 1 & 8031258 \\
        12 & 3 & 2 & 2849700 \\
        12 & 3 & 3 & 256 \\
        \graymidrule
        12 & 4 & 1 & 13811647 \\
        12 & 4 & 2 & 2518299 \\
        12 & 4 & 3 & 6 \\
        \graymidrule
        12 & 5 & 1 & 9523416 \\
        12 & 5 & 2 & 798474 \\
        \graymidrule
        12 & 6 & 1 & 2877957 \\
        12 & 6 & 2 & 96580 \\
        \graymidrule
        12 & 7 & 1 & 441457 \\
        12 & 7 & 2 & 5158 \\
        \graymidrule
        12 & 8 & 1 & 42402 \\
        12 & 8 & 2 & 211 \\
        \graymidrule
        12 & 9 & 1 & 3093 \\
        12 & 9 & 2 & 6 \\
        \graymidrule
        12 & 10 & 1 & 215 \\
        12 & 10 & 2 & 1 \\
        \graymidrule
        12 & 11 & 1 & 12 \\
        \graymidrule
        12 & 12 & 1 & 1 \\
        \bottomrule
    \end{tabular}
    \end{minipage}%
    \begin{minipage}[t]{0.33\textwidth}
    \centering
    \scriptsize
    \begin{tabular}[t]{C{0.6cm} C{0.6cm} C{0.6cm} R{1.8cm}}
        \toprule
        $n$ & $k$ & $d$ & $M$ \\
        \midrule
        13 & 1 & 1 & 649360 \\
        13 & 1 & 2 & 1191186 \\
        13 & 1 & 3 & 165529 \\
        13 & 1 & 4 & 21 \\
        \graymidrule
        13 & 2 & 1 & 13991607 \\
        13 & 2 & 2 & 13770830 \\
        13 & 2 & 3 & 188835 \\
        13 & 2 & 4 & 3 \\
        \graymidrule
        13 & 3 & 1 & 96173113 \\
        13 & 3 & 2 & 53880922 \\
        13 & 3 & 3 & 54257 \\
        \graymidrule
        13 & 4 & 1 & 240548389 \\
        13 & 4 & 2 & 74787131 \\
        13 & 4 & 3 & 859 \\
        \graymidrule
        13 & 5 & 1 & 228181293 \\
        13 & 5 & 2 & 36098973 \\
        13 & 5 & 3 & 5 \\
        \graymidrule
        13 & 6 & 1 & 86904718 \\
        13 & 6 & 2 & 6099111 \\
        \graymidrule
        13 & 7 & 1 & 14864068 \\
        13 & 7 & 2 & 392972 \\
        \graymidrule
        13 & 8 & 1 & 1371561 \\
        13 & 8 & 2 & 11931 \\
        \graymidrule
        13 & 9 & 1 & 86886 \\
        13 & 9 & 2 & 267 \\
        \graymidrule
        13 & 10 & 1 & 4684 \\
        13 & 10 & 2 & 6 \\
        \graymidrule
        13 & 11 & 1 & 263 \\
        \graymidrule
        13 & 12 & 1 & 13 \\
        \graymidrule
        13 & 13 & 1 & 1 \\
        \graymidrule
        14 & 1 & 1 & 4060337 \\
        14 & 1 & 2 & 9897944 \\
        14 & 1 & 3 & 2254065 \\
        14 & 1 & 4 & 539 \\
        \graymidrule
        14 & 2 & 1 & 140083911 \\
        14 & 2 & 2 & 200159716 \\
        14 & 2 & 3 & 7322214 \\
        14 & 2 & 4 & 65 \\
        \graymidrule
        14 & 3 & 1 & 1482677292 \\
        14 & 3 & 2 & 1279091878 \\
        14 & 3 & 3 & 6242945 \\
        14 & 3 & 4 & 2 \\
        \graymidrule
        14 & 4 & 1 & 5513412044 \\
        14 & 4 & 2 & 2803798715 \\
        14 & 4 & 3 & 975194 \\
        \graymidrule
        14 & 5 & 1 & 7453583787 \\
        14 & 5 & 2 & 2094974672 \\
        14 & 5 & 3 & 2118 \\
        \graymidrule
        14 & 6 & 1 & 3792773348 \\
        14 & 6 & 2 & 529029086 \\
        14 & 6 & 3 & 21 \\
        \graymidrule
        14 & 7 & 1 & 780207939 \\
        14 & 7 & 2 & 46309837 \\
        \graymidrule
        14 & 8 & 1 & 74486954 \\
        14 & 8 & 2 & 1582818 \\
        \graymidrule
        14 & 9 & 1 & 4110070 \\
        14 & 9 & 2 & 27189 \\
        \graymidrule
        14 & 10 & 1 & 172435 \\
        14 & 10 & 2 & 447 \\
        \graymidrule
        14 & 11 & 1 & 6918 \\
        14 & 11 & 2 & 7 \\
        \graymidrule
        14 & 12 & 1 & 320 \\
        14 & 12 & 2 & 1 \\
        \graymidrule
        14 & 13 & 1 & 14 \\
        \graymidrule
        14 & 14 & 1 & 1 \\
        \bottomrule
    \end{tabular}
    \end{minipage}%
    \caption{\textbf{Number of CSS codes up to $\bm{n = 14}$.} Entries are the number of distinct CSS codes of parameters $(n, k, d)$ up to equivalence under qubit permutations and global Hadamard duality ($\mathrm{\Pi H}$ equivalence) found through exhaustive code enumeration. We include $d = 1$ codes for completeness but they do not support error detection or correction capability.
    }
    \label{tab:enumeration_counts_all_css}
\end{table}
}

\clearpage
\pagebreak

\section{SAT-based code discovery}
\label{app:code_discovery}

While code enumeration provides an exhaustive coverage of $n$-qubit phantom codes, the maximum $n$ reachable is computationally limited. To explore higher $n$, we employed a code discovery method based on Boolean constraint satisfaction (i.e.~SAT solving) that produces phantom codes of desired $\db{n, k, d}$ parameters in an automated fashion.

We note that our work is not the first to use SAT-based methods for quantum code discovery. Ref.~\cite{tremblay2023finite} employs SAT to search for sparse quantum codes. By contrast, our approach is the first to search for codes via their gate sets.

\subsection{Review of standard form of stabilizer generator matrices}
\label{app:code_discovery/standard_form}

To start, we review Gottesman's standard form for the stabilizer generator matrix of an $\db{n, k}$ stabilizer code in symplectic form~\cite{gottesman1997stabilizer},
\begin{equation}
    H =
    \left(
    \begin{array}{ccc|ccc}
    \mathbb{I} & A_1 & A_2 & B & 0 & C \\
    0 & 0 & 0 & D & \mathbb{I} & E
    \end{array}
    \right),
    \label{eq:check_matrix_std_form_gen}
\end{equation}
which is parametrized by a rank $0 \leq r \leq n - k$. Letting $t \coloneqq n - r - k$, the submatrices $A_1 \in \mathbb{F}_2^{r \times t}$, $A_2 \in \mathbb{F}_2^{r \times k}$, $B \in \mathbb{F}_2^{r \times r}$, $C \in \mathbb{F}_2^{r \times k}$, $D \in \mathbb{F}_2^{t \times r}$, and $E \in \mathbb{F}_2^{t \times k}$. Each row of $H$ represents a stabilizer generator of the code in symplectic form (see \cref{app:basics/symplectic}). For $H$ to represent a valid code, the stabilizers must commute ($H \Omega \trans{H} = 0$), which amounts to the constraints
\begin{subequations}
    \begin{align}
        \trans{D} &= A_1 + A_2 \trans{E}, 
        \label{eq:check_matrix_std_form_comm_constraints_gen_1}
        \\
        B + C \trans{A_2} & \quad \text{symmetric}.
        \label{eq:check_matrix_std_form_comm_constraints_gen_2}
    \end{align}
    \label{eq:check_matrix_std_form_comm_constraints_gen}
\end{subequations}

Conveniently, a valid logical basis can also be written,
\begin{equation}
    L_\mathrm{x} =
    \left(
    \begin{array}{ccc|ccc}
    0 & E^{\mathrm T} & \mathbb{I} & C^{\mathrm T} & 0 & 0
    \end{array}
    \right), 
    \qquad
    L_\mathrm{z} =
    \left(
    \begin{array}{ccc|ccc}
    0 & 0 & 0 & A_2^{\mathrm T} & 0 & \mathbb{I}
    \end{array}
    \right),
    \label{eq:logicals_std_form_gen}
\end{equation}
with the anticommutation $L_\mathrm{x} \Omega \trans{L_\mathrm{z}} = \mathbb{I}$ guaranteed. Every qubit stabilizer code can be represented in standard form, thus taking $H, L_\mathrm{x}, L_\mathrm{z}$ in such a format for code discovery is without loss of generality. 

\subsection{SAT problem formulation for discovery of stabilizer phantom codes} 
\label{app:code_discovery/gen}

Each SAT problem instance determines if there exists an $\db{n, r, k, d}$ code that is phantom, and in positive cases, reports the solution. Enumerating the parameters $\db{n, r, k, d}$ therefore discovers phantom codes. To build a SAT problem instance, we declare the submatrices of the standard form, \cref{eq:check_matrix_std_form_gen}, as free variables, and impose the commutation constraints of \cref{eq:check_matrix_std_form_comm_constraints_gen} to ensure a valid stabilizer code. We then use \cref{prop:phantom_code_symplectic_gen} to constrain the code to be phantom. In particular, we declare the logical basis change $R \in \mathrm{Sp}(2k, \mathbb{F}_2)$ and permutation matrices $\{P^{a_i b_i}\}_{i = 1}^p$ as free variables, and impose the gate set constraints therein. 

Imposing the distance-$d$ constraint on the code is more subtle. Our approach is to impose distance lower- and upper-bound constraints simultaneously. To have distance at least $d$ means that every Pauli error of weight at most $d - 1$ must either generate nontrivial syndrome or is itself a stabilizer. This requirement can be expressed as
\begin{equation}
    \left( 
        \biglor_{i = 1}^{n - k} \left( H \Omega \trans{e} \right)_i
    \right)
    \lor
    \left[
        \left( 
            \neg \biglor_{i = 1}^{n - k} \left( H \Omega \trans{e} \right)_i
        \right)
        \land
        \left( 
            \neg \biglor_{i = 1}^{2 k} \left( L \Omega \trans{e} \right)_i
        \right)
    \right],
    \label{eq:code_discovery_gen_distance_constraint_lower_bound_1}
\end{equation}
for every $e \in \mathcal{E}_{\geq} \coloneqq \{e \in \mathcal{P}_n: \wt(e) \leq d - 1\}$ in binary symplectic representation. By applying De Morgan's laws and noting that the set of vectors $e$ is invariant under multiplication by $\Omega$, this can be simplified to
\begin{equation}
    \left( 
        \biglor_{i = 1}^{n - k} \left( H \trans{e} \right)_i
    \right)
    \lor
    \left(
        \neg \biglor_{i = 1}^{2 k} \left( L \trans{e} \right)_i
    \right),
    \label{eq:code_discovery_gen_distance_constraint_lower_bound_2}
\end{equation}
for every $e \in \mathcal{E}_{\geq}$---which we add as constraints. Next, to have distance at most $d$ means that there must exist a weight-$d$ Pauli error that generates no syndrome and is not a stabilizer. This can be expressed as
\begin{equation}
    \biglor_{e \in \mathcal{E}_{\leq}} 
    \Bigg\{
        \left(
            \neg \biglor_{i = 1}^{n - k} \left( H \Omega \trans{e} \right)_i
        \right)
        \land
        \left[
            \left(
                \biglor_{i = 1}^{n - k} \left( H \Omega \trans{e} \right)_i
            \right)
            \lor 
            \left(
                \biglor_{i = 1}^{2 k} \left( L \Omega \trans{e} \right)_i
            \right)
        \right]
    \Bigg\},
    \label{eq:code_discovery_gen_distance_constraint_upper_bound_1}
\end{equation}
where $\mathcal{E}_{\leq} \coloneqq \{e \in \mathcal{P}_n: \wt(e) = d\}$ in binary symplectic representation. Likewise, this can be simplified to
\begin{equation}\begin{split}
    \biglor_{e \in \mathcal{E}_{\leq}} 
    \left[
        \left(
            \neg \biglor_{i = 1}^{n - k} \left( H \trans{e} \right)_i
        \right)
        \land 
        \left(
            \biglor_{i = 1}^{2 k} \left( L \trans{e} \right)_i
        \right)
    \right],
    \label{eq:code_discovery_gen_distance_constraint_upper_bound_2}
\end{split}\end{equation}
which we add as a constraint. Thus, \cref{eq:code_discovery_gen_distance_constraint_lower_bound_2,eq:code_discovery_gen_distance_constraint_upper_bound_2} together enforces that the distance of the code is exactly $d$. The lower-level encoding of these variables and constraints into a form suitable for treatment by a SAT solver follows the prescription in \cref{app:sat_preliminaries}.

\subsection{SAT problem formulation for discovery of CSS phantom codes} 
\label{app:code_discovery/css}

Specializing the code discovery strategy to CSS codes affords key simplifications. First, $B = C = 0$ are fixed in the standard form of stabilizer generator matrices [\cref{eq:check_matrix_std_form_gen}] and the ``half-symplectic'' binary representations $H_\mathrm{x} = \mqty(\mathbb{I} & A_1 & A_2)$, $H_\mathrm{z} = \mqty(D & \mathbb{I} & E)$, $L_\mathrm{x} = \mqty(0 & \trans{E} & \mathbb{I})$, $L_\mathrm{z} = \mqty(\trans{A_2} & 0 & \mathbb{I})$ can be used. There is accordingly no need to declare $B$ and $C$ as free variables and only \cref{eq:check_matrix_std_form_comm_constraints_gen_1} applies as a commutation constraint. Second, we use the simpler \cref{prop:phantom_code_symplectic_css} to constrain the code to be phantom. In particular, there is no need for arbitrary logical basis rotations, and we declare only the permutation matrices $\{P^{a_i b_i}\}_{i = 1}^p$ as variables. 

Third, we impose distance constraints on both the $X$- and $Z$-logical sectors separately, which also affords us the flexibility of individually specifying the desired $X$- and $Z$-distances $d_\mathrm{x}$ and $d_\mathrm{z}$ of the code. Each SAT problem instance therefore determines if there exists an $\db{n, r, k, (d_\mathrm{x}, d_\mathrm{z})}$ CSS code that is phantom, and in positive cases, reports the solution. Explicitly, to enforce distance $d_\mu$ for $\mu \in \{X, Z\}$, we add the lower- and upper-bound constraints
\begin{subequations}\begin{align}
    & \left( 
        \biglor_{i = 1}^{r_{\nu}} \left( H_\mathrm{\nu} \trans{e} \right)_i
    \right)
    \lor
    \left(
        \neg \biglor_{i = 1}^{k} \left( L_\mathrm{\nu} \trans{e} \right)_i
    \right)
    \quad
    \forall e \in \mathcal{E}_{\geq}^{\mu}, \\
    & 
    \biglor_{e \in \mathcal{E}_{\leq}^{\mu}} 
    \left[
        \left(
            \neg \biglor_{i = 1}^{r_{\nu}} \left( H_\mathrm{\nu} \trans{e} \right)_i
        \right)
        \land 
        \left(
            \biglor_{i = 1}^{k} \left( L_\mathrm{\nu} \trans{e} \right)_i
        \right)
    \right],
    \label{eq:code_discovery_css_distance_constraints}
\end{align}\end{subequations}
where $\nu$ is the conjugate sector to $\mu$, and the vector sets $\mathcal{E}_{\geq}^{\mu} \coloneqq \{e \in \mathbb{F}_2^n: \wt(e) \leq d_\mu - 1\}$ and $\mathcal{E}_{\leq}^{\mu} \coloneqq \{e \in \mathbb{F}_2^n: \wt(e) = d_\mu\}$. 

We report results on the minimal block lengths of $k = 2,3$ CSS phantom codes as certified through SAT solving in \cref{tab:min_n_css_phantom_methods}, which is an expanded version of the smaller \cref{tab:min_n_phantom} of the main text. Examples of codes found by SAT that possess exceptional properties, such as high encoding rate, distance, or large logical gate sets are highlighted in \cref{tab:gate_for_phantom_codes} of the main text.

\subsection{SAT problem formulation for discovery of stabilizer and CSS codes without gate set constraints}

In \cref{tab:min_n_phantom} of the main text, we reported minimal block lengths of CSS codes in general. These were obtained by running the SAT solver for code discovery without gate set constraints, which turns the tool into one for generating codes with specified parameters (or certifying that none exists). Specifically, the problem formulation in \cref{app:code_discovery/gen} for stabilizer codes or \cref{app:code_discovery/css} for CSS codes are the same, including the distance constraints; but the permutation gate set or phantomness constraints of \cref{prop:phantom_code_symplectic_gen} or \cref{prop:phantom_code_symplectic_css} are excluded.

{\begin{table}[!ht]
    \centering
    \newcommand{\colwidth}{0.9cm}
    \renewcommand{\arraystretch}{1.2}
    \begin{tabular}{
        |C{0.5cm}|C{1cm}||
        C{\colwidth}|C{\colwidth}|C{\colwidth}|C{\colwidth}|C{\colwidth}|
        C{\colwidth}|C{\colwidth}|C{\colwidth}|C{\colwidth}|C{\colwidth}|
    }
    \hline
    $k$ & \diagbox{$d_\mathrm{x}$}{$d_\mathrm{z}$} & 2 & 3 & 4 & 5 & 6 & 7 & 8 & 9 & 10 
    \\\hline\hline
    \multirow{7}{*}{2} & 2 & 
        4 & 7 & 8 & 11 & 12 & 15 & 16 & 18--19 & 20 
    \\\cline{2-11}
    & 3 &  
        & 11 & 11 & 14 & 15 & ${\geq}\,17$ & ${\geq}\,19$ & ${\geq}\,20$ & ${\geq}\,21$
    \\\cline{2-11}
    & 4 &   
        & & 12 & 15 & 16 & ${\geq}\,19$ & ${\geq}\,20$ & ${\geq}\,21$ & ${\geq}\,22$
    \\\cline{2-11}
    & 5 &
        & & & 18 & 19 & ${\geq}\,20$ & ${\geq}\,21$ & ${\geq}\,22$ & ${\geq}\,22$
    \\\cline{2-11}
    & 6 &  
        & & & & 20 & ${\geq}\,21$ & ${\geq}\,22$ & &
    \\\cline{2-11}
    & 7 &
        & & & & & ${\geq}\,23$ & ${\geq}\,24$ & &
    \\\cline{2-11}
    & 8 &
        & & & & & & ${\geq}\,24$ & &
    \\\hline\hline
    \multirow{5}{*}{3} & 2 &
        ${\geq}\,15$ & 7 & 8 & 14 & 14 & 15 & 16 & 19--21 & 20--22
    \\\cline{2-11}
    & 3 &
        & 14 & 14 & ${\geq}\,15$ & ${\geq}\,16$ & 18--21 & ${\geq}\,19$ & ${\geq}\,21$ & ${\geq}\,21$ 
    \\\cline{2-11}
    & 4 &
        & & 15 & ${\geq}\,16$ & ${\geq}\,17$ & & & &
    \\\cline{2-11}
    & 5 & 
        & & & ${\geq}\,19$ & ${\geq}\,20$ & & & &
    \\\cline{2-11}
    & 6 & 
        & & & $\geq$19 & & & & &
    \\\hline
    \end{tabular}
    \caption{\textbf{Smallest $\bm{n}$ for CSS phantom codes with $\bm{k=2}$–$\bm{3}$}. Entries show the minimal $n$ known for each $(d_\mathrm{x}, d_\mathrm{z})$. Tables are symmetric along the diagonal since each code has a dual with $d_\mathrm{x}$ and $d_\mathrm{z}$ exchanged, assessed by deforming the code with Hadamards on every qubit ($H^{\otimes n}$); we therefore omit the $d_\mathrm{x} > d_\mathrm{z}$ entries (shaded grey). All $n \leq 14$ cells are from exhaustive code enumeration and the rest are from SAT code discovery. Ranges of $n$, for example $n_1$--$n_2$, are shown when SAT discovers a phantom code at size $n_2$ but cannot prove that no phantom code exists at size $n_1$ within time limits; lower bounds, for example $\geq n_{\mathrm{b}}$, are shown when SAT code discovery has proven that no phantom code exists at sizes $< n_{\mathrm{b}}$ but has not yet found a concrete solution. Time limit for SAT solving was constrained to be $14$ days.}
    \label{tab:min_n_css_phantom_methods}
\end{table}}

\section{Phantom quantum Reed--Muller codes}
\label{app:qrm}

Our numerical methods are tractable only for $k \le 4$. To reach higher $k$, we construct an infinite family of CSS phantom codes, starting from quantum Reed--Muller (qRM) codes. By fixing selected logical qubits to $\ket*{\overline{0}}$ or $\ket*{\overline{+}}$, we promote the corresponding logical operators to $Z$- or $X$-type stabilizers, respectively. The resulting codes have parameters $\db{2^m, m-l+1, \min(2^{m-l}, 2^l)}$. This construction is naturally described using the polynomial representation of qRM codes, which we briefly review below.

\subsection{Review of Reed--Muller codes and polynomial formalism}
\label{app:qrm/polymonial}

First, we briefly review classical Reed--Muller codes. Reed--Muller (RM) codes are a family of classical linear codes denoted $\RM(r, m)$, parameterized by the order $r$ and block length $2^m$, with code parameters $[2^m, \sum_{i=0}^r \binom{m}{i}, 2^{m-r}]$. The generators of RM codes can be constructed using the \textit{polynomial formalism}~\cite[Chapter~13]{macwilliams1977theory}. Since the dual of RM codes are themselves also RM codes, they can be constructed with the same formalism.

\subsubsection{Polynomial formalism}

Consider $m$ variables over the binary field \(x_1, x_2, \dots, x_m \in \mathbb{F}_2.\) A \emph{monomial} is a product of these variables that uses each variable at most once: for example, $x_1x_3$, $x_2x_4x_5$, or $1$. Here $1$ is the constant monomial, the product of zero variables. A \emph{polynomial} is a sum (modulo 2) of monomials, such as $1 + x_1 + x_2x_3$. An important example is the NOT monomial \(\bar{x}_i = 1 - x_i\), which evaluates to $1$ precisely when $x_i=0$.

It is useful to view a polynomial as a function $f : \mathbb{F}_2^m \to \mathbb{F}_2$. For instance, with input coordinates $(a_1,a_2,a_3) = (1,0,0)$, the polynomial $f(x_1,x_2,x_3) = x_1x_3$ evaluates to $f(1,0,0) = 1\cdot 0 = 0$. Since for binary variables we have $x^2_i=x_i$, every polynomial can be uniquely represented as a sum of square-free monomials.

Viewing polynomials as Boolean functions naturally associates to each polynomial a length-$2^m$ binary vector, obtained by evaluating it on all $2^m$ binary inputs in $\mathbb{F}_2^m$. For each input tuple $(a_1,\dots,a_m)$, we substitute these values into the polynomial $f(x_1,\dots,x_m)$ and record the output $f(a_1,\dots,a_m) \in \mathbb{F}_2$. Collecting all outputs in lexicographic order (from $00\cdots0$ to $11\cdots1$), which fixes a canonical ordering of coordinates, gives the \emph{evaluation vector}.

\begin{definition}[Evaluation vector]
Given a polynomial $f\in \F_2 [x_1, \dots, x_m]/\la x_1^2-x_1,\cdots,x_m^2-x_m\ra$, the evaluation vector is the vector of values $f$ takes on at all the $2^m$ possible coordinates. Note that the values are modulo $2$.
\end{definition}

For example, below we present the evaluation vector for all monomials up to $m=2$
\[
\begin{array}{c|cccc}
a_2 & 1 & 1 & 0 & 0 \\
a_1 & 1 & 0 & 1 & 0 \\ \hline
1 & 1 & 1 & 1 & 1 \\
x_1 & 1 & 0 & 1 & 0 \\
x_2 & 1 & 1 & 0 & 0 \\ 
x_1x_2 & 1 & 0 & 0 & 0
\end{array}
\]

We now review some salient properties of polynomials.
\begin{itemize}
    \item \textit{Algebraic conventions.} 
    All arithmetic is over $\mathbb{F}_2$, so the polynomial ring has characteristic two. In particular, for any polynomial $f$ we have $2f=0$, and hence addition and subtraction coincide. Moreover, since we quotient by the relations $x_i^2 = x_i$, every polynomial in $\mathbb{F}_2[x_1,\dots,x_m]/\langle x_1^2-x_1,\dots,x_m^2-x_m\rangle$ can be written as a sum of square-free monomials; higher powers satisfy $x_i^a = x_i$ for all $a \ge 1$.
    \item \textit{Overlap.}
    Given two polynomials $f,g$, their \emph{overlap} is defined as the bitwise AND of their evaluation vectors. Equivalently, the overlap is the evaluation vector of the product polynomial $fg$. For $m=2$, let $f(x_1,x_2)=x_1$ and $g(x_1,x_2)=x_2$. Their product is $fg=x_1x_2$, whose evaluation vector is $(0,0,0,1)$ in lexicographic order.
    \item \textit{Weight.} 
    The \textit{weight} $\mathrm{wt}(f)$ of a polynomial $f$ is the Hamming weight of its evaluation vector. The following properties will be useful:
    \begin{itemize}
        \item If $f$ is a monomial in $m$ variables that omits $i$ variables, then $\mathrm{wt}(f) = 2^i$.
        \item If a polynomial $f$ omits $j$ variables (i.e., no monomial term depends on them), then $\mathrm{wt}(f)$ is divisible by $2^j$.
    \end{itemize}
\end{itemize}

\subsubsection{Reed--Muller codes}

\begin{definition}[Reed--Muller codes]
\label{def:Reed--Muller-code}
For a positive integer $m$ and $r \leq m$, the Reed--Muller code $\RM(r, m)$ is the span of evaluation vectors of all monomials in $m$ variables of degree at most $r$.
\end{definition}
The code $\RM(r,m)$ has parameters $[2^m,\, \sum_{i=0}^r \binom{m}{i},\, 2^{m-r}]$. The distance arises from the highest-degree (degree-$r$) monomials, whose evaluation vectors have weight $2^{m-r}$, giving \( d = 2^{m-r}.\) The dual code of $\RM(r,m)$ is $\RM(m-r-1,m)$. Examples of generating monomials are
\[
\begin{aligned}
\mathrm{RM}(2,2): &\quad \{1,\, x_1,\, x_2,\, x_1x_2\},\\
\mathrm{RM}(1,3): &\quad \{1,\, x_1,\, x_2,\, x_3\}.
\end{aligned}
\]

All Reed–Muller (RM) codes share the same structure of encoding circuit as illustrated in \cref{fig:encoding-RM}~\cite{arikan2008rm}; different RM codes are obtained solely by choosing which input qubits are variable (initialized to either $\ket{0}$ or $\ket{1}$), with all remaining inputs fixed to $\ket{0}$. The variable inputs correspond to the generating monomials of the code. For example, $\RM(1,3)$ has variable inputs associated with ${1,x_1,x_2,x_3}$, while $\RM(2,2)$ uses ${1,x_1,x_2,x_1x_2}$; in both cases $k=4$. 

\begin{figure}[t]

\centering
\begin{quantikz}[row sep=0.5em, column sep=1.em]
\lstick{$x_1x_2x_3$: $\qquad$ 0}    & \targ{}   & \qw       & \targ{} & \qw & \qw & \qw & \targ{} &\\
\lstick{$x_1x_2$: $\qquad$ 0}    & \ctrl{-1}  & \targ{}  & \qw    & \qw & \qw & \targ{} & \qw &\\
\lstick{$x_1x_3$: $\qquad$ 0}    & \targ{}   & \qw       & \ctrl{-2} & \qw & \targ{} & \qw & \qw &\\
\lstick{$x_1$: $\quad$ 0/1}   & \ctrl{-1}  & \ctrl{-2} & \qw    & \targ{} & \qw & \qw & \qw &\\
\lstick{$x_2x_3$: $\qquad$ 0}    & \targ{}   & \qw       & \targ{} & \qw & \qw & \qw & \ctrl{-4} &\\
\lstick{$x_2$: $\quad$ 0/1}   & \ctrl{-1}  & \targ{}  & \qw    & \qw & \qw & \ctrl{-4} & \qw &\\
\lstick{$x_3$: $\quad$ 0/1}   & \targ{}   & \qw       & \ctrl{-2} & \qw & \ctrl{-4} & \qw & \qw &\\
\lstick{$1$: $\quad$ 0/1}   & \ctrl{-1}  & \ctrl{-2} & \qw    & \ctrl{-4} & \qw & \qw & \qw &\\
\end{quantikz}
\caption{\textbf{Encoding circuit of $\bm{\RM(1,3)}$}. Bits corresponding to monomials of degree more than 1 are fixed to 0, while the rest are variable. There are 4 variable bits, corresponding to the code dimension $k=4$.}
\label{fig:encoding-RM}
\end{figure}
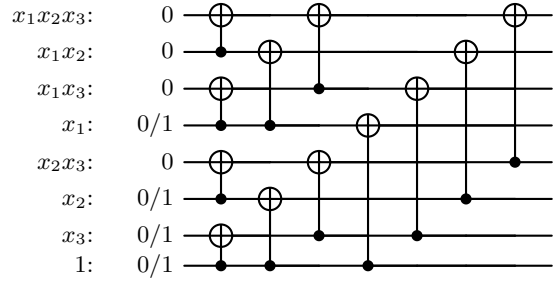

\subsubsection{Affine transformations}

The automorphism group of a Reed--Muller code arises from affine transformations of the variables.
Let $A \in \mathrm{GL}(m,\mathbb{F}_2)$ and $\vb{b} \in \mathbb{F}_2^m$, and define the affine map
\begin{equation}\label{eq:RM_affine_automorphism}
T(x) = Ax + \vb{b}.
\end{equation}
Such a transformation induces (i) a permutation of the $2^m$ coordinates via $a \mapsto T(a)$, and
(ii) an action on polynomials by substitution, $f(x) \mapsto f(T(x))$.
Since affine transformations preserve degree, this substitution maps degree-${\le} \, r$ polynomials to degree-${\le} \, r$ polynomials.
Moreover, applying the inverse of the induced coordinate permutation has the same effect on evaluation vectors as polynomial substitution.
Consequently, affine transformations act as automorphisms of $\RM(r,m)$, and
\[
\operatorname{Aut}(\RM(r,m)) = \operatorname{AGL}(m,\mathbb{F}_2).
\]

\begin{lemma}[Reed--Muller code automorphism~{\cite{macwilliams1977theory}}]
For $A \in \GL(m,\mathbb{F}_2)$ and $\vb{b} \in \mathbb{F}_2^m$, the coordinate permutation $\sigma \in S_{2^m}$ that maps $A\vb{x} + \vb{b}$ to $\vb{x}$ is a code automorphism, under which codewords transform as $f(\vb{x}) \mapsto f(T \vb{x})$.
This \emph{general affine group} formed by $A$ and $\vb{b}$ is the automorphism group of the RM codes in $m$ variables for any $1\le r\le m-2$.
\end{lemma}

The following example illustrates how transforming the monomial is equivalent to permuting bits. Consider
\begin{equation}
A =
\begin{pmatrix}
1 & 1 & 0\\
0 & 1 & 0\\
0 & 0 & 1
\end{pmatrix},
\qquad
\vb{b} =
\begin{pmatrix}
0\\0\\0
\end{pmatrix}.
\label{eq:ex1-affine-trans}
\end{equation}
As shown in \cref{fig:affine-trans}, this transformation is equivalent to swapping bits $(2,3)$ and $(6,7)$.

\begin{figure}[t]
\centering
\begin{minipage}[t]{0.28\textwidth}
\centering
\renewcommand{\arraystretch}{1.1}
\setlength{\tabcolsep}{3pt}
\scalebox{0.72}{
\begin{tabular}{c|cccccccc}
 & 7 & 6 & 5 & 4 & 3 & 2 & 1 & 0\\
 & 111 & 110 & 101 & 100 & 011 & 010 & 001 & 000\\
\hline
$1$   & 1 & 1 & 1 & 1 & 1 & 1 & 1 & 1\\
$x_1$ & 1 & 0 & 1 & 0 & 1 & 0 & 1 & 0\\
$x_2$ & 1 & 1 & 0 & 0 & 1 & 1 & 0 & 0\\
$x_3$ & 1 & 1 & 1 & 1 & 0 & 0 & 0 & 0\\
\end{tabular}}
\\[4pt]
{\footnotesize \textit{Original generators}}
\end{minipage}\hfill
\begin{minipage}[t]{0.34\textwidth}
\centering
\renewcommand{\arraystretch}{1.1}
\setlength{\tabcolsep}{3pt}
\scalebox{0.72}{
\begin{tabular}{c|cccccccc}
 & 7 & 6 & 5 & 4 & 3 & 2 & 1 & 0\\
 & 111 & 110 & 101 & 100 & 011 & 010 & 001 & 000\\
\hline
$1$   & 1 & 1 & 1 & 1 & 1 & 1 & 1 & 1\\
\cellcolor{yellow}$x_1{+}x_2$ & \cellcolor{yellow}0 & \cellcolor{yellow}1 & 1 & 0 & \cellcolor{yellow}0 & \cellcolor{yellow}1 & 1 & 0\\
$x_2$ & 1 & 1 & 0 & 0 & 1 & 1 & 0 & 0\\
$x_3$ & 1 & 1 & 1 & 1 & 0 & 0 & 0 & 0\\
\end{tabular}}
\\[4pt]
{\footnotesize \textit{Affine map applied to monomials}}
\end{minipage}\hfill
\begin{minipage}[t]{0.34\textwidth}
\centering
\renewcommand{\arraystretch}{1.1}
\setlength{\tabcolsep}{3pt}
\scalebox{0.72}{
\begin{tabular}{c|cccccccc}
 & \cellcolor{yellow}6 & \cellcolor{yellow}7 & 5 & 4 & \cellcolor{yellow}2 & \cellcolor{yellow}3 & 1 & 0\\
 & \cellcolor{yellow}110 & \cellcolor{yellow}111 & 101 & 100 & \cellcolor{yellow}010 & \cellcolor{yellow}011 & 001 & 000\\
\hline
$1$   & 1 & 1 & 1 & 1 & 1 & 1 & 1 & 1\\
$x_1$ & \cellcolor{yellow}0 & \cellcolor{yellow}1 & 1 & 0 & \cellcolor{yellow}0 & \cellcolor{yellow}1 & 1 & 0\\
$x_2$ & 1 & 1 & 0 & 0 & 1 & 1 & 0 & 0\\
$x_3$ & 1 & 1 & 1 & 1 & 0 & 0 & 0 & 0\\
\end{tabular}}
\\[4pt]
{\footnotesize \textit{Affine map applied by permuting bits}}
\end{minipage}
\caption{\textbf{Equivalence between affine transformation of variables and permutation of bits.} We illustrate an example of an affine transformation, as described in \cref{eq:ex1-affine-trans}, and a bit permutation that produces an identical action on the code.}
\label{fig:affine-trans}
\end{figure}

A second example is shown in \cref{fig:rm_aut}, where the generator matrix is for $\RM(1,2)$ and the variables are mapped as $T(x_2) = x_1 + x_2$ and $T(x_1) = x_1 + 1$.
Thus, the coordinates (i.e., columns in the generator matrix) are permuted according to $T^{-1}$: $00 \mapsto 11$, $11 \mapsto 10$, $10 \mapsto 01$, and $01 \mapsto 00$. 
One can verify that this is equivalent to transforming the polynomials (i.e., rows of the generator matrix): $1 \mapsto 1$ (the constant polynomial), $x_1 \mapsto T(x_1)=x_1+1$, and $x_2 \mapsto T(x_2)=x_1+x_2$.

\begin{figure}
\centering
\includegraphics[scale=.5]{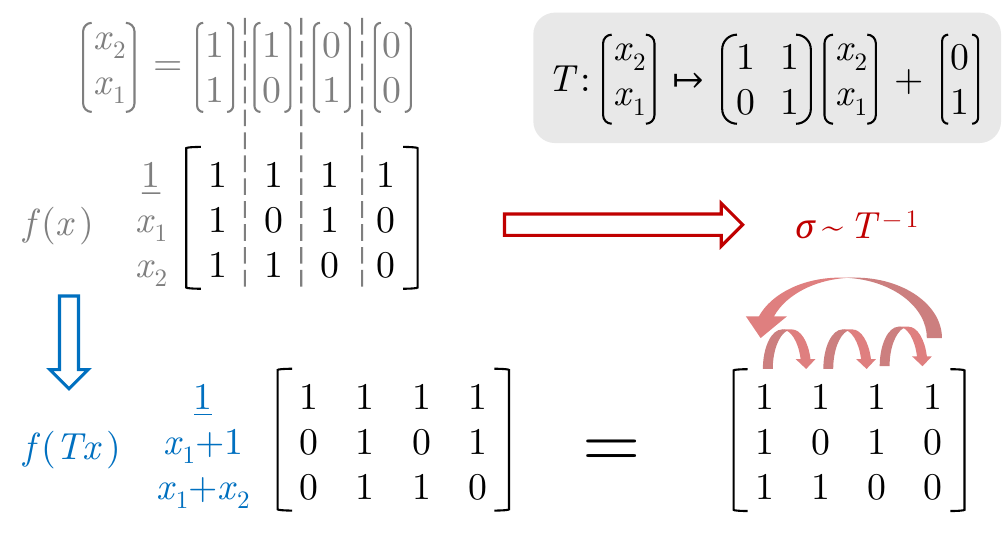}
\vspace{-4pt}
\caption{\label{fig:rm_aut}\textbf{Automorphism of a Reed--Muller code.}
An affine transformation $T$ to the variables induces two equivalent operations on the generator matrix: 1) permuting the coordinates by $T^{-1}$, and 2) substituting $x_i \mapsto T(x_i)$ in each polynomial.}
\end{figure}

\subsection{Quantum Reed--Muller code construction}
\label{app:qrm/construction}

The stabilizers and logical operators of quantum Reed--Muller (qRM) codes can be described via classical Reed--Muller codes.

\begin{definition}[Quantum Reed--Muller codes~\cite{gong2024computation}]
\label{def:general_qrm}
For positive $m$ and $l\le m$, the quantum code with $X$-type stabilizers given by $\RM(l-1,m)$ and $Z$-type stabilizers given by $\RM(m-l-1,m)$ has parameters $\db{2^m, \binom{m}{l}, \min(2^{m-l},\,2^l)}$.
\end{definition}

Following \cref{def:general_qrm}, on qRM codes, the degree-$l$ monomials serve as $X$-type logical operators, and the $Z$-type logical operators can be obtained by taking the corresponding negated complementing monomials: for a monomial $f$ of degree $l$, take $x_1 x_2 \dots x_m/f$ (a degree $m-l$ monomial) and replace each variable $x_i$ with its complement $\overline{x_i} = x_i + 1$.

While the codes in \cref{def:general_qrm} are valid CSS quantum codes (see e.g.~Ref.~\cite{gong2024computation} for proof)---i.e.~their stabilizers commute, and their logical operators have the correct commutation structure between themselves and commute with stabilizers---they are not guaranteed to be phantom. Yet, as we will show, there is a way to construct phantom codes out of these codes. 

\begin{proposition}[Hypercube codes are phantom]
    All codes with $l=1$ from \cref{def:general_qrm} are phantom. These codes have parameters $\db{2^D,D,2}$ and are also known as hypercube codes.
\end{proposition}
\begin{proof}
The automorphism of classical RM codes directly applies to this quantum code family.
Specifically, under an affine transformation $T$ as in \cref{eq:RM_affine_automorphism}, permuting the coordinates of physical qubits by $T^{-1}$ is equivalent to transforming the logical generator monomials by $T$.
The affine component $\vb{b}$ acts as a logical identity: when applied to a degree-$l$ monomial (i.e., an $X$-type logical), $\vb{b}$ preserves the degree-$l$ part (the logical) and only introduces lower-degree terms (corresponding to stabilizers).
Thus, distinct actions on the logical space are captured by $A \in \mathrm{GL}(m,\mathbb{F}_2)$.
For example, consider the $\db{8,3,2}$ code, and its higher-dimensional generalization, the $\db{2^D, D, 2}$ hypercube code ($m=D$ and $l=1$), where the $X$-type logical generators are the $D$ degree-1 monomials $x_1,\dots,x_D$.
This hypercube code is phantom, as every possible logical $\overline{\mathrm{CNOT}}$ circuit one-to-one corresponds to an $A$ describing a qubit permutation.
\end{proof}

An immediate way to build higher-distance phantom codes is to concatenate the $\db{2^D,D,(d_\rmx=2^{D-1},d_\rmz=2)}$ hypercube code with some $k=1$ code. 
For example, to obtain a code with approximately equal $X$- and $Z$-sector error-correction capability (i.e.~distances)\footnote{Another reason one might be interested in such balanced codes is that a transversal Hadamard together with qubit permutations might be a valid logical gate. However, we suspect that phantom codes with $ZX$-duality for $k>2$ may not exist. One can show that transversal $H^{\otimes n}$ implementing logical $\overline{H}^k$ up to any logical $\overline{\mathrm{CNOT}}$ for $k>2$ is not possible using \cref{thm:no_noncommuting_transversal_gates} and \cref{lemma:basis_change_matrix}, but it is not clear whether transversal $H$ on a subset of qubits and/or composed with permutations that are not necessarily a valid logical operation can overcome this barrier. In fact, we have not found any $k=3$ phantom code with $d_\rmx=d_\rmz$ and $H_\rmx,H_\rmz$ having the same rank through code enumeration up to $n = 14$ and non-exhaustive SAT search up to $n \sim 21$.}, one can balance the distance by 
concatenating with a $\db{2^{D-2},1,(d_\rmx=1,d_\rmz=2^{D-2})}$ phase-flip repetition code~\cite{hangleiter2025fault}. The resulting phantom code has parameters $\db{2^{2D-2},D,2^{D-1}}$ and possesses $2^{2D-2}-2^D+1$ $X$-type stabilizers and $2^{D}-D-1$ $Z$-type stabilizers. Despite being balanced in distances, these codes nonetheless contain significantly more stabilizers in one basis than the other, which translates to a practical difference in logical error rate performance in the two sectors.

Next, we show that it is possible to produce more balanced phantom codes by direct construction from the $\db{2^m, \binom{m}{l}, \min(2^{m-l},\,2^l)}$ qRM codes. The strategy is to retain a few logicals while carefully promoting the other logicals to either $X$- or $Z$-type stabilizers. We begin by considering the simplest case where all other logicals are fixed into $Z$-type stabilizers.

\begin{theorem}[Phantom quantum Reed--Muller codes]
\label{thm:phantom_qRM}
From the $\db{2^m, \binom{m}{l}, \min(2^{m-l},\,2^l)}$ qRM code, retain the following $m-l+1$ $X$-logical generators: $x_1x_2\dots x_{l-1}x_l$, $x_1x_2\dots x_{l-1}x_{l+1}$, ..., $x_1x_2\dots x_{l-1}x_m$, and fix the other logicals to the $|\overline{0}\rangle$ state---effectively extending the $Z$-type stabilizers with the negated complements of the other $X$-type logicals.
The resulting code is phantom with parameters
\[
\db{2^m, m-l+1, \min(2^{m-l}, 2^l)}.
\]
\end{theorem}

\begin{proof}
By construction, the code has $2^m$ physical qubits and $m-l+1$ logical qubits.
The $X$-type logicals and stabilizers set is a subcode of $\RM(l,m)$, so the $X$-distance of the new code cannot decrease---it in fact remains $2^{m-l}$, since the monomial $x_1x_2\dots x_l$ still has that weight.
The $Z$-type logicals and stabilizers remain within $\RM(m-l, m)$, keeping the $Z$-distance at $2^l$. 
To enable arbitrary $\overline{\mathrm{CNOT}}$s, it suffices to use the $\GL(m-l+1,\mathbb{F}_2)$ subgroup of $\GL(m,\mathbb{F}_2)$, acting invertibly on $\{x_l, x_{l+1}, \dots, x_m\}$ but fixing $x_1,\dots,x_{l-1}$. The action of this subgroup on the $X$-type logical monomials corresponds one-to-one to the set of possible permutation $\overline{\mathrm{CNOT}}$ circuits. 
\end{proof}

One can build a more balanced version of phantom qRM codes as follows. We still consider the above $\GL(m-l+1,\mathbb{F}_2)$ action which invertibly transforms $\{x_l, x_{l+1}, \dots, x_m\}$.
We can divide the original $X$-type logicals, namely the degree-$l$ monomials into invariant subsets (orbits) under this action (the invariance is up to degree $\le l-1$ monomials, i.e., stabilizers), then promote each orbit, either completely to $X$-type stabilizers, or (their negated complement) to $Z$-type stabilizers. To see why the resulting code is still phantom, note that one only needs to verify for the $X$-type sector that stabilizers (resp.~logicals) transform to stabilizers (resp.~logicals) up to stabilizers. Making the same choice on each orbit precisely guarantees this.

Taking the $\db{64,4,8}$ phantom qRM as an example, we divide the $\binom{6}{3}=20$ degree-$3$ monomials into four orbits: 
\begin{align}
\{x_1x_2x_3,x_1x_2x_4&,x_1x_2x_5,x_1x_2x_6\}, \label{eq:64-4-8_orbit_1}\\
\{x_1x_3x_4,x_1x_3x_5,x_1&x_3x_6,x_1x_4x_6,x_1x_5x_6\}, \label{eq:64-4-8_orbit_2}\\
\{x_2x_3x_4,x_2x_3x_5,x_2&x_3x_6,x_2x_4x_6,x_2x_5x_6\}, \label{eq:64-4-8_orbit_3}\\
\{x_3x_4x_5,x_3x_5x_6&,x_3x_4x_6,x_3x_4x_5\}. \label{eq:64-4-8_orbit_4}
\end{align}

\begin{definition}[Balanced $\db{64,4,8}$ phantom qRM code used in benchmarking]
\label{def:balanced_64-4-8}
Consider the $\db{64,4,8}$ phantom qRM code with the orbit~\eqref{eq:64-4-8_orbit_1} chosen to be logical representatives, the negated complement of~\eqref{eq:64-4-8_orbit_2} set to $Z$-type stabilizers, and \eqref{eq:64-4-8_orbit_3} and~\eqref{eq:64-4-8_orbit_4} retained as $X$-type stabilizers. The resulting code has $32$ $X$-type and $28$ $Z$-type stabilizers. This is the code used in our benchmarking study (\cref{sec:benchmarking}).
\end{definition}

We found in preliminary simulations that the $\db{64,4,8}$ code defined in~\cref{thm:phantom_qRM} has a logical $Z$-sector error rate ${\sim} \, 50 \times$ worse than the $X$-sector, under circuit-level depolarizing noise (see \cref{app:benchmarking/noise_model}). Therefore, we chose to use the more balanced version in \cref{def:balanced_64-4-8} for our benchmarking study, whose logical error rates in both sectors are similar at roughly the geometric mean of the less balanced one.

We give a few further remarks on our phantom qRM construction:
\begin{enumerate}
\item The specific choice of setting every orbit to $Z$-type stabilizers (except leaving one as the logicals) allows addressable diagonal logical gates such as $\overline{S}_i \overline{S}_j$ and $\overline{\mathrm{CZ}}_{ij}$ via a physical fold-type circuit ($S$ gates on the fold line and $\mathrm{CZ}$ on each pair of qubits mirrored across the fold line), as established in the next subsection (\cref{app:qrm/folds}).
\item The different balancing choices can be seen as different gauge-fixings of the parent $\db{2^m, \binom{m}{l}, \min(2^{m-l},\,2^l)}$ code. Code-switching between them can be transversally performed via Steane EC \cite{gong2024computation}, which can be understood as a kind of logical teleportation. Then one can implement the $\overline{\mathrm{CNOT}}$s using the most balanced version, and resort to a less balanced version for diagonal logical gates. Fault-tolerant state preparation schemes for both versions are needed to enable the transversal code-switching.
\item On the $\db{64,4,8}$ example, choosing both the first~\eqref{eq:64-4-8_orbit_1} and last~\eqref{eq:64-4-8_orbit_4} orbits are chosen as logicals is effectively coupling a phantom code with its dual (see \cref{sec:code_and_its_dual} for more details). While the resulting code is not phantom, it inherits phantom-like gates in the form of pairwise products of $\overline{\mathrm{CNOT}}$s; this code then supports permutation logical Hadamards, and transversal $S$ operation produces a logical $\overline{\mathrm{CZ}}$ action.
\end{enumerate}

We additionally remark that a punctured variant of the hypercube codes is also phantom. 

\begin{theorem}[Punctured hypercube codes are phantom]
Puncturing the coordinate $0\cdots 00$ from the hypercube codes yield a family of phantom codes with parameters $\db{2^D-1, D, (d_\rmx=2^{D-1}-1,d_\rmz=2)}$.
\end{theorem}
\begin{proof}
Recall that the original $\db{2^D,D,(d_\rmx=2^{D-1},d_\rmz=2)}$ hypercube code has a single $X$-type stabilizer supported on every qubit, and $Z$-type stabilizers $\RM(m-1,m)$. The punctured hypercube code again only has one $X$-type stabilizer, which is supported on all the $2^D-1$ qubits.
The $D$ $X$-type logical generators remain the same: they are still expressed as the degree-1 monomials $x_1, \dots$, and $x_D$ (these monomials evaluate to zero at the punctured coordinate $0\cdots 00$). The $X$-type logicals and stabilizers together form the punctured $\RM(1,D)^*$ code, and $\GL(D,\F_2)$ is its automorphism group~\cite[Ch.~13, Thm.~24]{macwilliams1977theory}. Therefore, this punctured version is also phantom for the same reason as the original version.
\end{proof}

For completeness, we state the $Z$-type logicals and stabilizers that form the shortened $\overline{\RM}(D-1,D)$ code: these are precisely the codewords from $\RM(D-1, D)$ that vanish at the coordinate $00\cdots 0$, with that coordinate subsequently deleted. Here, the paired $Z$-type logical of the $X$-type logical $x_i$ is its complement monomial $x_1\dots x_{i-1}x_{i+1}\dots x_D$ (no negation).

While the original $\db{2^D, D, 2}$ code supports $D^{\rm{th}}$ Clifford-hierarchy-level logical magic via transversal $Z^{1/2^{D-1}}$ rotations, the punctured $\db{2^D-1, D, 2}$ version loses the highest-level magic gate. 
The ${\le} \, (D-1)$ level magic gates persists (by addressing subcubes), and the logical $\overline{S}_i \overline{S}_j$ gates via folding described in the next subsection also applies to this punctured family.

Indeed, the $\db{7,3,(d_\rmx=3,d_\rmz=2)}$ phantom code in \cref{tab:gate_for_phantom_codes} is the punctured version of the $\db{8,3,2}$ hypercube code:
\begin{center}
\scalebox{0.72}{
\begin{tabular}{c|cccccccc}
 & 7 & 6 & 5 & 4 & 3 & 2 & 1 & \tikzmark{start}0\\
 & 111 & 110 & 101 & 100 & 011 & 010 & 001 & 000\\
\hline
$1$   & 1 & 1 & 1 & 1 & 1 & 1 & 1 & 1\\
$x_1$ & 1 & 0 & 1 & 0 & 1 & 0 & 1 & 0\\
$x_2$ & 1 & 1 & 0 & 0 & 1 & 1 & 0 & 0\\
$x_3$ & 1 & 1 & 1 & 1 & 0 & 0 & 0 & \tikzmark{end}0\\
\end{tabular}
\tikz[remember picture] \draw[overlay,draw=red] ([xshift=.25em,yshift=0.8em]pic cs:start) -- ([xshift=.25em,yshift=-0.2em]pic cs:end);}
\end{center}

Concatenating this code with the $\db{2,1,d_\rmx=1\;/\;d_\rmz=2}$ phase-flip repetition code yields a $\db{14,3,(d_\rmx=3,d_\rmz=4)}$ phantom code that has $8$ $X$-type and $3$ $Z$-type stabilizers, which was found also by our exhaustive code enumeration.
By running the gate solver described in \cref{app:gates/diagonal/fold}, we find fold-diagonal gates that implement $\smash{\overline{S}_i \overline{S}_j}$ or $\smash{\overline{\mathrm{CZ}}}$ on the Hadamard-dual of this code.
Note that there is also a fold-diagonal gate on the original $\db{7,3,(d_\rmx=3,d_\rmz=2)}$ phantom code, which can be derived by undoing the $\mathrm{CNOT}$s from the phase-flip repetition code encoding, applying $S$ gates on a subcube of the $\db{7,3,2}$, and redoing the encoding $\mathrm{CNOT}$s; these together can be compiled as $\mathrm{CNOT}_{12} S_2 \mathrm{CNOT}_{12} = \mathrm{CZ} S_1 S_2$.

There are two other $\db{14,3,(d_\rmx=3,d_\rmz=4)}$ phantom codes found by enumeration. They have $7/5$ $X$-type and $4/6$ $Z$-type stabilizers respectively, and they too admit fold-diagonal gates.

\subsection{Addressable diagonal logical gates via folding}
\label{app:qrm/folds}

\subsubsection{Involutions and associated fold-type circuits on qRM phantom codes}
\label{app:qrm/folds/SS_CZ}

Here we show that the phantom qRM code family of \cref{thm:phantom_qRM}, where all the extra logical operators are promoted to $Z$-type stabilizers, support $\smash{\overline{S}_i \overline{S}_j}$ and $\smash{\overline{\mathrm{CZ}}_{ij}}$ gates through fold-type circuits.

We begin with the $m=2l$ code family with parameters $\db{2^{2l},l+1,2^l}$; there, degree ${<} \, l$ monomials are both $X$- and $Z$-type stabilizers, the $X$-type logical generators are $x_1x_2\dots x_{l-1}x_l,\;x_1x_2\dots x_{l-1}x_{l+1},\;\dots,\;x_1x_2\dots x_{l-1}x_{2l}$, and the negated complement of the other degree $l$ monomials are $Z$-type stabilizers. We associate a fold-type circuit to an \emph{involution} $\tau$ permuting the $2^{2l}$ coordinates: single-qubit $S$ or its powers are applied to the fixed points, and $\mathrm{CZ}$ between every $(i,\tau(i))$ pair where $i\neq \tau(i)$. 

\begin{proposition}[Fold-diagonal logical gates on $m=2l$ phantom qRM codes]
\label{prop:phantom_qRM_fold}
For the $m=2l$ phantom qRM code family of \cref{thm:phantom_qRM}:
(a) the involution $\tau_{SS}$ that maps $x_i$ to $\overline{x}_{2l+1-i}$ for $i\in [1,2l]$ implements $\overline{S} \overline{S}$ on $X$-type logicals $x_1x_2\dots x_{l-1}x_l$ and $x_1x_2\dots x_{l-1}x_{l+1}$; 
(b) the involution $\tau_{\mathrm{CZ}}$ that maps $x_i$ to $\overline{x}_{2l+1-i}$ for $i\in [1,2l]\backslash \{l,l+1\}$ and $x_{l}=\overline{x}_l,\;x_{l+1}=\overline{x}_{l+1}$ implements $\overline{\mathrm{CZ}}$ between $X$-type logicals $x_1x_2\dots x_{l-1}x_l$ and $x_1x_2\dots x_{l-1}x_{l+1}$.
\end{proposition}

\begin{proof}
Since the physical gates are diagonal and commute with $Z$, the $Z$-type stabilizers of the code are left invariant; we need only check that the $X$-type stabilizers are mapped to stabilizers and that $X$-type logicals transform as desired. It is clear that both $\tau$ stated above are indeed involutions, i.e., $\tau^2=\mathbb{I}$.

The image of an $X$-type stabilizer generator $f$ (degree ${\le} \, l-1$ monomials) under both $\tau_{SS}$ and $\tau_{\rm{CZ}}$ is the product of itself and the $Z$-type $\tau(f)$, which is likewise a degree ${\le} \, l-1$ monomial and is hence a $Z$-type stabilizer. The phase accumulated is $i^{|f\wedge \tau(f)|}=1$: the product of $f$ and $\tau(f)$ is a degree ${\le} \, 2l-2$ monomial and lacks at least two variables from $\{x_1,\dots,x_{2l}\}$, so by the weight property (see \cref{app:qrm/polymonial}), $|f\wedge \tau(f)|$ is divisible by $4$. 

We now consider the logical action of the two involutions:
\begin{itemize}
    \item For $\tau_{SS}$, the $X$-type logical $x_1x_2\dots x_{l-1}x_l$ is mapped to the product of itself and the $Z$-type $\overline{x}_{2l}\overline{x}_{2l-1}\dots \overline{x}_{l+2}\overline{x}_{l+1}$, precisely its paired $Z$-type logical, with a phase $i$; and $x_1x_2\dots x_{l-1}x_{l+1}$ is mapped to the product of itself and $Z$-type $\overline{x}_{2l}\overline{x}_{2l-1}\dots \overline{x}_{l+2}\overline{x}_{l}$, again its paired $Z$-type logical and with a phase $i$. The other $X$-type logical generators are left invariant up to a $Z$-type stabilizer $\tau_{SS}(f)$. This is because $\tau_{SS}(f)$ would contain $\overline{x}_s$, where $1\le s\le l-1$, so the $Z$-type operator defined by $\tau_{SS}(f)$ cannot be the paired $Z$-type logical operator of any $X$-type logical.
    \item For $\tau_{\rm{CZ}}$, $x_1x_2\dots x_{l-1}x_l$ is mapped to the product of itself and the $Z$-type $\overline{x}_{2l}\overline{x}_{2l-1}\dots \overline{x}_{l+2}\overline{x}_{l}$, precisely the paired $Z$-type logical of the $X$-type logical $x_1x_2\dots x_{l-1} x_{l+1}$, with a trivial phase. The transformation of the other $X$-type logicals can be similarly verified.
\end{itemize}

\end{proof}

To give a concrete example, for the $\db{16,3,4}$ phantom qRM code, $\tau_{SS}$ maps each coordinate as [see~\cref{eq:RM_affine_automorphism}]
\begin{equation}
\begin{pmatrix}x_4\\x_3\\x_2\\x_1\end{pmatrix} \mapsto \begin{pmatrix}0&0&0&1\\0&0&1&0\\0&1&0&0\\1&0&0&0\end{pmatrix}\begin{pmatrix}x_4\\x_3\\x_2\\x_1\end{pmatrix}+\begin{pmatrix}1\\1\\1\\1\end{pmatrix}.
\label{eq:16_3_4_fold_S}
\end{equation}

As the qRM codes are phantom, $\smash{\overline{S}_i \overline{S}_j}$ and $\smash{\overline{\mathrm{CZ}}_{ij}}$ fold-type gates are realizable between any two logical qubits (logical $\smash{\overline{\mathrm{SWAP}}}$s implemented by qubit permutations can be performed on top of any $\tau_{SS}$ and $\tau_{\mathrm{CZ}}$ solution). To obtain an individually addressable $\overline{S}_i$ gate, a convenient method is to inject an $\overline{S}$ gate using a targeted $\overline{\mathrm{CNOT}}$ as shown in \cref{fig:6-2-2}(b), which can be implemented using two transversal $\overline{\mathrm{CNOT}}$s according to \cref{eq:cb_additive}. It is important to insert a round of (Steane) EC between the two transversal $\overline{\mathrm{CNOT}}$s, so that the targeted $\overline{\mathrm{CNOT}}$ gate is distance preserving.

\begin{figure*}
    \includegraphics[width=0.8\textwidth]{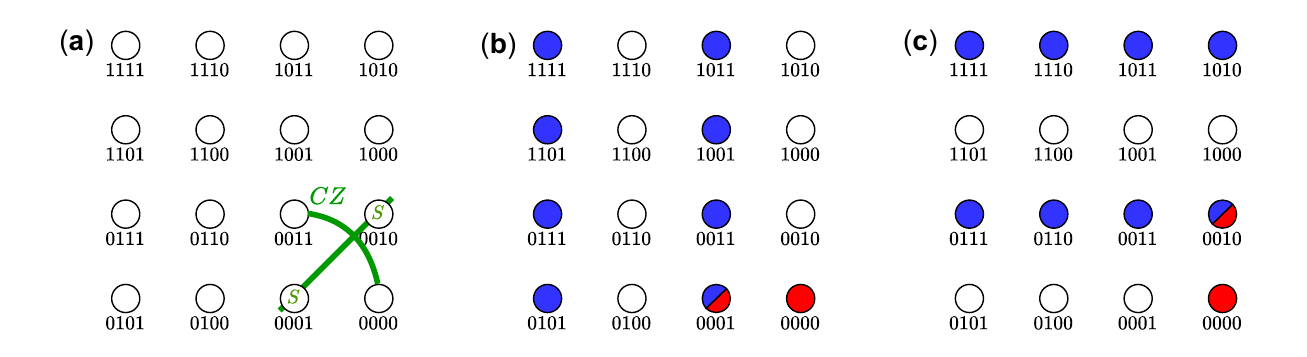}
    \caption{\textbf{Example of partial fold-type circuit on phantom qRM codes.} \textbf{(a)} An example partial fold-type circuit for the $\db{16,4,2}$ hypercube code, performing an $\overline{S} \overline{S}$ gate on the two $X$-type logicals $x_1, x_2$. The coordinates are shown in the order of $x_4x_3x_2x_1$, and this circuit only addresses the $\overline{x}_3\overline{x}_4$ subcube, i.e., the coordinates $x_4x_3x_2x_1=\text{00**}$. \textbf{(b)} The $X$-type logical $x_1$ is mapped into its paired $Z$-type logical $\overline{x}_2\overline{x}_3\overline{x}_4$. \textbf{(c)} The $X$-type logical $x_2$ is mapped into its paired $Z$-type logical $\overline{x}_1\overline{x}_3\overline{x}_4$.}
    \label{fig:Partial_fold}
\end{figure*}

Lastly, we remark that \cref{prop:phantom_qRM_fold} can be adapted to other types of phantom qRM codes discussed in this appendix, and comment on the use of fold gates in other scenarios: 
\begin{itemize}
    
    \item Note that \Cref{prop:phantom_qRM_fold} applies also to the $l<m/2$ phantom qRM codes of~\cref{thm:phantom_qRM}, but the fold-type circuit should only address the $C=\overline{x}_{2l+1}\overline{x}_{2l+2}\cdots \overline{x}_{m}$ subcube, i.e., the coordinates $x_m\dots x_1=\underbrace{\text{00}\cdots\text{0}}_{m-2l}\underbrace{\text{**}\cdots\text{*}}_{2l}$. 
    
    We refer to this as a \emph{partial} fold-type circuit. An example is shown in~\cref{fig:Partial_fold} for the $\db{16,4,2}$ code, where the circuit only addresses the subcube $\overline{x}_3\overline{x}_4$ and implements $\overline{S} \overline{S}$ on the $X$-type logicals $x_1$ and $x_2$. More generally, a monomial $f$ denoting an $X$-type logical or stabilizer generator transforms under the partial fold $\tau$ to the product of itself and a $Z$-type $\tau(f\wedge C)$. For example, one can verify that the $Z$-type contribution in the transformed $X$-type logical $x_1x_2\cdots x_{l-1}x_{l}$ is $\tau_{SS}(x_1x_2\cdots x_{l-1}x_{l}\overline{x}_{2l+1}\cdots \overline{x}_m)=\overline{x}_{2l}\overline{x}_{2l-1}\cdots\overline{x}_{l+2}\overline{x}_{l+1}\overline{x}_{2l+1}\cdots \overline{x}_m$, which is its paired $Z$-type logical. The action on other $X$-type monomials can be similarly verified and we omit them for brevity.

    \item We do not examine $l>m/2$, because one can instead consider $l'\leftarrow m-l$ and the resulting $\db{2^m,m-l'+1,\min(2^{m-l'},2^{l'})}$ code has the same distance but a higher $k$.
    
    \item On the $\db{2^D-1, D, (d_\rmx=2^{D-1}-1,d_\rmz=2)}$ punctured hypercube codes, to perform a fold-$\overline{S} \overline{S}$ gate on $X$-type logicals $x_1,x_2$, one uses $\tau_{SS}$ that maps $x_1$ to $x_2$, and $x_2$ to $x_1$, and restricts to the $x_3\cdots x_m$ subcube. For fold-$\overline{\mathrm{CZ}}$ on $x_1,x_2$, one can simply apply single-qubit $S$ gates to the $x_3\cdots x_m$ subcube.

    \item To perform an addressable $\overline{S}$ gate, one can use a targeted $\overline{\mathrm{CNOT}}$ to teleport in a $\overline{S}$ gate (see~\cref{fig:6-2-2/b}) from a freshly prepared all-plus state acted on by fold-$\overline{SS}$. This approach discussed further in the following subsection (\cref{app:qrm/folds/distillation}). Alternatively, one can decompose the desired $\overline{S}$ through the following circuit identity:
    \begin{equation*}
        \begin{quantikz}[row sep={0.65cm,between origins}, column sep={0.75cm,between origins}]
            & \gate{S} & \qw \\
            & \qw      & \qw 
        \end{quantikz}
        \;=\;
        \begin{quantikz}[row sep={0.65cm,between origins}, column sep={0.75cm,between origins}]
            & \targ{} & \gate{S} 
                \gategroup[wires=2,steps=1,style={draw,rounded corners,inner xsep=0pt,inner ysep=0pt},label style={label position=above,anchor=south,yshift=-6pt}]{fold}
                & \targ{} & \ctrl{1} & \qw & \qw \\
            & \ctrl{-1}  & \gate{S} & \ctrl{-1}  & \control{} & \gate{Z} & \qw
        \end{quantikz}
    \end{equation*}

    The $\overline{\mathrm{CZ}}$ can be performed through folding or automorphism (if available). In the case of the former, the logical performance of this decomposition suffers if the fold-$\overline{\mathrm{CZ}}$ is not distance-preserving.
    
\end{itemize}

\subsubsection{Distilling \texorpdfstring{$\overline{S} \overline{S}$}{SS} states on qRM and self-dual binarization-and-concatenation phantom codes}
\label{app:qrm/folds/distillation}

A subtlety is that fold-type gates are not distance-preserving in general, as they involve two-qubit physical interactions. Indeed, for example, we find by SAT solving (using built-in tools from \texttt{stim}~\cite{gidney2021stim}) that on the $\db{64,4,8}$ phantom qRM code, the fold-diagonal logical gates as described in \cref{prop:phantom_qRM_fold} suffer from reduced distance: $\tau_{SS}$ exhibits $(d_\mathrm{x}, d_\mathrm{z}) = (7, 4)$, while $\tau_{\mathrm{CZ}}$ exhibits $(d_\mathrm{x}, d_\mathrm{z}) = (6, 4)$. While the circuit-level distance $d = 4 > 3$ and therefore error-correction capability is still retained, the decrease in distance generally results in poorer logical error rate performance as compared to distance-preserving (e.g.~transversal or automorphism) logical gates.

\begin{figure}
    \includegraphics[width=0.6\linewidth]{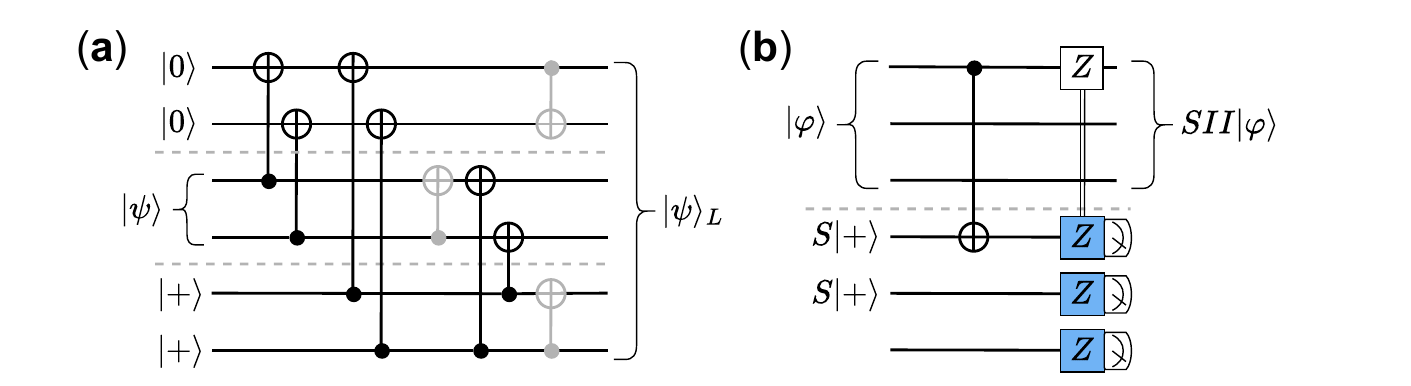}
    \phantomsubfloat{\label{fig:6-2-2/a}}
    \phantomsubfloat{\label{fig:6-2-2/b}}
    \vspace{-8pt}
    \caption{\textbf{Distillation scheme for $\overline{S} \overline{S}$ resource states and teleportation.} \textbf{(a)} Encoding circuit of $\db{6,2,2}$, to be implemented at the logical level using three $k = 2$ phantom codeblocks, as separated by the dashed grey lines. The grey $\overline{\mathrm{CNOT}}$ gates are permutation gates and can be implemented by qubit relabelling; only the three transversal $\overline{\mathrm{CNOT}}$s in black need to be performed physically. \textbf{(b)} Injection of a single $\overline{S}$ gate using a targeted $\overline{\mathrm{CNOT}}$. The ancilla codeblock hosting the $\overline{S} \overline{S}$ resource state is measured transversally in the $Z$-basis and a conditional $\overline{Z}$ correction performed. Besides the two logical qubits in the $\overline{S}|\overline{+}\ra$ state, the other logical qubit(s) of the ancilla codeblock can be in an arbitrary state that is not entangled with the first two.}
    \label{fig:6-2-2}
\end{figure}

To retain distance, we discuss an alternative scheme employing the following $\db{6,2,2}$ code to distill logical $\overline{S} \overline{S}$ resource states:
\begin{equation}
    H_\rmx=H_\rmz=\begin{pmatrix}1&1&0&1&1&0\\0&1&1&0&1&1\end{pmatrix},\quad L_\rmx=\begin{pmatrix}1&1&1&0&0&0\\0&0&0&1&1&1\end{pmatrix}.
\end{equation}

This code support a logical $\overline{S} \overline{S}$ gate implemented by transversal $S^\dagger$ operations, and can therefore serve as the outer code for distillation. In particular, we use three (resp.~two) codeblocks of $k = 2$ (resp.~$k > 2$) phantom codes to implement the distillation protocol illustrated in \cref{fig:6-2-2}. The encoding circuit of the $\db{6,2,2}$ code shown in \cref{fig:6-2-2} uses three interblock $\overline{\mathrm{CNOT}}$s between the phantom codeblocks. When $k = 2$ phantom codes are used, the interblock $\overline{\mathrm{CNOT}}$s are straightforwardly implementable transversally; for $k > 2$, we can again use transversal $\overline{\mathrm{CNOT}}$ gates rather than addressable interblock $\overline{\mathrm{CNOT}}$s (following \cref{lemma:phantom_css_interblock_cnot_circuits_two_codeblocks}), at the expense of ruining the other $k-2$ logicals; this is acceptable since targeted $\overline{\mathrm{CNOT}}$(s) will be used when teleporting in the $\overline{S}$ gate(s). After encoding, we perform fold-$\overline{S} \overline{S}$ on the phantom codes, which is realizable in the qRM and the self-dual binarization-and-concatenation codes (see \cref{app:binarize_concatenate}), and on various other numerically identified phantom codes (see \cref{tab:gate_for_phantom_codes}). Lastly, we perform unencoding of the $\db{6,2,2}$ outer code by running the $\overline{\mathrm{CNOT}}$s in reverse order, followed by transversal measurement of the first and third phantom codeblocks in the $Z$- and $X$-basis, respectively. The distillation is deemed successful when $|\overline{\text{00}}\ra$ and $|\overline{\text{++}}\ra$ are respectively recovered in the first and last codeblocks.

Although the distillation protocol entails transversal $\overline{\mathrm{CNOT}}$s, we expect it to bring in advantage
when folding is significantly noisier than the transversal $\overline{\mathrm{CNOT}}$ (e.g.~one to two orders of magnitude worse in logical error rate). Such a difference in logical error rates can arise at low physical error rates when the fold-type gates are not distance-preserving.

Lastly, we remark that the $\db{6,2,2}$ code used here is the same as that discussed in \cref{eq:binarized_3-1-2}, and is the smallest of the $\db{n,n-4,2}$ for  ($n$ even) ``$H$-code'' family introduced in Ref.~\cite{multilevel_jones}, for which logical Hadamards ($\smash{\overline{H}}^{\otimes k}$) can be implemented transversally. One can verify that $\smash{\overline{S}}^{\otimes k}$ gates are implemented by transversal $S^\dagger$ operations for this family, hence providing a route for higher-rate $\overline{S}$ distillation.

\subsection{Full logical Clifford group for qRM phantom codes}

The $m>2$ phantom qRM codes lack a transversal logical Hadamard gate. However, with the help of diagonal Clifford gates achieved through folding and targeted injection, an addressable logical Hadamard gate can be implemented. This thereby allows the implementation of the full logical Clifford group on these codes.

\begin{figure}
    \includegraphics[width=0.8\textwidth]{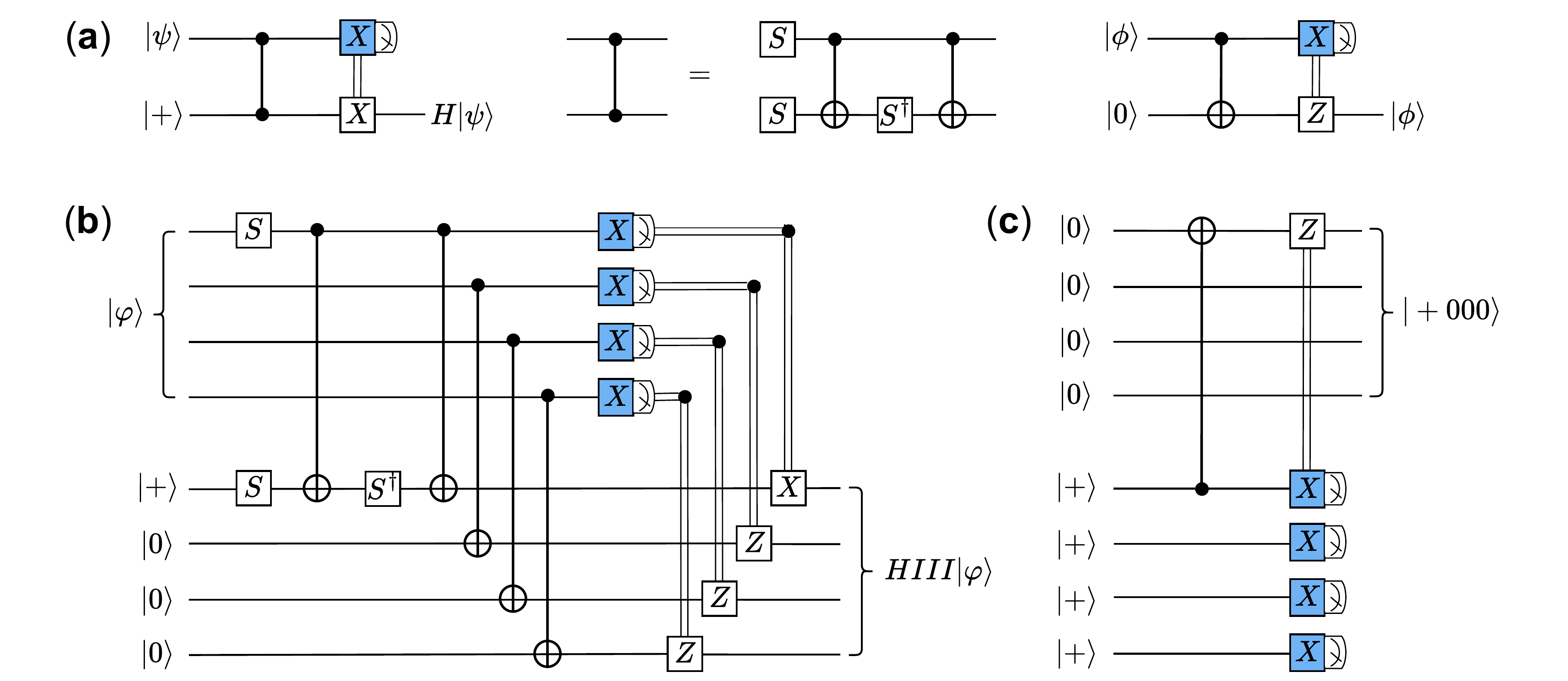}
    \phantomsubfloat{\label{fig:qRM_addressability/a}}
    \phantomsubfloat{\label{fig:qRM_addressability/b}}
    \phantomsubfloat{\label{fig:qRM_addressability/c}}
    \vspace{-8pt}
    \caption{\textbf{Addressable logical Hadamard via teleportation and mixed-basis state preparation.} \textbf{(a)} Gadget components for \textbf{(b)} teleporting a single logical Hadamard gate, where the first codeblock is measured transversally in the $X$-basis and corresponding logical Pauli corrections are performed. The $\overline{S}$ and $\smash{\overline{S}}^\dagger$ gates in (b) can be injected using a targeted $\overline{\mathrm{CNOT}}$ gate as~\cref{fig:6-2-2/b} and the $|\overline{\text{+000}}\ra$ mixed-basis state preparation following \textbf{(c)}. The scheme in (c) is used in our logical GHZ state preparation benchmarking (see~\cref{fig:benchmark/ghz}).}
    \label{fig:qRM_addressability}
\end{figure}

\Cref{fig:qRM_addressability/b} illustrates the central idea for a $k=4$ phantom code.
The protocol combines the three gadgets in \cref{fig:qRM_addressability/a}: Hadamard teleportation (left), substituting $\overline{\mathrm{CZ}}$ with $\overline{S}$, $\smash{\overline{S}}^\dag$ (middle), and $\overline{\mathrm{CNOT}}$, and the logical one-bit teleportation protocol (right). Both the teleportation gadgets measure the first codeblock in the $X$-basis, and combining them in \cref{fig:qRM_addressability/b} allows one to measure all logicals in the $X$-basis through a transversal $X$-basis measurement. The $\overline{S}$, $\smash{\overline{S}}^\dag$ gates can be performed by (distillation followed by) targeted injection as explained in the previous subsection. The mixed-basis state $|\overline{\text{+000}}\rangle$ can be prepared using a targeted $\overline{\mathrm{CNOT}}$ gate as shown in~\cref{fig:qRM_addressability/c}. The states of the other three logical qubits do not affect the protocol, so one could prepare the second codeblock in $|\overline{\text{++++}}\rangle$ and upon teleportation, measure and preselect on the other qubits being in $\ket{\overline{+}}$ to further improve the chance of having the first logical qubit in $\ket{\overline{+}}$ (and thereby a successful $\overline{H}$ teleported).

To perform an $\smash{\overline{H}}^{\otimes 2} \smash{\overline{\mathbb{I}}}^{\otimes 2}$ gate instead, we add $\overline{S}$ gates on the second logical qubit on both codeblocks, and another $\overline{\mathrm{CNOT}}$ between them. The $\overline{\mathrm{CNOT}}$ can again be achieved physically through two transversal $\overline{\mathrm{CNOT}}$ gates (see \cref{lemma:phantom_css_interblock_cnot_circuits_two_codeblocks}). Similar to the first logical qubit, the conditional correction on the second logical qubit changes to $X$-type.

On our phantom codes, this scheme of performing an addressable logical Hadamard has a lower overhead compared to the conventional scheme of utilizing the identity $(HS)^3 \propto \mathbb{I}$---concretely, for example, at $k=2$ and with the ability to perform addressable $\overline{S}$, the sequence $(H\otimes H)(S \otimes \mathbb{I})(H\otimes H)(S\otimes \mathbb{I})(H\otimes H)(S\otimes \mathbb{I})\propto \mathbb{I}\otimes H$ produces an addressable Hadamard.

\subsection{Magic gates on qRM phantom codes} 

\begin{figure*}
    \includegraphics[width=0.95\textwidth]{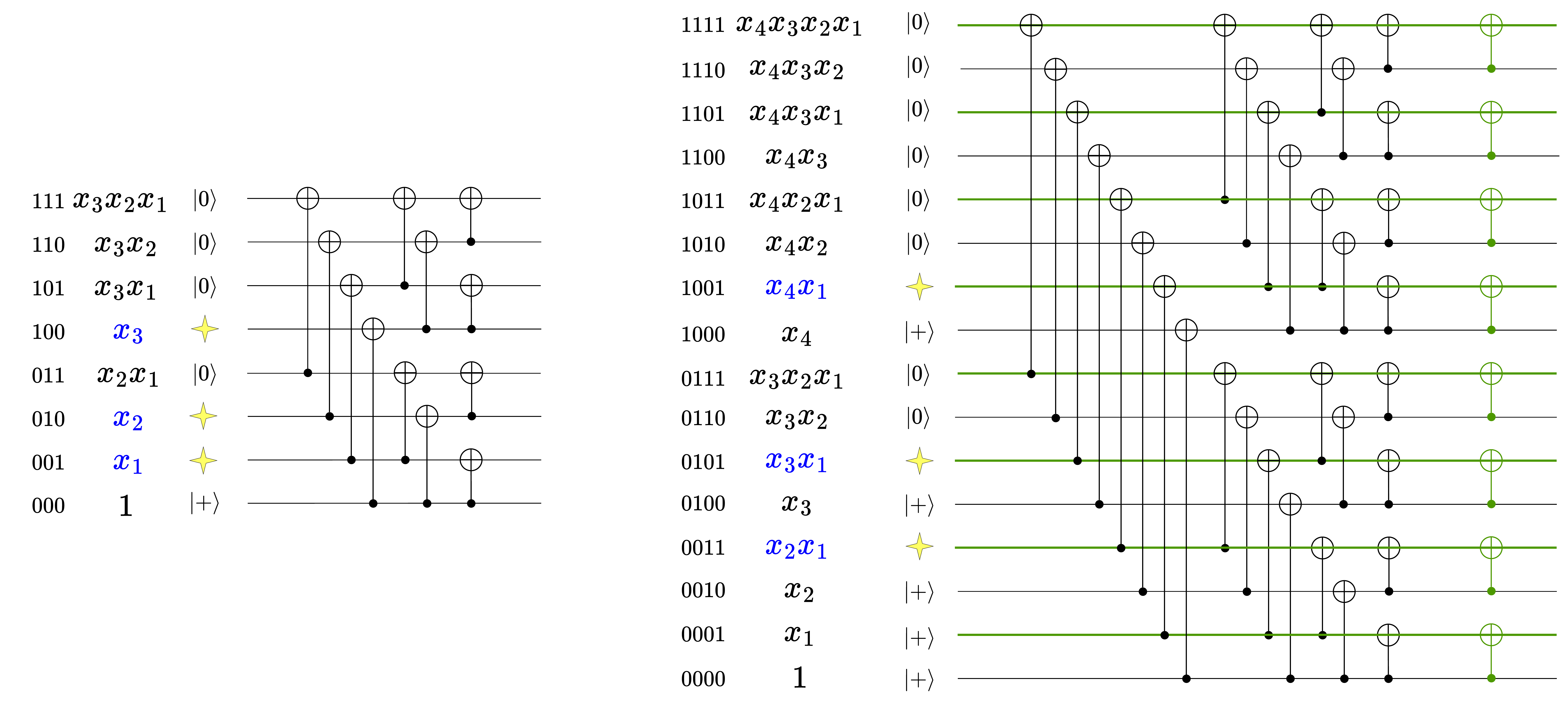}
    \caption{\textbf{Encoding circuits and decoupling scheme for phantom qRM codes}. The hypercube encoding circuit of the $\db{8,3,2}$ (left) and $\db{16,3,4}$ (right) phantom qRM codes. This recursively constructed encoding circuit applies in fact to the entire qRM code family~\cite{gong2024computation}. On the input side of the circuit, $x_3x_2x_1=111,\dots,000$ etc., label the coordinate of each qubit; each monomial indicates what an $X$ operator placed on this qubit propagates to through this encoding circuit. This can be thought of as stabilizer propagation: if a qubit is initialized to $|+\rangle$, which is stabilized by a single-qubit $X$, the propagation result of the $X$ is an $X$-type stabilizer of the code. Next, we take the $\db{16,3,4}$ (right) as an example to illustrate decoupling. The hypercube circuit couples qubits whose coordinates differ at $x_4$ in the first round using transversal $\mathrm{CNOT}$s, couples those that differ at $x_3$ in the next round, and so on. By undoing the last round of transversal $\mathrm{CNOT}$s (green) and focusing on the coordinates where $x_1=1$, the $\db{16,3,4}$ is decoupled to the $\db{8,3,2}$ code. For the $\db{2^m,m-l+1,\min(2^{m-l},2^l)}$ phantom code, we undo the $x_1,\dots,x_{l-1}$ layers of $\mathrm{CNOT}$s, and keep those coordinates where $x_1=1,\dots,x_{l-1}=1$; the resulting code is the $\db{2^{m-l+1},m-l+1,2}$ hypercube code.}
\label{fig:qRM_magic_decompose}
\end{figure*}

The infinite family of quantum Reed--Muller phantom codes we identified allows a distance-limited mechanism for magic gates. We start from the observation that the $\db{16,3,4}$ phantom code can be decoupled into two $\db{8,3,2}$ codes, as illustrated in \cref{fig:qRM_magic_decompose}, by undoing the encoding $\mathrm{CNOT}$s connecting the $x_1=0$ and $x_1=1$ subcubes. One can then perform a $\overline{\mathrm{CCZ}}$ gate using transversal physical $T$ gates on the target $\db{8,3,2}$ code, and then couple them back to restore the $\db{16,3,4}$ code. Alternatively, one can compile this $\mathrm{CNOT}$--$T$--$\mathrm{CNOT}$ sequence as $\smash{\mathrm{CS}^\dagger_{12} T_1 T_2}$ on the physical qubit pairs acted on.

More generally, our phantom qRM codes of parameters $\db{2^m,m-l+1,\min(2^{m-l},2^l)}$ (balanced or not) can be decoupled into the hypercube codes $\db{2^{m-l+1},m-l+1,2}$ by undoing the last $l-1$ layers of their hypercube encoding circuits and thereby only keeping the $x_1=x_2=\cdots=x_{l-1}=1$ subcube. There, logical $\mathrm{C}^{l_0}\mathrm{Z}$ gates are accessible for all $l_0 \le m-l$ through physical $Z$-basis rotations on the qubits of suitable subcubes~\cite{barg2024rm}. For example, the $\db{64,4,8}$ code reduces to a $\db{16,4,2}$ block, where a transversal $\smash{\sqrt{T}/\smash{\sqrt{T}}^\dagger}$ implements a $\overline{\mathrm{CCCZ}}$ gate, and transversal $T/T^\dagger$ on an $8$-qubit subcube gives a $\overline{\mathrm{CCZ}}$~\cite{barg2024rm}. The limitation is that the intermediate codes have a distance of only two.

Lastly, we comment that on the parent $\db{64,15,4}$ qRM code (where the $X$-type logicals are the $6\choose 2$ degree-2 monomials) of the $\db{64,5,4}$ phantom code, transversal $T$ gates lead to a logical hypergraph $\overline{\mathrm{CCZ}}$ circuit. The $\overline{\mathrm{CCZ}}$s are on every triplet of degree-2 monomials that multiply to $x_1x_2x_3x_4x_5x_6$~\cite{rengaswamy2020optimality,barg2024rm}. For example, $(x_1x_2,x_3x_4,x_5x_6)$, $(x_1x_2,x_3x_5,x_4x_6)$, etc.; each triple intersection has support at one coordinate $x_6=\cdots=x_1=1$. The phase polynomial framework in~\cref{app:gates/diagonal/transversal} allows a simpler proof of this result than in Refs.~\cite{rengaswamy2020optimality,barg2024rm}. 
However, the $\db{64,5,4}$ phantom qRM code, where $x_1x_2,x_1x_3,x_1x_4,x_1x_5,x_1x_6$ serve as the $X$-type logicals, does not have such a triplet. 
One could consider swapping the monomials in the parent $\db{64,15,4}$ qRM code (e.g.~swap $x_1x_3$ with $x_3x_5$, and $x_1x_4$ with $x_4x_6$) through interleaving qubit permutations with transversal $\overline{\mathrm{CNOT}}$s with an ancilla codeblock~\cite{gong2024computation}, applying transversal $T$ gates, and then swapping back. The disadvantage of this scheme, however, is the high depth of the circuit. Therefore, though in principle we can realize $\overline{\mathrm{CCZ}}$s of fault distance four this way, it might be cheaper in practice to distill $\overline{\mathrm{CCZ}}$ states through the distance-two decoupling scheme above and inject them where needed.

\subsection{Space-time efficient analog rotation}
\label{app:qrm/STAR}

Space-time efficient analog rotation (STAR) has been shown to be useful for implementing arbitrary small-angle single-qubit rotations in early fault-tolerant quantum computers~\cite{akahoshi2024partially, choi2023fault}.
In this section, we provide an implementation for the $\db{64,4,8}$ qRM phantom code and discuss some restrictions of STAR on related codes.

\subsubsection{STAR for the \texorpdfstring{$\db{64,4,8}$}{[[64,4,8]]} qRM phantom code}

\Cref{fig:star_teleportation} illustrates the STAR protocol for a small-angle $\overline{R}_Z$ gate. We discuss the rotation as applied to the first logical qubit of the code; rotations on other logical qubits work similarly (alternatively, as the code is phantom, logical swaps can be performed through qubit permutations). Concretely, physical $R_Z(\theta)$ rotations are applied on the support of the weight-8 logical operator $\overline{Z}_1$, 
\(
  R_Z^{\otimes 8}(\theta)
  = [\cos(\theta/2) \mathbb{I} - i \sin(\theta/2) Z]^{\otimes 8},
\)
which can be expanded as a sum of tensor-product terms with different Pauli weights.
At this point, a round of error detection is performed.
Note that since the code distance is $8$, $Z$-type Pauli operators with weight between $1$ and $7$ correspond to detectable errors.
After preselecting for no detected error, the induced logical operation is proportional to
\(
  \cos^8(\theta/2)\overline{\mathbb{I}}
  +
  \sin^8(\theta/2)\overline{Z}_1.
\)
Since the logical qubit is initialized in the $\ket*{\smash{\overline{+}}}$ state, this effective logical action is equivalent to $\overline{R}_Y(-\gamma)$ for an angle $\gamma$ given by
\(
  \tan(\gamma/2) = \tan^8(\theta/2).
\)

Following this resource state preparation, a transversal $\overline{\mathrm{CNOT}}$ from the target (data) codeblock to the ancillary codeblock is applied, followed by a fold-$\overline{S}_1 \overline{S}_2$ operation (see \cref{app:qrm/folds}), and finally the ancillary codeblock is measured in the $X$-basis.
Note that the $\overline{\mathrm{CNOT}}$s only act non-trivially on the first logical qubit in the ancillary codeblock since the other qubits are in $|+\rangle$ states.

The net effect of the circuit on the target codeblock is a $\overline{R}_Z(\pm \gamma)$ on the first logical qubit conditioned on the measurement outcome $m$ of the first logical qubit in the ancillary codeblock: with probability $1/2$ we obtain $m=1$ and realize the desired $Z$-rotation, while with probability $1/2$ we obtain $m=-1$, in which case we must compensate by applying a $2\gamma$ rotation.
This compensation procedure again succeeds with probability $1/2$, and otherwise requires a further ``fix-up'' rotation with twice the angle.
From this probabilistic structure, it follows that, on average, we need to perform the STAR circuit twice to implement a single logical $Z$-rotation gate.

\begin{figure}
    \includegraphics[width=0.45\textwidth]{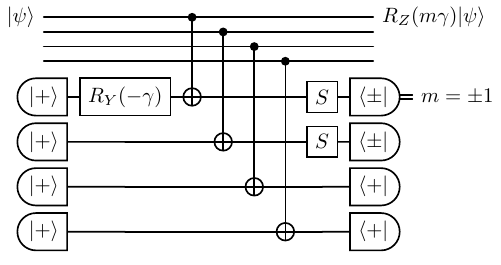}
    \caption{\textbf{Teleportation of STAR-prepared small-angle rotation resource state.} The circuit illustrated here is for a $k = 4$ phantom code (e.g.~the $\db{64,4,8}$ code). The target codeblock, on which we wish to implement a logical rotation, is depicted on top, and an ancillary codeblock is introduced below.
    The ancillary block is initialized to the all-$\ket*{\overline{+}}$ state, after which the STAR scheme is used to implement the $R_Y$ logical gate.}
    \label{fig:star_teleportation}
\end{figure}

\subsubsection{Limitations of STAR for CSS codes with even-weight stabilizer generators}

The reason for the $\overline{S}$ gates in \cref{fig:star_teleportation} is that the STAR implements an $\overline{R}_Y$ on the ancillary block, while we intend to realize $\overline{R}_Z$ on the target block. 
This is in fact a limitation due to the code structure.
We begin with the following fact:

\begin{lemma}
    \label{lemma:star}
    For a CSS code such that any pure-$X$ logical operator or stabilizer is of even weight and any pure-$Z$ logical operator or stabilizer is of even weight, any such logical operator is of even weight: a single $Z$-logical and any $X$-logical on the other logical qubits; any such logical operator are of odd weight: a single $Y$-logical and any $X$-logical on the other logical qubits.
\end{lemma}

\begin{proof} 
Without loss of generality, consider the operator $\overline{Z}_1 \cdot \overline{X}_{\mathcal{L}}$ where $\overline{X}_{\mathcal{L}}$ is any product of $\overline{X}$ operators excluding $\overline{X}_1$.
Such a logical operator corresponds to physical Pauli operators of the form $\mathcal{Z} \cdot \mathcal{X}$. 
Since $\mathcal{X}$ is pure-$X$ and $\mathcal{Z}$ is pure-$Z$, and they commute (because $\overline{Z}_1$ and $\overline{X}_{\mathcal{L}}$ commute), their support overlap on an even number locations, which turns to an even number of physical $Y$s and a $\pm$ sign.
Thus, $\overline{Z}_1 \cdot \overline{X}_{\mathcal{L}}$ is the remaining (even number of) $X$s in $\mathcal{X}$, the remaining (even number of) $Z$s in $\mathcal{Z}$ and an even number of $Y$s, which sums up to an even weight.

Without loss of generality, consider operator $\overline{Y}_1 \cdot \overline{X}_{\mathcal{L}} = i \overline{X}_{\mathcal{L}\cup\{1\}} \overline{Z}_1$ where, again, $\overline{X}_{\mathcal{L}}$ is any product of $\overline{X}$ operators excluding $\overline{X}_1$.
Such a logical operator corresponds to physical Pauli operators of the form $i\mathcal{X}\cdot \mathcal{Z}$.
Since $\mathcal{X}$ is pure-$X$, $\mathcal{Z}$ is pure-$Z$, and they anticommute (because $\overline{X}_{\mathcal{L}\cup\{1\}}$ and $\overline{Z}_1$ anticommute), their support overlap on an odd number of locations, which turns to an odd number of physical $Y$s and an odd number of $i$s.
Adding also the $i$ in the expression, the total sign is a $\pm$.
Thus, $\overline{Y}_1 \cdot \overline{x}_{\mathcal{L}}$ is the remaining (odd number of) $X$s in $\mathcal{X}$, the remaining (odd number of) $Z$s in $\mathcal{Z}$, and an odd number of $Y$s, which sums up to an odd weight.
\end{proof}

This restricts the accessible rotation bases of the STAR scheme, as we formalize below.

\begin{proposition}
    For a CSS code such that any pure-$X$ logical operator or stabilizer is of even weight and any pure-$Z$ logical operator or stabilizer is of even weight, the logical action of such a STAR protocol cannot be $\overline{R}_Z$: 1)~prepare logical all-$|+\rangle$ state, 2)~apply general single-qubit physical gates, 3)~perform error detection.
\end{proposition}   

\begin{proof}
Per \cref{lemma:star}, since any physical Pauli operator corresponding to $\overline{Z}_1 \overline{X}_{\mathcal{L}}$ with $1\notin\mathcal{L}$, has an even weight, the coefficient for such a term in $\smash{\prod_{q} e^{i\theta \vec \sigma_q\cdot \vec n_q}}$ is real, which rules out actions like $(\cos \alpha \overline{\mathbb{I}} +i\sin \alpha \overline{Z}_1 \overline{X}_{\mathcal{L}}) \ket*{\overline{+}}^{\otimes k}=(\cos \alpha \overline{\mathbb{I}}+i\sin \alpha \overline{Z}_1) \ket*{\overline{+}}^{\otimes k}$.
Moreover, since any $\overline{Y}_1 \overline{X}_{\mathcal{L}}$ with $1\notin\mathcal{L}$ has an odd weight, the coefficient for such a term is purely imaginary, which rules out actions like $(\cos\alpha \overline{\mathbb{I}} + \sin\alpha \overline{Y}_1 \overline{X}_{\mathcal{L}}) \ket*{\overline{+}}^{\otimes k} =(\cos \alpha \overline{\mathbb{I}}-i\sin \alpha \overline{Z}_1) \ket*{\overline{+}}^{\otimes k}$ (because $\overline{Y}_1 \overline{X}_{\mathcal{L}} \ket*{\overline{+}}^{\otimes k} = -i \overline{Z}_1 \ket*{\overline{+}}^{\otimes k}$).
\end{proof}

\section{Binarization and concatenation scheme for CSS phantom codes}
\label{app:binarize_concatenate}

Here, we construct a class of $k = 2$ CSS phantom codes with high distance. The overarching recipe is to take a qudit CSS code over $\GF(4)$ encoding one logical qudit, and encode every $\GF(4)$ qudit---implementable using two qubits via binarization---through the qubit $\db{4,2,2}_2$ code (i.e.~a code concatenation procedure). 

We refer to this scheme as binarization-and-concatenation (B\&C). We show that, if the qudit code is of parameters $\db{n,1,d}_4$, then the resulting qubit code from this B\&C construction has parameters $\db{4n,2,\ge 2d}_2$. We give a sufficient condition that the $\GF(4)$ code should satisfy in order for the resulting qubit code to be phantom. We then use a family of small $\GF(4)$ cyclic codes that naturally satisfy this condition to concretely construct a family of phantom qubit codes; these codes are (Hermitian) self-dual and have extremal distances at small block lengths $n$.

We begin with an introduction to $\GF(4)$ codes and the binarization step in \cref{app:binarize_concatenate/binarization}, and then explain the motivation for concatenating with the $\db{4,2,2}_2$ code in \cref{app:binarize_concatenate/concatenation}, thereby deriving the condition on the $\GF(4)$ code to guarantee that the resulting qubit code is phantom. Next, we introduce a family of small $\GF(4)$ cyclic codes in \cref{app:binarize_concatenate/small_cyclic} and show that they satisfy this condition. Finally, in \cref{app:binarize_concatenate/gates}, we explore the logical gates of the resulting qubit codes.

\subsection{\texorpdfstring{$\GF(4)$}{GF(4)} CSS quantum codes and binarization}
\label{app:binarize_concatenate/binarization}

Recall that $\GF(4) \cong \F_2[z] / (z^2 + z + 1) \cong \{0,1,\omega,\omega^2\}$ where $\omega^2=\omega+1$. We establish the following basics:
\begin{itemize}[noitemsep]
    \item The conjugate $\overline{x}$ of an element $x \in \GF(4)$ is its square: $\overline{x} = x^2$.
    \item The trace of an element $x\in\GF(4)$ is $\tr(x):=x+x^2\in\{0,1\}$. The trace is an $\mathbb{F}_2$-linear function, meaning for $a,b\in\F_2$, $x,y\in\GF(4)$, $\tr(ax+by)=a\tr(x)+b\tr(y)$.
\end{itemize}

We work with $\GF(4)$ CSS quantum codes~\cite[Ch.~8.3.1]{gottesman2024surviving}, whose $X$- and $Z$-type stabilizer groups are themselves $\GF(4)$-linear classical codes. By $\GF(4)$-linear, we not only require stabilizers to be closed under multiplication, but also demand that if $\bigotimes_{j=1}^n Z^{\gamma_j}$ is a $Z$-type stabilizer, $\gamma_j\in\GF(4)$, then for any $\gamma_0\in \GF(4)$, $\bigotimes_{j=1}^n Z^{\gamma_0 \gamma_j}$ is also a $Z$-type stabilizer; and likewise for $X$-type stabilizers. That a $Z$-type stabilizer $\otimes_{j=1}^n Z^{\gamma_j}$ and a $X$-type stabilizer $\otimes_{j=1}^n X^{\eta_j}$ commute means that $\smash{\sum_{j=1}^n\gamma_j\eta_j}=0$, $\gamma_j,\eta_j\in\GF(4)$.

A $\GF(4)$-qudit can be implemented using two qubits through binarization~\cite[Ch.~8.1.2]{gottesman2024surviving}. This process can be performed as follows. Taking the self-dual normal basis $\{\omega,\omega^2\}$, where $\tr(\omega\omega)=\tr(\omega^2\omega^2)=1$ and $\tr(\omega\;\omega^2)=0$, we note that any $\alpha\in \GF(4)$ can be written as $\smash{\alpha=\alpha_1\omega+\alpha_2\omega^2}$ where $\smash{\alpha_1=\tr(\alpha\omega),\alpha_2=\tr(\alpha\omega^2)\in\F_2}$. In this way, expressing $\alpha$ as $(\alpha_1, \alpha_2)$ therefore provides a bijective map between $\GF(4)$ and two $\F_2$ spaces. Explicitly:
\begin{itemize}[noitemsep]
    \item $X^\alpha$ (resp.~$Z^\alpha$) where $\alpha\in\GF(4)$ is mapped to $X^{\alpha_1}X^{\alpha_2}$ (resp.~$Z^{\alpha_1}Z^{\alpha_2}$).
    \item An $X$-type stabilizer $\otimes_{j=1}^n X^{\eta_j}$ where $\eta_j\in\GF(4)$ is mapped to $X^{\eta_{1,1}}X^{\eta_{1,2}}\cdots X^{\eta_{n,1}}X^{\eta_{n,2}}$, by expanding each $\smash{\eta_j=\eta_{j,1}\omega+\eta_{j,2}\omega^2}$, $\eta_{j,1},\eta_{j,2}\in \F_2$. Likewise for $Z$-type stabilizers.
\end{itemize}

Indeed, the commutation relationship of $X$- and $Z$-type stabilizers on the $\GF(4)$ code translates to the qubit level, 
\begin{equation}
    \sum_{j=1}^n \gamma_j\eta_j=0
    \implies
    \sum_{j=1}^n\left[\gamma_{j,1}\eta_{j,1}+\gamma_{j,2}\eta_{j,2}\right]=0,
\end{equation}
and therefore the resulting code after binarization is a valid CSS qubit code. This result can be seen by expanding the left-hand side $\smash{\sum_{j=1}^n (\gamma_{j,1}\omega+\gamma_{j,2}\omega^2)(\eta_{j,1}\omega+\eta_{j,2}\omega^2)}=0$, taking the trace on both sides, and invoking the $\F_2$-linearity of the trace.

\medskip
\textit{Examples.} We illustrate an example on the following $\db{3,1,2}_4$ code, defined by stabilizer generator matrices
\begin{equation}
    H_\rmx
    = H_\rmz
    = \begin{pmatrix}1&\omega&\omega^2\\\end{pmatrix}.
\end{equation}

We remind that, since the code is $\GF(4)$-linear, each stabilizer multiplied by $\omega$ or $\omega^2$ is also a stabilizer. For each row, we binarize each $\GF(4)$ entry as described above. For example, $\smash{\begin{pmatrix}\omega&\omega^2&1\\\end{pmatrix}}$ becomes $\smash{\begin{pmatrix}10 & 01 & 11\end{pmatrix}}$, and $\smash{\begin{pmatrix}\omega^2&1&\omega\\\end{pmatrix}}$ becomes $\smash{\begin{pmatrix}01 & 11 & 10\end{pmatrix}}$. The resulting qubit code after binarization has parameters $\db{6,2,2}_2$ and are defined by stabilizer generator matrices
\begin{equation}
\label{eq:binarized_3-1-2}
    H_\rmx^{\rm{bin}}
    = H_\rmz^{\rm{bin}}
    = \begin{pmatrix}
        10 & 01 & 11 \\
        01 & 11 & 10
    \end{pmatrix},
\end{equation}
where we retained linearly independent rows. Likewise, we can binarize the following $\db{5,1,3}_4$ code,
\begin{equation}
    H_\rmx=\begin{pmatrix}1&1&1&1&0\\0&1&\omega&\omega^2&1\end{pmatrix},
    \qquad 
    H_\rmz=\begin{pmatrix}1&1&1&1&0\\0&1&\omega^2&\omega&1\end{pmatrix}.
\end{equation}
to obtain a $\db{10,2,3}_2$ code,
\begin{equation}
    H_\rmx^{\rm{bin}}=\begin{pmatrix}
        10 & 10 & 10 & 10 & 00\\
        01 & 01 & 01 & 01 & 00\\
        00 & 10 & 01 & 11 & 10\\
        00 & 01 & 11 & 10 & 01
    \end{pmatrix},
    \qquad 
    H_\rmz^{\rm{bin}}=\begin{pmatrix}
        10 & 10 & 10 & 10 & 00\\
        01 & 01 & 01 & 01 & 00\\
        00 & 10 & 11 & 01 & 10\\
        00 & 01 & 10 & 11 & 01
    \end{pmatrix}.
\end{equation}

These two examples will be reused throughout the following subsections.

\subsection{Concatenation with the \texorpdfstring{$\db{4,2,2}$}{[[4,2,2]]} code}
\label{app:binarize_concatenate/concatenation}

The binarized $\db{6,2,2}_2$ and $\db{10,2,3}_2$ codes are not phantom (the phantomness of a given code can be concretely checked by a SAT formulation, see \cref{app:enumeration}). To obtain phantom codes, we concatenate the binarized code with the $\db{4,2,2}_2$ code, by passing every qubit pair encoding one $\GF(4)$-qudit through the $\db{4,2,2}_2$ encoding circuit.
\begin{wrapfigure}{l}{0.22\textwidth}
\begin{quantikz}[row sep={0.4cm,between origins}, column sep={0.6cm,between origins}]
    \lstick{$|0\ra$}  & \targ{}    & & \targ{} &  \\
    \lstick[2]{inputs} & \ctrl{-1} & \targ{}    && \\
    & \targ{}         &           & \ctrl{-2}  &  \\
    \lstick{$|+\ra$}  & \ctrl{-1}  & \ctrl{-2} &&  
\end{quantikz}
\end{wrapfigure}
i.e., $\db{4,2,2}_2$ is used at the inner level to encode the two qubits representing a $\GF(4)$-qudit at the outer level. 
The binarization and concatenation taken as a whole perform the following mapping from a $\GF(4)$-qudit to four qubits: $\smash{X^\omega\mapsto XXII}$ and $\smash{Z^{\omega^2}\mapsto IIZZ}$; $\smash{X^{\omega^2}\mapsto XIXI}$ and $\smash{Z^\omega\mapsto IZIZ}$; and $\smash{X^1\mapsto IXXI}$ and $\smash{Z^1\mapsto IZZI}$, because $1=\omega^2+\omega$. As a standard consequence of code concatenation, the resulting qubit code also contains sets of weight-four $X$- and $Z$-type stabilizer generators arising from the stabilizers of the $\db{4,2,2}_2$ codeblocks used for each $\GF(4)$ qudit.

\begin{proposition}[Distance of binarization-and-concatenation codes]
\label{prop:B&C_distance}
    A $\GF(4)$ CSS quantum code of distance $d$ produces a CSS qubit code of distance at least $2d$ after binarization and concatenation with the $\db{4,2,2}_2$ code.
\end{proposition}

\begin{proof}
    As seen above, the weight of any single nontrivial Pauli operator on a $\GF(4)$-qudit is doubled to weight-two on two qubits through the binarization and concatenation procedure. Therefore, if all logical operators on the $\GF(4)$-qudit code are of weight ${\geq} \, d$, then all logical operators on the resulting qubit code are of weight ${\geq} \, 2d$.
\end{proof}

The $\db{6,2,2}_2$ and $\db{10,2,3}_2$ codes yield $\db{12,2,4}_2$ and $\db{20,2,6}_2$ codes after concatenation, which are phantom, and are in fact independently discovered by our SAT-based code discovery effort (see \cref{app:code_discovery}). We delay the rigorous proof of phantomness for the concatenated codes after introducing the $\GF(4)$ quadratic residue code family in \cref{app:binarize_concatenate/small_cyclic}, in which the $\db{3,1,2}_4$ and $\db{5,1,3}_4$ serve as the smallest examples. 

For now, we give the intuition behind why the binarization-and-concatenation procedure works in creating phantom codes. A central fact we utilize is that the following two operations on a $\GF(4)$-qudit $\gamma$: 
\begin{enumerate}[noitemsep]
    \item $\gamma\mapsto\alpha\gamma$ where $\alpha\in\{1,\omega,\omega^2\}$, and
    \item the Frobenius transform $\gamma\mapsto \gamma^2$,
\end{enumerate}
together generate the invertible transforms $\GL(2,\mathbb{F}_2)$ on the two qubits implementing $\gamma$. On the logical space of the $\db{4,2,2}_2$ phantom code, $\GL(2,\mathbb{F}_2)$ is the space of $\overline{\mathrm{CNOT}}$ logical circuits and are implementable by physical qubit permutations. To be explicit, consider the transformation $\gamma \mapsto \omega^2\gamma$ as an example. In the self-dual normal basis $\{\omega, \omega^2\}$, this is expanded as $\gamma_1\omega+\gamma_2\omega^2\mapsto \gamma_1\omega^3+\gamma_2\omega^4=(\gamma_1+\gamma_2)\omega+\gamma_1\omega^2$ where $\gamma_1, \gamma_2 \in \F_2$; therefore binarizing the $\GF(4)$-qudit into two qubits, the transformation corresponds to $(\gamma_1,\gamma_2) \mapsto (\gamma_1+\gamma_2,\gamma_1)$, which is precisely the action of a $\overline{\mathrm{CNOT}}$ when the two qubits are encoded in the logical space of the $\db{4,2,2}_2$ code. Likewise, one finds that $\gamma \mapsto \omega\gamma$ corresponds to $(\gamma_1,\gamma_2) \mapsto (\gamma_2,\gamma_1+\gamma_2)$, and the Frobenius transform $\gamma_1\omega+\gamma_2\omega^2\mapsto \gamma_1\omega^2+\gamma_2\omega$ corresponds to a swap of qubits, $(\gamma_1,\gamma_2) \mapsto (\gamma_2,\gamma_1)$.

Let us consider an arbitrary $\db{n,1}_4$ code. In general, we can write the computational-basis logical states of a $k = 1$ $\GF(4)$ CSS code as $\smash{|\gamma\ra_\mathrm{L}=\sum_{s_\mathrm{x}\in\rs(H_\mathrm{x})}|\gamma l_\mathrm{x}+s_\mathrm{x}\ra}$, where $\gamma\in\GF(4)$ and $\bl_\mathrm{x}$ is the logical $X$ operator of the code represented as a $\GF(4)$ vector. For the resulting qubit code after binarization and concatenation to be phantom, it must be that $|\gamma\ra_\mathrm{L}=|\gamma_1\omega+\gamma_2\omega^2\ra_\mathrm{L}$ on the qubit code can be transformed into $|\gamma_1'\omega+\gamma_2'\omega^2\ra_\mathrm{L}$ using qubit permutations, where $\trans{(\begin{matrix} \gamma_1 & \gamma_2 \end{matrix})} \mapsto A \trans{(\begin{matrix} \gamma_1' & \gamma_2' \end{matrix})}$ for any invertible $A\in\GL(2,\mathbb{F}_2)$. This means that the qubit code can implement a complete set of individually addressable $\overline{\mathrm{CNOT}}$ gates via qubit permutations.

One sees that any $\db{n,1}_4$ code, after going through the B\&C process with $\db{4,2,2}$, can do the $|\gamma\ra_L\mapsto|\alpha\gamma\ra_L$ using just permutations: map each qudit $\beta$ to $\alpha\beta$ by permuting the $\db{4,2,2}$ code it is encoded in. By $\GF(4)$-linearity, the stabilizers are preserved. The logical $X$ operator $\gamma\;l_\rmx$ is mapped to $\alpha\gamma\; l_\rmx$.
Recall that $\gamma\mapsto\alpha\gamma$ where $\alpha\in\{1,\omega,\omega^2\}$, and $\gamma\mapsto \gamma^2$ (Frobenius transform) together generate the invertible transform $\GL(2,\mathbb{F}_2)$ on the two qubits implementing $\gamma$. Therefore, if we can find a way to implement $|\gamma\ra_L\to|\gamma^2\ra_L$ using permutation on the B\&C code, then the resulting B\&C code is phantom. 
For this purpose, we give the following sufficient condition on the $\db{n,1}_4$ code.

\begin{proposition}
\label{prop:phantom_B&C_condition}
Call $H_\rmx$ (resp. $H_\rmz$) the $X$ (resp. $Z$) stabilizer of the $\GF(4)$-linear CSS code $\db{n,1}_4$, and $l_\rmx\in\F_4^n$ an $X$ logical representative. Denote the coordinate-wise conjugation of the two matrices and vector $l_\rmx$ by $\overline{H_\rmx}$, $\overline{H_\rmz}$ and $\overline{l_\rmx}$. If there exists a permutation $\pi$ of the $n$ coordinates such that the stabilizers are preserved, i.e., $\text{rs}(\pi(\overline{H_\rmx}))=\text{rs}(H_\rmx)$ and $\text{rs}(\pi(\overline{H_\rmz}))=\text{rs}(H_\rmz)$, and $\pi(\overline{\gamma l_\rmx})$ equals $\gamma^2l_\rmx$ up to $X$ stabilizers, then the resulting B\&C code is phantom.
\end{proposition}
\begin{proof}
We first do a coordinate-wise conjugation $c_i\mapsto c_i^2$, $i\in [n]$ using permutation on $\db{4,2,2}$, this operation transforms $H_\rmx$ to $\overline{H_\rmx}$, $H_\rmz\to \overline{H_\rmz}$ and $\gamma l_\rmx\to \overline{\gamma l_\rmx}=\gamma^2\overline{l_\rmx}$.

Next, we apply $\pi$ to each qudit coordinate. On the B\&C code, this is effectively permuting the $\db{4,2,2}$ blocks associated to each qudit coordinate.

\end{proof}

\subsection{Small cyclic codes}
\label{app:binarize_concatenate/small_cyclic}

We now present a concrete family of $\GF(4)$-linear codes that (i) yield good distances and (ii) satisfy the sufficient condition of \cref{prop:phantom_B&C_condition}, guaranteeing that the corresponding B\&C qubit codes are phantom. The two running examples from \cref{app:binarize_concatenate/binarization}, namely the $\db{3,1,2}_4$ and $\db{5,1,3}_4$ codes, are the smallest members of this family.

\textit{Cyclic codes over $\GF(4)$.} 
We briefly recall standard facts about cyclic codes following~\cite[Ch.~7]{macwilliams1977theory}. Cyclic codes can be defined over any finite field $F$; in the present context we take $F = \GF(4)$. A linear code $\mathscr{C}\subseteq F^n$ is \emph{cyclic} if it is closed under cyclic shifts: whenever $(c_0,c_1,\dots,c_{n-1})\in\mathscr{C}$, then $(c_{n-1},c_0,\dots,c_{n-2})\in\mathscr{C}$.
Identify $(c_0,\dots,c_{n-1})$ with the residue class of the polynomial
\[
c(x)=c_0+c_1x+\cdots+c_{n-1}x^{n-1}\in F[x]/(x^n-1).
\]
Then $\mathscr{C}$ is an ideal in $F[x]/(x^n-1)$, hence $\mathscr{C}=\langle g(x)\rangle$ for a unique monic \emph{generator polynomial} $g(x)$ dividing $x^n-1$.

Let $h(x):=(x^n-1)/g(x)$ be the \emph{parity-check polynomial}. 
Any $c(x)\in\mathscr{C}$ can be written as $c(x)=f(x)g(x)$ for some $f(x)\in F[x]/(x^n-1)$, and satisfies $h(x)c(x)=0$.

The Euclidean dual $\mathscr{C}^\perp$ is also cyclic, with generator polynomial
\[
g^\perp(x)=x^{\deg h(x)}h(x^{-1}).
\]
We will also use the \emph{conjugate} code $\overline{\mathscr{C}}$, obtained by applying the Frobenius automorphism entrywise:
\[
(c_0,\dots,c_{n-1})\in\mathscr{C}\quad\Longmapsto\quad(\overline{c}_0,\dots,\overline{c}_{n-1})=(c_0^2,\dots,c_{n-1}^2)\in\overline{\mathscr{C}}.
\]

The generator polynomial of a cyclic code of length $n$ over $F$ must be a factor of $x^{n}-1$. Since we are interested in the case $F=\GF(4)$, let $m$ be the smallest integer such that $n$ divides $4^m-1$. $\GF(4^m)$ is the splitting field of $x^n-1$, i.e., $x^n-1=\prod_{i=0}^{n-1}(x-\xi^i)$, where $\xi\in\GF(4^m)$ is a primitive $n^{th}$ root of unity, 

The \emph{cyclotomic coset} mod $n$ over $\GF(4)$ which contains $s$ is $C_s^{(4)}=\{s,4s,4^2s,\dots,4^{k(s)-1}s\}$, where $k(s)$ is the least integer such that $4^{k(s)}\equiv s\mod n$. Note that $m=k(1)$.

A polynomial $f(x)=\prod_{i\in K}(x-\xi^i)$ for $K$ a subset of $\{0,1,\dots,n-1\}$ has coefficients in $\GF(4)$ if and only if $k\in K\implies 4k\;(\text{mod}\; n)\in K$. Therefore, for the generator polynomial $g(x)=\prod_{i\in K}(x-\xi^i)$ of a cyclic code $\mathscr{C}$ to have coefficients in $\GF(4)$, $K$ is a union of cyclotomic cosets. The $n^{th}$ roots of unity $\{\xi^i:i\in K\}$ are called \emph{the zeros of the code}. $c(x)$ belongs to $\mathscr{C}$ iff $c(\xi^i)=0$ for all $i\in K$. 
\begin{lemma} 
\label{lemma:cyclic_code_zeros}
Given a cyclic code $\mathscr{C}$ with generator polynomial $g(x)=\prod_{i\in K}(x-\xi^i)$,\\
(a) the zeros of the dual code $\mathscr{C}^\perp$ are $\{\xi^i:-i\notin K\}=\{\xi^i:i\in -(\{0,1,\dots,n-1\}\backslash K)\mod n\}$;\\
(b) the zeros of the conjugate code $\overline{\mathscr{C}}$ are $\{\xi^{2i}:i\in K\}$.
\end{lemma}
\begin{proof}
For (a), the dual code $\mathscr{C}^\perp$ is cyclic and has generator polynomial $g^\perp(x)=x^{\deg h(x)}h(x^{-1})$.\\
To see (b), note that for $(c_0,c_1,\dots,c_{p-1})\in \mathscr{C}$, one has $i\in K \Rightarrow c(\xi^i)=\sum_{j=0}^{n-1}c_j \xi^{ij}=0$, then the conjugated codeword $(c_0^2,c_1^2,\dots,c_{n-1}^2)$ satisfies $\sum_{j=0}^{n-1}c_j^2\xi^{2ij}=(\sum_{j=0}^{n-1}c_j\xi^{ij})^2=0$ ($2$ times everything is zero in $\GF(4)$), so $\xi^{2i}$ is a zero of the conjugated code.
\end{proof}
The classical QR code $Q_p$ is the $[p, \frac{1}{2}(p+1)]$ cyclic code with generator polynomial $\prod_{i\in R}(x-\xi^i)$, where $p$ is an odd prime, and $R$ is the subset of $\{1,\dots,p-1\}$ that are quadratic residue modulo $p$, i.e., can be written as $j^2\mod p$ for some $j$. Note that we exclude $0$ from $R$. Denote by $N$ the nonresidues modulo $p$, then $R\bigsqcup N=\{1,\dots,p-1\}$.
The parity check generator of $Q_p$ is $(x-1)\prod_{i\in N}(x-\xi^i)$. Therefore, the generator polynomial of $Q_p^\perp$ is $(x-1)\prod_{i\in N}(x-\xi^{-i})=\prod_{i\in -N\cup\{0\}}(x-\xi^{i})$. 

Next, we use the classical QR codes to construct a $\GF(4)$ CSS $\db{p,1}_4$ quantum QR code. We choose the $X$ stabilizers to be $Q_p^\perp$. We show that, depending on $p\mod 8$, either taking $Q_p^\perp$ for the $Z$ stabilizers (self-orthogonal) or taking the Hermitian conjugate of $Q_p^\perp$ (hermitian self-orthogonal) satisfies the CSS constraint. This follows from~\cite{macwilliams1978self}, which we summarize below.

By the law of quadratic reciprocity: $\genfrac(){}{0}{2}{p}=(-1)^{\frac{p^2-1}{8}}$, $\genfrac(){}{0}{-1}{p}=(-1)^{\frac{p-1}{2}}$, which means $2$ is a quadratic residue mod $p$ if $p\equiv \pm 1\mod 8$, and $-1$ is a quadratic residue mod $p$ if $p\equiv 1,5\mod 8$.
If $a$ and $b$ are both nonresidues mod $p$, or $a,b$ are both residues, then $ab$ is a residue. If one of $a,b$ is a residue and the other is a nonresidue, then $ab$ is a nonresidue.
\begin{proposition}
\label{prop:GF4_QR_construction}
It is possible to construct a $\GF(4)$ CSS $\db{p,1}_4$ quantum quadratic residue code for a prime number $p$ that equals $3,5,7$ modulo $8$. Take $X$ stabilizers to be $Q_p^\perp$ generated by $\prod_{i\in -N\cup\{0\}}(x-\xi^{i})$.\\
(a) For $p=8k+3$, $Q_p^\perp$ is self-orthogonal, i.e., $Q_p^\perp\subset (Q_p^\perp)^\perp = Q_p$, and we can take $Z$ stabilizers to be $Q_p^\perp$ as well.\\
(b) For $p=8k-1$ or $p=8k-3$, $Q_p^\perp$ is hermitian self-orthogonal, i.e., $\overline{Q_p^\perp}\subset  Q_p$, and we can take $Z$ stabilizers to be $\overline{Q_p^\perp}$.\\
We call the constructed quantum code in (a) self-dual and in (b) hermitian self-dual. In both cases, an $X$-logical representative can be taken to be $\prod_{i\in -N}(x-\xi^i)$.
\end{proposition}
\begin{proof}
To show that one cyclic code is contained in the other, we just need to show that its zero set contains the other.
For primes of the form $8k+3$, we have $-1$ is a nonresidue and hence $-N=R$. The zeros of $Q_p^\perp$ are $\xi^i$ where $i\in -N\cup\{0\}$. Recall that the zeros of $Q_p$ are $\xi^j$ where $j\in R$. Since $R=-N\subset (-N\cup\{0\})$, we have $Q_p^\perp\subset Q_p$. An $X$-logical representative belongs to $Q_p\backslash Q_p^\perp$ and one can see $\prod_{i\in -N}(x-\xi^i)$ is indeed a valid choice.
For $p=8k-1$, $-1$ is a nonresidue but $2$ is a residue. For $p=8k-3$, $-1$ is a residue but $2$ is a nonresidue. Therefore, $-2N=R$ in both cases.
Therefore, by~\cref{lemma:cyclic_code_zeros}(b), the zeros of $\overline{Q_p^\perp}$ are $\xi^i$ where $i\in -2N\cup\{0\}=R\cup\{0\}$, while the zeros of $Q_p$ are $\{\xi^j\;|\;j\in R\}$. An $X$-logical representative belongs $(\overline{Q_p^\perp})^\perp\backslash Q_p^\perp=\overline{Q_p}\backslash Q_p^\perp$ and the zeros of $\overline{Q_p}$ are $\{\xi^j\;|\;j\in-N\cup\{0\}\}$, so $\prod_{i\in -N}(x-\xi^i)$ again works.

\end{proof}
Note that for $p=8k-1$, we can always find a binary parity check matrix~\cite[Thm.~38]{macwilliams1978self} for $Q_p^\perp$. This is not very interesting because one has binary QR codes of the same parameter, and we can obtain a phantom code by encoding each qubit of the $\db{4,2,2}$ code using this binary QR code. In other words, B\&C does not give any advantage over the usual $K=1$ concatenation.

For primes of the form $8k+1$, the usual $R$ formed by residues will not work because both $-1$ and $2$ are residues.
However, it is sometimes possible to construct a hermitian self-dual quantum cyclic code of this length. For example, using the $[17,9,7]_4$ in~\cite[Table~VI]{macwilliams1978self} with $K=C_1^{(4)}\bigcup C_3^{(4)}$, one can construct a $\db{17,1,7}_4$ hermitian self-dual quantum cyclic code, and in general this is possible if $C_s^{(4)}\neq -2 C_s^{(4)}$ for all $1\le s\le p-1$~\cite[Cor.~34]{macwilliams1978self}. Our observations for logical gates in the next two subsections also apply to these hermitian self-dual cyclic codes.

The qudit distance of the $\GF(4)$ quantum QR codes we used to construct those B\&C codes in~\cref{tab:gate_for_phantom_codes} is lower bounded by the distance listed in \cite[Table~VI]{macwilliams1978self} minus one. Therefore, by~\cref{prop:B&C_distance}, the qubit distance of the B\&C code is at least twice that of the qudit code.
We numerically confirm that these are the actual code distances as listed in~\cref{tab:gate_for_phantom_codes}. 

\subsubsection{Proof of phantom property}
\label{sec:B&C_phantom_proof}
We are now ready to prove that the B\&C codes constructed from the above $\GF(4)$ cyclic codes of length $p$ (odd) are phantom. Recall from~\cref{prop:phantom_B&C_condition} that we only need to show that the logical level Frobenius transform $|\gamma\ra_L$ to $|\gamma^2\ra_L$ can be achieved through permutation only. Here we only show $|\gamma\ra_L\mapsto|\alpha\gamma^2\ra_L$ for some constant $\alpha\in\{1,\omega,\omega^2\}$ that depends on the code; this suffices because we can compose it with the $|\gamma\ra_L\mapsto|\alpha^{-1}\gamma\ra_L$ transform obtained from permutations inside each $\db{4,2,2}$ codeblock.

At the qudit level, we first do a coordinate-wise conjugation $c_j\mapsto c_j^2$, then we permute the coordinates as $j\mapsto 2j\mod p$.
At the qubit level of the B\&C code, the first operation can be realized through permutations inside each $\db{4,2,2}$ code, and the second can be achieved by permuting across the $\db{4,2,2}$ codeblocks. 

Going back to the qudit picture, let us first show that the $X$ and $Z$ stabilizers are preserved, again we just need to verify the zeros. Take an arbitrary stabilizer $(c_0,c_1,\dots,c_{p-1})$ and let $\xi^i$ be a zero of it, i.e., $\sum_{j=0}^{p-1} c_j\xi^{ij}=0$, then the two steps maps $c_j\mapsto c_j^2\mapsto c_{j/2}^2$, where $j/2=\frac{p+1}{2}j\mod p$. The polynomial of the transformed codeword becomes $c'(x)=\sum_{j=0}^{p-1}c_{j/2}^2 x^{j}$. One can see that $\xi^i$ is still a zero of this codeword, because $c'(\xi^i)=\sum_{j=0}^{p-1}c_{j/2}^2 \xi^{ij}=\sum_{j=0}^{p-1}c_j^2\xi^{2ij}=(\sum_{j=0}^{p-1}c_j\xi^{ij})^2=0$. Since $(\sum_{j=0}^{p-1}c_j\xi^{ij})^2=0$ implies $\sum_{j=0}^{p-1}c_j\xi^{ij}=0$, one can similarly show that if $\xi^i$ is a zero of $c'(x)$, then it is also a zero of $c(x)$. Consequently, the zeros of the common divisor, i.e., the generator polynomial, of the stabilizer codewords, are preserved under this mapping. Therefore, the stabilizers are preserved.

Next, we show that the $X$-logical $\gamma l_\rmx$ is mapped to $\alpha\gamma^2 l_\rmx$ up to stabilizers, for a constant $\alpha\in\{1,\omega,\omega^2\}$. Recall from~\cref{prop:GF4_QR_construction} that the $X$-stabilizers are generated by $\prod_{i\in -N\cup\{0\}}(x-\xi^i)$ and we can choose $l(x)=\prod_{i\in -N}(x-\xi^i):=\sum_{i=0}^{p-1}l_ix^i$ to represent the $X$ logical $l_\rmx$. By the codeword to polynomial correspondence, we can write $l_\rmx=(l_0,\cdots,l_{p-1})$, and $l_0,\cdots,l_{p-1}\in\GF(4)$. Moreover, $l(\xi^i)=\sum_{j=0}^{p-1}l_j \xi^{ij}=0$ for $i\in -N$. Next, we show that the constant $\alpha$ is $l(1)=\sum_{i=0}^{p-1}l_i$. Note that $l(1)\in\GF(4)$ because $l_0,\cdots,l_{p-1}\in\GF(4)$; $l(1)$ is not zero because $\xi^0=1$ is not a root of $l(x)$.

Under the two-step transform above, $\gamma l_\rmx$ transforms as $(\gamma l_0,\dots,\gamma l_{p-1})\mapsto (\gamma^2 l_0^2,\dots,\gamma^2 l_{p-1}^2)\mapsto (\gamma^2 l_{0/2}^2,\dots,\gamma^2 l_{(p-1)/2}^2)$. Again, the new generating polynomial $l'(x)=\sum_{j=0}^{p-1}\gamma^2 l_{j/2}^2 x^j$ has the zeros exactly at the same places as $l(x)$. Moreover, $\gamma l_\rmx$ is mapped to $l(1)\gamma^2 l_\rmx$, i.e., $\alpha=l(1)=\sum_{j=0}^{p-1}l_j\in \{1,\omega,\omega^2\}$. This is because $l'(x)$ and $l(1)\gamma^2 l_\rmx$ differs by an $X$-stabilizer: $\Delta l(x):=l'(x)-l(1)\gamma^2 l(x)=\gamma^2(\sum_{j=0}^{p-1}l_{j/2}^2 x^j - l(1) \sum_{j=0}^{p-1}l_j x^j)$ can be divided by $\prod_{i\in -N\cup\{0\}}(x-\xi^i)$. For $i\in -N$, one has $\Delta l(\xi^i)=\gamma^2(0-0)=0$ and $\Delta l(\xi^0)=\gamma^2 (\sum_{j=0}^{p-1} l_{j/2}^2 - l(1)\sum_{j=0}^{p-1} l_j)=\gamma^2(l(1)^2-l(1)^2)=0$.

\subsubsection{Logical gates}
\label{app:binarize_concatenate/gates}

Next, we establish other Clifford logical gates that the B\&C codes admit besides $\overline{\mathrm{CNOT}}$; they are by no means exhaustive. Denote by $H_\rmx$ (resp. $H_\rmz$) and $l_\rmx$ (resp. $l_\rmz$) the $X$ (resp. $Z$) type PCMs and logical of the $\db{p,1}_4$ $\GF(4)$ codes; entries are in $\GF(4)$.
For the code obtained after B\&C, denote by $H'_\rmx$ and $H'_\rmz$ its binary PCMs. 

The per qudit $\db{4,2,2}$ encoding gives extra stabilizers $XXXX$ and $ZZZZ$, and it maps the $X$ and $Z$ operators on each qudit to a Pauli string on four qubits as follows (up to these extra stabilizers). $X^\omega\mapsto XXII\equiv IIXX$ and $Z^{\omega^2}\mapsto IIZZ$; $X^{\omega^2}\mapsto XIXI\equiv IXIX$ and $Z^\omega\mapsto IZIZ$; $X^1\mapsto IXXI$ and $Z^1\mapsto IZZI$. 

For simplicity, denote the per qudit $\db{4,2,2}$ encoding as $\enc(\cdot)$ and rewrite the above as, e.g., $\enc(X^\omega)=1100\equiv 0011=\enc(Z^{\omega^2})$.
One can see that in the binary PCM, $X^\gamma$ is mapped to the equivalent thing as $Z^{\gamma^2}$, i.e., $\enc(X^\gamma)=\enc(Z^{\gamma^2})$, $\forall\gamma\in\GF(4)$.

The two logical pairs of the B\&C code can be written as $(\enc_\rmx(\omega l_\rmx),\;\enc_\rmz(\omega^2 l_\rmx))$ and $(\enc_\rmx(\omega l_\rmz),\;\enc_\rmz(\omega^2 l_\rmz))$. Here, $\enc_{\rmx/\rmz}(l)$ for $l=(l_0,\dots,l_{p-1})\in \F_4^p$ means $\enc_\rmx(l):=(\enc(X^{l_0}),\dots,\enc(X^{l_{p-1}}))\in \F_2^{4p}$ for $X$-type and $\enc_\rmz(l):=(\enc(Z^{l_0}),\dots,\enc(Z^{l_{p-1}}))\in \F_2^{4p}$ for $Z$-type. We can shorthand the relation $\enc(X^\gamma)=\enc(Z^{\gamma^2})$ as $\enc_\rmx(\gamma)=\enc_\rmz(\gamma^2)$.

\begin{proposition}
For the phantom B\&C codes constructed from the hermitian self-dual $\GF(4)$ codes (not necessarily cyclic), transversal Hadamard implements logical $HH$ up to $\mathrm{SWAP}_{12}$, and transversal $S$ implements logical CZ up to Pauli corrections.
\end{proposition}
\begin{proof}
By definition, $H_\rmz=\overline{H_\rmx}$, and by $\enc_\rmx(\gamma)=\enc_\rmz(\gamma^2)$, we have $\text{rs}(H'_\rmx)=\text{rs}(H'_\rmz)$. Therefore, stabilizers are preserved under transversal Hadamard. 
For the logical action, 
first notice that we can take $l_\rmz=\overline{l_\rmx}$ because $l_\rmz\in H_\rmx^\perp\backslash H_\rmz=H_\rmx^\perp\backslash\overline{H_\rmx}$ and $l_\rmx\in H_\rmz^\perp\backslash H_\rmx=\overline{H_\rmx}^\perp\backslash H_\rmx$. Then one can see that the $X$-type logical $\enc_\rmx(\omega l_\rmx)$ is mapped to $Z$-type logical $\enc_\rmx(\omega l_\rmx)=\enc_\rmz(\omega^2\overline{l_\rmx})=\enc_\rmz(\omega^2 l_\rmz)$ under transversal Hadamard. Similarly, the $X$-type logical $\enc_\rmx(\omega^2 l_\rmx)$ is mapped to the $Z$-type logical $\enc_\rmz(\omega l_\rmz)$. Therefore, the logical action of transversal Hadamard is $HH$ and $\mathrm{SWAP}_{12}$.

Since $\text{rs}(H'_\rmx)=\text{rs}(H'_\rmz)$, one can also see that stabilizers are preserved under transversal $S$ but possibly up to a phase. We show that the phase is $+1$ and so the codespace is preserved as follows. An $X$ stabilizer $s_\rmx$ of the $\GF(4)$ code is mapped to itself times a $Z$ stabilizer $\overline{s_\rmx}$ (on the qudit level) with a phase $i^{2\wt_\omega(s_\rmx)+2\wt_{\omega^2}(s_\rmx)+2\wt_{1}(s_\rmx)}$ (on the qubit level), where $\wt_{\omega}(s_\rmx),\wt_{\omega^2}(s_\rmx),\wt_1(s_\rmx)$ count the number of $\omega,\omega^2,1$ appearing in $s_\rmx$. Since $s_\rmx\cdot\overline{s_\rmx}=0$, we have $\wt_\omega(s_\rmx)+\wt_{\omega^2}(s_\rmx)+\wt_{1}(s_\rmx)$ even. 
The logical CZ action resulting from transversal $S$ can be verified similarly to the proof above for the transversal Hadamard.
\end{proof}
This proposition is relevant for $p=8k+5$ and $p=8k+1$ if a hermitian self-dual quantum code can be constructed. Since CZ, $HH$ and phantom CNOTs do not generate the full Clifford group, we want to find a way to do logical $SS$ via fold diagonal gates. However, although the solver in~\cite{webster2023transversal} says logical $SS$ from fold is indeed possible for $\db{20,2,6}$ constructed from $\db{5,1,3}_4$, we do not know how to generalize these findings to larger codes in this family. It is known that the automorphism group of the classical extended QR codes contains a subgroup isomorphic to $\text{PSL}(2,\F_p)$, and the base code\footnote{When puncturing the extended QR codes at $\infty$ to construct a $K=1$ quantum code, one of the generator $y\mapsto-\frac{1}{y}$ of $\text{PSL}(2,\F_p)$ becomes unavailable. The only possible qudit coordinate involution is $y\mapsto -y$ if $-1$ is a residue.} $[6,3,4]_4$ of $\db{20,2,6}$ is the only known case over $\GF(4)$ in which the group is bigger than this~\cite[Ch.~16]{macwilliams1977theory}. Therefore, it is also possible that $\db{20,2,6}$ is the only one in the family that allows fold diagonal gates.
\begin{proposition}
For the phantom B\&C codes constructed from the self-dual $\GF(4)$ QR codes of length $p=8k+3$, transversal Hadamard together with the qudit coordinate permutation $i\mapsto 2i\mod p$ implements logical $HH$ up to $\mathrm{SWAP}_{12}$.
Logical $SS$ can be achieved up to Pauli corrections through folding each $\db{4,2,2}$ codeblock using the involution $\tau_{SS}$: CZ between qubits $2$ and $3$ and $S$ on qubits $1$ and $4$. 
Logical CZ can also be achieved up to Pauli corrections through folding; $\tau_{\rm{CZ}}$ is the involution on the qudit coordinate $i \mapsto -i\mod p$.
\end{proposition}
\begin{proof}
Since $\enc_\rmx(\gamma)=\enc_\rmz(\gamma^2)$, the qubit-level transversal Hadamard transforms the qudit PCM $H_\rmx$ into its conjugate on the qudit level. Since $H_\rmz=H_\rmx$, we need to counteract this conjugation through coordinate permutation. As we have seen before, $i\mapsto 2i\mod p$ effectively doubles the exponents of the zeros of a cyclic code and transforms the code into its conjugate. Therefore, the codespace is preserved under this SWAP-transversal Hadamard.
For the logical action, note that we can take $l_\rmz=l_\rmx$ for the self-dual codes.
The $X$-logical $\enc_\rmx(\omega l_\rmx)$ is mapped to the $Z$-type $\enc_\rmx(\omega l_\rmx)=\enc_\rmz(\omega^2\overline{l_\rmx})=\enc_\rmz(\omega^2\overline{l_\rmz})$ under transversal Hadamard, and the coordinate permutation further maps $\omega^2\overline{l_\rmz}$ to $l(1)\omega^2 l_\rmz=\omega^2 l_\rmz$ (up to some $Z$ stabilizers). The constant $l(1)=\sum_{i=0}^{p-1}l_i$, where $l(x)=\sum_{i\in-N}(x-\xi^i)=\sum_{i=0}^{p-1}l_i x^i$, appearing in \cref{sec:B&C_phantom_proof} for this coordinate permutation turns out to be $1$ for self-dual codes constructed from $p=8k+3$ QR codes. 

One can verify $l(1)=\omega+\omega^2+0=1$ for $p=3$ directly. For $p>3$, $l(1)=\prod_{i\in -N}(1-\xi^i)=\prod_{i\in -N} \xi^i(\xi^{-i}-1)=\xi^{-\sum_{i\in N}i} \prod_{i\in N}(1-\xi^i)=\prod_{i\in N}(1-\xi^i)$, where in the second equality we use $1=-1$ in $\GF(4)$ and in the last step $\sum_{i\in N}i\equiv\sum_{i\in R}i=\sum_{j=1}^{(p-1)/2}j^2=p(p^2-1)/24\equiv 0\mod{p}$. Since $-1$ is a nonresidue for $p\equiv 3\mod 4$, one has $-N=R$, and $l(1)^2=\prod_{i=1}^{p-1}(1-\xi^i)=\frac{\xi^{p}-1}{\xi-1}=1\Rightarrow l(1)=1$. 

To summarize the logical action of transversal $H$ together with the coordinate permutation $i\mapsto 2i\mod p$, the $X$-logical $\enc_\rmx(\omega l_\rmx)$ transforms to $Z$-logical $\enc_\rmz(\omega^2 l_\rmz)$. One can similarly show that the $X$-logical $\enc_\rmx(\omega^2 l_\rmx)$ transforms to $Z$-logical $\enc_\rmz(\omega l_\rmz)$.

Under the fold gate following involution $\tau_{SS}$, each $X$-type qudit $\gamma$ is mapped to$\tau(\enc_\rmx(\gamma))=\enc_\rmx(\gamma^2)=\enc_\rmz(\gamma)$. Therefore, the per-coordinate folding maps an $X$ stabilizer $\enc_\rmx(s_\rmx)$ to itself and the $Z$-type 
$\enc_\rmz(s_\rmx)$ and a phase $i^{\wt_\omega(s_\rmx)+\wt_{\omega^2}(s_\rmx)+2\wt_1(s_\rmx)}=\pm 1$ because $s_\rmx\cdot s_\rmx=0$; one can fixes the $-1$ phase by applying destabilizers~\cite{sayginel2025fault}. The $X$-logical $\enc_\rmx(\omega l_\rmx)$ is mapped to itself times its pairing $Z$-logical $\enc_\rmz(\omega l_\rmx)=\enc_\rmz(\omega l_\rmz)$ with a phase $i^{\wt_\omega(\omega l_\rmx)+\wt_{\omega^2}(\omega l_\rmx)+2\wt_1(\omega l_\rmx)}=i^{\wt_1(l_\rmx)+\wt_w(l_\rmx)+2\wt_{\omega^2}(l_\rmx)}=\pm i$. The last step follows because $l(1)=1\Rightarrow \wt_\omega(l_\rmx)\equiv\wt_{\omega^2}(l_\rmx)\not\equiv\wt_1(l_\rmx)\mod{2}$. We can apply a $Z$-logical operator to fix the phase if it is $-i$. One can similarly show that the other $X$-logical is mapped to itself times its pairing $Z$-logical with a $\pm i$ phase. The logical effect of this fold gate is thus $SS$ up to Pauli corrections.

Lastly, we show that the fold gate following the involution $\tau_{\rm{CZ}}$ implements logical CZ. Note that under $i\mapsto -i\mod p$, a zero of a polynomial is mapped to its inverse. Therefore, the generator polynomial for $X$/$Z$ stabilizers $\prod_{i\in-N\cup\{0\}}(x-\xi^i)$ is mapped to $\prod_{i\in N\cup\{0\}}(x-\xi^i)=\prod_{i\in-2N\cup\{0\}}(x-\xi^i)$, the last step $N=-2N$ is because both $-1$ and $2$ are nonresidues modulo $p=8k+3$.
Also, recall that coordinate-wise conjugation doubles the exponent of zeros. Therefore, by $\enc_\rmx(\gamma)=\enc_\rmz(\gamma^2)$, each $X$ stabilizer is mapped to itself times a $Z$ stabilizer. Since the only fixed point is the qudit at coordinate $0$, the phase accumulated is either $1$ or $-1$, and the latter case can be fixed by applying a destabilizer. 
The $X$-logical $\enc_\rmx(\omega l_\rmx)=(\enc_\rmx(\omega l_0),\dots,\enc_\rmx(\omega l_{p-1}))$ is mapped to itself times $Z$ type $(\enc_\rmx(\omega l_{-0}),\dots,\enc_\rmx(\omega l_{-(p-1)}))=(\enc_\rmz(\omega^2 l^2_{-0}),\dots,\enc_\rmz(\omega^2 l^2_{-(p-1)}))$. The polynomial of this new codeword $\omega^2\sum_{i=0}^2 l_{-i}^2 x^i$ has the same zero as $\omega^2l_\rmx=\omega^2l_\rmz=\omega^2\prod_{i\in -N}(x-\xi^i)$ (again because of both $-1$ and $2$ are nonresidues). Moreover, their difference is a $Z$-stabilizer, i.e., $(x-\xi^0)$ divides the difference, because $l(1)=1$. 
Therefore, the $X$-logical $\enc_\rmx(\omega l_\rmx)$ is mapped to itself and the $Z$-logical $\enc_\rmz(\omega^2 l_\rmz)$. Similar arguments hold for the other $X$-logical $\enc_{\rmx}(\omega^2 l_\rmx)$. The logical effect of the fold gate is thus CZ.
\end{proof}

\section{Other code constructions}
\label{app:other_constructions}

Here we discuss additional avenues for code construction. First, in \cref{app:other_constructions/hypergraph}, we describe the construction of CSS phantom codes from suitable pairs of classical linear codes via the hypergraph product. In \cref{app:other_constructions/gluing_422}, we examine a family of $d = 2$ phantom codes built by gluing $\db{4,2,2}$ codes. Lastly, in \cref{app:other_constructions/codes_from_phantom}, we examine codes built from phantom codes. These codes are not phantom as they support products rather than individually addressable $\overline{\mathrm{CNOT}}$ gates implemented by qubit permutations, but are related to prior literature.

\subsection{Hypergraph product phantom codes}
\label{app:other_constructions/hypergraph}

To start, we first review the notion of automorphisms on classical linear codes. 

\begin{definition}[Code automorphism of a classical linear code]
For an $[n,k]$ linear code over $\mathbb{F}_2$ with generator matrix $G \in \mathbb{F}_2^{k \times n}$,
a code automorphism is a coordinate permutation $\sigma \in \S_n$ such that $G \sigma = V G$ for some invertible matrix $V \in \GL(k, \mathbb{F}_2)$.
\end{definition}

In other words, permuting the coordinates of bits by $\sigma$ maps a codeword basis to another codeword basis, and the codespace spanned by these bases is preserved.

\begin{definition}[Tanner graph automorphism of a classical linear code~\cite{berthusen2025automorphism}]
For an $[n,k]$ linear code over $\mathbb{F}_2$ with parity-check matrix $H\in \mathbb{F}_2^{m \times n}$, a Tanner graph automorphism is a code automorphism such that (additionally) $H \sigma = \pi H$ where $\pi \in \S_m$.
\end{definition}

To illustrate the difference, consider the $[3,1]$ repetition code with generator and parity-check matrices
\begin{equation}
    G = \begin{bmatrix} 1 & 1 &  1 \end{bmatrix},
    \qquad
    H = \begin{bmatrix} 1 & 1 & 0 \\ 0 & 1 & 1 \end{bmatrix}.
\end{equation}

Any cyclic permutation of the three coordinates leaves $G$ unchanged and so is a code automorphism. However, a rightward cyclic shift of the coordinates maps $H$ to
\begin{equation}
    H' = \begin{bmatrix} 0 & 1 & 1 \\ 1 & 0 & 1 \end{bmatrix},
\end{equation}
which is not the same as $H$ up to row permutations, and therefore this permutation is not a Tanner graph automorphism of the code. In contrast, swapping the first and the last coordinates preserves $H$ up to swapping the two rows, and is a Tanner graph automorphism. This highlights the distinction: Tanner graph automorphisms are a subset of code automorphisms.

We construct CSS phantom codes using the hypergraph product of classical linear codes with Tanner graph automorphisms, as a corollary to Ref.~\cite[Thm.~4.2]{berthusen2025automorphism}.

\begin{theorem}[Phantom hypergraph product codes]
\label{thm:phantom_hypergraph_product_codes}
    Let $H_1$ be an $m_1 \times n_1$ parity-check matrix of an $[n_1, k, d_1]$ classical linear code whose Tanner graph automorphism group is $\GL(k, \mathbb{F}_2)$. Let $H_2$ be an $m_2 \times n_2$ parity-check matrix of an $[n_2, 1, d_2]$ classical linear code. Then, the hypergraph product of $H_1$ and $H_2$ yields a CSS phantom code with parameters
    \begin{equation}
        \db{n_1 n_2 + m_1 m_2, \, k, \, \min(d_1, d_2)}.
    \end{equation}
\end{theorem}

\begin{proof}
In the hypergraph product construction, the $X$- and $Z$-type stabilizer generator matrices take the forms
\begin{equation}
    H_\mathrm{x} = (H_1 \otimes I_{n_2} \;\;|\;\; I_{m_1} \otimes H_2^\top),
    \qquad
    H_\mathrm{z} = (I_{n_1} \otimes H_2 \;\;|\;\; H_1^\top \otimes I_{m_2}),
\end{equation}
where the vertical bar separates the ``left'' data sector with $n_1 n_2$ qubits from the ``right'' data sector with $m_1 m_2$ qubits. We adopt the canonical logical basis
\begin{equation}
    L_\mathrm{x} = (\{e_j\} \otimes g_2 \;\;|\;\; 0),
    \qquad
    L_\mathrm{z} = (G_1 \otimes \epsilon \;\;|\;\; 0), 
\end{equation}
where $G_1$ is a $k \times n_1$ generator matrix for the $[n_1,k, d_1]$ code,  $\epsilon$ is a length-$n_2$ row vector outside the row span of $H_2$,  $\{e_j\}$ denotes a $k \times n_1$ matrix of linearly independent unit vectors outside the row span of $H_1$, and $g_2$ is the only codeword of the $[n_2,1, d_2]$ code. 

Because the commutation relations of logical Pauli operators are preserved under physical qubit permutations, it suffices to track the action on $L_\mathrm{z}$ to verify that a permutation implements a desired logical Clifford. In particular, for a $\overline{\mathrm{CNOT}}$ from logical qubit $a$ to $b$, take
\begin{equation}
    V = \mathbb{I}_k + e_{ba},
\end{equation}
where $e_{ba}$ is the matrix with a $1$ in the $(b,a)$ position and zeros elsewhere. Note that $V \in {\rm GL}_k(\mathbb{F}_2)$. Since the $[n_1,k,d_1]$ code admits the full $\GL(k, \mathbb{F}_2)$ Tanner graph automorphism group, for such a $V$ there exist permutations $\sigma_V \in S_{n_1}$ and $\pi_V \in S_{m_1}$ such that
\begin{equation}
    G_1 \sigma_V = V G_1, 
    \qquad 
    H_1 \sigma_V = \pi_V H_1 .
\end{equation}

By Ref.~\cite[Thm.~4.2]{berthusen2025automorphism}, the qubit permutation
\begin{equation}
    \rho_V = (\sigma_V \otimes \mathbb{I}_{n_2}) \oplus (\pi_V \otimes \mathbb{I}_{m_2})
\end{equation}
preserves the stabilizer group and acts on the logical operators as
\begin{equation}
    L_\mathrm{z} \rho_V 
    = (G_1 \sigma_V \otimes \epsilon \, \mathbb{I}_{n_2} \ | \ 0) 
    = (V G_1 \otimes \epsilon \ | \ 0) 
    = V L_\mathrm{z},
\end{equation}
noting that $\epsilon$ is a single row. Thus, $\rho_V$ implements a $\overline{\mathrm{CNOT}}$ from qubit $a$ to $b$.
\end{proof}

As a result of \cref{thm:phantom_hypergraph_product_codes}, we can construct CSS phantom codes from the hypergraph product of a classical simplex code with a repetition code. Let $H_1$ be the parity-check matrix of the $[n_1 = 2^k - 1, k, 2^{k-1}]$ simplex code, where the $m_1 = n_1(n_1-1)/6$ rows correspond to all weight-$3$ checks (i.e.~all weight-3 codewords of the dual of simplex code, the Hamming code). Let $H_2$ be the parity-check matrix of the $[n_2 = 2^{k-1}, 1, 2^{k-1}]$ repetition code, with $m_2 = n_2 - 1$ rows. Then the hypergraph product of $H_1$ and $H_2$ yields a CSS phantom code with parameters
\begin{equation}
    \db{n_1 n_2 + m_1 m_2 = \mathcal{O}(2^{3k}), k, 2^{k-1}}.
\end{equation}

The smallest members of this family are the $\db{7,2,2}$ and $\db{49,3,4}$ codes.

\subsection{Glued \texorpdfstring{$\db{4,2,2}$}{[[4,2,2]]} codes as a family of \texorpdfstring{$d = 2$}{d=2} phantom codes} 
\label{app:other_constructions/gluing_422}

Detailed examination of our exhaustive code enumeration (see \cref{app:enumeration}) results revealed that a class of $n$-optimal CSS phantom codes for a given $d_\mathrm{z}$, when $k = 2$ and $d_\mathrm{x} = 2$, possess a common underlying structure and can be understood as gluings of the the $\db{4,2,2}$ code. The Hadamard-dual of these codes exchange $d_\mathrm{x}, d_\mathrm{z}$ distances (thereby fixing $d_\mathrm{z} = 2$ instead). We describe this construction below.

\begin{theorem}[Glued $\mathbf{\db{4,2,2}}$ phantom codes]
    There exist CSS phantom codes with parameters $\db{n=4m,k=2,d_\text{x}=2,d_\text{z}=2m}$ and $\db{n=4m-1,k=2,d_\text{x}=2,d_\text{z}=2m-1}$ for any integer $m \geq 1$.
\end{theorem}

\begin{proof}
We illustrate the constructions in \cref{fig:gluing_422}. First we describe the $n=4m$ construction. We use $m$ copies of the $\db{4,2,2}$ code and glue them together with $X$-type stabilizers to form a ``tube''. Labelling the physical qubits as shown in \cref{fig:gluing_422/a}, a basis for the logical operators of the glued code is $\overline{X}_1 = X_0 X_1 \equiv \cdots \equiv X_{4i} X_{4i+1}$, $\overline{X}_2 = X_1 X_2 \equiv \cdots \equiv X_{4i+1} X_{4i+2}$ for all integer $0 \leq i < m$, where $\equiv$ denotes equivalence up to stabilizers, and $\overline{Z}_1 = Z_1 Z_2 Z_5 Z_6 \cdots Z_{4m-3} Z_{4m-2}$, $\overline{Z}_2 = Z_0 Z_1 Z_4 Z_5 \cdots Z_{4m-4} Z_{4m-3}$. The $\overline{\mathrm{CNOT}}_{12}$ gate can be implemented by swapping the qubits $(1,2), (5,6), \ldots, (4m-3,4m-2)$, and $\overline{\mathrm{CNOT}}_{21}$ by $(0,1), (4,5), \ldots, (4m-4,4m-3)$. From this construction, we can further remove a physical qubit, as shown in \cref{fig:gluing_422/b}. Consequently, an $X$-type stabilizer is removed and a $Z$-type stabilizer is reshaped to a triangle. This accounts for the $n = 4m - 1$ family of codes.
\end{proof}

\begin{figure}[!ht]
    \centering
    \includegraphics[
        scale=0.5,
        trim=0 5.8cm 0 0,
        clip
    ]{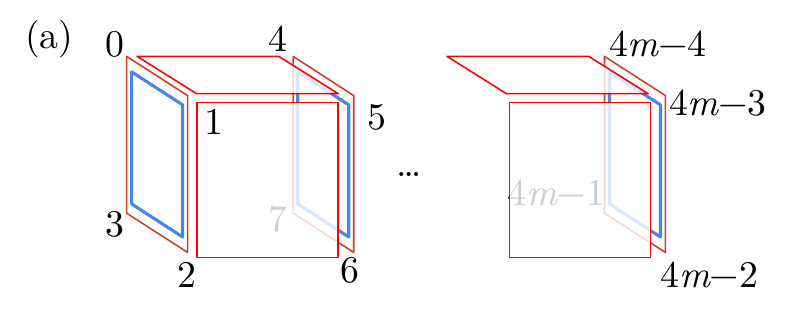}
    \qquad
    \includegraphics[
        scale=0.5,
        trim=0 0 0 5.8cm,
        clip
    ]{figures/gluing422.pdf}
    \phantomsubfloat{\label{fig:gluing_422/a}}
    \phantomsubfloat{\label{fig:gluing_422/b}}
    \caption{\textbf{Glued $\db{4,2,2}$ phantom codes.} Red and blue polygons denote $X$- and $Z$-type stabilizers, respectively.}
    \label{fig:gluing_422}
\end{figure}

\subsection{Codes built from phantom codes}
\label{app:other_constructions/codes_from_phantom}

As discussed in the main text, the $\smash{\db{2^{2r}, \binom{2r}{r}, 2^r}}$ quantum Reed--Muller codes~\cite{gong2024computation, reichardt2024demonstration} implement products of $\overline{\mathrm{CNOT}}$ gates across logical qubits via qubit permutations. However, they do not support individually addressable $\overline{\mathrm{CNOT}}$ gates, except in the special case of $r=1$, which reduces to the hypercube codes. These codes (for $r > 1$) therefore lie outside our definition of phantom codes. Here we comment on constructions for such kinds of ``phantom-like'' codes to connect with existing literature, starting from genuine phantom codes.

\subsubsection{Connecting a CSS code and its Hadamard-dual} 
\label{sec:code_and_its_dual}

We describe a construction that generates code families from phantom codes, though the resulting codes are not phantom. This scheme generalizes the well-known procedure that produces the $\db{16,6,4}$ tesseract code from the $\db{8,3,2}$ code and its dual.

\begin{proposition}[Connecting a CSS code and its Hadamard-dual]
    Consider connecting an $\db{n,k,(d_\mathrm{x},d_\mathrm{z})}$ CSS code (the primal code) and its Hadamard-dual, by performing transversal $\mathrm{CNOT}$ gates controlled on the primal and targeting the dual codeblock. The resulting code has distance $d' = d'_\mathrm{x} = d'_\mathrm{z}$ given by
    \begin{equation}
        \label{eq:distance_of_connected_code}
        d' = \min \left[ d_\mathrm{z}, \min_{\substack{
            l_\mathrm{x} \in \mathcal{L}_\mathrm{x} \\ 
            g_\mathrm{z} \in \rs(H_\mathrm{x})^\perp}}
            \left(2|l_\mathrm{x}|+|g_\mathrm{z}|-2|l_\mathrm{x} \cap g_\mathrm{z}|\right)
        \right],
    \end{equation}
    where the stabilizer generators of the primal code are denoted $H_\mathrm{x}, H_\mathrm{z}$ in symplectic representation, $\mathcal{L}_\mathrm{x} = \rs(H_\mathrm{z})^\perp \backslash \rs(H_\mathrm{x})$ is accordingly the vector space of $X$-type logical operators of the primal code, and $\rs(H_\mathrm{x})^\perp$ is the vector space of $Z$-type stabilizers and logical operators of the primal code. 
\end{proposition}

\begin{proof}
    The stabilizer generator matrix of the primal code in the symplectic representation (see \cref{app:basics/symplectic}) has the structure $\diag(H_\mathrm{x} \, | \, H_\mathrm{z})$, where the $H_\mathrm{x}$ and $H_\mathrm{z}$ blocks describe the $X$- and $Z$-type stabilizer generators respectively. The stabilizer generator matrix of the dual code is accordingly $\diag(H_\mathrm{z} \, | \, H_\mathrm{x})$. Taking the tensor product of these codes (i.e.~placing them side-by-side) gives an $\db{2n,2k}$ code with the stabilizer generator matrix $\diag(H_\mathrm{x}, H_\mathrm{z} \, | \, H_\mathrm{z}, H_\mathrm{x})$. Applying transversal $\mathrm{CNOT}$ gates between these two codes, where the $\mathrm{CNOT}$ gates are controlled by qubits in the primal code and target corresponding qubits in dual code, yields an $\db{2n,2k}$ code with the stabilizer generator matrix
    \begin{equation}
        H = \left(\begin{array}{cc|cc}
            H_\mathrm{x} & H_\mathrm{x} & & \\
            & H_\mathrm{z} & & \\ \hline
            & & H_\mathrm{z} & \\
            & & H_\mathrm{x} & H_\mathrm{x}
        \end{array}\right).
    \end{equation}

    To observe that \cref{eq:distance_of_connected_code} describes the distance of the resulting code $H$, we write out explicitly the set of logical operators on the code. By definition, the $Z$-type logical operators for the code must commute with all $X$-type stabilizers but are not $Z$-type stabilizers themselves. Denoting the bitstring representing a $Z$-type logical operator as $(u,v)$, where both $u$ and $v$ are of length $n$, the former condition means $H_\mathrm{x} \trans{v} = 0$ and $H_\mathrm{x} (\trans{u} + \trans{v}) = 0$---i.e., $v \in \rs(H_\mathrm{z})^\perp$ and $u + v \in \rs(H_\mathrm{x})^\perp$. For simplicity, denote $v$ as $g_\mathrm{x}$ and $u + v$ as $g_\mathrm{z}$. Analogous conditions apply for the $X$-type logical operators of the code. The vector space of logical operators of the code $H$ can thereby be written:
    \begin{equation}\begin{split}
        \mathcal{L}_\mathrm{x}'
            &= \left\{ (g_\mathrm{x}, g_\mathrm{x} + g_\mathrm{z}) \;\big|\; 
                    g_\mathrm{x} \in \rs(H_\mathrm{z})^\perp, g_\mathrm{z}\in \rs(H_\mathrm{x})^\perp \right\}
                \backslash
                \left\{ (s_\mathrm{x}, s_\mathrm{x} + s_\mathrm{z}) \;\big|\; 
                    s_\mathrm{x} \in \rs(H_\mathrm{x}), s_\mathrm{z} \in \rs(H_\mathrm{z}) \right\}, \\
        \mathcal{L}_\mathrm{z}'
            &= \left\{ (g_\mathrm{x} + g_\mathrm{z},g_\mathrm{x}) \;\big|\; 
                    g_\mathrm{x} \in \rs(H_\mathrm{z})^\perp, g_\mathrm{z}\in \rs(H_\mathrm{x})^\perp \right\}
                \backslash
                \left\{ (s_\mathrm{z} + s_\mathrm{x},s_\mathrm{x}) \;\big|\; 
                    s_\mathrm{x} \in \rs(H_\mathrm{x}), s_\mathrm{z} \in \rs(H_\mathrm{z}) \right\}.
    \end{split}\end{equation}

    The weight of a logical operator is $w = |g_\mathrm{z}|+2|g_\mathrm{x}|-2|g_\mathrm{x} \cap g_\mathrm{z}|$ for both $X$- and $Z$-types. We consider the minimum weight by cases:
    \begin{itemize}
        \item In the case that $g_\mathrm{x} \coloneqq l_\mathrm{x} \in \rs(H_\mathrm{z})^\perp \backslash \rs(H_\mathrm{x}) = \mathcal{L}_\mathrm{x}$, that is, $g_\mathrm{x}$ is a nontrivial logical operator of the primal code, the weight $w$ reduces to the $(2|l_\mathrm{x}|+|g_\mathrm{z}|-2|l_\mathrm{x} \cap g_\mathrm{z})$ term in \cref{eq:distance_of_connected_code}.
        \item Otherwise $g_\mathrm{x} \coloneqq s_\mathrm{x} \in \rs(H_\mathrm{x})$. Then $w=|g_\mathrm{z}| + 2(|s_\mathrm{x}|-|s_\mathrm{x}\cap g_\mathrm{z}|) \ge |g_\mathrm{z}|$ where equality holds when $s_\mathrm{x}=0$, and moreover $g_\mathrm{z} \notin \rs(H_\mathrm{z})$ in order for $(g_\mathrm{x}, g_\mathrm{x} + g_\mathrm{z})$ or $(g_\mathrm{x} + g_\mathrm{z},g_\mathrm{x})$ to be a logical operator (otherwise the operator is simply a stabilizer). That is, $g_\mathrm{z}$ is a $Z$-type logical operator of the primal code, and therefore $\abs{g_\mathrm{z}} \geq d_\mathrm{z}$. This accounts for the $d_\mathrm{z}$ term in \cref{eq:distance_of_connected_code}.
    \end{itemize}

    In fact, given a set of logical Pauli operators for the primal code, $\{l_{\mathrm{x},i}\}_{i=1}^k$ and $\{l_{\mathrm{z},i}\}_{i=1}^k$, a logical basis for the resulting code can be explicitly written as:
    \begin{align}
        l'_{\mathrm{x},i} =
        \begin{cases}
            (l_{\mathrm{x},i},\, l_{\mathrm{x},i}), & i=1,\dots,k, \\
            (0,\, l_{\mathrm{z},i-k}), & i=k+1,\dots,2k,
        \end{cases}
        \qquad \qquad
        l'_{\mathrm{z},i} =
            \begin{cases}
            (l_{\mathrm{z},i},\, 0), & i=1,\dots,k, \\
            (l_{\mathrm{x},i-k},\, l_{\mathrm{x},i-k}), & i=k+1,\dots,2k.
        \end{cases}
    \end{align}
\end{proof}

As expressed in \cref{eq:distance_of_connected_code}, the distance $d'$ of the resulting code depends on the stabilizer structure of the primal code---specifically, on the overlaps between logical operators of one Pauli type and stabilizers and logical operators of the opposite type. Consequently, it is not possible to make (meaningful) guarantees on $d'$ knowing only the $\db{n,k,(d_\mathrm{x},d_\mathrm{z})}$ parameters of the primal code in general. We make further comments on particular cases below:

\begin{itemize}
    
    \item That $d_\mathrm{z}$ appears as a term in \cref{eq:distance_of_connected_code} implies that one must pick a primal code with $d_\mathrm{x} < d_\mathrm{z}$ for the construction to be possibly productive in distance. In the opposite case, where $d_\mathrm{x} \geq d_\mathrm{z}$, we have $d' \leq d_\mathrm{z} = d$ and the resulting code has distance at most that of the primal.
    

    \item As a concrete example, consider the $\db{8,3,(d_\mathrm{x}=2, d_\mathrm{z}=4)}$ hypercube code as the primal. Its $X$-type stabilizers are supported on the faces of a cube, while the single $Z$-type stabilizer acts on all qubits. One finds from the geometrical structure that $|L_\mathrm{x} \cap g_\mathrm{z}| \leq 2$, and the resulting distance is thereby $d' = \min(d_\mathrm{z}, 2 d_\mathrm{x}) = 4$. The resulting code is the $\db{16,6,4}$ tesseract code.

    \item A second illustrative example arises from the $\db{7,3,(d_\mathrm{x}=2,d_\mathrm{z}=3)}$ punctured hypercube code as the primal. Here, every $g_\mathrm{z}$ has weight at least three, and the stabilizer group contains only weight-four nontrivial operators. Consequently,
    \begin{equation}
        2|l_\mathrm{x}|+|g_\mathrm{z}|-2|l_\mathrm{x}\cap g_\mathrm{z}| 
        \ge |g_\mathrm{z}| + 2(|l_\mathrm{x}|-|l_\mathrm{x}\cap g_\mathrm{z}|)
        \ge |g_\mathrm{z}|
        \ge 3,
    \end{equation}
    with equality achieved when $l_\mathrm{x}$ lies entirely within $g_\mathrm{z}$ and $g_\mathrm{z}$ is a weight-three logical $\overline{Z}$ operator. As $d_\mathrm{z}=3$, the construction produces a $\db{14,6,(d_\mathrm{x}=3,d_\mathrm{z}=3)}$ code.  

\end{itemize}

Individually addressable permutation $\overline{\mathrm{CNOT}}$ gates on the primal code implies pairwise products of $\overline{\mathrm{CNOT}}$ gates implemented by qubit permutations on the resulting code. Specifically, let $\pi$ denote a qubit permutation implementing $\overline{\mathrm{CNOT}}_{ij}$ on the primal code. Then the qubit permutation $\pi \oplus \pi$---where the same permutation $\pi$ acts on the two sets of $n$ qubits arising from the primal and dual codes---implements $\smash{\overline{\mathrm{CNOT}}_{ij} \overline{\mathrm{CNOT}}_{(j+k)(i+k)}}$ on the resulting code.

\subsubsection{Two-sub-lattice CSS codes by doubling a non-CSS code}

A non-CSS phantom code can give rise to a CSS code by doubling the number of qubits, though the CSS code will not be phantom. This follows from the well-known doubling procedure:

\begin{theorem}[Non-CSS to CSS code doubling; Thm.~1 in Ref.~\cite{kovalev13ldpc}]
    Suppose an $\db{n,k,d}$ stabilizer code has parity check matrix $(A|B)$, then the CSS code with parity matrix 
    \begin{equation}
    \left(\begin{array}{cc|cc} 
    A & B & & \\
    & & B & A\end{array}\right)
    \end{equation}
    has parameter $\db{2n,2k,d'}$ where $d\le d'\le 2d$.
\end{theorem}

Given a logical basis of the original non-CSS code $\{l_{\mathrm{x},i} = (u_i|v_i)\}_{i=1}^k$ and $\{l_{\mathrm{z},i} = (s_i|t_i)\}_{i=1}^k$, a logical basis for the resulting CSS code can be written as
\begin{align}
    l'_{\mathrm{x},i} = (u_i,v_i|0,0),
    \qquad
    l'_{\mathrm{x},i+k} = (s_{i},t_{i}|0,0),
    \qquad
    l'_{\mathrm{z},i} = (0,0|t_i,s_i),
    \qquad
    l'_{\mathrm{z},i+k} = (0,0|v_{i},u_{i}),
\end{align}
for all $1 \le i \le k$. 

Likewise, individually addressable permutation $\overline{\mathrm{CNOT}}$ gates on the original non-CSS code implies pairwise products of $\overline{\mathrm{CNOT}}$ gates implemented by qubit permutations on the resulting CSS code. Specifically, let $\pi$ denote a qubit permutation implementing $\overline{\mathrm{CNOT}}_{ij}$ on the original code. Then the qubit permutation $\pi \oplus \pi$---where the same permutation $\pi$ acts on first $n$ and last $n$ qubits---implements $\smash{\overline{\mathrm{CNOT}}_{ij} \overline{\mathrm{CNOT}}_{(j+k)(i+k)}}$ on the resulting code. 

\section{Gate search beyond phantom}
\label{app:gates}

\subsection{Automorphism logical Clifford gates}
\label{app:gates/auto}

Here we detail our procedure to exhaustively identify all automorphism logical Clifford gates on the phantom codes generated---i.e.~logical gates implemented by qubit permutations and local (single-qubit) physical Clifford gates. Consider a stabilizer code on $n$ physical qubits defined by the stabilizer generator matrix $H \coloneqq [H^{(\mathrm{x})} \mid H^{(\mathrm{z})}]$ in symplectic form (see \cref{app:basics/symplectic}), where each row encodes a stabilizer generator. We extend this stabilizer generator matrix to a three-block representation~\cite{sayginel2025fault} by appending the entry-wise (modulo-two) sum $H^{(\mathrm{x})} \oplus H^{(\mathrm{z})}$, giving  
\begin{equation}
    H_\mathcal{E} \coloneqq
    \left[ H^{(\mathrm{x})} \mid H^{(\mathrm{z})} \mid H^{(\mathrm{x})} \oplus H^{(\mathrm{z})} \right].
\end{equation}

This form is convenient as both $H$ and $S$ physical gates act as coordinate permutations. The $H$ gate maps $X \leftrightarrow Z$ under conjugation and thus exchanges a column between $H^{(\mathrm{x})}$ and $H^{(\mathrm{z})}$, while the $S$ gate maps $X \mapsto Y$ and $Z \mapsto Z$, and therefore exchanges a column between $H^{(\mathrm{x})}$ and $H^{(\mathrm{x})} \oplus H^{(\mathrm{z})}$.

The row span $\langle H_\mathcal{E}\rangle$ defines a binary linear code. Its automorphism group $\mathrm{Aut}(\langle H_\mathcal{E}\rangle)$ consists of all column permutations that preserve this span. To ensure that these automorphisms correspond to valid physical Clifford circuits, we restrict attention to those that also lie in the automorphism group of  
\begin{equation}
    B \coloneqq [\mathbb{I} \mid \mathbb{I} \mid \mathbb{I}],
\end{equation}
which guarantees that each automorphism corresponds to a circuit composed only of $H$, $S$, and $\mathrm{SWAP}$ gates~\cite[Theorem 6]{sayginel2025fault}. The correspondence between column permutations and physical Clifford gates is straightforward:
\begin{itemize}[noitemsep]
    \item $(i,\,n+i)$, which affects qubit $i$ in blocks $H^{(\mathrm{x})}$ and $H^{(\mathrm{z})}$, implements $H_i$.
    \item $(n+i,\,2n+i)$, which affects qubit $i$ in blocks $H^{(\mathrm{z})}$ and $H^{(\mathrm{x})} \oplus H^{(\mathrm{z})}$, implements $S_i$ up to a Pauli correction.
    \item $(i,j)(n+i,n+j)(2n+i,2n+j)$, which affects qubits $i$ and $j$ in all three blocks, implements $\mathrm{SWAP}_{ij}$.
\end{itemize}

Thus, by computing
\begin{equation}
    \mathrm{Aut}(\langle H_\mathcal{E}\rangle) \cap \mathrm{Aut}(\langle B\rangle),
\end{equation}
we obtain the set of logical Clifford operators realizable as circuits built from $\{H,S,\mathrm{SWAP}\}$ that preserve the codespace. These circuits are precisely those generated by qubit permutations and local
physical Clifford gates. We find these automorphism groups using the computational algebra system \texttt{MAGMA}~\cite{bosma1997magma}.

A subtlety is that the logical gates found from $\mathrm{Aut}(\langle H_{\mathcal E}\rangle)\cap \mathrm{Aut}(\langle B\rangle)$ are exact only up to physical Pauli factors and global phase. Indeed, let $\overline{U}$ denote the Clifford unitary induced by a selected permutation in $\mathrm{Aut}(\langle H_{\mathcal E}\rangle)\cap \mathrm{Aut}(\langle B\rangle)$. Because binary symplectic data fixes a Clifford only up to Paulis and phases~\cite{gottesman1998theory}, $\overline{U}$ can send a stabilizer to a signed product of stabilizers, not necessarily positive. While stabilizer sign changes can in principle be frame-tracked in software, at least up till non-Clifford logical gates in the circuit that are sensitive to stabilizer signs, here we ensure an exact logical gate by tacking physical Pauli corrections that restores all stabilizers to be positive onto the gate implementation. These Pauli corrections are found through standard Clifford-tableau computation.

In general, the logical action associated with a physical circuit found in this procedure is dependent on the logical basis chosen on the code. On general stabilizer codes, this means that to answer the question of whether a desired logical gate $\overline{U}$ is available via qubit permutations and physical local Cliffords, one must enumerate over all possible logical bases of the code and list the automorphism gates in each---an expensive task. In the case of CSS logical bases, however, this issue does not arise as there is a notion of logical basis independence:

\begin{proposition}[CSS logical basis independence of automorphism logical gates]
    If a CSS phantom code admits any automorphism logical gate $\overline{O}$ in some CSS logical basis, then it admits that gate in every CSS logical basis. 
\end{proposition}

\begin{proof}
Changes between any pair of CSS logical bases on a code correspond to (i.e.~can be effected by) a logical $\mathrm{CNOT}$ circuit. Suppose the code admits the automorphism logical gate $\overline{O}$ in a CSS logical basis $L_\mathrm{x}, L_\mathrm{z}$. Now consider an arbitrary CSS logical basis $L_\mathrm{x}', L_\mathrm{z}'$, which is related to the first by a $\mathrm{CNOT}$ circuit $C$ such that $L_\mathrm{x}' = F_{\mathrm{xx}}(C) L_\mathrm{x}$ and $L_\mathrm{z}' = F_{\mathrm{zz}}(C) L_\mathrm{z} = \invtrans{F_{\mathrm{xx}}(C)} L_\mathrm{z}$ (see \cref{app:basics/symplectic} for notation conventions). Then $\overline{O}$ can be performed in the $L_\mathrm{x}', L_\mathrm{z}'$ logical basis as follows: (1) transform the logical basis to $L_\mathrm{x}, L_\mathrm{z}$ by performing $C^\dag$, (2) perform $\overline{O}$ as known, and (3) transform the logical basis back by performing $C$. Because the code is phantom, the entire circuit $C$ can be implemented via qubit permutations. Therefore, overall, $\overline{O}$ is performed using only qubit permutations and single-qubit physical Clifford gates---an automorphism logical gate.
\end{proof}

Therefore, on the CSS phantom codes that we investigate, choosing an arbitrary CSS logical basis for analysis of their automorphism logical gate sets suffices. We use the standard form logical basis (see \cref{app:code_discovery/standard_form}) in our implementation. We also note that, on a phantom code, the existence of an automorphism gate $\overline{O}_{ij \cdots k}$ on some set of distinct logical qubits $\{i, j, \ldots, k\}$ immediately implies the existence of automorphism $\overline{O}_{i' j' \cdots k'}$ on all distinct $\{i', j', \ldots, k'\}$, as logical swaps between logical qubits can be implemented by qubit permutations. For example, automorphism gates $\overline{H}_i$, $\overline{S}_i$, $\overline{H}_i \overline{S}_i \overline{H}_i$ when present are always supported on all logical qubits $i$, and $\mathrm{CZ}_{ij}$, $\mathrm{CXX}_{ij} = \overline{H}_i \overline{H}_j \mathrm{CZ}_{ij} \overline{H}_i \overline{H}_j$ when present are always supported on all ordered pairs of distinct logical qubits $(i, j)$.

\subsection{Logical Clifford and magic diagonal gates}
\label{app:gates/diagonal}

\subsubsection{Review of phase polynomials}
\label{app:gates/diagonal/phase_polynomial}

We first give a short review of the formalism of phase polynomials~\cite{amy2014polynomial,campbell2017unified}, which is convenient for treating diagonal unitary gates and will be used repeatedly in the following subsections. To begin, we consider single-qubit gates. The following association between phase monomials in the formal binary variable $x \in \{0, 1\}$ and single-qubit diagonal gates in the Clifford hierarchy holds:
\begin{equation}\begin{split}
    x \, \bmod 2 & \rightarrow Z, \\
    x \, \bmod 4 & \rightarrow S, \\
    x \, \bmod 8 & \rightarrow T. \\
\end{split}\end{equation}

A gate on the $l^\text{th}$ level of the Clifford hierarchy is associated with a phase monomial taken modulo $N = 2^l$. Controlled gates acting on multiple qubits can be expressed as follows:
\begin{equation}\begin{split}
    2x_1 x_2 \, \bmod 4 & \rightarrow \mathrm{CS}_{12}, \\
    2x_1 x_2 \, \bmod 8 & \rightarrow \mathrm{CS}_{12}, \\
    4x_1 x_2 x_3 \, \bmod 8 & \rightarrow \mathrm{CCZ}_{123}, \\
    8x_1 x_2 x_3 x_4 \, \bmod 8 & \rightarrow \mathrm{CCCZ}_{1234},
\end{split}\end{equation}
where the subscripts on the formal binary variables $x_1, x_2, \ldots$ denote the qubits on which the gates act. Assembling phase monomials into polynomials corresponds to taking products of the diagonal gates, taken modulo $N = 2^l$ where $l$ is the highest level of the Clifford hierarchy the gates in the circuit belong to:
\begin{equation}\begin{split}
    x_1 + x_2 + 2x_3 \, \bmod 4 & \rightarrow S_1 S_2 Z_3, \\
    x_1 + 3 x_2 + 4 x_4 \, \bmod 8 & \rightarrow T_1 T_2^3 Z_4.
\end{split}\end{equation}

Diagonal gates commute with each other. To describe diagonal circuits at the third level of the Clifford hierarchy (e.g.~$T$, $\mathrm{CS}$, $\mathrm{CCZ}$), taking modulus $N = 8$ is sufficient; for the fourth level (e.g.~$T^{1/2}$, $\mathrm{CCCZ}$), taking modulus $N = 16$ is sufficient. Lastly, we remark that the binary addition (i.e.~XOR) of formal variables can be turned into products and vice versa, a common transformation used when manipulating phase polynomials:
\begin{equation}\begin{split}
    2xy \, \bmod 8 
        &= x + y + 7(x \oplus y) \, \bmod 8, \\
    4x_1x_2x_3 \, \bmod 8 
        &= x_1 + x_2 + 7(x_1\oplus x_2) + x_3 + 7(x_1 \oplus x_3)
            + 7(x_2 \oplus x_3) + (x_1 \oplus x_2 \oplus x_3) \, \bmod 8.
\label{eq:product_to_sum_examples}
\end{split}\end{equation}

\subsubsection{Logical diagonal gates via products of single-qubit \texorpdfstring{$Z$}{Z}-basis rotations}
\label{app:gates/diagonal/transversal}

A computational- ($Z$-) basis logical state of an $\db{n,k,d}$ CSS code can be written as 
\begin{equation}\begin{split}
    \ket{\overline{\vb{x}}} 
    = \sum_{\vb{y} \in \mathbb{F}_2^{r_\mathrm{x}}} \ket{\vb{x} L_\mathrm{x} \oplus \vb{y} H_\mathrm{x}},
\end{split}\end{equation}
where $\vb{x}$ is the length-$k$ binary row vector labelling the state and $H_\mathrm{x} \in \mathbb{F}_2^{r_\mathrm{x} \times n}$ and $L_\mathrm{x} \in \mathbb{F}_2^{k \times n}$ are the $X$-type stabilizer generator matrix and logical matrix of the code respectively (see \cref{app:basics/phantom_css} for introduction to this convention). We assume without loss of generality that $H_\mathrm{x}$ has full row rank. Then, acting on the state with a physical operator 
\begin{equation}\begin{split}
    \vb*{\Gamma} = \sum_{i=1}^n \Gamma_i x_i \, \bmod N,
\end{split}\end{equation}
which is a product of single-qubit $Z$-basis rotations, yields the state
\begin{equation}\begin{split}
    \vb*{\Gamma} \ket{\overline{\vb{x}}} 
    = \sum_{\vb{y} \in \mathbb{F}_2^{r_\mathrm{x}}} 
        \omega^{\sum_{i=1}^n \Gamma_i (\vb{x} L_\mathrm{x} \oplus \vb{y} H_\mathrm{x})_i \bmod N} 
        \ket{\vb{x} L_\mathrm{x} \oplus \vb{y} H_\mathrm{x}},
\end{split}\end{equation}
where $\omega = e^{2i\pi / N}$ and as before $N = 2^l$ for probing the $l^\text{th}$ level of the Clifford hierarchy. Defining
\begin{equation}\begin{split}
    {\wt}_{\vb*{\Gamma}} (\vb{v}) 
        \coloneqq \sum_{i = 1}^n \Gamma_i v_i \, \bmod N,
    \qquad 
    \vb{v} \in \mathbb{F}_2^n,
\end{split}\end{equation}
we write more succinctly
\begin{equation}\begin{split}
    \vb*{\Gamma} \ket{\overline{\vb{x}}} 
    = \sum_{\vb{y} \in \mathbb{F}_2^{r_\mathrm{x}}} \omega^{{\wt}_{\vb*{\Gamma}} (\vb{x} L_\mathrm{x} \oplus \vb{y} H_\mathrm{x})} \ket{\vb{x} L_\mathrm{x} \oplus \vb{y} H_\mathrm{x}}.
\end{split}\end{equation}

For $\vb*{\Gamma}$ to implement a valid logical diagonal gate, each logical state $\ket{\overline{\vb{x}}}$ should merely pick up a phase, which demands that all phase factors in the sum above coincide. Therefore, we require that ${\wt}_{\vb*{\Gamma}} (\vb{x} L_\mathrm{x} \oplus \vb{y} H_\mathrm{x})$ is invariant for all $\vb{y} \in \mathbb{F}_2^{r_\mathrm{x}}$. We now follow the treatment described in Ref.~\cite[Appendix~D]{campbell2017unified}. To make use of this invariance constraint, we first rewrite $\vb{x} L_\mathrm{x} \oplus \vb{y} H_\mathrm{x}$ as
\begin{equation}\begin{split}
    \vb{x} L_\mathrm{x} \oplus \vb{y} H_\mathrm{x}
    = x_1\vb{g}^1\oplus\dots \oplus x_k\vb{g}^k\oplus 
        y_1\vb{g}^{k+1}\oplus \dots\oplus y_{m-k}\vb{g}^{m} 
    \equiv \bigoplus_{j=1}^m z_j\vb{g}^j,
\end{split}\end{equation}
where $\vb{g}^i$ are the rows of the stacked matrix $G = \binom{L_\mathrm{x}}{H_\mathrm{x}}$, and $\bz \equiv (x_1, x_2, ..., x_k, y_1, y_2, ..., y_{r_\mathrm{x}})$ is a length-$m$ binary vector, $m \coloneqq k + r_\mathrm{x}$. Then
\begin{equation}
    {\wt}_{\vb*{\Gamma}} (\vb{x} L_\mathrm{x} \oplus \vb{y} H_\mathrm{x}) = \sum_{i=1}^n \Gamma_i \left(\bigoplus_{j=1}^m z_j g_i^j\right) \, \bmod N.
    \label{eq:weight_of_gamma}
\end{equation}

Now we make use of the property~\cite[Eq.~D7]{campbell2017unified}
\begin{equation}
    \bigoplus_j a_j=\sum_j a_j-2\sum_{j<l}a_j a_l+4\sum_{j<l<r}a_j a_l a_r-8\sum\cdots,
    \label{eq:sum_to_product}
\end{equation}
for binary variables $a_i$ (a proof can be found in Ref.~\cite{webster2022xp}). This allows us to expand the $\bigoplus$ in \cref{eq:weight_of_gamma}, after which terms with coefficients that are multiples of $N$ can be discarded as they are congruent to $0 \, \bmod N$. For example, to explore the Clifford hierarchy up to the third level ($N = 8$), only the first three terms of \cref{eq:sum_to_product} need to be retained,
\begin{equation}\begin{split}
    {\wt}_{\vb*{\Gamma}} (\vb{x} L_\mathrm{x} \oplus \vb{y} H_\mathrm{x}) 
        = \sum_{i = 1}^n \Gamma_i \left[ 
            \sum_{j = 1}^m z_j \vb{g}_i^j 
            - 2 \sum_{1 \leq j < l\leq m} z_j z_l \vb{g}_i^j \vb{g}_i^l
            + 4\sum_{1\leq j < l < r \leq m} z_j z_l z_r \vb{g}_i^j \vb{g}_i^l \vb{g}_i^r
            \right] \, \bmod 8
\end{split}\end{equation}
and employing the definition of ${\wt}_{\vb*{\Gamma}}(\cdot)$, we can write this as
\begin{equation}\begin{split}
    {\wt}_{\vb*{\Gamma}} (\vb{x} L_\mathrm{x} \oplus \vb{y} H_\mathrm{x}) 
    = \sum_{i = 1}^n {\wt}_{\vb*{\Gamma}}(\vb{g}^j)z_j 
        - \sum_{1 \leq j < l\leq m} 2{\wt}_{\vb*{\Gamma}} (\vb{g}^j \wedge \vb{g}^l) z_j z_l
        + \sum_{1 \leq j < l < r \leq m} 4{\wt}_{\vb*{\Gamma}} (\vb{g}^j \wedge \vb{g}^l \wedge \vb{g}^r) z_j z_l z_r
        \,\, \bmod 8,
\end{split}\end{equation}
where the symbol $\wedge$ denotes element-wise product of binary vectors, $(\vb{u} \wedge \vb{v})_i = u_i v_i$. Since we require this expression to be invariant with respect to $\vb{y}$, the rows of $G$ must satisfy:
\begin{itemize}
    \item ${\wt}_{\vb*{\Gamma}} (\vb{g}^j) \equiv 0 \, \bmod 8, \ \forall j \in \{k+1,...,m\}$;
    \item $2{\wt}_{\vb*{\Gamma}} (\vb{g}^j \wedge \vb{g}^l) \equiv 0 \, \bmod 8, \ \forall j, l \in \{1,...,m\}$, with $j \neq l$ and not both in $\{1,...,k \}$;
    \item $4{\wt}_{\vb*{\Gamma}} (\vb{g}^j \wedge \vb{g}^l \wedge \vb{g}^r) \equiv 0 \, \bmod 8, \ \forall j, l, r \in \{1,...,m\}$, with $j \neq l \neq r \neq j$ and not all three in $\{1,...,k \}$.
\end{itemize}

Finding solutions for $\vb*{\Gamma}$ that satisfy this condition, and which therefore implement valid logical diagonal gates, can be reduced to finding the kernel of a matrix $M$ whose entries in $\mathbb{Z}_8$ comprising the rows:
\begin{itemize}
    \item $\vb{g}^j$ for $j \in \{k+1,k+2,...,m \}$;
    \item $2(\vb{g}^j \wedge \vb{g}^l)$ for $j,l \in \{1,...,m \}$, $j\ne l$ not both in $\{1,...,k \}$;
    \item $4(\vb{g}^j \wedge \vb{g}^l \wedge \vb{g}^r)$ for $j,l,r \in \{1,...,m\}$, $j \ne l \ne r \ne j$ not all three in $\{1,...,k\}$.
\end{itemize}

The generators of the kernel of $M$ (modulo $8$) describe all possible ways of acting with a series of $T$ gates and their powers on every qubit such that the codespace of the code is preserved. Lastly, the logical action of any valid combination of $T$ gates amounts to a phase factor $\exp[i \pi f_T(x_1,x_2,...,x_k)/4]$ on the logical states---in general $\exp[2i\pi f(\vb{x})/N]$ for an arbitrary level of the Clifford hierarchy---which can be recovered by evaluating the polynomial
\begin{equation}
    f_T(\vb{x}) = \sum_{j = 1}^k {\wt}_{\vb*{\Gamma}} (\vb{g}^j)x_j + \sum_{1 \leq j < l\leq k} 2{\wt}_{\vb*{\Gamma}} (\vb{g}^j \wedge \vb{g}^l) x_j x_l + \sum_{1 \leq j < l < r \leq k}  4{\wt}_{\vb*{\Gamma}} (\vb{g}^j \wedge \vb{g}^l \wedge \vb{g}^r) x_j x_l x_r \,\, \bmod 8.
\end{equation}

This method can be used to exhaustively identify logical diagonal gates implemented by products of single-qubit $Z$-basis rotations at any level of the Clifford hierarchy. For example, at the second level of the Clifford hierarchy where the physical single-qubit rotations are powers of $S$, the matrix $M$ has entries in $\mathbb{Z}_4$ and comprises the rows:
\begin{itemize}
    \item $\vb{g}^j$ for $j \in \{k+1,k+2,...,m \}$;
    \item $2(\vb{g}^j \wedge \vb{g}^l)$ for $j,l \in \{1,...,m \}$, $j\ne l$ not both in $\{1,...,k \}$;
\end{itemize}
and the overall phase on logical states is now given by $\exp[i \pi f_S(x_1,x_2,...,x_k)/2]$ where
\begin{equation}
    f_S(\vb{x}) = \sum_{j = 1}^k {\wt}_{\vb*{\Gamma}} (\vb{g}^j)x_j 
        + \sum_{1 \leq j < l\leq k} 2{\wt}_{\vb*{\Gamma}} (\vb{g}^j \wedge \vb{g}^l) x_j x_l \,\, \bmod 4.
\end{equation}

\subsubsection{Logical diagonal fold gates}
\label{app:gates/diagonal/fold}

The above method of finding logical diagonal gates implemented by products of single-qubit $Z$-basis rotations (i.e.~transversally) can be combined with the embedded code technique~\cite{sayginel2025fault, webster2023transversal}, to more generally find logical diagonal gates implemented by products of physical single- and multi-qubit diagonal gates. In particular, this is useful for identifying gates that are implemented with physical $S$ and $\mathrm{CZ}$ in depth one. This type of logical gates is termed ``fold'' gates in the literature, and is distinguished from transversal gates as they involve physical multi-qubit operations within the codeblock and is therefore not a priori guaranteed to be distance-preserving. 

\begin{figure}[!ht]
    \centering
    \[
    \begin{quantikz}[row sep={0.5cm,between origins}, column sep={0.7cm,between origins}]
    \lstick{$q_1$} & \ctrl{2} & \qw       & \qw      & \qw       & \ctrl{2} & \qw \\
    \lstick{$q_2$} & \qw      & \ctrl{1}  & \qw      & \ctrl{1}  & \qw      & \qw \\
    \lstick{$\ket{0}$} & \targ{}   & \targ{}    & \gate{S} & \targ{}    & \targ{}   & \qw
    \end{quantikz}
    \;=\;
    \begin{quantikz}[row sep={0.7cm,between origins}, column sep={0.7cm,between origins}]
    \lstick{$q_1$} & \gate{S} & \ctrl{1}  & \qw \\
    \lstick{$q_2$} & \gate{S} & \ctrl{-1} & \qw
    \end{quantikz}
    \]
    \caption{\textbf{Circuit identity underlying code embedding for $\bm{\mathrm{CZ}}$ gates.} Here $\{q_1, q_2\}$ is an arbitrary pair of data qubits on the original code, and the third qubit initialized in $\ket{0}$ is an ancillary qubit that is introduced. Coupling the ancillary qubit to $\{q_1, q_2\}$ and performing an $S$ gate and a reset ($R$) on the ancilla is equivalent to performing $S$ and a $\mathrm{CZ}$ gate on $\{q_1, q_2\}$. This follows from the phase polynomial identity $q_1\oplus q_2 = q_1+q_2+2q_1q_2\mod 4$.}
    \label{fig:extended_code_S}
\end{figure}
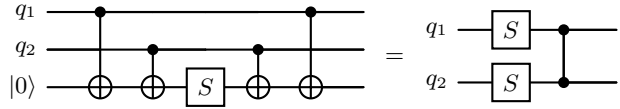

To introduce the embedded code approach, we establish the circuit identity in \cref{fig:extended_code_S}. The idea is that by suitably introducing an ancillary qubit for every pair of data qubits on the original code, single-qubit $S$ gates on the ancillary qubit can be made equivalent to $S$ and $\mathrm{CZ}$ gates on the data qubits. These ancillary qubits are entangled with data qubits on the original code via $\mathrm{CNOT}$s, effectively defining a larger ``embedded'' code---these $\mathrm{CNOT}$ operations extend the encoding circuit of the original code to produce the encoding circuit of the embedded code. The stabilizers of the embedded code are accordingly fixed by the stabilizers of the original code and the $Z$ stabilizers of the ancillary qubits, propagated through the $\mathrm{CNOT}$ operations. To be concrete, we outline explicitly the steps to construct the embedded code:
\begin{enumerate}
    \item Add an ancillary qubit for each pair of data qubits $\{q_i, q_j\}$ on the original code.
    \item Propagate each stabilizer of the original code through the $\mathrm{CNOT}$ gates with the ancilla qubits.
    \item For each ancillary qubit $x$ associated with data qubits $\{q_i, q_j\}$ of the original code, add a stabilizer $Z_x Z_i Z_j$.
\end{enumerate}

This procedure produces an $\db{n + \binom{n}{2},k,d'}$ embedded code for an original $\db{n,k,d}$ code. A logical basis on the embedded code can likewise be found by propagating a logical basis on the original code through the $\mathrm{CNOT}$ operations. Fold-diagonal logical Clifford gates on the original code can then be found by identifying diagonal logical Clifford gates on the embedded code implemented by products of single-qubit $Z$-basis rotations, through the method described in \cref{app:gates/diagonal/transversal}, and mapping the single-qubit operations on ancillary qubits where present back onto data qubits of the original code via the circuit identity of \cref{fig:extended_code_S}. 

This Z-basis rotations on embedded code approach is proposed in~\cite{webster2023transversal}, and gives a fold-type logical gate~\cite{kubica2015unfolding,breuckmann2024fold} when restricted to a depth-one physical implementation. A later work~\cite{sayginel2025fault} also allows one to do permutations on the embedded code whose action on the original code translates to CNOT gates. However, we only use the former approach because a logical phase gate with some $S$ action is enough to enable the full Clifford group for phantom codes, see~\cref{fig:qRM_addressability}.
The advantage of restricting to depth-one physical implementation is that the fault-distance of this gadget is at least half of the code distance (since two-qubit gates are allowed but with disjoint support).

However, although any element in the kernel of $M$ (see~\cref{app:gates/diagonal/transversal}) applied to the embedded code gives a valid diagonal logical operation, their physical implementation may not be depth-one when translated back to the original code.
Due to this enormous search space, it is almost intractable to identify all the fold diagonal gates; it is also inefficient to even find one for larger codes (unless one constrains which pairs to add the gadget).
Therefore, only partial search results for those $n\lesssim 20$ phantom codes are shown in~\cref{tab:gate_for_phantom_codes}.
The fold diagonal gates for the larger phantom codes are algebraically constructed (or unknown), see~\cref{app:qrm/folds} and~\cref{app:binarize_concatenate/gates}.

\section{Numerical benchmarking of logical performance}
\label{app:benchmarking}

\subsection{Noise model}
\label{app:benchmarking/noise_model}

We employ a circuit-level noise model with relative error rates calibrated to recent neutral-atom array experiments~\cite{bluvstein2025fault}, detailed in \cref{tab:noise_model_details}.

\begin{table}[!ht]
\begin{tabular}{p{4cm} p{4.5cm} l}
\toprule
Error mechanism    & Model                 & Relative rates \\ 
\midrule
Idle error         & 1-qubit depolarizing  & $p/300$        \\
1-qubit gate error & 1-qubit depolarizing  & $p/15$         \\
2-qubit gate error & 2-qubit depolarizing  & $p$            \\
Measurement error  & Classical bit-flip    & $5p/3$         \\
Reset error        & Qubit bit-/phase-flip & $p$            \\ 
\bottomrule
\end{tabular}
\caption{\textbf{Noise model used in benchmarking simulations.} Gate errors are applied after the physical gates, measurement errors are applied on the classical measurement outcomes, and reset errors are applied on physical qubits after reset and initialization.}
\label{tab:noise_model_details}
\end{table}

Depolarizing error channels of strength $p$ entail a probability $1 - p$ of leaving the input state unperturbed, and bit- and phase-flip error channels of strength $p$ impose $X$ and $Z$ errors, respectively, with probability $p$. 

We take $p = 3 \times 10^{-3}$ for error rates on current-generation hardware based on recent experiments (Ref.~\cite{bluvstein2025fault}, and e.g.~Refs.\cite{manetsch2025tweezer, muniz2025repeated, senoo2025high}), and $p = 1 \times 10^{-3}$ and $p = 5 \times 10^{-4}$ for near-term (within the next $1$--$2$ years) and future error rates, respectively, informed by ongoing and projected progress on neutral-atom platforms. Many alternative effective and microscopic error models are possible~\cite{bilokur2024thermodynamic, dasu2025breaking}; our choice here is a pragmatic choice given the clear compatibility of neutral atoms with phantom codes.

\subsection{Surface code implementation details}
\label{app:benchmarking/surface}

\subsubsection{Circuit structure and matching-based correlated decoding} 
\label{app:benchmarking/surface/decoding}

We employ a version of the recent matching-based correlated decoder~\cite{cain2025fast} for the surface code, which follows from a line of work developing correlated decoding and algorithmic fault tolerance for logical circuits containing transversal logical gates~\cite{cain2024correlated,zhou2025low}. Here, we use the term ``correlated decoding'' to mean decoding techniques that take into account error propagation through the physical circuit underlying the computation. For example, through a transversal $\overline{\mathrm{CNOT}}$ gate, an $X_i$ error on the $i^\text{th}$ data qubit of the control codeblock would propagate into $X_i$ errors on \emph{both} the control and target codeblocks. Correlated decoding enables drastic improvements in logical performance over simpler propagation-blind independent decoding of codeblocks (e.g.~a ${\sim} 2 \times$ enhancement in threshold of a $\overline{\mathrm{CNOT}}$ gate~\cite{cain2024correlated}). To suitably discuss decoding, we use the language of detectors~\cite{gidney2021stim,derks2025designing}, which is a space-time generalization of stabilizers and refer to products of Pauli operators that detect Pauli errors in a region of a circuit.

For our purpose, the relevant advance of Ref.~\cite{cain2025fast} is a framework to perform correlated decoding on surface codes through minimum-weight perfect matching (MWPM), which is fast in practice and highly accurate, in place of slower integer programming or less accurate clustering decoders used in earlier works~\cite{cain2024correlated,zhou2025low}. While surface codes are matchable in memory experiments with code-capacity or circuit-level noise, where decoding can be performed through standard MWPM, the same is not true for logical circuits with transversal gates performed between codeblocks. The difficulty is that logical gates between codeblocks in general introduce hyperedges in the decoding graph, where a single error mechanism (e.g.~a measurement error) flips more than two detectors (i.e.~checks). The decoding hypergraph, unlike simple graphs, cannot be treated efficiently with standard MWPM. 

The method proposed in Ref.~\cite{cain2025fast} is to split the decoding hypergraph into multiple graphs, each deriving from the lightcone of a reliable logical Pauli product on the logical circuit, such that each decoding subproblem can be treated via MWPM. The overall decoding predictions of logical observable (i.e.~$\overline{X}_1$, $\overline{X}_2$, $\overline{Z}_1 \overline{Z}_2$, etc.) flips are then assembled from the outcomes of these subproblems. There are different ways of performing this decomposition---for example, one can start from the logical observables measured at the end of the circuit and examine their backward lightcones, or the logical stabilizers defining the initial logical state and examine their forward lightcones~\cite[Sec.~2.2 and App.~B]{cain2025fast}. We employ the latter strategy. As is consistent with Ref.~\cite{cain2025fast}, we adopt the convention that a single syndrome extraction (SE) round is performed on the two codeblocks involved in every transversal $\overline{\mathrm{CNOT}}$ before and after the gate. An overview of the decoding procedure is as follows:

\begin{enumerate}
    
    \item On a logical circuit over $K$ logical qubits on $K$ codeblocks of the surface code, we 
    denote by $\mathcal{S}^0 = \{S_i^0\}_{i = 1}^K$ a basis of logical Pauli stabilizers of the initial logical state. For example, the initial state $\ket{\overline{+00}}$ is stabilized by $\smash{\{\overline{X}_1, \overline{Z}_2, \overline{Z}_3\}}$, where the subscripts label the codeblocks.

    \item The SE rounds on the logical circuit can be organized into layers. The $t = 0$ layer is the first round of SE, which initializes the surface code codeblocks. For example, to prepare a codeblock in the $\smash{\ket*{\overline{0}}}$ (resp.~$\smash{\ket*{\overline{+}}}$) state, the data qubits are initialized in $\smash{\ket{0}^{\otimes n}}$ (resp.~$\smash{\ket{+}^{\otimes n}}$) and the first SE round performed. Subsequent $t \geq 1$ SE rounds are performed during the logical computation. The logical circuit terminates with transversal measurement of all codeblocks, which readouts the logical Pauli observables of the codes and also serves as the last SE round.
    
    \item The decoding hypergraph comprises all error mechanisms that can occur on the circuit as defined by the noise model (e.g.~depolarizing noise, measurement and reset errors) as hyperedges and detectors as vertices. We breakdown the detectors into two types:
    
    \begin{itemize}
        \item \textit{Initialization.} The $t = 0$ SE rounds themselves present as detectors. On a codeblock initialized in $\smash{\ket*{\overline{0}}}$ (resp.~$\smash{\ket*{\overline{+}}}$), all $Z$-type (resp.~$X$-type) stabilizers measured at $t = 0$ are detectors.
        \item \textit{Computation.} For all $t \geq 1$, detectors are declared by propagating each stabilizer $\smash{s^t_i}$ of each codeblock measured in the SE round backward through the previous logical gate, which results in a product of stabilizers $\smash{\{s^{t - 1}_j\}}$ generically on multiple codeblocks measured in the previous SE round; the product of the set of measured stabilizers $\smash{\{s^{t - 1}_j\} + s^{t}_i}$ is declared as a detector.
    \end{itemize}

    Error mechanisms (hyperedges) are incident on detectors (vertices) that they flip. The logical gates present in our context are mostly transversal $\overline{\mathrm{CNOT}}$ gates between codeblocks.

    \item We produce $K$ decoding graphs from the hypergraph. For each $S_l^0 \in \mathcal{S}^0$, we filter the decoding hypergraph into a corresponding graph by keeping only detectors in the forward lightcone of $S_l^0$. Specifically, the initialization detectors on codeblocks in $\supp(S_l^0)$ are included. Moreover, propagating $S_l^0$ forward through logical gates to each SE layer $t \geq 1$ to obtain $\smash{S_l^t = \mu_{i_1} \mu_{i_2} \cdots}$ where $\smash{\mu_{i_\ell} \in \{X, Z\}}$, all detectors of the form $\smash{\{s^{t - 1}_j\} + s^{t}_{i_\ell}}$ where $\smash{i_\ell \in \supp(S_l^t)}$ and $\smash{s^{t}_{i_\ell}}$ is the same Pauli type as $\smash{\mu_{i_\ell}}$ are included. All other detectors are deleted from the hypergraph. This reduces all hyperedges into edges~\cite{cain2025fast}.

    \item Each simulated noisy shot of the circuit produces a vector of binary measurement outcomes and thereby binary detector outcomes (whether or not each is flipped). Each of the $K$ decoding graphs are decoded through standard MWPM with this same vector of detector outcomes. The decoding result of the decoding graph built from the forward lightcone of $S_l^0 \in \mathcal{S}$ reports on whether the logical Pauli product $S_l^{\mathrm{f}}$ is flipped at the end of the circuit, where $S_l^{\mathrm{f}}$ is $S_l^0$ forward-propagated through the entire logical circuit.
    
    Note that $\mathcal{S}^{\mathrm{f}} = \{S_i^{\mathrm{f}}\}_{i = 1}^K$ is a complete basis of logical stabilizers for the terminal logical state. Therefore, whether or not any deterministic Pauli logical observable (e.g.~$\overline{X}_1 \overline{Z}_2$) is flipped at the end of the circuit can be predicted from the decoding outcomes of the $K$ decoding graphs.
    
\end{enumerate}

We use the sparse blossom MWPM algorithm in \texttt{pymatching}~\cite{higgott2025sparseblossom} to decode the decoding graphs.

\subsection{Phantom code implementation details}
\label{app:benchmarking/phantom}

\subsubsection{Fault-tolerant state preparation protocol}
\label{app:benchmarking/phantom/state_prep}

Our state preparation protocols for the phantom $\db{64,4,8}$ code builds on Ref.~\cite{gong2024computation}, which is originally proposed for the Golay code in Ref.~\cite{paetznick2013golay}. The main idea is to employ a four-to-one logical state certification protocol for fault tolerance, which can be understood as a concatenated two-to-one (classical) state certification circuit with the two levels certifying the $X$ and $Z$ bases. There are two differences to Ref.~\cite{gong2024computation}. First, we abstract away AOD movements as considered in Ref.~\cite{gong2024computation} as we assume qubit permutations are to be compiled away as far as possible. Second, as our present context concerns $k > 1$ codes, we can take advantage of state automorphisms---i.e.~qubit permutations that preserve the logical $\ket*{\smash{\overline{0}}^k}$ or $\ket*{\smash{\overline{+}}^k}$ state---which are larger than the code automorphism group.

We illustrate an example for the $\db{64,4,8}$ phantom qRM code defined in \cref{thm:phantom_qRM}, where all extra logicals are set to $Z$-type stabilizers. This code is not the one we use for benchmarking (it has $22$ $X$-type and $38$ $Z$-type stabilizer generators), but it allows us to explain the concept of the $|\overline{0000}\ra$ state automorphism more easily. Whereas we use the top left $4 \times 4$ block of $A\in\GL(6,\F_2)$ manipulating $\{x_6,x_5,x_4,x_3\}$ [see~\cref{eq:RM_affine_automorphism}] for permutation $\overline{\mathrm{CNOT}}$s, in the $|\overline{0000}\ra$ state, the monomials $\{x_1x_2x_3,x_1x_2x_4,x_1x_2x_5,x_1x_2x_6\}$ take definite $+1$ values in the $Z$ basis and can therefore be treated as additional $Z$-type stabilizers. Therefore, all the degree $\le 3$ (resp.~$\le 2$) monomials are set to $Z$-type (resp.~$X$-type) stabilizers for this \emph{state} on the code. We can then use the full $\GA$ automorphism on the six variables $x_6,\dots,x_1$. In particular, $A$ does not have to be block-diagonal; any $A\in\GL(6,\F_2)$ is a $|\overline{0000}\ra$ state automorphism. We state without proof that $\left(\begin{smallmatrix} B & C \\ 0_{2 \times 4} & D \end{smallmatrix}\right) \in \GL(6,\F_2)$, where $B\in\F_2^{4\times 4},\; C\in\F_2^{4\times 2},\; D\in\F_2^{2\times 2}$, serves as state automorphisms for $|\overline{\text{++++}}\ra$.

For the balanced $\db{64,4,8}$ phantom qRM code (see \cref{def:balanced_64-4-8}) that we perform numerical benchmarking on, we search over the following state automorphisms for the $|\overline{\text{0000}}\ra$ and $|\overline{\text{++++}}\ra$ state, respectively ($B$ is a $4\times 4$ matrix):
\begin{equation}
    \left(\begin{array}{cc}
        B & \begin{smallmatrix}
            * & 0\\ * & 0 \\ * & 0 \\ * & 0
        \end{smallmatrix} \\
        \begin{smallmatrix}
            * & * & * & * \\ * & * & * & *
        \end{smallmatrix} & 
        \begin{smallmatrix}
            * & 0\\ * & *
        \end{smallmatrix}
    \end{array}\right)\in\GL(6,\F_2),
    \qquad
    \left(\begin{array}{cc}
        B & \begin{smallmatrix}
            * & 0\\ * & 0 \\ * & 0 \\ * & 0
        \end{smallmatrix} \\
        \begin{smallmatrix}
            0 & 0 & 0 & 0 \\ * & * & * & *
        \end{smallmatrix} & 
        \begin{smallmatrix}
            * & 0\\ * & *
        \end{smallmatrix}
    \end{array}\right)\in\GL(6,\F_2),
\end{equation}
where $B \in \F_2^{4 \times 4}$ and $*$ indicating free entries are both searched over.

For a concrete implementation of the state-preparation protocol, we choose four permutations for each to minimize the number of malignant faults. We call an order-$s$ fault (configuration) \emph{malignant} if it results in a strict preselection acceptance (i.e.~no syndromes measured on the ancillary codeblock), but yet leaves a weight ${>}s$ residual error on the output codeblock (here the weight of the error is the minimum allowing multiplication by any stabilizer of the code \emph{state}).
For our distance-8 code, full fault-tolerance requires that there are no order-${<}4$ malignant faults, so that the logical failure rate is guaranteed to scale at least as well as $p^4$ when varying the physical error rate $p$. 

We were unable to exactly satisfy this strict demand through our numerical search. 
Our state preparation protocol for $|\overline{\text{0000}}\ra$ has no order-two and ${<} 100$ order-three malignant faults, and that for $|\overline{\text{++++}}\ra$ has one order-two and ${\sim} 200$ order-three malignant faults. 
This difference in the number of malignant faults is expected, as the state automorphism space is larger for $|\overline{\text{0000}}\ra$ than $|\overline{\text{++++}}\ra$. 
Moreover, we did not find the $+\vb{b}$ in \cref{eq:RM_affine_automorphism}, where $\vb{b}\in\F_2^6$, to help with eliminating the order-two malignant faults.
Nevertheless, since the number of order-two/three malignant faults is small, we were able to observe $p^4$ scaling in logical failure rate in the $p = 5\times 10^{-4}$ to $p = 3 \times 10^{-3}$ regime in our numerical simulations.

\subsubsection{Interfacing fault-tolerant state-preparation with logical circuit simulations}
\label{app:benchmarking/phantom/interface}

To ensure an accurate simulation of noise on prepared codeblocks and their subsequent propagation into the logical circuit being executed, we performed (Monte Carlo) simulations of the state-preparation protocol described above (\cref{app:benchmarking/phantom/state_prep}) for $|\overline{\text{0000}}\ra$ and $|\overline{\text{++++}}\ra$ logical states and recorded large number of samples of the residual error on the output codeblocks under strict and relaxed preselection. Then, when running the numerical simulation of the logical circuit, we inject these errors on each (otherwise perfectly initialized) ancillary codeblock employed.

\subsubsection{Spatiotemporal sliding-window list and most-likely-error decoding}
\label{app:benchmarking/phantom/decoding}

Steane EC (with fault-tolerant state preparation) is known to be single-shot as the Steane syndrome extraction gadget involves only transversal operations. However, there are error correlations between Steane rounds that can be exploited in decoding. This notion is similar to \emph{correlated decoding} as introduced on surface codes (see \cref{app:benchmarking/surface/decoding}), with a structural difference being that transversal $\overline{\mathrm{CNOT}}$ gates are present in the Steane gadgets themselves.

We give a simple example to motivate this decoding consideration. Consider two consecutive rounds of bit-flip (i.e.~$X$-error-detecting) Steane EC on a single data codeblock. Before these two rounds, there is the residual $X$-type error $x\in\F_2^n$ on the codeblock. For simplicity, assume that only the transversal measurement in the Steane gadget is noisy (i.e.~the transversal $\overline{\mathrm{CNOT}}$s and ancillary codeblocks are perfectly noiseless), and errors $m_1,m_2\in\F_2^n$ occurred on the measurements. Then we observe $x\oplus m_1\oplus c_1$ and $x\oplus m_2 \oplus c_2$ in the two transversal measurements, where $c_1,c_2\in\F_2^n$ are random codewords. A naïve decoding of the first Steane round independent of the second would return the $X$-type correction $x \oplus m_1$ to be applied to the data codeblock; however, it is clearly more optimal to avoid performing the $m_1$ correction, as $m_1$ is a measurement error on the ancillary codeblock and does not correspond to qubit errors on the data codeblock. A better strategy is therefore to decode $x\oplus m_1\oplus c_1$ and $x\oplus m_2\oplus c_2$ jointly, to produce only $x$ as the correction. This simple observation---to decode across two Steane rounds---yielded two-orders-of-magnitude improvement in logical failure rate in our numerical tests. 

This explains the use of sliding-window decoding where each window spans two consecutive Steane EC rounds. To also take into account error propagation across transversal $\overline{\mathrm{CNOT}}$ gates on the logical circuit (à la correlated decoding), we set each window to span the control and target codeblock of each transversal $\overline{\mathrm{CNOT}}$ gate and their Steane EC rounds before (``leading'') and after (``trailing'') the gate. Each window is decoded by either list or most-likely-error (MLE) decoding:

\begin{itemize}
    
    \item \textit{MLE decoding.} We construct the circuit-level detector error model (i.e.~parity check matrix) for each window, and commit to faults that are inferred to have happened before the transversal $\overline{\mathrm{CNOT}}$. The corresponding corrections are propagated through the circuit (i.e~physical Pauli-frame tracking) to update measurement results in later Steane rounds. In this way the decoding for later windows depend on results of earlier windows.
    
    We implement the MLE decoding per window by integer programming~\cite{bluvstein2024logical}. In short, we construct a mixed-integer linear program that maximizes the probability (more specifically log-likelihood) of the inferred physical error given the noise model, subject to the constraint that the error generates syndromes consistent with those observed. These problems are solved with the high-performance integer programming solver \texttt{Gurobi}. 

    To simplify the detector error model so that decoding is faster, we chose to model the ancillary codeblocks for Steane EC as noiseless logical states (e.g.~through the hypercube encoding circuit) subject to single-qubit depolarizing noise on each qubit. The strength of this noise is set to $0.8 p$ for strict preselection and $2.5 p$ for relaxed preselection; these scales were chosen by probing the residual error weight distributions from the state preparation protocol. Moreover, we modelled the residual error on the data codeblock prior to the leading Steane EC round as single-qubit depolarizing noise of strength $2p$; this heuristic value was not specially tuned. We emphasize that the simplified noise modelling here is only for the purpose of decoding; the circuit simulation injects errors sampled from simulations of the state-preparation protocol (see \cref{app:benchmarking/phantom/interface}).

    \item \textit{List decoding.} We adapt the list decoder developed in Ref.~\cite{gong2024improved} for single-round code-capacity depolarizing-noise decoding to the setting of correlated decoding across two Steane EC rounds. We illustrate the central idea using the example described previously. The decoder takes as input $y_1 = x \oplus m_1 \oplus c_1$ and $y_2 =x \oplus m_2 \oplus c_2$. Let us only focus on the error on one qubit, so that $x,m_1,m_2\in \F_2$, and ignore the random codewords. Assuming that the errors are well-modelled by the binary symmetric channel, $x\sim\mathrm{BSC}(p)$ and $m_1,m_2\sim\mathrm{BSC}(p)$, the correlation between $y_1=x\oplus m_1$ and $y_2=x\oplus m_2$ is, in fact, structurally similar to $X$ and $Z$ error correlations in depolarizing noise. This similarity allows the code-capacity depolarizing-noise list decoder to be repurposed. 
    
    In particular, we associate depolarizing noise with a quaternary alphabet: $\mathbb{I}=(0,0)$, $X=(1,0)$, $Z=(0,1)$ and $Y=(1,1)$. In our present task of decoding across two Steane rounds, we associate $(1,0)$ with $m_1=1,\;m_2=0,\; x=0$; $(0,1)$ with $m_1=0,\;m_2=1,\;x=0$; and, most importantly, $(1,1)$ with $x=1$ and $m_1=m_2=0$. In the last case, it is also possible that $(1,1)$ arises from $x=0$ and $m_1=m_2=1$, but this arises from a higher-order fault process occurring with probability $p^2$ (compared to $p$) and so can be disregarded to good approximation. That is, only when the list decoder infers the fault $(1,1)$ given the inputs $y_1$ and $y_2$ (likewise interpreted in quaternary) do we attribute a correction to the residual error $x$.

    This correlated list decoder takes $\mathcal{O}(L n\log n)$ time per window, where $L$ is the list size. We set the list size heuristically to $16$ for the $\db{64,4,8}$ code. We used the \texttt{aff3ct}~\cite{Cassagne2019a} toolbox in the implementation of list decoding.

\end{itemize}

In practice, we observe that the list decoder is ${>} 100 \times$ faster, though ${\sim}$half as accurate, per window than MLE decoding in our benchmarking. The difference in decoding accuracy is expected, because MLE is able to draw on qualitatively more correlation information: MLE accesses both bit-flip and phase-flip Steane EC syndromes, on both the control and target codeblocks, together for decoding. The correlated list decoder, on the other hand, separates the decoding for $X$ and $Z$ errors. For $X$ errors, it (1) decodes across the two bit-flip EC rounds (before and after the transversal $\overline{\mathrm{CNOT}}$) of the control codeblock, (2) commits corrections to the leading EC round, (3) propagates the correction to update the trailing bit-flip syndromes of the target codeblock, and finally (4) decodes across the two bit-flip Steane EC rounds of the target codeblock; and similarly for $Z$ errors, where it decodes across the phase-flip EC rounds of the target codeblock first. In principle, with a larger alphabet (than quaternary), it is possible for the list decoder to be generalized to take these correlations---between $X$ and $Z$ error types and between control and target codeblocks) into account. We leave this possible improvement for future work.

\subsection{\texorpdfstring{$\overline{\mathrm{CNOT}}$}{CNOT} circuit benchmark}
\label{app:benchmarking/cx}

In \cref{sec:benchmarking/cx,fig:benchmark/cx} we benchmarked a $\overline{\mathrm{CNOT}}$ circuit of varying depth on four logical qubits, hosted on either four surface code codeblocks or a single $\db{64,4,8}$ phantom qRM codeblock. These circuits take the following form: (1) data codeblock(s) initialization in $\ket*{\overline{\text{0000}}}$ or $\ket*{\overline{\text{++++}}}$, (2) repeated $\overline{\mathrm{CNOT}}$s in the arrangement shown in \cref{fig:benchmark/cx}, (3) transversal measurement of deterministic $\{\overline{Z}_1, \overline{Z}_2, \overline{Z}_3, \overline{Z}_4\}$ for the $\ket*{\overline{\text{0000}}}$ initial state, or $\{\overline{X}_1, \overline{X}_2, \overline{X}_3, \overline{X}_4\}$ for the $\ket*{\overline{\text{++++}}}$ initial state, where the subscripts denote logical qubits. 

For the surface code, codeblock initializations are performed in the standard fashion with physical Pauli-frame tracking (i.e.~à la algorithmic fault tolerance)---see \cref{app:benchmarking/surface/decoding}---and all $\overline{\mathrm{CNOT}}$s are transversal between codeblocks. For the phantom code, codeblock initialization is via our fault-tolerant state preparation protocol (see \cref{app:benchmarking/phantom/state_prep}), and all $\overline{\mathrm{CNOT}}$s are via in-block qubit permutations that are compiled away by relabelling qubits. For both codes, transversal measurement of $Z$-type (resp.~$X$-type) logical operators entail measuring all qubits of the data codeblock(s) in the $Z$ (resp.~$X$) basis. The transversal measurement data also provides the last round of syndrome information.

We reported the logical failure rate per logical qubit ($\overline{\mathrm{LFR}}$) in \cref{fig:benchmark/cx}. This $\overline{\mathrm{LFR}}$ was computed as the average LFR across all logical qubits in both Pauli bases. On each circuit simulation shot starting from the $\ket*{\overline{\text{0000}}}$ (resp.~$\ket*{\overline{\text{++++}}}$) initial state, the $i^\text{th}$ logical qubit is deemed to have suffered a logical failure when decoding of its $\overline{Z}_i$ (resp.~$\overline{X}_i$) logical observable predicts an incorrect flip when compared to ground-truth (i.e.~the actual logical qubit state). On $M$ shots for each Pauli basis $\mu \in \{X, Z\}$, let $\smash{m^{(\mu)}_i}$ denote the number of logical failures for the $i^\text{th}$ logical qubit. Then 
\begin{equation}
    \overline{\mathrm{LFR}} 
    = \frac{1}{8} \sum_{i = 1}^4 \sum_{\mu \in \{X, Z\}} \frac{m^{(\mu)}_i}{M}.
\end{equation}

\subsection{Logical GHZ state preparation benchmark}
\label{app:benchmarking/ghz}

In \cref{sec:benchmarking/ghz,fig:benchmark/ghz} we investigated logical GHZ state preparation on up to $K = 64$ logical qubits. The logical circuit begins with the $\ket*{\overline{\text{+00}\cdots\text{0}}}$ initial state and unitarily prepares the GHZ state through a log-depth $\overline{\mathrm{CNOT}}$ circuit (which is, in fact, structurally identical to the encoding circuit for hypercube codes). All $\overline{\mathrm{CNOT}}$s on the surface code are transversal, while the first three are in-block and all remaining are transversal for the phantom code.

Preparation of mixed-basis logical states, in particular the $\ket*{\overline{\text{+000}}}$ state needed here, on $k > 1$ codes is generally nontrivial. Our approach is to fault-tolerantly prepare $\ket*{\overline{\text{0000}}}$ or $\ket*{\overline{\text{++++}}}$ on two codeblocks and teleport one logical qubit, as shown in \cref{fig:qRM_addressability/c}. The single $\overline{\mathrm{CNOT}}$ for teleportation requires two transversal $\overline{\mathrm{CNOT}}$s to implement (see \cref{lemma:phantom_css_interblock_cnot_circuits_two_codeblocks}), and we insert a round of Steane EC in between. The window for decoding encompasses this gadget and involves two $\overline{\mathrm{CNOT}}$s, but is only slightly larger than the usual single-$\overline{\mathrm{CNOT}}$ window (as described in \cref{app:benchmarking/phantom/decoding}). This is because there is no leading EC before the first $\overline{\mathrm{CNOT}}$, since both codeblocks are freshly prepared; also, there is no trailing EC after the second $\overline{\mathrm{CNOT}}$ on the $\ket*{\overline{\text{++++}}}$ codeblock as it is transversally measured. We additionally preselect the $\ket*{\overline{\text{+000}}}$ preparation conditioned on the last three logical qubits of the $\ket*{\overline{\text{++++}}}$ codeblock being decoded correctly to $\ket*{\overline{\text{+++}}}$; the rejection rate arising from this check is vastly subdominant at ${<}0.05\%$.

We use direct fidelity estimation, as proposed in e.g.~Ref.~\cite{flammia2011direct} and used in e.g.~Refs.~\cite{koh2025readout,baumer2023efficient}, applied at the logical level to measure the logical fidelity of the prepared GHZ state. The logical fidelity of the prepared state $\rho$ is by definition
\begin{equation}\begin{split}
    \overline{\mathcal{F}}
    = \bra{\overline{\mathrm{GHZ}}_K} \rho \ket{\overline{\mathrm{GHZ}}_K},
    \qquad
    \ket{\overline{\mathrm{GHZ}}_K} 
    = \frac{\ket{\smash{\overline{0}}^K} + \ket{\smash{\overline{1}}^K}}{\sqrt{2}},
\end{split}\end{equation}
where states and inner products between states are understood to be with respect to the logical Hilbert space. Suppose in each circuit shot, a basis $\mathcal{B}$ of logical stabilizer generators (e.g.~$\{\overline{Z}_1 \overline{Z}_2, \overline{Z}_2 \overline{Z}_3, \ldots, \overline{Z}_{K-1} \overline{Z}_K, \overline{X}_1 \overline{X}_2 \cdots \overline{X}_k\}$ of the logical GHZ state is measured on $\rho$; for noiseless $\rho$, all such stabilizers should yield outcomes $+1$, but in the presence of noise some will yield $-1$ from shot to shot. Across $M$ shots of the circuit, let the number of $-1$ logical stabilizer observations be $m$. Then the logical infidelity can be shown to be
\begin{equation}\begin{split}
    1 - \overline{\mathcal{F}}
    = \frac{m}{K M}.
\end{split}\end{equation}

Note that, unlike e.g.~Refs.~\cite{flammia2011direct,baumer2023efficient}, we measure a complete basis $\mathcal{B}$ of stabilizer generators of the state rather than sampling a subset of stabilizers. This removes the consideration of sampling uncertainty. 

To measure $\mathcal{B}$, which involves products of logical operators across codeblocks, on the surface and phantom codes, we (1) add a noiseless round of syndrome extraction (i.e.~bare-ancilla for surface and Steane for phantom code) to enable decoding and (2) perform measurements of all the logical Pauli products in $\mathcal{B}$. This structure is a subcircuit proxy---in actual applications the prepared logical GHZ state will be consumed in a larger circuit, presumably terminating in logical measurements, and QEC cycles and logical measurement data of the broader circuit will be decoded together with those of the GHZ state preparation. Algorithmic fault tolerance on surface codes does not generally work with open time-boundaries, hence necessitating the closure of our benchmark circuit with (1). We further comment that (2) is equivalent to noiselessly executing the inverse of the logical GHZ state preparation circuit followed by the unencoding circuit of the codeblocks, at which point the logical state can be read off from physical qubits. We expect a deterministic $\ket*{\overline{\text{+00}\cdots\text{0}}}$ in the absence of noise.

The closure of the decoding time-boundary with (1) in principle provides the decoder with more information than would have been available in reality on a larger circuit. To prevent this boundary effect from causing a drastic underestimation of logical infidelity on the phantom code, we exclude windows across the noiseless SE round and the preceding round when performing sliding-window decoding---i.e.~the noiseless SE round is decoded separately with the \emph{uncorrelated} list decoder. However, fault-tolerant correlated decoding on the surface code requires time-boundaries closed with transversal measurements, and we cannot separate the layers. Therefore, in actuality, the surface code is given an unfair decoding advantage relative to the phantom code in our benchmark.

\subsection{Trotterized quantum simulation circuit benchmark}

In \cref{sec:benchmarking/trotter,fig:benchmarking/trotter} we benchmarked a Trotterized quantum simulation circuit for a many-body Hamiltonian on up to $K = 64$ logical qubits. As described in the main text, for both the surface and phantom codes, we selected rotation angles $\theta = \phi = \pi/2$ in each Trotter step so that noisy circuit simulation is tractable. In an actual quantum simulation experiment, small angles would be used so that the Trotter error is controlled (or alternatively large angles tuned away from $\pi / 2$ to access classically intractable Floquet dynamics). Small-angle rotations can be performed in a resource-efficient way on both the surface and phantom codes through STAR injection---see e.g.~Ref.~\cite{akahoshi2024partially} for the surface code and \cref{app:qrm/STAR} for the phantom code.

Each Trotter block implementing $\smash{\exp(-i\theta Z^{\otimes 8})}$ comprises a forward and inverse $\overline{\mathrm{CNOT}}$ ladder on eight logical qubits, sandwiching a $\overline{Z}$ rotation. Similar to logical GHZ state preparation, these $\overline{\mathrm{CNOT}}$ ladders are performed in log-depth. On the surface code all $\overline{\mathrm{CNOT}}$s are transversal; on the phantom code only a single transversal $\overline{\mathrm{CNOT}}$ is needed, with all other $\overline{\mathrm{CNOT}}$s being in-block. The sliding-window decoding discussed in \cref{app:benchmarking/phantom/decoding} directly applies for the phantom code.

We start from the $\ket{\overline{\text{++++0000}\cdots\text{++++0000}}}$ logical state, and transversally measure a complete basis of logical stabilizers $\{\overline{X}_1, \overline{X}_2, \overline{X}_3, \overline{X}_4, \overline{Z}_5, \overline{Z}_6, \overline{Z}_7, \overline{Z}_8, \ldots, \overline{Z}_{K-3}, \overline{Z}_{K-2}, \overline{Z}_{K-1}, \overline{Z}_K\}$ at the end of the circuit. We likewise use direct fidelity estimation (see \cref{app:benchmarking/ghz}) to obtain the logical state infidelity.

\subsection{Benchmarking results at $p = 3 \times 10^{-3}$ physical error rate}
\label{app:benchmarking/current}

\begin{figure}[!ht]
    \centering
    \includegraphics[width=0.85\linewidth]{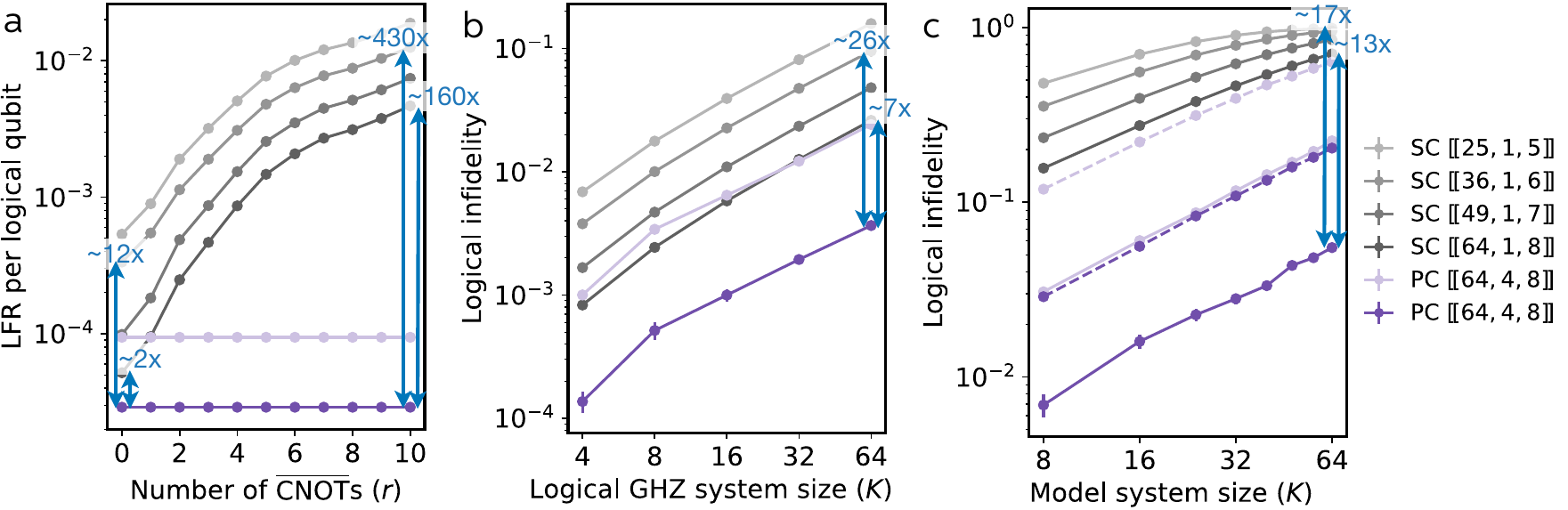}
    \phantomsubfloat{\label{fig:currenterr/repeated}}
    \phantomsubfloat{\label{fig:currenterr/ghz}}
    \phantomsubfloat{\label{fig:currenterr/trot}}
    \caption{\textbf{Additional numerical benchmarking results of phantom and surface codes at $\bm{p = 3 \times 10^{-3}}$ physical error rates.} \textbf{(a)} Repeated in-block $\overline{\mathrm{CNOT}}$ circuits on four logical qubits hosted on a single $\db{64,4,8}$ phantom codeblock or four surface codeblocks ($d = 5$--$8$). \textbf{(b)} Logical GHZ state preparation on $K = 4$--$64$ logical qubits, hosted on multiple $\db{64,4,8}$ phantom or surface codeblocks. \textbf{(c)} Trotterized many-body quantum simulation with eight Trotter steps. In (a)--(c), dark and pale purple denote strict and relaxed preselection, respectively; in (c), solid and dashed lines denote sliding-window MLE and correlated list decoding, respectively. Error bars are 98\% confidence intervals.}
    \label{fig:currenterr}
\end{figure}

We benchmark the phantom code against surface codes matched either by spatial footprint or by code distance; the matching considerations are detailed in \cref{sec:benchmarking}, and not repeated here. An important difference here is that the preselection rate for fault-tolerant preparation of codeblocks of the phantom code is $1.3\%$ (resp.~$5.1\%$) for the strict (resp.~relaxed) case.

We begin with logical state preparation alone, corresponding to the $r=0$ limit of \cref{fig:currenterr/repeated}. In this setting, the phantom code achieves a ${\sim}12\times$ lower logical failure rate than the $d=6$ surface code with comparable footprint. Even relative to the larger $d=8$ surface code—which requires roughly twice as many physical qubits—the phantom code retains a ${\sim}2\times$ advantage.

As in-block $\overline{\mathrm{CNOT}}$s are added, the contrast between the codes sharpens. For the surface code, the logical failure rate grows linearly with circuit depth, reflecting the accumulation of physical operations. In contrast, the phantom code’s logical failure rate remains essentially constant, since its $\overline{\mathrm{CNOT}}$s require no physical gates. This leads to improvements of up to ${\sim}430\times$ over the $d=6$ surface code and ${\sim}160\times$ over the $d=8$ surface code for the deepest (depth-10) circuits studied (\cref{fig:currenterr/repeated}). The logical gate time of the phantom code remains effectively zero after state preparation, while it increases linearly with circuit depth for the surface code.

We observe the same behaviour in our more structured circuits. For logical GHZ state preparation at $K=64$, the phantom code achieves a ${\sim}26\times$ reduction in logical infidelity relative to the $d=6$ surface code (\cref{fig:currenterr/ghz}). Even when compared to the larger $d=8$ surface code, the phantom code maintains a ${\sim}7\times$ infidelity advantage while using approximately $1.8\times$ fewer physical qubits. Notably, this advantage persists across all values of $K$, despite the increasing fraction of interblock $\overline{\mathrm{CNOT}}$s.

Finally, for Trotterized dynamics, spatiotemporal sliding-window MLE decoding yields a ${\sim}17\times$ (resp.~${\sim}13\times$) reduction in logical infidelity over the $d=6$ (resp.~$d=8$) surface code (\cref{fig:currenterr/trot}). The phantom code continues to outperform the $d=8$ surface code even under relaxed preselection. Although the faster sliding-window correlated list decoder trades decoding accuracy for speed, the phantom code still exceeds the $d=8$ surface code under strict preselection and the $d=7$ surface code under relaxed preselection in logical fidelity.

\end{document}